\documentclass[11pt]{article}
\usepackage[utf8]{inputenc}
\usepackage{amsmath,amsfonts,amsthm,amssymb,url, bbm}
\usepackage{xcolor}
\usepackage{fullpage}
\usepackage{thmtools, thm-restate}
\usepackage{algorithm,algpseudocode}
\usepackage[T1]{fontenc}
\usepackage[colorlinks=true, linkcolor=red, urlcolor=blue, citecolor=gray]{hyperref}
\usepackage{mathtools}
\usepackage{footnote}
\usepackage{graphicx}
\usepackage{caption}
\usepackage{subcaption}
\usepackage{tikz}
\usetikzlibrary{calc,arrows.meta,positioning}
\usetikzlibrary{decorations.pathreplacing}
\usetikzlibrary{arrows}

\usepackage[color=blue!10,linecolor=blue,textsize=footnotesize]{todonotes}

\newtheorem{theorem}{Theorem}
\newtheorem{fact}{Fact}
\newtheorem{lemma}{Lemma}
\newtheorem{corollary}{Corollary}
\newtheorem{definition}{Definition}[section]
\newtheorem{rthm}{Theorem}
\newtheorem{rcrlry}{Corollary}
\newenvironment{reptheorem}[1]
  {\renewcommand\therthm{\ref*{#1}}\rthm}
  {\endrthm}
\newenvironment{repcorollary}[1]
  {\renewcommand\thercrlry{\ref*{#1}}\rcrlry}
  {\endrcrlry}

\newcount\Comments  
\newcommand{\kibitz}[2]{\ifnum\Comments=1\textcolor{#1}{#2}\fi}

\hypersetup{linktocpage}
\definecolor{ForestGreen}{RGB}{34,139,34}

\definecolor{Orangish}{RGB}{210,105,30}
\definecolor{Cyanish}{RGB}{64,224,208}

\usepackage{color}

\title{Improved Spectral Density Estimation via Explicit and Implicit Deflation}
\author{
   Rajarshi Bhattacharjee\\UMass Amherst\\ \texttt{rbhattacharj@cs.umass.edu}
	\and 
	Rajesh Jayaram\\Google Research\\ \texttt{rkjayaram@google.com }
	\and 
   Cameron Musco\\UMass Amherst\\ \texttt{cmusco@cs.mass.edu}
	\and 
   Christopher Musco\\New York University\\ \texttt{cmusco@nyu.edu}
	\and 
	Archan Ray\\Sloan Kettering Institute\\ \texttt{talk2archan@gmail.com}
}
\date{}

  \usepackage{nth}
  \usepackage{intcalc}

\definecolor{darkgreen}{HTML}{1b7837}

\newcommand{\bv}[1]{\mathbf{#1}}

\DeclareUnicodeCharacter{2212}{-}

\DeclareMathOperator*{\argmin}{arg\,min}

\DeclareMathOperator{\range}{range}
\newcommand{\R}{\mathbb{R}}

\newcommand{\E}{\mathbb{E}}
\newcommand{\poly}{\mathop\mathrm{poly}}

\DeclareMathOperator{\rank}{rank}

\newcommand{\Vb}{\bv V}
\newcommand{\Ub}{\bv U}
\newcommand{\Qb}{\bv Q}
\newcommand{\Tb}{\bv T}

\newcommand{\vb}{\bv v}

\newcommand{\Xb}{\bv X}

\newcommand{\Ab}{\mathbf{A}}
\newcommand{\Bb}{\mathbf{B}}
\newcommand{\Cb}{\mathbf{C}}
\newcommand{\s}{\text{s}}
\newcommand{\q}{\text{q}}
\newcommand{\W}{\text{W}}
\newcommand{\zb}{\mathbf{z}}
\newcommand{\ub}{\mathbf{u}}

\newcommand{\xb}{\mathbf{x}}
\newcommand{\Pb}{\mathbf{P}}

\newcommand{\Zb}{\mathbf{Z}}
\newcommand{\Ib}{\mathbf{I}}
\newcommand{\n}{{n \times n}}

\newcommand{\tth}{\textsuperscript{th} }

\DeclareMathOperator{\tr}{tr}
\newcommand{\norm}[1]{\|#1\|}

\begin{document}

\sloppy

\begin{titlepage}
\maketitle
 \thispagestyle{empty}
\begin{abstract}
We study algorithms for approximating the spectral density (i.e., the eigenvalue distribution) of a symmetric matrix $\mathbf A \in \R^{n \times n}$ that is accessed through matrix-vector product queries. Recent work has analyzed  popular Krylov subspace methods for this problem,
showing that they output an $\epsilon \cdot \|\bv A\|_2$ error approximation to the spectral density in the Wasserstein-$1$ metric using $O(1/\epsilon)$ matrix-vector products. By combining a previously studied Chebyshev polynomial moment matching method with a \emph{deflation} step that approximately projects off the largest magnitude eigendirections of $\mathbf A$ before estimating the spectral density, we give an improved error bound of $\epsilon \cdot \sigma_{\ell}(\mathbf A)$ using $O(\ell \log n+ 1/\epsilon)$ matrix-vector products, where $\sigma_\ell(\mathbf A)$ is the $\ell^{th}$ largest singular value of $\mathbf A$. In the common case when $\mathbf A$ exhibits fast singular value decay and so $\sigma_\ell(\mathbf A) \ll \norm{\mathbf A}_2$, our bound can be much stronger than prior work. We also show that it is nearly tight: any algorithm giving error $\epsilon \cdot \sigma_\ell(\mathbf A)$ must use $\Omega(\ell+1/\epsilon)$ matrix-vector products.

We further show that the popular Stochastic Lanczos Quadrature (SLQ) method  essentially matches the above bound \emph{for any choice of parameter $\ell$}, even though SLQ itself is parameter-free and performs no explicit deflation. Our bound helps to explain the strong practical performance and observed `spectrum adaptive' nature of SLQ, and motivates a simple variant of the method that achieves an even tighter error bound.
Technically, our results require a careful analysis of how eigenvalues and eigenvectors are approximated by (block) Krylov subspace methods, which may be of independent interest. Our error bound for SLQ  leverages an analysis of the method  that views it as an implicit polynomial moment matching method, along with recent results on low-rank approximation with single-vector Krylov methods. We use these results to show that the method can perform `implicit deflation' as part of moment matching.

\end{abstract}

\end{titlepage}

\tableofcontents
\thispagestyle{empty}

\clearpage

\section{Introduction}

\emph{Spectral density estimation (SDE)} is a fundamental task in computational linear algebra.
Given a symmetric matrix $\bv A \in \R^{n \times n}$ with eigenvalues $\lambda_1(\bv A),\ldots,\lambda_n(\bv A)$, the goal is to approximate $\bv A$'s eigenvalue distribution (i.e., its spectral density) $s_{\bv A}$, which is the distribution that places probability mass $1/n$ at each of $\bv A$'s $n$ eigenvalues. Formally, letting $\delta(\cdot)$ be the Dirac delta function,
\begin{align}\label{eq:spectralDensity}
s_{\bv A}(x) = \frac{1}{n} \sum_{i=1}^n \delta(x-\lambda_i(\bv{A})).
\end{align}
Understanding $s_\bv{A}$ can provide important information about the input matrix $\bv A$. To what degree does it exhibit low-rank structure (i.e., have a decaying eigenvalues)? How close is the spectrum, or some part of the spectrum, to that of a random matrix? Does $\bv {A}$ have many repeated or nearly repeated eigenvalues that may indicate anomalies or other interesting structure? As such, spectral density estimation is applied throughout the sciences \cite{Skilling:1989,SilverRoder:1994,SbierskiTrescherBergholtz:2017,SchnackRichterSteinigeweg:2020}, network science \cite{FarkasDerenyiBarabasi:2001,EikmeierGleich:2017,DongBensonBindel:2019}, machine learning and deep learning in particular \cite{RameshLeCun:2018,PenningtonSchoenholzGanguli:2018,MahoneyMartin:2019,GhorbaniKrishnanXiao:2019}, numerical linear algebra \cite{Di-NapoliPolizziSaad:2016,LiXiErlandson:2019}, and beyond.

The spectral density $s_{\bv A}$ can be computed directly by performing a full eigendecomposition of $\bv A$, in $O(n^\omega)$ time, for $\omega \approx 2.37$ being the exponent of fast matrix multiplication \cite{parlett1998symmetric,PanChen:1999,DemmelDumitriuHoltz:2007}. However, when $\bv A$ is very large, or in settings where $\bv A$ can only be accessed through a small number of queries, we often seek an approximation $\tilde s_{\bv A}$, such that $\tilde s_{\bv A}$ and $s_{\bv A}$ are close in some metric. In this work we will focus on the Wasserstein-$1$ (i.e., earth mover's) distance, $W_1(s_{\bv A}, \tilde s_{\bv A})$, which has been studied in a number of recent works giving formal approximation guarantees for SDE \cite{Cohen-SteinerKongSohler:2018,chen2021analysis,braverman:2022,ChenTrogdonUbaru:2022,JinMuscoSidford:2023,JinKarmarkarMusco:2024}. When $\tilde s_{\bv A}$ is the uniform distribution over approximate eigenvalues $\tilde \lambda_1(\bv A),\ldots,\tilde \lambda_n(\bv A)$, and when we order both sets of eigenvalues in decreasing order,  $W_1(s_{\bv A},\tilde s_{\bv A}) = \frac{1}{n} \sum_{i=1}^n |\lambda_i(\bv A)-\tilde \lambda_i(\bv A)|$. I.e., it is the average absolute error of our eigenvalue estimates. More generally, when $\tilde s_{\bv A}$ is any distribution, $W_1(s_{\bv A},\tilde s_{\bv A})$ is the minimum cost of transforming  $s_{\bv A}$ into $\tilde s_{\bv A}$, where moving probability $\delta$ from $x$ to $y$ incurs cost $\delta \cdot |x - y|$.

\subsection{Matrix-Vector Query Algorithms for SDE}

Given its practical importance, efficient  algorithms for SDE have been widely studied \cite{BenoitRoyerPoussigue:1992,Wang:1994,WeisseWelleinAlvermann:2006,LinSaadYang:2016,chen2021analysis,braverman:2022}. A large fraction of these methods operate in the \emph{matrix-vector query model}: they only access the input matrix $\bv A \in \R^{n \times n}$ through multiplication on the left or right with a sequence of (possibly adaptively chosen) query vectors $\bv x_1,\ldots,\bv x_m \in \R^n$. The goal is to minimize the number of queries $m$, which typically dominate the runtime cost. 

Matrix-vector query algorithms encompass both linear sketching methods (when queries are chosen non-adaptively) and Krylov subspace methods (when queries are of the form $\bv x, \bv A\bv x, \bv A^2 \bv x,\ldots$, for some starting vector $\bv x$, or set of starting vectors). Beyond spectral density estimation, they are the dominant algorithms in practice for many linear algebraic problems, including eigenvalue and eigenvector computation \cite{parlett1998symmetric}, low-rank approximation \cite{HalkoMartinssonTropp:2011,Musco:2015}, linear system solving \cite{LehoucqSorensenYang:1998,Saad:2003}, and beyond. Such methods typically have low-memory overhead, since even when $\bv A$ is very large, they only need to store the outputs of the matrix-vector products. Further, they can often take advantage of highly optimized software and hardware for matrix-vector multiplication, including parallel hardware like GPUs, and faster matrix-vector multiplication routines when $\bv A$ is sparse or structured. Moreover, matrix-vector query algorithms are applicable in settings where $\bv A$ cannot be efficiently materialized, but can be efficiently multiplied by vectors. This is the case e.g., when $\bv A$ is the Hessian of a neural network \cite{Pearlmutter:1994,GhorbaniKrishnanXiao:2019} or a function of some other matrix that can be efficiently applied to vectors using e.g., an iterative method \cite{ubaru2017fast}.

Recently, matrix-vector query algorithms have received significant attention in theoretical work on numerical linear algebra since in many cases, it is possible to prove (nearly) matching query complexity upper and lower bounds for central problems like trace estimation \cite{meyer2021hutch++}, low-rank approximation \cite{SimchowitzEl-AlaouiRecht:2018,BakshiClarksonWoodruff:2022,BakshiNarayanan:2023}, linear regression \cite{BravermanHazanSimchowitz:2020}, structured matrix approximation \cite{HalikiasTownsend:2022,DharangutteMusco:2023,AmselChenHalikias:2024}, and beyond \cite{SunWoodruffYang:2021,NeedellSwartworthWoodruff:2022,SwartworthWoodruff:2023}.

Most state-of-the-art matrix-vector query algorithms for spectral density estimation are Krylov subspace methods that fall into two general classes.

\medskip

\noindent\textbf{Moment Matching.} The first class of methods  approximates $s_{\bv A}$ by approximating its polynomial moments. I.e., $\E_{s_\bv{A}} [p(x)] = \frac{1}{n} \sum_{i=1}^n p(\lambda_i(\bv A)) = \frac{1}{n}\tr(p(\bv A))$, where $p$ is a low-degree polynomial. We can employ stochastic trace estimation methods like Hutchinson's method \cite{Girard:1987,Hutchinson:1990} to approximate this trace using just a small number of  matrix-vector products with $p(\bv A)$ and in turn $\bv A$, since if $p$ has degree $k$, a single matrix-vector product with $p(\bv A)$ can be performed using $k$ matrix vector products with $\bv A$. After approximating the moments for a set of low-degree polynomials (e.g., the first $k$ monomials, or the first $k$ Chebyshev polynomials), we can let $\tilde s_{\bv A}$ be a distribution that matches these moments as closely as possible, and thus should closely match $s_{\bv A}$.

Moment matching methods include the popular Kernel Polynomial Method (KPM)  \cite{SilverRoder:1994,Wang:1994,WeisseWelleinAlvermann:2006}  and its  variants \cite{CovaciPeetersBerciu:2010,LinSaadYang:2016,braverman:2022,Chen:2023}. Several works also use moment matching to give sublinear time SDE methods for graph adjacency matrices \cite{Cohen-SteinerKongSohler:2018,braverman:2022,JinMuscoSidford:2023}, leveraging structure to estimate moments faster than with matrix-vector queries.

\medskip

\noindent\textbf{Lanczos-Based Methods.} The second class of methods computes a small number of approximate eigenvalues of $\bv A$ using the Lanczos method, and lets $\tilde s_{\bv A}$ be a distribution supported on these eigenvalues, with appropriately chosen probability mass placed at each. The canonical method of this form is Stochastic Lanczos Quadrature (SLQ) \cite{chen2021analysis}. Many other variants have also been studied. Some place probability mass not just at the approximate eigenvalues, but on Gaussian or other simple distributions centered at these eigenvalues \cite{LambinGaspard:1982,BenoitRoyerPoussigue:1992,LinSaadYang:2016,HaydockHeineKelly:1972}.

\subsection{Existing Bounds}\label{sec:existing}

While matrix-vector query algorithms for SDE have been studied for decades, theoretical guarantees on their approximation error in terms of the distance between the true spectral density $s_{\bv A}$ and the approximate density $\tilde s_{\bv A}$ have only recently been formalized.
Braverman et al. \cite{braverman:2022} analyze a  Chebyshev Moment Matching method, which can be thought of as a simple variant of KPM, showing that the method can compute $\tilde s_{\bv A}$ satisfying $W_1(s_{\bv A},\tilde s_{\bv A}) \le \epsilon \cdot \norm{\bv A}_2$ with  probability $\ge 1-\delta$ using just $O(b/\epsilon)$ matrix vector products, where $b = \max(1, \frac{1}{n \epsilon^2} \log^2 \frac{1}{\epsilon \delta} \log^2 \frac{1}{\epsilon})$. Note that $b = 1$ in the common case when $\epsilon = \tilde \Omega(1/\sqrt{n})$. Here $\norm{\bv A}_2$ denotes the spectral norm of $\bv A$ -- i.e., its  largest eigenvalue magnitude. They prove a similar guarantee for KPM itself, but with a worse dependence on $\epsilon$. Chen et al. \cite{chen2021analysis,ChenTrogdonUbaru:2022} prove that the Lanczos-based SLQ method gives essentially the same approximation bound: error $\epsilon \cdot \norm{\bv A}_2$ using $O(1/\epsilon)$ matrix-vector products when $\epsilon = \tilde \Omega(1/\sqrt{n})$.\footnote{Work on eigenvalue estimation \cite{AndoniNguyen:2013,BhattacharjeeDexterDrineas:2024,SwartworthWoodruff:2023}  can also  give guarantees for SDE. However they are generally weaker than those discussed above. E.g., \cite{SwartworthWoodruff:2023} shows how to approximate all eigenvalues of symmetric $\bv A$ to additive error $\epsilon \norm{\bv A}_F$ using $O(1/\epsilon^2)$ matrix-vector products, which is optimal. Letting $\tilde s_{\bv A}$ be the uniform distribution over their approximate eigenvalues,  we obtain the somewhat weak bound of $W_1(s_{\bv A},\tilde s_{\bv A}) \le \epsilon \norm{\bv A}_F$.} 

The above error bounds for KPM, Chebyshev Moment Matching, and SLQ help to justify the effectiveness of these methods in practice. However, in many cases, they can be loose. A bound of $W_1(s_{\bv A},\tilde s_{\bv A}) \le \epsilon \cdot \norm{\bv A}_2$ roughly corresponds to estimating each eigenvalue to average error $\epsilon \cdot \norm{\bv A}_2$. Many matrices however exhibit spectral decay: most of their eigenvalues are much smaller in magnitude than their largest (i.e., than $\norm{\bv A}_2$). Thus, this error bound does not guarantee that $\tilde s_{\bv A}$ effectively captures information about $\bv A$'s smaller magnitude eigenvalues.   

\subsection{Our Results}

Our main contribution is to show that both moment matching and Lanczos based methods for SDE can achieve improved bounds on $W_1(s_{\bv A},\tilde s_{\bv A})$ that depend on $\sigma_{l+1}(\bv A)$, the $(l+1)^{st}$ largest singular value of $\bv A$ (i.e., the $(l+1)^{st}$ largest eigenvalue magnitude) for some  parameter $l$, instead of $\norm{\bv A}_2$. For matrices that exhibit spectral decay and thus have $\sigma_{l+1}(\bv A) \ll \sigma_1(\bv A) = \norm{\bv A}_2$, our bounds can be much stronger than those given in prior work. 

\vspace{-.25em}
\subsubsection{Improved SDE via Moment Matching with Explicit Deflation}\label{sec:initialDeflation}

Our first contribution is a modification of the moment matching method of \cite{braverman:2022} that first `deflates' off any eigenvalue of $\bv A$ with magnitude significantly larger than $\sigma_{l+1}(\bv A)$, before estimating the spectral density. Specifically, the method uses a block Krylov subspace method \cite{Musco:2015,tropp:2018} to first compute highly accurate approximations to the $p$ largest magnitude eigenvalues of $\bv A$, for some $p \le l$, along with an orthonormal matrix $\bv Z \in \R^{n \times p}$ with columns approximating the corresponding eigenvectors. 
It uses moment matching to estimate the spectral density of $\bv{A}$ projected away from these approximate eigendirections $(\bv I - \bv{ZZ}^T)\bv A (\bv I-\bv{ZZ}^T)$, achieving error $\epsilon \sigma_{l+1}(\bv A)$ since this matrix has spectral norm bounded by $O(\sigma_{p+1}(\bv A)) = O(\sigma_{l+1}(\bv A))$ if $\bv Z$ is sufficiently accurate. It then modifies this approximate density to account for the probability mass at the top $p$ eigenvalues. While block Krylov methods are well understood for the closely related tasks of low-rank approximation and singular value approximation \cite{Musco:2015,Musco:2018,BakshiNarayanan:2023}, our work requires a careful analysis of eigenvalue/eigenvector approximation with these methods that may be of independent interest -- see Section \ref{sec:techOverview} for details.
Overall, the above approach gives the following result:
\begin{theorem}[SDE with Explicit Deflation]\label{thm:sde1}
  Let $\Ab \in \R^\n$ be symmetric. For any $\epsilon \in (0,1)$,  $l \in [n]$, and any constants $c_1,c_2 > 0$, Algorithm~\ref{alg:sde} performs $O\left(l\log n+\frac{b}{\epsilon} \right)$ matrix-vector products with $\Ab$ where $b=\max\left(1,\frac{1}{n\epsilon^2}\log^2 \frac{n}{\epsilon}\log^2 \frac{1}{\epsilon} \right)$ and computes a probability density function $\tilde s_{\bv A}$ such that, with probability at least $1-\frac{1}{n^{c_1}}$, \vspace{-.75em} $$\W_1(s_{\Ab},\tilde s_{\bv A}) \leq \epsilon \cdot \sigma_{l+1}(\Ab) +\frac{\|\Ab \|_2}{n^{c_2}}.$$
\end{theorem}

  As compared to the result of \cite{braverman:2022}, Theorem \ref{thm:sde1} uses $O(l \log n)$ additional matrix-vector products -- these are used to compute the approximate top eigenvalues and eigenvectors for deflation. However, the method gives a significantly improved error bound of roughly $\epsilon \sigma_{l+1}(\bv A)$. The additive error $\norm{\bv A}_2/n^{c_2}$ can be thought of as negligible -- comparable e.g., to round-off error when  directly computing $s_{\bv A}$ using a full eigendecomposition in finite precision arithmetic \cite{Banks:2022tn}.
  
  We further show that our algorithm is optimal amongst all matrix-vector query algorithms, up to logarithmic factors and the negligible additive error term. Formally:

  \begin{theorem}[SDE Lower Bound]\label{thm:lower_bound}
 Any (possibly randomized) algorithm that given symmetric $\bv A \in \R^{n \times n}$ outputs $\tilde s_{\bv A}$ such that, with probability at least $1/2$, $W_1(s_\bv{A},\tilde s_{\bv A}) \le \epsilon \sigma_{l+1}(\Ab)$ for $\epsilon \in (0,1)$ and $l \in [n]$ must make $\Omega \left(l +\frac{1}{\epsilon} \right)$ (possibly adaptively chosen) matrix-vector queries to $\Ab$.
\end{theorem}
Theorem \ref{thm:lower_bound} leverages an existing lower bound for distinguishing Wishart matrices of different ranks, previously used to give matrix-vector query lower bounds for the closely related problem of eigenvalue estimation \cite{SwartworthWoodruff:2023}.

\medskip

\noindent\textbf{Application to Schatten-1 (Nuclear) Norm Estimation.}
  Even for matrices that don't exhibit spectral decay, by balancing the $O(l \log n)$ and $O(b/\epsilon)$ terms we can leverage Theorem \ref{thm:sde1} in applications that require understanding the small magnitude eigenvalues of $\bv A$, where previous SDE bounds gave weak results. For example, consider the estimating the Schatten-1 (nuclear) norm $\norm{\bv A}_1 =\sum_{i=1}^n |\lambda_i(\bv A)| = \sum_{i=1}^n \sigma_i(\bv A)$. It is not hard to show that we can estimate $\norm{\bv A}_1$ to relative error $\epsilon \norm{\bv A}_1$ given an approximate spectral density $\tilde s_{\bv A}$ with $W_1(s_{\bv A},\tilde s_{\bv A}) \le \frac{\epsilon \norm{\bv A}_1}{n}$ (see~\cite{Cohen-SteinerKongSohler:2018}, Theorem B.1 in~\cite{braverman:2022}). We can find such a density by applying Theorem \ref{thm:sde1} with $l = \frac{\sqrt{n}}{\epsilon}$ and $\epsilon' = \frac{1}{\sqrt{n}}$ so that $\epsilon' \cdot  \sigma_{l+1}(\bv A) \le \frac{\epsilon' \norm{\bv A}_1}{l} \le \frac{\epsilon \norm{\bv A}_1}{n}$. Doing so yields the following corollary:

\begin{corollary}[Schatten-$1$ Norm Estimation]\label{cor:schatten}
Let $\Ab \in \R^\n$ be symmetric. For any $\epsilon \in (0,1)$ and any constant $c > 0$, there exists an algorithm that performs $O\left(\frac{\sqrt{n} \log n}{\epsilon} + \sqrt{n} \log^4 n \right)$ matrix vector products with $\Ab$ and computes $M$ such that, with probability at least $1-\frac{1}{n^c}$, $|M-\norm{\bv A}_1| \le \epsilon \norm{\bv A}_1$.   
\end{corollary}
Prior work on SDE \cite{chen2021analysis,braverman:2022} could only give error $\epsilon \sqrt{n} \norm{\bv A}_2$ using a comparable number of matrix-vector products, and thus was not able to achieve a relative error guarantee. We note that the $\sqrt{n}$ dependence of Corollary \ref{cor:schatten} matches the best known matrix-vector product algorithms for Schatten-1 norm estimation \cite{Musco:2018}, while the $\epsilon$ dependence improves on prior work. 

\subsubsection{Implicit Deflation Bounds for Stochastic Lanczos Quadrature}

Our second contribution is to show that the popular Stochastic Lanczos Quadrature (SLQ) method for SDE \cite{LinSaadYang:2016,chen2021analysis} nearly matches the improved error bound of Theorem \ref{thm:sde1} \emph{for any choice of $l$}, even though SLQ is  `parameter-free' and performs no explicit deflation step. This result helps to justify the strong practical performance of SLQ and the observed `spectrum adaptive' nature of this method as compared to standard moment matching-based methods like KPM \cite{chen2021analysis}.

A key idea  used to achieve this bound is to view SLQ as an implicit moment matching method as in \cite{chen2021analysis,ChenTrogdonUbaru:2022}, and to analyze it similarly to KPM and other explicit moment matching methods. We combine this analysis approach with recent work on low-rank approximation with single-vector (i.e., non-block) Krylov methods \cite{meyer:2024} to show that SLQ can perform `implicit deflation' as part of moment matching to achieve the improved error bound. See Section \ref{sec:techOverview} for details.
Formally, our error bound for SLQ is as follows: 
\begin{theorem}[SDE with SLQ]\label{thm:slq}
Let $\Ab\in\R^\n$ be symmetric and consider any $l \in [n]$, and $\epsilon, \delta \in (0,1)$. Let $g_{\min}=\min_{i \in [l]} \frac{\sigma_i(\Ab)-\sigma_{i+1}(\Ab)}{\sigma_i(\Ab)}$ and $\kappa=\frac{\|\Ab \|_2}{\sigma_{l+1}(\bv A)}$. Algorithm~\ref{alg:slq} (SLQ) run for $m = O(l\log \frac{1}{g_{\min}}+\frac{1}{\epsilon}\log \frac{n \cdot \kappa}{\delta})$ iterations performs $m$ matrix vector products with $\bv A$ and outputs a probability density function $\tilde s_{\bv A}$ such that, with probability at least $1-\delta$, for a fixed constant $C$,
$$W_1(s_{\Ab},\tilde s_{\bv A}) \leq \epsilon \cdot \sigma_{l+1}(\bv A) + \frac{C\log (n/\epsilon)\log (1/\epsilon) }{\sqrt{n}} \cdot \sigma_{l+1}(\Ab) + \frac{Cl\log (1/\epsilon)\sqrt{\log (l/\delta)}}{n}\|\Ab \|_2.$$
\end{theorem}
Theorem \ref{thm:slq} essentially matches our result for moment matching with explicit deflation (Theorem \ref{thm:sde1}) up to some small caveats, discussed below. First, the number of matrix vector products has a logarithmic dependence on the minimum  gap $g_{\min}$ amongst the top $l$ singular values as well as the condition number $\kappa=\frac{\| \Ab\|_2}{\sigma_{l+1}(\bv A)}$. The dependence on the minimum gap is inherent, as non-block Krylov methods like SLQ require a dependence on $g_{\min}$ in order to perform deflation/low-rank approximation \cite{meyer:2024}. We note that by adding a random perturbation to $\bv A$ with spectral norm bounded by $\frac{\norm{\bv A}_2}{\poly(n)}$, one can ensure that both $g_{\min} \ge \frac{1}{\poly(n)}$ and $\kappa \leq \poly(n)$ with high probability, and thus replace the $O(l \log \frac{1}{g_{\min}})$ term  with an $O(l \log n)$ and the $O(\frac{\log (n \kappa)}{\epsilon})$ term with an $O(\frac{\log n}{\epsilon})$ term, matching Theorem \ref{thm:sde1}. See e.g., \cite{meyer:2024}.

Second, Theorem \ref{thm:slq} has an additional error term of size $\tilde O(\sigma_{l+1}(\bv A)/\sqrt{n})$. This term is lower order whenever $\epsilon = \tilde \Omega(1/\sqrt{n})$. Further, we believe that it can be removed by using a variant on SLQ that is popular in practice, where the densities output by multiple independent runs of the method are averaged together to produce $\tilde s(\bv A)$. See Section \ref{sec:slq} for further discussion.

Finally,  Theorem \ref{thm:slq} has an additional error term of size $\tilde O(\norm{\bv A}_2 \cdot l/n)$. In the natural case when we run for $m \ll n$ iterations and thus $l \ll n$, this term will be small. However, it cannot be avoided: even for a matrix with rank $\le l$ with well-separated eigenvalues, while the Lanczos method will converge to near-exact approximations to these eigenvalues (with error bounded by $\frac{\norm{\bv A}_2}{n^c}$), the probability distribution output by SLQ will not place mass exactly $1/n$ at these approximate eigenvalues and thus will not achieve SDE error $O(\frac{\norm{\bv A}_2}{n^c})$ -- see Section \ref{sec:techOverview} for further details.

This limitation motivates us to introduce a simple variant of SLQ, which we call \emph{variance reduced SLQ}, which places mass exactly $1/n$ at any eigenvalue computed by  Lanczos that has converged to sufficiently small error. This variant gives the following stronger error bound:
\begin{theorem}[SDE with Variance Reduced SLQ]\label{thm:varslq}
Let $\Ab\in\R^\n$ be symmetric and consider any $l \in [n]$, and $\epsilon,\delta \in (0,1)$. Let $g_{\min}=\min_{i \in [l]} \frac{\sigma_i(\Ab)-\sigma_{i+1}(\Ab)}{\sigma_i(\Ab)}$ and $\kappa=\frac{\|\Ab \|_2}{\sigma_{l+1}(\bv A)}$. Algorithm~\ref{alg:slq+} run for $m = O(l\log \frac{1}{g_{\min}}+\frac{1}{\epsilon}\log \frac{n \cdot \kappa}{\delta})$ iterations performs $m$ matrix vector products with $\bv A$ and outputs a probability density function $\tilde s_{\bv A}$ such that, with probability at least $1-\delta$, for any constant $c$ and a fixed constant $C$,
$$W_1(s_{\Ab},\tilde s(\bv A)) \leq \epsilon \cdot \sigma_{l+1}(\bv A) + \left (\frac{C \log(n/\epsilon) \log(1/\epsilon)}{\sqrt{n}} + \frac{C l \log(1/\epsilon) \sqrt{\log (l/\delta)}}{n} \right) \cdot \sigma_{l+1}(\bv A) + \frac{\norm{\bv A}_2}{n^{c}}.$$
\end{theorem}

\subsection{Technical Overview}\label{sec:techOverview}

We next overview the main techniques used to achieve our improved SDE bounds for moment matching with deflation  (Theorem \ref{thm:sde1}) and  SLQ (Theorems \ref{thm:slq} and \ref{thm:varslq}).

\subsubsection{Eigenvalue Deflation for SDE}

As discussed in Section \ref{sec:initialDeflation}, the key idea behind Theorem \ref{thm:sde1} is to apply \emph{eigenvalue deflation}. Assume that we are given $\bv Z \in \R^{n \times l}$ with columns equal to the eigenvectors of $\bv A$ corresponding to its $l$ largest magnitude eigenvectors. Then we can write $\bv A = \bv A_l + \bv A_{l,\perp}$, where $\bv A_l = \bv{ZZ}^T \bv A \bv {ZZ}^T$ and $\bv A_{l,\perp} =  (\bv I- \bv{ZZ}^T) \bv A (\bv I - \bv {ZZ}^T)$. Referring to \eqref{eq:spectralDensity}, since $\bv A_l$ an $\bv A_{l,\perp}$ have $n$ eigenvalues in total that are set to $0$ due to projecting off $\bv Z$, we have:
\begin{align}\label{eq:breakdown}
s_{\bv A} = s_{\bv A_l} + s_{\bv A_{l,\perp}} - \delta(0),
\end{align}
where $\delta(x)$ denotes the Dirac distribution that places probability mass one at $x$. If we have $\bv Z$ in hand, we can exactly compute $s_{\bv A_l}$, whose only non-zero eigenvalues are exactly the eigenvalues corresponding to the eigenvectors in $\bv Z$. Further, we can approximate $s_{\bv A_{l,\perp}}$ using an existing SDE algorithm -- e.g., the Chebyshev moment matching method of \cite{braverman:2022}, which will give error $\epsilon \cdot \norm{\bv A_{l,\perp}}_2$ using $O(1/\epsilon)$ matrix vector products when $\epsilon = \tilde \Omega(1/\sqrt{n})$.\footnote{We will focus on the case $\epsilon = \tilde \Omega(1/\sqrt{n})$  and so $b = 1$ in Thm. \ref{thm:sde1} throughout the technical overview for simplicity.} We have $\norm{\bv A_{l,\perp}}_2 = |\lambda_{l+1}(\bv A)| = \sigma_{l+1}(\bv A)$, where $\lambda_{l+1}(\bv A)$ is the $(l+1)^{st}$ largest magnitude eigenvalue of $\bv A$. Thus, combining this approximation with \eqref{eq:breakdown}, we can approximate $s_{\bv A}$ to error $\epsilon \cdot \sigma_{l+1}(\bv A)$ as desired.

We note that eigenvalue deflation is widely applied throughout numerical linear algebra to problems like linear system solving \cite{BurrageErhelPohl:1998,FrankVuik:2001,GonenOrabonaShalev-Shwartz:2016,FrangellaTroppUdell:2023}, trace estimation \cite{GambhirStathopoulosOrginos:2017,Lin:2017,meyer2021hutch++}, norm estimation \cite{Musco:2018}, and beyond \cite{ChapmanSaad:1997} -- the approach is useful whenever the complexity or approximation error of solving a problem depends on some feature of its eigenspectrum (e.g. the spectral norm, Frobenius norm, or condition number) that can be improved by removing the largest magnitude eigenvalues from the matrix.

\subsubsection{Error Analysis of Deflation}\label{sec:deflationIntro}

Of course, when implementing deflation for SDE, we cannot exactly compute $\bv Z \in \R^{n \times l}$ spanning the top $l$ eigenvectors of $\bv A$. Instead we will approximate $\bv Z$, in our case, using a standard block Krylov subspace method (Algorithm \ref{alg:krylov}). It is well known that when $\bv Z$ is approximated in this way, using $O(\log n)$ iterations of the block Krylov method, and thus $O(l \log n)$ matrix-vector products in total, we can still ensure that $\norm{(\bv I- \bv{ZZ}^T) \bv A (\bv I - \bv {ZZ}^T)}_2 = O(\sigma_{l+1}(\bv A))$ \cite{Musco:2015}. This suffices to obtain the error bound of $\epsilon \cdot \sigma_{l+1}(\bv A)$ in approximating the spectral density of $(\bv I- \bv{ZZ}^T) \bv A (\bv I - \bv {ZZ}^T)$. 

The key challenge however is that when $\bv Z$ does not exactly span a set of eigenvectors, \eqref{eq:breakdown} no longer holds. Further, the spectral density of $\bv{ZZ}^T \bv A \bv{ZZ}^T$ will no longer exactly match the spectral density of $\bv{A}_l$ -- it will place mass at the eigenvalues of $\bv{ZZ}^T \bv A \bv{ZZ}^T$, equal to the eigenvalues of $\bv Z^T \bv A \bv Z$, which approximate, but don't exactly match, the top eigenvalues of $\bv A$.

Thus, our proof requires showing that $\bv Z$ is very close to spanning a subspace of eigenvectors, up to the small $\frac{\norm{\bv A}_2}{n^c}$ additive error of Theorem \ref{thm:sde1}. 
To do so, in our block Krylov algorithm (Algorithm \ref{alg:krylov} line 7) we let $\bv Z$ contain only the approximate eigenvectors that have converged to  very small \emph{residual error} -- i.e., for which $\bv A  \bv{\tilde z}_j \approx \lambda_j  \bv{\tilde z}_j$. Via standard backward error analysis bounds for eigenvector approximation (see Theorem \ref{thm:bkwd_ptr}),  this is enough to handle  the above two issues. However, we can now no longer be sure that $\norm{(\bv I- \bv{ZZ}^T) \bv A (\bv I - \bv {ZZ}^T)}_2$ is  small -- what if, e.g., no eigenvectors have converged and thus $\bv Z$ is empty? Or e.g., if just the top eigenvector has failed to converge.

Our main technical contribution is to argue that that for some $p \le l$ with $\sigma_{p+1}(\bv A) = O(\sigma_{l+1}(\bv{A}))$, the set of converged eigenvectors (i.e., the columns of $\bv Z$) will contain approximations to at least the top $p$ magnitude eigenvectors of $\bv A$, ensuring that $\norm{(\bv I- \bv{ZZ}^T) \bv A (\bv I - \bv {ZZ}^T)}_2 = O(\sigma_{p+1}(\bv A)) =  O(\sigma_{l+1}(\bv A))$ as needed.
Theorems \ref{thm:converge}  and \ref{thm:eig_error} give our convergence bounds for the top $p$ eigenvectors, and Theorem \ref{thm:sigma_S} states the ultimate resulting bound on $\norm{(\bv I- \bv{ZZ}^T) \bv A (\bv I - \bv {ZZ}^T)}_2$. 

Our work fits into a line of work that focuses on more refined eigenvalue/singular value approximation bounds for block Krylov subspace methods -- see e.g. \cite{Musco:2015,Musco:2018,drineas2018structural,Tropp:2022}. We believe that Theorems  \ref{thm:converge}  and \ref{thm:eig_error} may be of independent interest, outside our application to SDE.

\subsubsection{SLQ and its Existing Analysis}

We next turn our attention to the Stochastic Lanczos Quadrature (SLQ) method, which is detailed in Algorithm \ref{alg:slq}. This is a popular and extremely simple, parameter-free algorithm for SDE based on the classic Lanczos method (Algorithm \ref{alg:lanczos-slq}). SLQ uses Lanczos to compute an orthonormal basis $\bv Q$ for the Krylov subspace $\{\bv g, \bv A \bv g, \bv A^2 \bv g,\ldots, \bv A^m \bv g\}$ for random starting vector $\bv g \in \R^n$ -- typically, $\bv g$ is Gaussian. The algorithm then computes the eigenvalues of $\bv T = \bv Q^T \bv A \bv Q$ and lets $s_{\bv{\tilde A}} = \sum_{i=1}^m w_j \cdot  \delta(x - \lambda_j(\bv T))$ where $w_j$ are appropriately chosen weights: $w_j = (\bv g^T \bv Q \bv v_j)^2$ where $\bv v_j$ is the $j^{th}$ eigenvector of $\bv T$. I.e., the spectral density of $\bv A$ is approximated with a reweighted spectral density of $\bv T$.
Lanczos constructs $\bv Q$ such that $\bv T$ is tridiagonal -- this makes computing its eigenvalues very efficient, but is otherwise not critical in the analysis.

The key idea behind SLQ (and Lanczos-based algorithms in general) is that, for any polynomial with degree $< m$, one can show that $p(\bv A) \bv g = \bv Q p(\bv T) \bv {Q}^T \bv g$. Using this fact, if we consider any low-degree polynomial moment for the SLQ approximate spectral density, we have:
\begin{align}
\E_{\tilde s_{\bv A}} [p(x)] = \sum_{j=1}^m w_j \cdot p(\lambda_j(\bv T)) &= \sum_{j=1}^m \bv{g}^T \bv Q \bv {v}_j \bv v_j^T \bv Q^T \bv g \cdot p(\lambda_j(\bv T))\nonumber\\ 
&= \bv g^T \bv Q \left (\sum_{j=1}^m \bv v_j \bv v_j^T p(\lambda_j(\bv T) )\right ) \bv Q^T \bv g\nonumber \\ &= \bv g^ T \Qb p(\bv T) \Qb^T\bv g = \bv g^T p(\bv A) \bv g.\label{eq:slqIntro}
\end{align} 
We can observe that when $\bv g$ has i.i.d. Gaussian entries with variance $\frac{1}{n}$, $\E_{\bv g}[\bv g^T p(\bv A) \bv g] = \frac{1}{n} \tr(p(\bv A)) = \frac{1}{n}\sum_{i=1}^n p(\lambda_i) =  \E_{s_\bv A}[p(x)]$. That is, the approximate spectral density output by SLQ matches all low-degree polynomial moments of the true spectral density in expectation.

To formally argue that $\tilde s_{\bv A}$ is close to $s_{\bv A}$ in Wasserstein distance, Chen, Trogden, and Ubaru \cite{chen2021analysis} argue that any two distributions supported on $[-\norm{\bv A}_2, \norm{\bv A}_2]$ that exactly match on any polynomial moment of degree $O(1/\epsilon)$ have Wasserstein distance at most $\epsilon \cdot \norm{\bv A}_2$. This allows them to argue that $W_1(\bar s_{\bv A}, \tilde s_{\bv A}) \le \epsilon \cdot \norm{\bv A}_2$ where $\bar s_{\bv A} = \sum_{i=1}^n \bv \eta_i \cdot \lambda_i(\bv A)$ with $\eta_i  = (\bv g^T \bv u_i)^2$, for $\bv u_i$ being the $i^{th}$ eigenvector of $\bv A$. Note that $\bar s_{\bv A}$ is constructed so that for any $p$, $\E_{\bar s_{\bv A}}[p(x)] = \bv g^T p(\bv A) \bv g$ and thus the low-degree moments of $\tilde s_{\bv A}$ and $\bar s_{\bv A}$ match by \eqref{eq:slqIntro}. \cite{chen2021analysis} next uses a concentration bound to argue that the CDFs of $\bar s_{\bv A}$ and $s_{\bv A}$ are close, and therefore that these two distributions are close in Wasserstein distance. They conclude via triangle inequality that $\tilde s_{\bv A}$ is close to $s_{\bv A}$.

\subsubsection{Moment Matching-Based Analysis of SLQ}

Our analysis follows a similar approach to  \cite{chen2021analysis}, also viewing SLQ as a moment matching method, but analyzing it in essentially an identical manner to the Chebyshev moment matching method of \cite{braverman:2022}. 
We consider the first $m = O(1/\epsilon)$ \emph{Chebyshev polynomials of the first kind}: $T_0,T_1,\ldots,T_m$. These polynomials are sine-like functions that are bounded in magnitude by $1$ on $[-1,1]$ (in the SDE setting, we rescale them to be bounded on $[-\norm{\bv A}_2,\norm{\bv A}_2]$). \cite{braverman:2022} show that any two distributions that are close in their Chebyshev moments are also close in Wasserstein distance. They then approximate the moments of $s_{\bv A}$ and find a distribution $\tilde s_{\bv A}$ matching these moments (and thus satisfying $W_1(s_{\bv A},\tilde s_{\bv A}) \le \epsilon \norm{\bv A}_2$) by solving a constrained optimization problem.

 Like many moment maching methods, \cite{braverman:2022} use Hutchinson's trace estimator \cite{Girard:1987,Hutchinson:1990} to approximate the Chebyshev moments -- i.e., they approximate $\frac{1}{n} \tr(T_i(\bv A)) \approx \frac{1}{b} \sum_{j=1}^b \bv g_j^T T_i(\bv A) \bv g_j$ where each $\bv g_j$ has i.i.d. Gaussian entries with variance $\frac{1}{n}$. They argue that since $T_i$ is bounded, Hutchinson's method is highly accurate, and thus for $\epsilon = \tilde \Omega(1/\sqrt{n})$, they can in fact just set $b = 1$. So their approximate moments are just of the form $\bv g^T T_i(\bv A)\bv g$ for a single random $\bv g$.
 
 Recall that by  \eqref{eq:slqIntro} these are exactly the moments of $\tilde s_{\bv A}$ output by SLQ! So, SLQ `automatically' finds a distribution matching the approximate Chebyshev moments of \cite{braverman:2022}, and thus satisfies $W_1(\tilde s_{\bv A},s_{\bv A}) \le \epsilon \cdot \norm{\bv A}_2$.\footnote{Even when $\epsilon = \tilde o(1/\sqrt{n})$ and so we need  $b > 1$ repetitions of Hutchinson's to estimate the Chebyshev moments in \cite{braverman:2022}, a similar analysis follows from averaging together the densities output by $b$ independent runs of SLQ.} 
 We prove this bound in Section \ref{sec:slq1}, Theorem \ref{thm:slq_main}.

\subsubsection{Implicit Deflation with SLQ}

Armed with the above moment matching-based analysis of SLQ, we next show that the method performs implicit deflation and thus nearly matches the bound of Theorem \ref{thm:sde1}. The key idea is as follows: the Chebyshev moments used in the analysis above are the polynomials $T_1(x/\norm{\bv A}_2), T_2(x/\norm{\bv A}_2),\ldots$ -- i.e., the Chebyshev polynomials scaled to be bounded on the range $[-\norm{\bv A}_2,\norm{\bv A}_2]$, which contains the eigenvalues of both $\bv A$ and $\bv T$. This scaling is what leads to the error term scaling with $\norm{\bv A}_2$.

In our explicit deflation approach, after deflating off the largest magnitude eigenvalues we apply moment matching to  $(\bv I-\bv{ZZ}^T) \bv A(\bv I-\bv{ZZ}^T) $. We thus approximate the Chebyshev moments of this matrix where the polynomials are scaled to the (possibly much smaller range) $[- \norm{(\bv I-\bv{ZZ}^T) \bv A(\bv I-\bv{ZZ}^T)}_2, \norm{(\bv I-\bv{ZZ}^T) \bv A(\bv I-\bv{ZZ}^T)}_2]$ of width $O(\sigma_{l+1}(\bv A))$. This leads to error scaling with $\sigma_{l+1}(\bv A)$. We cannot use these moments to approximate $s_{\bv A}$  directly: these scaled Chebyshev polynomials blow up outside the range on which they are bounded and the moments would be dominated by eigenvalues with magnitudes $\ge \sigma_{l+1}(\bv A)$ and so not informative for approximating $s_{\bv A}$.

\medskip

\noindent\textbf{Deflated Polynomial Moments.} 
To handle this issue, we use that the density output by SLQ approximates $s_{\bv A}$ on \emph{any low-degree polynomial moment}. Instead of the scaled Chebyshev polynomials, we use a set of ``deflated polynomials'', denoted $t_1,\ldots,t_m$ which are approximately equal to the scaled Chebyshev polynomials on the range $[- O(\sigma_{l+1}(\bv A)), O(\sigma_{l+1}(\bv A))]$ and have roots placed outside this range so they are equal to zero for any eigenvalue with magnitude significantly larger than $\sigma_{l+1}(\bv A)$. Constructing such polynomials is somewhat delicate --  naively placing the roots at large eigenvalues would distort the Chebyshev polynomials on the range of interest. This distortion needs to be canceled using a separate Chebyshev damping polynomial. Our analysis follows ideas from work studying the convergence of single-vector Krylov subspace methods like Lanczos for eigenvector/singular vector approximation \cite{Saad:1980}, and in particular, recent work proving low-rank approximation guarantees for these methods \cite{meyer:2024}.

Once we show that  $t_1,\ldots,t_m$ can be constructed as described, we know that SLQ matches the moments of $s_{\bv A}$ with respect to these polynomials. This is enough to argue that the output $\tilde s_{\bv A}$ must closely match the mass of $s_{\bv A}$ on the eigenvalues with magnitude $\le \sigma_{l+1}(\bv A)$. But it says nothing about the large eigenvalues, which contribute nothing to these moments.  We need to separately argue about this part of the density.

\medskip

\noindent{\textbf{Top Eigenvalue Approximation.} To do so, we directly look at the form of $\tilde s_{\bv A}$. Following a similar strategy to our proof for explicit deflation and again building on recent work studying the convergence of single-vector Krylov methods like Lanczos  \cite{meyer:2024}, we argue that  any eigenvalue $\lambda_j(\bv T)$ with magnitude significantly larger than $\sigma_{l+1}(\bv A)$ is extremely close to $\lambda_j(\bv A)$. This finding is not surprising: the Lanczos method is a popular choice for approximating outlying eigenvalues \cite{parlett1998symmetric}. Thus, on the large magnitude eigenvalues, the error that $\tilde s_{\bv A}$ incurs  vs. $s_{\bv A}$ is roughly: $$w_j \lambda_j(\bv T) - \frac{1}{n} \lambda_j(\bv A) \approx \left (w_j - \frac{1}{n} \right ) \lambda_j(\bv A) \le \left (w_j - \frac{1}{n} \right )  \norm{\bv A}_2.$$
Recall that $w_j = (\bv g^T \bv Q \bv v_j)^2$. Further, since our approximations to the top eigenvectors have converged, $\bv Q \bv v_j \approx 
\bv u_j$, where $\bv u_j$ is the eigenvector of $\bv A$ corresponding to $\lambda_j(\bv A)$. Since $\bv u_j$ is a fixed unit vector and $\bv g$ has random Gaussian entries with variance $\frac{1}{n}$, $w_j$ is simply the square of a Gaussian with variance $\frac{1}{n}$. I.e., $\E[w_j] = 1/n$ and we can use a Chi-Squared concentration bound to argue that $|w_j-\frac{1}{n}| = \tilde O(\frac{1}{n})$ with high probability. Accounting for the $l$ top eigenvalues, this leads to the $\tilde O(l/n \cdot \norm{\bv A}_2)$ error term in Theorem \ref{thm:slq} and completes our analysis.

\subsubsection{Variance Reduced SLQ}

 The above argument immediately suggests a simple improvement to SLQ that can avoid the $\tilde O(l/n \cdot \norm{\bv A}_2)$ error term due to the large magnitude eigenvalues. If an eigenvector has converged and thus we know that $|\lambda_j(\bv T) - \lambda_j(\bv A)|$ is small, then in the weighted spectral density $\tilde s_{\bv A}$ output by Lanczos, we know that we should set $w_j = \frac{1}{n}$ rather than setting $w_j$ to be the square of a Gaussian with variance $\frac{1}{n}$. We formalize this approach in Algorithm \ref{alg:slq+} and show that it obtains the stronger error bound of Theorem \ref{thm:varslq}, with only a negligible additive term depending on $\|\Ab \|_2$. 

\medskip

\noindent\textbf{A Remark on Numerical Stability.} We note that there exist results proving that the Conjugate Gradient (CG) algorithm for solving linear systems, which is closely related to Lanczos,  can perform implicit eigenvalue deflation to achieve faster convergence rates, analogous to our analysis for SDE \cite{SpielmanWoo:2009,AxelssonLindskog:1986}. These bounds are known to suffer from numerical stability issues -- when CG is implemented in finite precision arithmetic, it may not be able to apply the necessary polynomials to achieve these convergence rates. It is likely that our bounds also suffer from stability issues. However 1) in finite precision, the method likely works well as long as the  number of deflated eigenvalues $l$ is relatively small and 2) our analysis directly applies to SLQ implemented using a stable algorithm to compute a basis for the Krylov subspace -- e.g., one that performs full reorthogonalization at each step. Using such an algorithm will not affect the matrix-vector query complexity of the algorithm, and only gives a small runtime overhead when the number of iterations is relatively small.

\subsection{Roadmap}

The remainder of the paper is organized as follows: In Section \ref{sec:prelim} we define basic notation and preliminary results used throughout our proofs. In Section \ref{sec:explicit} we analyze our explicit moment matching method with deflation, culminating in the proof of Theorem \ref{thm:sde1}. In Section \ref{sec:slq} we give our analysis of SLQ and prove Theorems \ref{thm:slq} and \ref{thm:varslq}. In Section \ref{sec:lowerBound} we prove our matching lower bound, Theorem \ref{thm:lower_bound}, which shows that our SDE bounds are near-optimal. Finally, in Section \ref{sec:exp}, we report the results of some numerical experiments comparing various moment matching and Lanczos-based algorithms for SDE.

\section{Notation and Preliminaries}\label{sec:prelim}

We first introduce notation and foundational results that we will use throughout.
 
 \subsection{Basic Notation} 
For any integer $n$, let $[n]$ denote the set $\{1,2,\ldots,n\}$. We write matrices and vectors in bold literals -- e.g., $\bv A$ or $\bv{x}$. For a vector $\bv{x}$, we let $\norm{\bv{x}}_2$ denote its Euclidean norm. 
 The Frobenius norm and spectral (i.e., operator) norm of a matrix $\Ab$ are denoted by $\|\Ab \|_F$ and $\|\Ab \|_2$ respectively. We often use that for any two matrices $\Ab, \Bb$ of appropriate dimensions, the spectral norm is sub-multiplicative, i.e., $\|\Ab\Bb\|_2 \leq \|\Ab\|_2\|\Bb\|_2$. The column span of a matrix $\Ab$ is denoted by $\range(\Ab)$.
 
 \subsection{Linear Algebra Preliminaries}\label{sec:linAlg}
 
\noindent\textbf{SVD and Eigendecomposition.}
The singular values of $\bv{A} \in \R^{n \times n}$ are denoted by $\sigma_i(\bv{A})$ for $i \in [n]$ with $\sigma_1(\Ab) \geq \ldots \ge \sigma_n(\Ab)$. We denote the singular value decomposition (SVD) of $\bv{A} \in \R^{n \times n}$ by $\Ub \bv{\Sigma} \Vb^T$ where $\bv{\Sigma} \in \R^{n \times n}$ is a diagonal matrix with $\bv{\Sigma}_{ii}=\sigma_i(\Ab)$ and $\Ub \in \R^{n \times n}$ and $\Vb \in \R^{n \times n}$ are orthonormal matrices with columns containing $\bv A$'s left and right singular vectors respectively. 

Our main theorems are concerned with the eigenspectrum of symmetric matrices. We denote the eigenvalues of a symmetric matrix $\bv{A} \in \R^{n \times n}$ by $\lambda_i(\bv A)$ for $i \in [n]$ such that  $|\lambda_1(\bv A)| \ge \ldots \ge |\lambda_n(\bv A)|$. Observe that we have $|\lambda_i(\Ab)| = \sigma_i(\Ab)$ for $i \in [n]$. We denote the eigendecomposition of a symmetric  $\bv A \in \R^{n \times n}$ by $\Ub \bv{\Lambda} \Ub^T$ where $\bv{\Lambda} \in \R^{n \times n}$ is a diagonal matrix such that $\bv{\Lambda}_{ii}=\lambda_i(\Ab)$ and $\Ub \in \R^{n \times n}$ has orthonormal columns equal to the corresponding eigenvectors of $\Ab$. 
 
 For  $p \in [n]$, let $\Ub_p \in \R^{n \times p}$ (or $\Vb_p \in \R^{n \times p}$) denote the matrix containing the first $p$ columns of $\Ub$ (or $\Vb$) -- i.e., the singular vectors or the eigenvectors corresponding to the largest $p$ singular values or eigenvalues of $\Ab$ by magnitude. Similarly, let $\Ub_{p,\perp } \in \R^{n \times (n-p)}$ (or $\Vb_{p,\perp} \in \R^{n \times (n-p)}$) denote the matrix which containing the last $n-p$ columns of $\Ub$ (or $\Vb$) -- i.e., the singular vectors or the eigenvectors corresponding to the smallest $n-p$ singular values or eigenvalues of $\Ab$ by magnitude.  Let $\bv{\Sigma}_p \in \R^{p \times p}$ (and $\bv{\Sigma}_{p,\perp} \in \R^{(n-p) \times (n-p)}$) denote the matrix containing the top $p$ (or the bottom $n-p$) singular values of $\Ab$ along its diagonal. For symmetric $\Ab \in \R^{n \times n}$ with eigendecomposition $\bv U \bv \Lambda \bv U^T$, define $\bv \Lambda_p \in \R^{p \times p}$ and $\bv \Lambda_{p,\perp} \in \R^{(n-p) \times (n-p)}$ analogously. 
 
 The best $p$-rank approximation to $\Ab$ in the spectral or Frobenius norms is denoted by $\Ab_p$ and is given by $\Ab_p=\Ub_p \bv{\Sigma}_p \Vb_p^T$. The pseudoinverse of $\bv A$ is denoted by $\Ab^{\dag}$ and given by $\Ab^\dag = \bv V \bv \Sigma^\dag \bv U$, where $\bv \Sigma^\dag_{ii} = 1/\sigma_i(\bv A)$ if $\sigma_i(\bv A) > 0$ and $\bv \Sigma^\dag_{ii} = 0$ otherwise. 

 We will frequently use Weyls's inequality which states that a small perturbation of a symmetric matrix will not significantly change its eigenvalues.

\begin{fact}[Weyls' inequality~\cite{weyl1912asymptotic}]\label{fact:weyl}
For any two symmetric matrices $\Ab \in \R^{\n}$ and $\Bb^{\n}$,
\begin{align*}
    \max_{i \in [n]} |\lambda_i(\Ab) -\lambda_i(\Bb)| \leq \| \Ab-\Bb \|_2.
\end{align*}
\end{fact}

 \medskip

\noindent\textbf{Matrix functions.} Given a function $\phi: \R \rightarrow \R$, we define its application to a matrix $\bv A$ with SVD $\bv A = \bv U \bv \Sigma \bv V^T$ as $\phi(\Ab)=\Ub\phi(\bv{\Sigma})\Vb^T$, where $\phi(\bv{\Sigma})$ is formed by applying $\phi$ element-wise to the diagonal entries  of $\bv{\Sigma}$ (i.e., to the singular values of $\Ab$). We overload notation and, for symmetric $\bv A$ with eigendecomposition $\Ab = \bv U \bv \Lambda \bv U^T$, also let $\phi(\Ab)=\Ub \phi(\bv{\Lambda})\Ub^T$, where $\phi(\bv{\Lambda})$ is formed by applying $\phi$ element-wise to the diagonal entries of $\bv{\Lambda}$ (i.e., to the eigenvalues of $\Ab$). We  will specify which specific form we are using when needed.

\medskip

\noindent\textbf{Krylov Subspaces.}
All algorithms analyzed in this work are Krylov subspace methods -- we introduce basic notation for Krylov subspaces below.
\begin{definition}[Block Krylov Subspace]\label{def:krylov}
The Krylov subspace of a matrix $\Ab \in \R^{n \times n}$ with respect to a starting block $\bv{X} \in \R^{n \times l}$ is given by:\begin{align*}
    \mathcal{K}_q(\bv A, \bv X) = \left[\Ab\Xb, (\Ab\Ab^T)\Ab\Xb, \ldots, (\Ab\Ab^T)^q\Ab\Xb\right].
\end{align*}
Here, $l$ is called the block size and $q$ is the depth or the number of iterations. 
\end{definition}
Notice that we require $O(ql)$ matrix-vector products with $\Ab$ to generate a Krylov subspace $\mathcal K_q(\bv A,\bv X)$ with depth $q$ and block size $l$. 
We will typically denote an orthonormal basis of the Krylov subspace by $\Qb \in \R^{n \times r}$ where $r$ is the dimension to the column span of $\mathcal{K}_q(\bv A,\bv X)$.

 \subsection{Moment Matching and Wasserstein Distance Preliminaries} 

\noindent\textbf{Inner Product between Functions.} For any two functions $f:[-L,L] \rightarrow \R$ and $g:[-L,L] \rightarrow \R$, the inner product between $f$ and $g$ is defined as $\langle f,g \rangle=\int_{-L}^{L} f(x) g(x) dx$.

\medskip

\noindent\textbf{Chebyshev polynomials.} Both our explicit moment matching method and our error analysis for SLQ are based on analyzing approximations to the  Chebyshev polynomial moments of the spectral density $s_{\bv A}$. We now describe some basic properties of Chebyshev polynomials.  We will denote the $k$th Chebyshev polynomial of the first kind by $T_k$. These polynomials are defined by the recurrence:
\begin{align*}
    T_0(x)=1 \text{ , } T_1(x)=x \text{ , } T_k(x)=2x \cdot T_{k-1}(x)-T_{k-2}(x) \text{ for } k \geq 2.
\end{align*}
We use the well known fact that the Chebyshev polynomials are bounded between $[-1,1]$ i.e., $\max_{x \in [-1,1]}|T_k(x)| \leq 1$. These polynomials also have an explicit expression of the form: 
\begin{equation}\label{Eq:T(x)}
    T_q(x)=\frac{1}{2}[(x+\sqrt{x^2-1})^q+(x-\sqrt{x^2-1})^q].
\end{equation}
Let $w(x) := \frac{1}{\sqrt{1-x^2}}$. It is well known that $\langle T_k, w \cdot T_k \rangle=\int_{-1}^{1} \frac{1}{\sqrt{1-x^2}}T^2_k(x) dx =\frac{\pi}{2}$ for $k>0$, $\langle T_0, w \cdot T_0 \rangle=\pi$ and $\langle T_i, w \cdot T_j \rangle=0$ for $i \neq j$.
We define the $k^{th}$ \emph{normalized Chebyshev polynomial} to be $\Bar{T}_k := T_k/\sqrt{\langle T_k, w \cdot T_k  \rangle}$. Note that we have $|\Bar{T}_k(x)| \leq \sqrt{\frac{2}{\pi}}$ for $x \in [-1,1]$.
Also note that $T_k(-x)=T_k(x)$ when $k$ is even and $T_k(-x)=-T_k(x)$ when $k$ is odd.

We now describe the Chebyshev Series expansion of a function:
\begin{definition}[Chebyshev Series Expansion]
The Chebyshev series expansion for a function $f :[-1,1] \rightarrow \R$ is given by:
\begin{align*}
    \sum_{k=0}^{\infty} \langle f, w \cdot \Bar{T}_k \rangle \cdot \Bar{T}_k.
\end{align*}
If f is Lipschitz
continuous then this expansion converges absolutely and uniformly to $f$ \cite{trefethen2019approximation}.
\end{definition}

\medskip

\noindent\textbf{Wasserstein Distance.} The Wasserstein-$1$ distance, also known as the earth mover's distance, is a way to measure distance between two distributions. We will use the dual formulation of this distance given by Kantorovich-Rubinstein theorem~\cite{kantorovich1957functional}. In particular, for any two probability densities $s$ and $q$ supported on $[-L,L]$ for some $L \in \R^{+}$  the Wasserstein-$1$ distance is given by:
\begin{align*}
    \W_1(s,q)=\sup_{h \in \text{1-Lip}} \int_{-L}^{L} h(x)(s(x)-q(x))dx,
\end{align*}
where $h \in \text{1-Lip}$ means we are optimizing the integral over all 1-Lipschitz functions.

\section{SDE via Moment Matching with Explicit Deflation}\label{sec:explicit}

In this section, we introduce our deflation-based approach to SDE (Algorithm~\ref{alg:sde}) which combines a block Krylov method for deflation  (Algorithm~\ref{alg:krylov}) with Chebyshev moment matching (Algorithm 2 of~\cite{braverman:2022}) for SDE. Pseudocode for this algorithm is given below.

\begin{algorithm}[H] 
\caption{Spectral Density Estimation with Deflation}
\label{alg:sde}
\begin{algorithmic}[1]
\Require{Symmetric $\bv A \in \mathbb{R}^{n\times n}$, error $\epsilon \in (0,1)$, confidence $\delta \in (0,1)$, block size $l$.}
\State Let $\Zb \in \R^{n \times s}$, $\bv{\tilde \Lambda} \in \R^{s \times s}$ be the outputs of Algorithm~\ref{alg:krylov} (Block Krylov iteration) with inputs $\Ab$, block size $l$, and constant $\beta  $. \label{alg2:line1}
\State Let $q_1$ be the spectral density corresponding to $\bv{\tilde \Lambda}$: i.e., $q_1(x)=\frac{1}{s}\sum_{i =1}^{s}\delta(x-(\Tilde{\bv{\Lambda}})_{ii})$. \label{alg2:q-def}
\State Let $\Pb=\Ib-\Zb\Zb^T$ and let $L$ be an upper bound on $\| \Pb\Ab\Pb\|_2$ such that $\| \Pb\Ab\Pb\|_2 \leq L \leq 2\| \Pb\Ab\Pb\|_2$.\label{alg2:line3}
Run Algorithm 2  of~\cite{braverman:2022} (Hutchinson-based Chebyshev moment estimation) with input  matrix $\frac{1}{L}\Pb\Ab\Pb$, number of moments $N=\frac{c_1}{\epsilon}$, and number of repetitions of Hutchinson's method $b= \max \left(1,\frac{c_2}{n\epsilon^2}\log^2 \frac{1}{\epsilon \delta}\log^2 \frac{1}{\epsilon} \right)$, where $c_1,c_2 > 0$ are sufficiently large constants. Let the approximate moments $\tilde{\tau}_1, \ldots, \tilde{\tau}_N$ denote the output of this algorithm. \label{alg2:line4}
\State Set $\hat{\tau}_i \rightarrow \frac{1}{n-s}(n \cdot \tilde{\tau}_i-s \cdot \Bar{T}_i(0))$ for $i \in [N]$ where $\Bar{T}_i(x)$ is the $i$'th \textit{normalized} Chebyshev polynomial. \label{alg2:line5}
\State Run Algorithm 1 of~\cite{braverman:2022} (Chebyshev moment matching) with the modified approximate moments $\hat{\tau}_1, \ldots, \hat{\tau}_N$ as inputs. Let $q'_2$ be the output of this algorithm.\label{alg2:line6}
\State Define $q_2$ as $q_2(x)=q'_2(x/L)$ if $x \in [-L,L]$.
\State \Return the spectral density $q=\frac{s \cdot q_1 + (n-s) \cdot q_2}{n}$. 
\end{algorithmic}
\end{algorithm}

\noindent\textbf{Description of Algorithm \ref{alg:sde}.}
As discussed in Section \ref{sec:deflationIntro}, 
Algorithm~\ref{alg:sde} first uses a randomized block Krylov method (Algorithm~\ref{alg:krylov}) to compute a set of approximate eigenvectors and eigenvalues for $\bv A$, denoted by $\bv Z$ and $\bv{\tilde \Lambda}$ respectively. Importantly for our error analysis, Algorithm~\ref{alg:krylov} only returns approximate eigenvectors that satisfy a convergence condition (line~\ref{alg1:if-condition} of Algorithm~\ref{alg:krylov}), which ensures that  $\Ab\zb_i \approx \tilde{\lambda}_i\zb_i$. 

Algorithm ~\ref{alg:sde} computes the spectral density $q_1$ of these converged approximate eigenvalues. It then \emph{deflates} $\Ab$ using the corresponding converged eigenvectors in line~\ref{alg2:line4} by computing $\Pb\Ab\Pb$ where $\Pb=\bv{I}-\Zb\Zb^T$. The algorithm then approximates the spectral density of $\Pb\Ab\Pb$ corresponding to the \textit{non-deflated} eigenvalues of $\Ab$ using a moment matching method. To do so, Algorithm~\ref{alg:sde} first computes estimates $\tilde{\tau}_1, \ldots, \tilde{\tau}_N$ of  the normalized Chebyshev moments of $\Pb\Ab\Pb$ (after normalizing $\Pb\Ab\Pb$ so that its spectral norm is bounded by 1), using Hutchinson's method (Algorithm 2 of  \cite{braverman:2022}). Since $\Pb\Ab\Pb$ contains $s$ zero eigenvalues corresponding to the deflated eigenvectors $\Zb \in \R^{n \times s}$, the moments need to adjusted in line~\ref{alg2:line5} so that they do not take into account this mass of zero eigenvalues. These adjusted moments, $\hat{\tau}_1,\ldots,\hat{\tau}_N$, are then passed as an input to Algorithm 2 of~\cite{braverman:2022} in line~\ref{alg2:line6}. This algorithm computes   a spectral density $q_2$ whose moments are equal to the approximate moments $\hat{\tau}_i$ in line~\ref{alg2:line5}. The final output of Algorithm~\ref{alg:sde} is obtained by combining $q_1$ and $q_2$ after appropriately reweighting them. 

\medskip

\noindent\textbf{Outline of Error Analysis.} 
The remainder of this section is dedicating to analyzing Algorithm \ref{alg:sde}, ultimately culminating in the proof of Theorem \ref{thm:sde1}, which bounds the Wasserstein error of the output spectral density estimate by $\epsilon \cdot \sigma_{l+1}(\bv A) + \frac{\norm{\bv A}_2}{n^{c}}$ for any constant $c$. Our analysis breaks down into the following steps:
\begin{itemize}
    \item In Section \ref{subsec:errblk} we analyze the error of the block Krylov based deflation step (Algorithm \ref{alg:krylov}). The main technical results of this section are Theorems \ref{thm:converge} and \ref{thm:eig_error}, which argue that, for some $k \le l$ with $\sigma_{k+1}(\bv A) = O(\sigma_{l+1})$, the method finds highly accurate approximations to at least the top $k$ eigenvectors and eigenvalues of $\bv A$. In particular, these approximate eigenpairs meet the convergence condition of the algorithm and thus are returned as columns of $\bv Z$. With this fact established, we can prove the main export of the section, Theorem \ref{thm:sigma_S}, which shows that $\norm{\bv P \bv A \bv P}_2 = O(\sigma_{l+1}(\bv A))$. That is, when we deflate off the converged approximate eigenvectors, we reduce the spectral norm of the matrix to $O(\sigma_{l+1}(\bv A))$, allowing us to obtain $\epsilon \cdot \sigma_{l+1}(\bv A)$ error when approximating the spectral density of $\bv P \bv A \bv P$ with Chebyshev moment matching.
    \item In Section \ref{sec:deflate} we give the error analysis of Algorithm \ref{alg:sde} itself. This analysis is fairly straightforward -- using existing backward stability analysis for eigenpair approximation, we can argue that, since $\bv P$ only deflates off highly accurate approximations to $\bv A$'s eigenvectors, our combined spectral density $q$ has Wasserstein error roughly equal to (up to an additive  $\frac{\norm{\bv A}_2}{n^{c}}$ term) the Wasserstein error of approximating the spectral density of $\bv P \bv A \bv P$ with moment matching. As discussed above, this error is bounded by $O(\epsilon \cdot \norm{\bv P \bv A \bv P}_2) = \epsilon \cdot \sigma_{l+1}(\bv A)$, allowing us to achieve our final error bound of $\epsilon \cdot \sigma_{l+1}(\bv A) + \frac{\norm{\bv A}_2}{n^{c}}$.
\end{itemize}

\begin{savenotes}
\begin{algorithm}[H] 
\caption{Block Krylov Iteration for Deflation}
\label{alg:krylov}
\begin{algorithmic}[1]
\Require{Symmetric $\bv A \in \mathbb{R}^{n\times n}$, block size $l \in [n]$, iterations $q=O(\log n)$, constant $\beta>0$}
\State Let $\bv{X} \in \R^{n \times l}$ be a starting block with independent $\mathcal{N}(0,1)$ Gaussian entries.
\State Compute $\mathcal{K}_q = \left[\Ab\Xb, (\Ab\Ab^T)\Ab\Xb, \ldots, (\Ab\Ab^T)^q\Ab\Xb\right]$.\label{alg1:krylov-subspace}
\State Orthonormalize the columns of $\mathcal{K}_q$ to get $\Qb \in \R^{n \times r}$ where $r$ is the rank of $\mathcal{K}_q$ \label{line:q}.
\State Compute $\bv{T}=\bv{Q}^T\bv{A}\bv{Q}$ and let the eigenvectors of $\bv{T}$ be $\vb_1, \ldots, \vb_r$ corresponding to eigenvalues $|\lambda_1(\bv{T})| \geq \ldots \geq |\lambda_r(\bv{T})|$.
\State Set $S=\{\}$.
\For{$j=1, \ldots, r$}
\If{$\|\Ab \Qb \vb_j-  \lambda_j(\Tb) \Qb\vb_j\|_2 \leq \frac{\|\Ab \|_2}{n^{\beta}}$\footnote{Any upper bound on $\|\Ab \|_2$ off by at most constant multiplicative factors suffices here. We can compute such an upper bound eusing $O(\log n)$ matrix-vector products via e.g., the power method or a single-vector Krlyov method.} }\label{alg1:if-condition}
\State $S=S \cup \{j \}$
\EndIf
\EndFor
\State \Return $\bv{Z}_S=\bv{Q}\bv{V}_S$ and $\Tilde{\bv{\Lambda}}_S$ where $\bv{V}_S \in \R^{r \times |S|}$ contains all eigenvectors with indices in $S$ and $\Tilde{\bv{\Lambda}}_S \in \R^{|S| \times |S|}$ is a diagonal matrix containing the corresponding eigenvalues of $\bv{T}$. 
\end{algorithmic}
\end{algorithm}
\end{savenotes}

\subsection{Error bounds for Deflation via Block Krylov }\label{subsec:errblk}

In this section, we analyze Algorithm~\ref{alg:krylov} and prove that it outputs highly accurate approximations to the top eigenvalues and eigenvectors of $\Ab$. 

Recall that we let $\bv A = \bv U \bv{\Sigma} \bv V^T$ denote the SVD of $\bv A$ -- see Section \ref{sec:linAlg}.
We start by proving Lemma \ref{thm:angle_bound}, a generalization of Theorem 2.1 of~\cite{drineas2018structural}, which bounds the error $\|(\bv{I}-\bv{Q}\bv{Q}^T)\bv{U}_p \|_F^2$ of projecting the top $p$ left singular vectors of $\bv A$ onto the span of the block Krylov subspace, for any $p \le l$\footnote{In the  numerical linear algebra literature, this quantity is often denoted as $\|(\bv{I}-\bv{Q}\bv{Q}^T)\bv{U}_p \|_F^2 = \sin(\Qb, \Ub_p)$}.
Notice that, for any $q$ and for any polynomial $\phi(x)$ of degree $\leq 2q+1$ containing only terms with odd powers, $\phi(\Ab)\bv{X}$ lies exactly in the span of the Krylov subspace $\mathcal{K}_q$.  Lemma~\ref{thm:angle_bound} bounds $\|(\bv{I}-\bv{Q}\bv{Q}^T)\bv{U}_p \|_F^2$ in  terms of the norms of the $\phi(\bv{\Sigma}_p)$, $\phi(\bv{\Sigma}_{l,\perp})$ and $\bv{V}^T_{l,\perp}\bv{X}(\bv{V}^T_{l}\bv{X})^{-1}$ for any such polynomial $\phi(x)$ which is not zero at any of the top $p$ singular values of $\Ab$. Later, by choosing such a polynomial $\phi(x)$ of degree at most $2q+1$ with odd powers so that it is large at the top singular values of $\Ab$ and small at the rest, we can bound the projection error.

\begin{lemma}[Angle between subspaces, generalization of Theorem 2.1 of~\cite{drineas2018structural}]\label{thm:angle_bound}
Let $\bv{A} \in \R^{n \times n}$ have $\rank(\Ab) >l$ and SVD given by $\bv A = \bv{U}\bv{\Sigma}\bv{V}^T$. Let $\Xb \in \R^{n \times l}$ be such that $\rank(\bv{V}_l^T\bv{X})=l$ and let $\Qb$ be an orthonormal basis of the depth $q$ Krylov subspace $\mathcal{K}_q(\bv A,\bv X)$ (see Def. \ref{def:krylov}). For any $p \leq l$ and any polynomial $\phi(x)$ of degree $2q+1$ with odd powers only, such that $\phi(\bv{\Sigma}_p)$ is non-singular,
\begin{align*}
    \|(\bv{I}-\bv{Q}\bv{Q}^T)\bv{U}_p \|_F^2 \leq \|\phi(\bv{\Sigma}_{l,\perp}) \|^2_2 \cdot \| \phi(\bv{\Sigma}_p)^{-1}\|^2_2 \cdot \|\bv{V}^T_{l,\perp}\bv{X}(\bv{V}^T_{l}\bv{X})^{-1} \|^2_2.
\end{align*}
\end{lemma}
\begin{proof}

Let $\bv{\Phi}=\phi(\Ab)\bv{X}=\bv{U}\phi(\bv \Sigma)\bv{V}^T\bv{X}$. As explained above, since $\phi(x)$ consists of only odd powers and has degree at most $2q+1$, the columns of $\bv{\Phi}$ lie in the span of the Krylov susbspace $\mathcal{K}_q(\bv A,\bv X)$ and thus $\range(\bv{\Phi})\subseteq \range(\mathcal{K}_q(\bv A,\bv X)) = \range(\bv Q)$. $\bv{\Phi}\bv{\Phi}^{\dag}$ and $\bv Q \bv Q^T$ are the orthogonal projectors onto $\range(\bv{\Phi})$ and $\range(\mathcal{K}_q(\bv A,\bv X))$ respectively. Since $\text{range}(\bv{\Phi}) \subseteq \text{range}(\bv{Q})$ we have:
\begin{align}\label{eq:up}
    \|(\bv{I}-\bv{Q}\bv{Q}^T)\bv{U}_p \|_F^2 \leq \|(\bv{I}-\bv{\Phi}\bv{\Phi}^{\dag})\bv{U}_p \|_F^2.
\end{align}
Thus, to prove the lemma it suffices to upper bound the righthand side of \eqref{eq:up}.
We write $\bv{\Phi}=\bv{\Phi}_{l}+\bv{\Phi}_{l,\perp}$ where $\bv{\Phi}_{l}=\phi(\Ab_l)\bv{X}$ and $\bv{\Phi}_{l, \perp}=\phi(\Ab_{l,\perp})\bv{X}$.  Since $p \le l$ and $\phi$ is non-zero on the top $p$ singular values, we have $\bv{\Phi}_l\bv{\Phi}_l^{\dag}\Ub_p=\Ub_p$. We upper bound the righthand side of \eqref{eq:up} by:
\begin{align*}
 \|(\bv{I}-\bv{\Phi}\bv{\Phi}^{\dag})\bv{U}_p \|_F^2 = \|\bv{U}_p-\bv {\Phi}(\bv{\Phi}^{\dag}\bv{U}_p) \|_F^2 = \min_{\bv{\Psi} \in \R^{\ell \times p}}\|\bv{U}_p-\bv{\Phi}\bv{\Psi} \|_F^2 \leq \|\bv{U}_p - \bv{\Phi}(\bv{\Phi}_{l}^{\dag}\bv{U}_p)\|_F^2.  
\end{align*}
Then, 
replacing $\bv{\Phi}$ by $\bv{\Phi}_{l}+\bv{\Phi}_{l,\perp}$ and using that $\bv{\Phi}_l\bv{\Phi}_l^{\dag}\Ub_p=\Ub_p$ we get:
\begin{align*}
  \|(\bv{I} -\bv{\Phi}\bv{\Phi}^{\dag})\bv{U}_p \|_F^2 &\leq   \|\bv{U}_p - \bv{\Phi}(\bv{\Phi}_{l}^{\dag}\bv{U}_p)\|_F^2 \\
  &= \|(\bv{I} - \bv{\Phi}_{l}\bv{\Phi}_{l}^{\dag})\bv{U}_p-\bv{\Phi}_{l,\perp}\bv{\Phi}_l^{\dag}\bv{U}_p\|_F^2 \\
  &= \| \bv{\Phi}_{l,\perp}\bv{\Phi}_{l}^{\dag} \bv{U}_p\|_F^2 \\
  &=  \| \bv{U}_{l,\perp}\phi(\bv{\Sigma}_{l,\perp})\bv{V}_{l,\perp}^T\bv{X}(\bv{V}_l^T\bv{X})^{-1}\phi(\bv{\Sigma}_l)^{-1}\bv{U}_l^T\bv{U}_p\|_F^2 \\
  &= \|\bv{U}_{l,\perp}\phi(\bv{\Sigma}_{l,\perp})\bv{V}_{l,\perp}^T\bv{X}(\bv{V}_l^T\bv{X})^{-1} \phi(\bv{\Sigma}_p^{-1}) \|_F^2 \\
  &= \|\phi(\bv{\Sigma}_{l,\perp})\bv{V}_{l,\perp}^T\bv{X}(\bv{V}_l^T\bv{X})^{-1} \phi(\bv{\Sigma}_p^{-1}) \|_F^2 \\
  &\leq \|\phi(\bv{\Sigma}_{l,\perp}) \|_F^2 \cdot \| \phi(\bv{\Sigma}_p^{-1})\|_F^2 \cdot\|\bv{V}_{l,\perp}^T\bv{X}(\bv{V}_l^T\bv{X})^{-1} \|_F^2. 
\end{align*}
In line 4 we used that $\bv{V}_l^T\bv{X} \in \R^{l \times l}$ has rank $l$ by assumption, and is thus invertible.
Combined with \eqref{eq:up} the above bound completes the proof. 
\end{proof}

Using Lemma~\ref{thm:angle_bound} we can show that for any $p$ with $\sigma_p(\bv A)$ larger than $\sigma_{l+1}(\bv A)$ by a constant multiplicative factor, the Krylov subspace $\mathcal{K}_q(\bv A, \bv X)$ generated with random width-$l$ starting block  $\bv X \in \R^{n \times l}$ and with depth $q = O(\log n)$ approximately spans the top $p$ singular vectors of $\bv A$.

\begin{lemma}[Convergence of block Krylov subspace to top singular vectors]\label{lem:sin_bound}
Let $\bv{A} \in \R^{n \times n}$ have $\rank(\Ab) >l$ and SVD given by $\bv A = \bv{U}\bv{\Sigma}\bv{V}^T$. Let $\Xb \in \R^{n \times l}$ be a matrix with independent $\mathcal{N}(0,1)$ Gaussian entries and let $\Qb$ be an orthonormal basis for the depth $q=O(\log n)$ Krylov subspace $\mathcal{K}_q(\bv A,\bv X)$ generated by $\bv X$ (see Def. \ref{def:krylov}). Then, for any $p \le l$, such that $\sigma_p(\Ab) \geq \frac{3}{2}\sigma_{l+1}(\Ab)$, for any constants $c,c' >0$, with probability at least $1-\frac{1}{n^{c'}}$, $$\|\Ub_{p} - \Qb\Qb^T\Ub_{p}\|_F^2 \leq \frac{1}{n^{c}}.$$ \end{lemma}
\begin{proof}
Since $\Xb \in \R^{n \times l}$ is a random Gaussian matrix, with probability one, $\rank(\Vb_l^T\Xb)=l$. Thus, we can apply Lemma~\ref{thm:angle_bound} to give, for any degree $2q+1$ polynomial $\phi$ consisting of only odd powers where $\phi(\bv{\Sigma}_p)$ is non-singular,
\begin{align}\label{eq:sine}
    \|\Ub_{p} - \Qb\Qb^T\Ub_{p}\|_F^2  \leq \|\phi(\bv{\Sigma}_{l,\perp})\|_2^2 \cdot \|\phi(\bv{\Sigma}_{p})^{-1}\|_2^2 \cdot \|\Vb_{l,\perp}^T\Xb(\Vb_l^T\Xb)^{\dag}\|_F^2.
\end{align}
We will now bound each term on the righthand side of \eqref{eq:sine}. First, we bound $\|\phi(\bv{\Sigma}_{l,\perp})\|_2$ and $\|\phi(\bv{\Sigma}_{p})^{-1}\|_2$ similarly to Lemma 2.4 of \cite{drineas2018structural}. We consider a gap amplifying polynomial $\phi(x)$ of degree $2q+1$ consisting only of odd powers as defined in Lemma 4.5 of~\cite{drineas2018structural} with parameters $\alpha=\sigma_{l+1}(\Ab)$ and gap $\gamma=\frac{\sigma_{p}(\Ab)}{\sigma_{l+1}(\Ab)}-1$. 
Observe that we have,
\begin{align}\label{Eq:phi_1}
   \|\phi(\bv{\Sigma}_{p})^{-1}\|_2 = \max_{1 \leq i \leq p}  \phi(\sigma_i(\Ab))^{-1} \leq \max_{1 \leq i \leq p} \sigma_i^{-1}(\Ab)=\sigma_{p}^{-1}(\Ab),
\end{align}
where we used the fact that $\sigma_i(\Ab)>0$ for $i \leq p \leq l$ as $\text{rank}(\Ab) >l$
and the fact that for any $i \leq p$, $\phi(\sigma_i(\Ab)) \geq \sigma_i(\Ab)$ as $\sigma_i(\Ab) \geq \sigma_p(\Ab) \geq (1+\gamma)\sigma_{l+1}(\Ab)=(1+\gamma)\alpha$ which follows from Lemma 4.5 of~\cite{drineas2018structural}.
Also, as $\sigma_i(\Ab) \leq \sigma_{l+1}(\Ab)$ for any $i \geq l+1$, from Lemma 4.5 of~\cite{drineas2018structural} we have:
\begin{align}\label{Eq:phi_2}
   \|\phi(\bv{\Sigma}_{l,\perp})\|_2 =\max_{i \geq l+1} |\phi(\sigma_i(\Ab))| \leq \frac{4\sigma_{l+1}(\Ab)}{2^{(2q+1)\min(\sqrt{\gamma},1 )}}.
\end{align}

Finally, we bound the middle term of \eqref{eq:sine}, $\|\Vb_{l,\perp}^T\Xb(\Vb_l^T\Xb)^{-1}\|_F^2 \leq \|\Vb_{l,\perp}^T\Xb \|_F^2 \cdot \|(\Vb_l^T\Xb)^{-1} \|_2^2$. By rotational invariance of the Gaussian distribution, $\bv V_l^T \bv X$ and $\bv V_{l,\perp}^T \bv X$ are $l \times l$ and $n-l \times l$ random Gaussian matrices respectively. Thus, by Corollary 5.35 of~\cite{vershynin2018high}, with probability $1-\frac{1}{n^{c_3}}$ we have $\|\Vb_{l,\perp}^T\Xb \|_2^2 \leq c_1 n$ and $\sigma^2_l(\Vb_l^T\Xb) \geq c_2 n$ for constants $c_1$, $c_2$ and $c_3$. So, we have $\|\Vb_{l,\perp}^T\Xb \|_F^2 \leq c_1n^2$ and $\|(\Vb_l^T\Xb)^{-1} \|_2^2 \leq \sigma^{-2}_l(\Vb_l^T\Xb)\leq \frac{1}{c_2n}$ which overall gives us that
\begin{align}\label{Eq:vbnorm}
    \|\Vb_{l,\perp}^T\Xb(\Vb_l^T\Xb)^{-1}\|_F^2 \leq c''n
\end{align}
for some constant $c''$. Plugging~\eqref{Eq:vbnorm},~\eqref{Eq:phi_2} and~\eqref{Eq:phi_1} back into~\eqref{eq:sine}, we get:
\begin{align*}
     \|\Ub_{p} - \Qb\Qb^T\Ub_{p}\|_F^2 \leq \frac{c''\sigma^2_{l+1}(\Ab)}{\sigma^2_{p}(\Ab)}\frac{n}{2^{(4q+2)\min(\sqrt{\gamma},1 )}},
\end{align*}
for some constant $c''$. Since $\sigma_p(\Ab) \geq \frac{3}{2}\sigma_{l+1}(\Ab)$, the gap is given by $\gamma=\frac{\sigma_{p}(\bv A)}{\sigma_{l+1}(\bv A)}-1 \geq \frac{1}{2}$. This gives us $\min(\sqrt{\gamma},1) \geq \frac{1}{\sqrt{2}}$.  Finally, choosing $q=C\log n$ where $C$ is a large enough constant we obtain the final bound.
\end{proof}

Lemma \ref{lem:sin_bound} establishes that the Krylov subspace generated by a random starting block $\bv X \in \R^{n \times l}$ will approximately span any singular vector corresponding to a singular value significantly larger than $\sigma_{l+1}(\bv A)$. Intuitively,  projection onto this subspace should thus preserve the largest singular values (and the largest magnitude eigenvalues) of $\bv A$. Below we prove several Lemmas that allow us to argue this formally. We consider the matrix $\Qb\Qb^T\Ab\Qb\Qb^T$ -- the projection of $\bv A$ onto the Krylov subspace on both the left and right.

\subsubsection{Eigenvector Alignment}\label{subsec:eigvec}

In this section, we first prove that the eigenvectors corresponding to the large magnitude eigenvalues of $\Qb\Qb^T\Ab\Qb\Qb^T$ are also approximate eigenvectors of $\Ab$. We first prove that if there exists some $k \in [l]$ such that $\sigma_k(\bv A)$ is larger than $\sigma_{l+1}(\Ab)$ by at least a constant multiplicative factor, then there exists some $k \leq p \leq l$ such that the top $p$ singular vectors of $\Ab$ approximately span the top $p$ singular vectors of $\Qb\Qb^T\Ab\Qb\Qb^T$ (i.e. the  matrix $\Ab$ projected on both sides to the Krylov subspace). Note that if no such $k$ exists, then $\norm{\bv A}_2 = O(\sigma_{l+1}(\bv A))$ and thus there is nothing to gain from deflation.

\begin{lemma}[Projection of $\Zb_p$ on $\Ub_p$]\label{Lem:eig_vec}
Consider the setting of Lemma~\ref{lem:sin_bound} where $\bv A \in \R^{n \times n}$ is symmetric.  Let $\alpha=\max \left(\frac{\| \Ab|_2}{n^{c/2}},\sigma_{l+1}(\Ab) \right)$ where $c>0$ is the constant in Lemma~\ref{lem:sin_bound}. Let $k \in [l]$ be such that $\sigma_{k}(\Ab) \geq 2\alpha$ and $\sigma_{k+1}(\Ab) < 2\alpha$. Let $\bv{z}_1, \ldots, \bv{z}_n$ be the eigenvectors of $\Qb\Qb^T\Ab\Qb\Qb^T$ corresponding to its eigenvalues $|\lambda_1(\Qb\Qb^T\Ab\Qb\Qb^T)| \geq \ldots \geq |\lambda_n(\Qb\Qb^T\Ab\Qb\Qb^T)|$. Let $c_1, c'>0$ be some large constants. Then, there exists some $k \leq p \leq l$ such that $\sigma_p(\Ab) \geq \frac{3}{2}\alpha$, $\sigma_p(\Ab)-\sigma_{p+1}(\Ab) \geq \frac{\|\Ab \|_2}{2n^{c/2+1}}$, and, letting $\bv{Z}_p \in \R^{n \times p}$  have $\zb_1, \ldots, \zb_p$ as its columns, $$\|\Ub_p\Ub_p^T\Zb_p -\Zb_p \|_2 \leq \frac{1}{n^{c_1}},$$ with probability at least $1-\frac{1}{n^{c'}}$. 
\end{lemma}

\begin{proof}
Let $l_1$ be the largest index with $k \leq l_1 \leq l$ and $\sigma_{l_1}(\bv A) \geq \frac{3}{2}\alpha$, i.e., we must have $\sigma_{l_1+1}(\Ab) < \frac{3}{2}\alpha$. Note that since $\sigma_k(\Ab) \geq 2\alpha$, such an $l_1$ must exist.  Then, $\sigma_{k}(\Ab)-\sigma_{l_1+1}(\Ab) \geq  \alpha/2 \geq \frac{\|\Ab \|_2}{2n^{c/2}}$ where we use the fact that $\alpha \geq \frac{\|\Ab \|_2}{n^{c/2}}$.  Since there are at most $l_1$ indices in the range $[k,l_1]$, there must be some $p \in [k,l_1]$ such that:

\begin{align}\label{Eq:gap1}
    \sigma_p(\Ab)-\sigma_{p+1}(\Ab) \geq \frac{\sigma_k(\bv A)-\sigma_{l_1+1}(\bv A)}{l_1}  \geq \frac{\|\Ab \|_2}{2n^{c/2+1}},
\end{align}
where we loosely bound $l_1 \le n$. This establishes the first claim of the Lemma. We now prove the second claim.  

Since $\sigma_p(\Ab) \geq \frac{3}{2}\alpha \geq \frac{3}{2}\sigma_{l+1}(\Ab)$, from Lemma~\ref{lem:sin_bound}, we get that $\bv{Q}\bv{Q}^T\bv{U}_p=\bv{U}_p+\bv{E}$ where $\|\bv{E} \|_2 \leq \frac{1}{n^c}$ with rpobability at least $1-\frac{1}{n^{c'}}$ for some constant $c'>0$.  Let $\Zb \in \R^{n \times n}$ be an orthonormal matrix containing all eigenvectors $\zb_1, \ldots, \zb_n$ of $\bv{Q}\bv{Q}^T \bv{A}\bv{Q}\bv{Q}^T$ as its columns. Since the columns of $\Zb$ form an orthonormal basis of $\R^n$, there exists a matrix $\bv{C} \in \R^{n \times p}$ such that $\bv{U}_p=\bv{Z}\bv{C}$ and $\bv{C}^T\bv{C}=\bv{I}_p$ where $\bv{I}_p$ is the $p \times p$ identity matrix. Then, we have:
\begin{align}\label{Eq:sum}
    \bv{U}_p=\bv{Z}\bv{C}=\bv{Z}_p\bv{C}_1+\Zb_{p,\perp}\bv{C}_2,
\end{align} where $\bv{C}=[\bv{C}_1; \bv{C}_2]$ for $\bv{C}_1 \in \R^{p \times p}$ and $\bv{C}_2 \in \R^{n-p \times p}$. We will now prove that $\|\bv{C}_2 \|_2$ is very small and $\bv{C}_1$ is very close to the identity matrix. This implies that $\bv{U}_p$ approximately spans $\bv{Z}_p$ which proves our second claim.

First observe that we can write $\bv{Q}\bv{Q}^T\bv{A}\bv{Q}\bv{Q}^T \bv{U}_p$ as:
\begin{align}\label{Eq:qq}
    \bv{Q}\bv{Q}^T\bv{A}\bv{Q}\bv{Q}^T \bv{U}_p
    &= \Qb\Qb^T \Ab (\Ub_p+\bv{E}) \nonumber \\
   &= \Qb\Qb^T \Ub_p\bv{\Lambda}_p + \Qb\Qb^T\Ab\bv{E} \nonumber \\
   &= (\Ub_p +\bv{E})\bv{\Lambda}_p+ \Qb\Qb^T\Ab\bv{E} \nonumber\\
   &= \Ub_p \bv{\Lambda}_p + \bv{E}\bv{\Lambda}_p+ \Qb\Qb^T\Ab\bv{E}.
\end{align}
Thus, we get:
\begin{align}\label{Eq:qq1}
     \bv{Q}\bv{Q}^T\bv{A}\bv{Q}\bv{Q}^T \bv{U}_p &= \bv{Z}\bv{C}\bv{\Lambda}_p+\bv{E}\bv{\Lambda}_p+\bv{Q}\bv{Q}^T\bv{A}\bv{E}.
\end{align}
Next, observe that, since $\bv {Z}$ has columns equal to the eigenvalues of $\bv{Q}\bv{Q}^T\bv{A}\bv{Q}\bv{Q}^T$, we can also write $\bv{Q}\bv{Q}^T\bv{A}\bv{Q}\bv{Q}^T \bv{U}_p$ as:
\begin{align}\label{Eq:qq2}
    \bv{Q}\bv{Q}^T\bv{A}\bv{Q}\bv{Q}^T \bv{U}_p = \bv{Q}\bv{Q}^T\bv{A}\bv{Q}\bv{Q}^T\bv{Z}\bv{C}= \bv{Z}\bv{\Tilde{\Lambda}}\bv{C},
\end{align}
where $\bv{\Tilde{\Lambda}} \in \R^{n \times n}$ denotes the diagonal matrix containing the eigenvalues of $\bv{Q}\bv{Q}^T\bv{A}\bv{Q}\bv{Q}^T$ on its diagonal. Thus, combining~\eqref{Eq:qq1} and~\eqref{Eq:qq2} we get that $\bv{Z}\bv{\Tilde{\Lambda}}\bv{C} = \bv{Z}\bv{C}\bv{\Lambda}_p+\bv{E}\bv{\Lambda}_p+\bv{Q}\bv{Q}^T\bv{A}\bv{E}$ or 
\begin{align*}
    \bv{Z}\bv{\Tilde{\Lambda}}\bv{C}-\bv{Z}\bv{C}\bv{\Lambda}_p=\bv{E}\bv{\Lambda}_p+\bv{Q}\bv{Q}^T\bv{A}\bv{E}.
\end{align*}
Let $\bv{E}'=\bv{E}\bv{\Lambda}_p+\bv{Q}\bv{Q}^T\bv{A}\bv{E}$. Using triangle inequality, we get that $\|\bv{E}' \|_2 \leq \| \bv{E}\bv{\Lambda}_p\|_2 +\| \bv{Q}\bv{Q}^T\bv{A}\bv{E}\|_2$. Now,  since $\| \bv{E}\|_2 \leq \frac{1}{n^c}$ and $\|\bv{\Lambda}_p \|_2=\|\Ab \|_2$, we have $\| \bv{E}\bv{\Lambda}_p\|_2 \leq \|\bv{E} \|_2 \|\bv{\Lambda}_p \|_2 \leq \frac{\|\Ab\|_2}{n^c}$. Similarly, since $\bv Q$ is orthonormal, $\| \bv{Q}\bv{Q}^T\bv{A}\bv{E}\|_2  \leq \frac{\|\Ab\|_2}{n^c}$.  So, we get $\|\bv{E}' \|_2 \leq \frac{2\|\Ab\|_2}{n^c}$ which implies that $\|\bv{Z}\bv{C}\bv{\Lambda}_p-\bv{Z}\bv{\Tilde{\Lambda}}\bv{C} \|_2 \leq \frac{2\| \Ab\|_2}{n^c}$ or, equivalently, since $\bv{Z}$ has orthonormal columns,
\begin{align*}
    \|\bv{C}\bv{\Lambda}_p-\bv{\Tilde{\Lambda}}\bv{C} \|_2 \leq \frac{2\| \Ab\|_2}{n^c}.
\end{align*}

For any $i \in [p]$, the $i^{th}$ column of the $n \times p$ matrix $\bv{C}\bv{\Lambda}_p-\bv{\Tilde{\Lambda}}\bv{C}$ is given by $\lambda_i(\bv{A})\bv{C}_{:,i}-\bv{\Tilde{\Lambda}}\bv{C}_{:,i}$ where $\bv{C}_{:,i}$ is the $i$th column of $\bv{C}$. So, we have $\|\lambda_i(\bv{A})\bv{C}_{:,i}-\bv{\Tilde{\Lambda}}\bv{C}_{:,i}\|_2 \leq \frac{2\| \Ab\|_2}{n^c}$ for all $i \in [p]$. Using the definition of the $l_2$ norm of the vector, we have $\sqrt{\sum_{j=1}^{n}(\lambda_i(\Ab)-\lambda_j(\bv{Q}\bv{Q}^T\bv{A}\bv{Q}\bv{Q}^T))^2\bv{C}^2_{ji}} \leq \frac{2\|\Ab \|_2}{n^c}$. This implies that for all $i \in [p]$ and $j \in [n]$,
\begin{align}\label{eq:gap}
    |\lambda_i(\Ab)-\lambda_j(\bv{Q}\bv{Q}^T\bv{A}\bv{Q}\bv{Q}^T)||\bv{C}_{ji}| \leq \frac{2\|\Ab \|_2}{n^c}.
\end{align}Now, by the minimax principle of singular values, $|\lambda_j(\bv{Q}\bv{Q}^T\bv{A}\bv{Q}\bv{Q}^T)| \leq |\lambda_{p+1}(\Ab)|$ for all $j\geq p+1$. Thus, for any $i \in [p]$ and $j \in \{p+1,\ldots,n\}$, by the triangle inequality, $|\lambda_i(\Ab)-\lambda_j(\bv{Q}\bv{Q}^T\bv{A}\bv{Q}\bv{Q}^T)| \geq |\lambda_i(\Ab)|-|\lambda_j(\bv{Q}\bv{Q}^T\bv{A}\bv{Q}\bv{Q}^T)| \geq |\lambda_p(\Ab)|-|\lambda_{p+1}(\Ab)| \geq \frac{\|\Ab \|_2}{2n^{c/2+1}}$ where the last step follows from the first claim~\eqref{Eq:gap1}. Thus, from~\eqref{eq:gap}, we get that for any $i \in [p]$ and $j \in \{p+1,\ldots,n \}$, $$|\bv{C}_{ji}| \leq \frac{4}{n^{c/2-1}}.$$ Note that from~\eqref{Eq:sum}, we have that $\bv{C}_2$ contains all $\bv{C}_{ji}$ such that  $j \in \{p+1,\ldots,n \}$ and $i \in [p]$. So, we can bound:
\begin{align}\label{Eq:c2}
    \| \bv{C}_2\|_2 \leq \|\bv{C}_2 \|_F  \leq  \sqrt{\sum_{i=1}^p\sum_{j=p+1}^{n} \bv{C}_{ji}^2} \leq \frac{4}{n^{c/2-2}}.
\end{align}
Next, since $\bv{C}^T\bv{C}=\bv{I}_p= \bv{C}_1^T\bv{C}_1+\bv{C}_2^T\bv{C}_2$, we have   $\|\bv{C}_1^T\bv{C}_1 -\bv{I}_p\|_2 = \|\bv{C}_2^T\bv{C}_2 \|_2=\|\bv{C}_2\|_2^2 \leq \frac{16}{n^{2(c/2-2)}}$. Thus, we can write $\bv{C}_1^T\bv{C}_1=\bv{I}+\bv{E}_1$ where $\| \bv{E}_1\|_2 \leq \frac{16}{n^{2(c/2-2)}}$. Observe that $\bv{C}_1\bv{C}_1^T$ and $\bv{C}_1^T\bv{C}_1$ have the same eigenvalues, and thus we also have 
\begin{align}\label{Eq:c1c1}
    \bv{C}_1\bv{C}^T_1=\bv{I}+\bv{E}'',
\end{align}
where $\| \bv{E}''\|_2 \leq \frac{16}{n^{2(c/2-2)}}$. Since we have $\|\bv{C}_1 \|_2=\sqrt{\|\bv{C}_1\bv{C}_1^T \|_2}=\sqrt{\|\bv{I}+\bv{E}_1 \|_2} \leq \sqrt{1+\|\bv{E} \|_1} \leq 1+\frac{4}{n^{c/2-2}}$ (where the second to last step follows using triangle inequality), we also have that: 
\begin{align}\label{Eq:c1}
    \bv{C}_1=\bv{I}+\bv{E}_2,
\end{align}
where $\| \bv{E}_2\|_2 \leq \frac{4}{n^{c/2-2}}$. Then, we have:
\begin{align*}
    \Ub_p\Ub_p^T\Zb_p -\Zb_p&= \Ub_p(\bv{C}_1^T\Zb_p^T+ \bv{C}_2^T\Zb_{p,\perp}^T)\Zb_p -\Zb_p \\
    &= \Ub_p \bv{C}_1^T -\Zb_p \\
    &= \bv{Z}_p\bv{C}_1\bv{C}_1^T+\Zb_{p,\perp}\bv{C}_2\bv{C}_1^T-\Zb_p\\
    &=\bv{Z}_p(\bv{I}+\bv{E}'')+\Zb_{p,\perp}\bv{C}_2\bv{C}_1^T -\Zb_p\\
    &= \Zb_p\bv{E}''+\Zb_{p,\perp} \bv{C}_2 \bv{C}_1^T \\
    &= \Zb_p\bv{E}'' + \Zb_{p,\perp}\bv{C}_2(\bv{I}+\bv{E}^T_2) \\
    &= \Zb_p\bv{E}'' + \Zb_{p,\perp}\bv{C}_2 + \Zb_{p,\perp}\bv{C}_2\bv{E}^T_2.
\end{align*}
The first and third equality above follows from~\eqref{Eq:sum}, the fourth equality follows from~\eqref{Eq:c1c1} and the sixth equality follows from~\eqref{Eq:c2}. Thus, using triangle inequality and spectral submultiplicativity, we get $\|\Ub_p\Ub_p^T\Zb_p -\Zb_p \|_2\leq \|\bv{Z}_p\bv{E}'' \|_2+ \|\Zb_{p,\perp}\bv{C}_2 \|_2+\|\Zb_{p,\perp}\bv{C}_2\bv{C}_1^T \|_2\leq \|\bv{E}'' \|_2 + \|\bv{C}_2 \|_2 +\| \bv{C}_1\|_1\|\bv{C}_2 \|_2 \leq \frac{16}{n^{c-4}}+ \frac{4}{n^{c/2-2}} +\frac{16}{n^{c-4}} \leq \frac{36}{n^{c/2-2}}$ where the second inequality follows from the fact that $\bv{Z}_p$ and $\bv{Z}_{p,\perp}$ are orthonormal matrices and the third inequality follows from the error bounds in~\eqref{Eq:c2},~\eqref{Eq:c1c1} and~\eqref{Eq:c1}. Choosing the constant $c_1$ to be suitably large enough so that $\|\Ub_p\Ub_p^T\Zb_p -\Zb_p \|_2$ is bounded by $\frac{1}{n^{c_1}}$ completes the proof.
\end{proof}

The next theorem proves the main result of this subsection, i.e., the top eigenvectors of $\Qb\Qb^T\Ab\Qb\Qb^T$ are also approximately the eigenvectors of $\Ab$, provided there is some $\sigma_k(\bv A)$ that is constant factor larger than $\sigma_{l+1}(\bv A)$, as assumed in Lemma~\ref{Lem:eig_vec}.

\begin{theorem}[Convergence error of top $p$ eigenvectors]\label{thm:converge}

Consider the setting of Lemma~\ref{Lem:eig_vec}. Let $c_1,c'>0$ be some large constants. Let $\Tilde{\bv{\Lambda}}_p \in \R^{p \times p}$ be a diagonal matrix containing the corresponding eigenvalues $\lambda_1(\bv{Q}\bv{Q}^T\Ab\bv{Q}\bv{Q}^T), \ldots, \lambda_p(\bv{Q}\bv{Q}^T\Ab\bv{Q}\bv{Q}^T)$ on its diagonal. Then, we have $$\|\Ab\Zb_p-\Zb_p\Tilde{\bv{\Lambda}}_p \|_2 \leq \frac{\|\Ab \|_2}{n^{c_1}} .$$
with probability at least $1-\frac{1}{n^{c'}}$.
\end{theorem}
\begin{proof}
From Lemma~\ref{Lem:eig_vec}, $\Zb_p=\Ub_p\Ub_p^T\Zb_p+\bv{E}_1$ where $\|\bv{E}_1\|_2 \leq \frac{1}{n^{c_2}}$ with probability at least $1-\frac{1}{n^{c''}}$ for some constant $c_2,c''$. Also, from Lemma~\ref{lem:sin_bound}, we have $\bv{U}_p=\bv{Q}\bv{Q}^T\bv{U}_p+\bv{E}_2$ where $\|\bv{E}_2 \|_2 \leq \frac{1}{n^{c_3}}$ with probability at least $1-\frac{1}{n^{c''}}$ for some constant $c_3$. Observe that:
\begin{align*}
    \Ab\Zb_p &\stackrel{a}{=} \Ab(\Ub_p\Ub_p^T\Zb_p+\bv{E}_1) \\
    &= \Ab\Ub_p\Ub_p^T\Zb_p + \Ab \bv{E}_1 \\
    &\stackrel{b}{=} \Ub_p\Ub_p^T\Ab\Ub_p\Ub_p^T\Zb_p + \Ab \bv{E}_1 \\
    &\stackrel{c}{=} (\bv{Q}\bv{Q}^T\bv{U}_p+\bv{E}_2)\Ub_p^T\Ab\Ub_p\Ub_p^T\Zb_p + \Ab \bv{E}_1 \\
    &= \bv{Q}\bv{Q}^T\bv{U}_p\Ub_p^T\Ab\Ub_p\Ub_p^T\Zb_p +\bv{E}_2\Ub_p^T\Ab\Ub_p\Ub_p^T\Zb_p + \Ab \bv{E}_1 \\
    &\stackrel{d}{=} \bv{Q}\bv{Q}^T\Ab\Ub_p\Ub_p^T\Zb_p+\bv{E}_2\Ub_p^T\Ab\Ub_p\Ub_p^T\Zb_p + \Ab \bv{E}_1 \\
    &\stackrel{e}{=} \bv{Q}\bv{Q}^T\Ab(\bv{Q}\bv{Q}^T\bv{U}_p+\bv{E}_2)\Ub^T_p\bv{Z}_p+ \bv{E}_2\Ub_p^T\Ab\Ub_p\Ub_p^T\Zb_p + \Ab \bv{E}_1\\
    &=\bv{Q}\bv{Q}^T\Ab\bv{Q}\bv{Q}^T\bv{U}_p\Ub^T_p\bv{Z}_p+\bv{Q}\bv{Q}^T\Ab\bv{E}_2\Ub^T_p\bv{Z}_p+\bv{E}_2\Ub_p^T\Ab\Ub_p\Ub_p^T\Zb_p + \Ab \bv{E}_1.\\
    &\stackrel{f}{=}\bv{Q}\bv{Q}^T\Ab\bv{Q}\bv{Q}^T\bv{Z}_p-\bv{Q}\bv{Q}^T\Ab\bv{Q}\bv{Q}^T\bv{E}_1+\bv{Q}\bv{Q}^T\Ab\bv{E}_2\Ub^T_p\bv{Z}_p+\bv{E}_2\Ub_p^T\Ab\Ub_p\Ub_p^T\Zb_p + \Ab \bv{E}_1.
\end{align*}
In the above set of equations, (a) follows by replacing $\bv{Z}_p$ with $\Ub_p\Ub_p^T\Zb_p+\bv{E}_1$, (b) and (d) follows from the fact that $\Ab\Ub_p\Ub_p^T=\Ub_p\Ub_p^T\Ab\Ub_p\Ub_p^T$ as $\Ub_p$ contains the eigenvectors of $\Ab$, (c) and (e) follows from replacing $\Ub_p$ in the first term by $\bv{Q}\bv{Q}^T\bv{U}_p+\bv{E}_2$, (f) follows from replacing the $\Ub_p\Ub_p^T\Zb_p$ in the first term by $\Zb_p-\bv{E}_1$.
Now, $\bv{Q}\bv{Q}^T\Ab\bv{Q}\bv{Q}^T\bv{Z}_p=\Zb_p\Tilde{\bv{\Lambda}}_p$. Also, using spectral submultiplicativity, $\|\bv{Q}\bv{Q}^T\Ab\bv{Q}\bv{Q}^T\bv{E}_1 \|_2 \leq \|\Ab \|_2 \|\bv{E}_1 \|_2 \leq \frac{\|\Ab \|_2}{n^{c_2}}$. Similarly, each of the last three terms in the final step above can be bounded by $\max \big(\frac{\|\Ab \|_2}{n^{c_2}}, \frac{\|\Ab \|_2}{n^{c_3}}\big)$. Thus, using triangle inequality, and after adjusting the constants $c_1$ and $c'$ appropriately, we have:
\begin{align*}
    \|\Ab\Zb_p -\Zb_p\Tilde{\bv{\Lambda}}_p \|_2 \leq \frac{\|\Ab \|_2}{n^{c_1}},
\end{align*}
with probability at least $1-\frac{1}{n^{c'}}$.
\end{proof}

\subsubsection{Eigenvalue Alignment}\label{subsec:eigval}

In this section, we show that as a consequence of Lemma~\ref{lem:sin_bound}, the large magnitude eigenvalues of $\Ab$ are approximated by those of $\Qb\Qb^T\Ab\Qb\Qb^T$. We first state a result from~\cite{parlett1998symmetric} (Theorem 11.5.1) which states that for a symmetric matrix  $\bv{D} \in \R^{n \times n}$ and any orthonormal matrix $\bv{C} \in \R^{n \times m}$ the eigenvalues of $\bv{D}$ can be put into one-to-one correspondence with those of  $\bv{C}^T\bv{D}\bv{C}$ such that the error is bounded by the spectral norm of $\bv{DC}-\bv{C}\bv{C}^T\bv{DC}$.
\begin{theorem}[Theorem 11.5.1 of ~\cite{parlett1998symmetric}]\label{thm:parlett}
Let $\bv{D} \in \R^{n \times n}$ be a symmetric matrix and $\bv{C} \in \R^{n \times m}$ be a matrix with orthonormal columns. Then, there exists $m$ eigenvalues of $\bv{D}$, $\{\alpha_i~\lvert~i=1,\ldots,m\}$ such that for $i \in [m]$:
\begin{align*}
    |\alpha_{i}-\lambda_i(\bv{C}^T\bv{D}\bv{C})| \leq \|\bv{DC}-\bv{C}\bv{C}^T\bv{DC} \|_2.
\end{align*} 
\end{theorem}

We now prove a couple of Lemmas which we will use along with Theorem~\ref{thm:parlett} for our main theorem. 

\begin{lemma}\label{lem:spec1}
Consider the setting of Lemma~\ref{lem:sin_bound} where $\bv A \in \R^{n \times n}$ is  symmetric. Let the error bound for Lemma~\ref{lem:sin_bound} hold for a constant $c>0$. For some constant $c' >0$, with probability at least $1-\frac{1}{n^{c'}}$, 
$$\|\Qb\Qb^T \Ab \Qb \Qb^T \Ub_p-\Ub_p\Ub_p^T\Qb\Qb^T\Ab\Qb\Qb^T\Ub_p \|_2 \leq \frac{\|\Ab \|_2}{n^{c-1}}.$$
\end{lemma}
\begin{proof}
Recall that we denote the eigendecomposition of a symmetric matrix $\Ab$ by $\Ab = \Ub \bv{\Lambda}\Ub^T$.  $\Ab$'s eigenvectors (i.e., the columns of $\Ub$) are equal to its singular vectors, and recall that $\bv U_p \in \R^{n \times p}$ has columns equal to the $p$ eigenvectors corresponding to the $p$ largest magnitude eigenvalues (i.e., the $p$ singular vectors corresponding to the $p$ largest singular values). From Lemma~\ref{lem:sin_bound}, we have $\Qb\Qb^T\Ub_p= \Ub_p+\bv{E}$ where $\bv{E} \in \R^{n \times p}$ has $\|\bv{E} \|_2 \leq \frac{1}{n^{c}}$ with probability at least $1-\frac{1}{n^{c''}}$. Thus, from~\eqref{Eq:qq}, we get
\begin{align}
   \Qb\Qb^T \Ab \Qb \Qb^T \Ub_p 
   &= \Ub_p \bv{\Lambda}_p + \bv{E}\bv{\Lambda}_p+ \Qb\Qb^T\Ab\bv{E}.
\end{align}
Also, we have:
\begin{align*}
   \Ub_p\Ub_p^T\Qb\Qb^T\Ab\Qb\Qb^T\Ub_p &= \Ub_p (\Ub_p^T+ \bv{E}^T)\Ab (\Ub_p + \bv{E}) \\
   &= \Ub_p \Ub_p^T \Ab \Ub_p + \Ub_p \Ub_p^T \Ab \bv{E} +\Ub_p \bv{E}^T\Ab\Ub_p+\Ub_p \bv{E}^T\Ab\bv{E} \\
   &= \Ub_p \bv{\Lambda}_p+  \Ub_p \Ub_p^T \Ab \bv{E} +\Ub_p \bv{E}^T\Ab\Ub_p+\Ub_p \bv{E}^T\Ab\bv{E}.
\end{align*}
Then, using triangle inequality and spectral submultiplicativity, we have that $$\|\Qb\Qb^T \Ab \Qb \Qb^T \Ub_p- \Ub_p\Ub_p^T\Qb\Qb^T\Ab\Qb\Qb^T\Ub_p\|_2 \leq\norm{\bv \Lambda_p}_2 \norm{\bv{E}}_2 + 3 \norm{\Ab}_2 \norm{\bv{E}}_2+ \norm{\Ab} \norm{\bv{E}}_2^2 \leq \frac{5\|\Ab \|_2}{n^{c}}.$$
This completes the proof.
\end{proof}

Next we prove that the eigenvalues of $\Ab$ are well approximated by those of $\bv{U}_p^T\bv{Q}\bv{Q}^T\bv{A}\bv{Q}\bv{Q}^T\bv{U}_p$.

\begin{lemma}\label{lem:bound3}
Consider the setting of Lemma~\ref{lem:sin_bound} where $\bv A \in \R^{n \times n}$ is  symmetric. Let the error bound for Lemma~\ref{lem:sin_bound} hold for a constant $c>0$. For some constant $c' >0$, with probability at least $1-\frac{1}{n^{c'}}$, we have for all $i \in [p]$, $$| \lambda_i(\bv{U}_p^T\bv{Q}\bv{Q}^T\bv{A}\bv{Q}\bv{Q}^T\bv{U}_p)-\lambda_i(\bv{A})| \leq \frac{\|\bv{A} \|_2}{n^{c-1}}.$$
\end{lemma}
\begin{proof}
To prove the theorem we will show that $\bv{U}_p^T\bv{Q}\bv{Q}^T\bv{A}\bv{Q}\bv{Q}^T\bv{U}_p $ is close to $\bv U_p^T\bv A \bv U_p = \bv \Lambda_p$ in the spectral norm, and thus has nearby eigenvalues. Observe that 
\begin{align}\label{eq:bound1}
       \bv{U}_p^T\bv{A}\bv{U}_p- \bv{U}_p^T\bv{Q}\bv{Q}^T\bv{A}\bv{Q}\bv{Q}^T\bv{U}_p &=\bv{U}_p^T(\Ib-\Qb\Qb^T)\Ab\Qb\Qb^T\Ub_p+\Ub_p^T\Qb\Qb^T\Ab(\Ib-\Qb\Qb^T)\Ub_p \nonumber \\
      &+ \Ub_p^T(\Ib-\Qb\Qb^T)\Ab(\Ib-\Qb\Qb^T)\Ub_p. 
\end{align}
Since all the terms on the righthand side above contain the term $(\bv I-\bv Q\bv Q^T)\bv U_p$, we can use Lemma~\ref{lem:sin_bound} to show that they are small. For the first term, we have:
\begin{align*}
    \|\bv{U}_p^T(\Ib-\Qb\Qb^T)\Ab\Qb\Qb^T\Ub_p \|_2 &\leq \| \bv{U}_p^T(\Ib-\Qb\Qb^T)\|_2 \cdot \|\Ab \|_2 \cdot \|\Qb\Qb^T\Ub_p \|_2 \leq \frac{\|\bv{A} \|_2}{n^c},
\end{align*}
where we also use the fact $\Qb$ and $\Ub_p$ are orthonormal matrices and that $ \|\bv{U}_p^T(\Ib-\Qb\Qb^T) \|_2  \le  \|\bv{U}_p^T(\Ib-\Qb\Qb^T) \|_F \le \frac{1}{n^{c}}$ by Lemma \ref{lem:sin_bound}. Similarly we can bound the other two terms in \eqref{eq:bound1} by $\frac{\|\Ab \|_2}{n^{c}}$. Applying triangle inequality, we obtain:
\begin{align*}
    \| \bv{U}_p^T\bv{A}\bv{U}_p- \bv{U}_p^T\bv{Q}\bv{Q}^T\bv{A}\bv{Q}\bv{Q}^T\bv{U}_p\|_2 \leq \frac{3\|\bv{A} \|_2}{n^{c}}.
\end{align*}
Observe that $\bv U_p^T \bv A \bv U_p = \bv \Lambda_p$ has eigenvalues $\lambda_1(\bv A) \ldots, \lambda_p(\bv A)$. So, by Weyl's inequality (Fact~\ref{fact:weyl}), we have that $|\lambda_i(\Ab)-\lambda_i(\bv{U}_p^T\bv{Q}\bv{Q}^T\bv{A}\bv{Q}\bv{Q}^T\bv{U}_p)| \leq  \| \bv{U}_p^T\bv{A}\bv{U}_p- \bv{U}_p^T\bv{Q}\bv{Q}^T\bv{A}\bv{Q}\bv{Q}^T\bv{U}_p\|_2 \leq \frac{3\|\bv{A} \|_2}{n^{c}}$ for $i \in [p]$. This completes the proof.
\end{proof}

The next theorem is the main result of this section which states that the top $p$ eigenvalues of $\Ab$ are approximated by the top $p$ eigenvalues of $\Qb\Qb^T \Ab \Qb \Qb^T$ up to some permutation for some $k \leq p \leq l$ as long as the conditions in Lemma~\ref{Lem:eig_vec} are satisfied. i.e., there exists some $k \in [l]$ with at least some constant multiplicative gap between $\sigma_k(\Ab)$ and $\sigma_{l+1}(\Ab)$. 

\begin{theorem}[Eigenvalue alignment]\label{thm:eig_error}
Consider the setting of Theorem~\ref{thm:converge} and Lemma~\ref{Lem:eig_vec}. Let $c_1,c'>0$ be some constants. Then, there exists a permutation $S:[p] \to [p]$ such that for every $i \in [p]$, we have $$|\lambda_i(\Ab)-\lambda_{S(i)}(\Qb\Qb^T \Ab \Qb \Qb^T)| \leq \frac{\|\Ab \|_2}{n^{c_1}},$$ with probability at least $1-\frac{1}{n^{c'}}$.
\end{theorem}
\begin{proof}
We apply Theorem~\ref{thm:parlett} to first prove that each of the top $p$ eigenvalues of $\Ab$ has an eigenvalue of $\Qb\Qb^T\Ab \Qb\Qb^T$ close to it. Let $\bv{C}=\Ub_p$ and $\bv{D}=\Qb\Qb^T\Ab\Qb\Qb^T$. Let $c>0$ be the constant in the statements of Lemmas~\ref{lem:sin_bound},~\ref{lem:spec1} and~\ref{lem:bound3}. Then, $\bv{C}^T\bv{D}\bv{C}=\Ub_p^T\Qb\Qb^T\Ab\Qb\Qb^T\Ub_p$ and by Theorem~\ref{thm:parlett}, there exist $p$ eigenvalues $\alpha_1, \ldots, \alpha_p$ of $\Qb\Qb^T\Ab\Qb\Qb^T$ such that for $i \in [p]$:
    \begin{align}\label{eq:permute}
        |\alpha_{i}-\lambda_i(\Ub_p^T\Qb\Qb^T\Ab\Qb\Qb^T\Ub_p)| &\leq \|\bv{D}\bv{C}-\bv{C}\bv{C}^T\bv{D}\bv{C} \|_2 \nonumber \\ 
        &=\|\Qb\Qb^T\Ab\Qb\Qb^T\Ub_p-\Ub_p\Ub_p^T\Qb\Qb^T\Ab\Qb\Qb^T\Ub_p \|_2 \nonumber \\
        &\leq \frac{\|\Ab \|_2}{n^{c-1}}.
    \end{align}
The last inequality hold  with probability at least $1-\frac{1}{n^{c'}}$ for some constant $c'>0$ due to Lemma~\ref{lem:spec1}. From Lemma~\ref{lem:bound3}, with probability at least $1-\frac{1}{n^{c'}}$, we get that for all $i \in [p]$, we have:
\begin{equation}\label{eq:1}
    |\lambda_i(\Ub_p^T\Qb\Qb^T\Ab\Qb\Qb^T\Ub_p)-\lambda_i(\Ab)| \leq \frac{\|\Ab \|_2}{n^{c-1}}.
\end{equation}
Combining~\eqref{eq:permute} and~\eqref{eq:1}, and using triangle inequality we have for every $i \in [p]$ (for some constant $c_4>0$), 
\begin{align}\label{Eq:eig_gap1}
    |\alpha_i-\lambda_i(\bv{A})| \leq \frac{2\|\Ab \|_2}{n^{c-1}}.
\end{align}
We now prove that the $\alpha_i$'s (for all $i \in [p]$) are a permutation over $\lambda_1(\Qb\Qb^T\Ab\Qb\Qb^T),\ldots, \lambda_p(\Qb\Qb^T\Ab\Qb\Qb^T)$. Suppose that this is not true, i.e., there exists some $j \in [p]$ such that $\alpha_j=\lambda_{p+r}(\Qb\Qb^T\Ab\Qb\Qb^T)$ for some $r \geq 1$. So we have $|\lambda_{p+r}(\Qb\Qb^T\Ab\Qb\Qb^T)-\lambda_j(\Ab)| \leq \frac{2\|\Ab \|_2}{n^{c-1}}$. Now, by the minimax principle we have 
\begin{align*}
    |\lambda_{p+r}(\Qb\Qb^T\Ab\Qb\Qb^T)| \leq |\lambda_{p+1}(\Qb\Qb^T\Ab\Qb\Qb^T)| \leq |\lambda_{p+1}(\Ab)|.
\end{align*}
Using triangle inequality, we have:
\begin{align*}
    |\lambda_{p}(\Ab)|- |\lambda_{p+1}(\Ab)| \leq |\lambda_j(\Ab)|-|\lambda_{p+r}(\Qb\Qb^T\Ab\Qb\Qb^T)| \leq |\lambda_j(\Ab)-\lambda_{p+r}(\Qb\Qb^T\Ab\Qb\Qb^T)| \leq \frac{2\|\Ab \|_2}{n^{c-1}}.
\end{align*}
From Lemma~\ref{Lem:eig_vec}, we have $\sigma_p(\Ab)-\sigma_{p+1}(\Ab) \geq \frac{\|\Ab \|_2}{2n^{c/2+1}}$ for some constant $c$. For a large enough $c$, we have $\frac{2\|\Ab \|_2}{n^{c-1}} < \frac{\|\Ab \|_2}{2n^{c/2+1}}$ and thus, we have a contradiction. So we must have $\gamma_j=\lambda_i(\Qb\Qb^T\Ab\Qb\Qb^T)$ for some $i \in [p]$. Choosing the constant $c_1$ to  be suitably gives us the bound.
\end{proof}

\subsubsection{Bounding the spectral norm after deflation}\label{subsec:spectral}

In this section, we bound the spectral norm of the matrix $\Ab$ after deflating its top subspace using the converged eigenvectors from the randomized block Krylov algorithm (Algorithm~\ref{alg:krylov}). More formally, let $\Zb_S=\Qb\Vb_S$ be the output of Algorithm~\ref{alg:krylov}. Let $\Pb=\bv{I}-\Zb_S \Zb_S^T$, i.e., the projection matrix onto the subspace orthogonal to $\Zb_S$. Then, we bound $\|\Pb\Ab\Pb \|_2$ based on the fact that $\Zb_S$ must contain the top $p$ eigenvectors of $\Ab$ for some $p \leq l$ as proven in Theorems~\ref{thm:converge} and~\ref{thm:eig_error}, provided the assumptions in those theorems hold for $\Ab$. We first bound the spectral norm after deflating exactly the top $p$ subspace of $\Ab$.

\begin{lemma}[Spectral norm bound after deflation]\label{lem:pap}
Consider the setting of Theorems~\ref{thm:converge} and~\ref{thm:eig_error}. Let $c_1,c'>0$ be some constants. Let $\Zb_p \in \R^{n \times p}$ be as defined in Theorem~\ref{thm:converge}. Then, $$\|(\Ib-\Zb_p\Zb_p^T) \Ab (\Ib-\Zb_p\Zb_p^T) \|_2 \leq \sigma_{p+1}(\Ab)+ \frac{\|\Ab \|_2}{n^{c_1}},$$ with probability at least $1-\frac{1}{n^{c'}}$.
\end{lemma}
\begin{proof}
Let the bounds from Theorem~\ref{thm:converge} and~\ref{thm:eig_error} hold with some constant $c_2>0$. Assume for contradiction, that $\| (\Ib-\Zb_p\Zb_p^T) \Ab (\Ib-\Zb_p\Zb_p^T)\|_2 > \sigma_{p+1}(\Ab)+\frac{\|\Ab \|_2}{n^{c_2/2-1}}$. Then, there exists some eigenvector $\bv{x}$ of $(\Ib-\Zb_p\Zb_p^T) \Ab (\Ib-\Zb_p\Zb_p^T)$ with corresponding eigenvalue $\lambda$ such that $|\lambda|> \sigma_{p+1}(\Ab)+\frac{\|\Ab \|_2}{n^{c_2/2-1}}$. Since any eigenvector corresponding to a nonzero eigenvalue of $(\Ib-\Zb_p\Zb_p^T) \Ab (\Ib-\Zb_p\Zb_p^T)$ must lie in the column space of $\Ib-\Zb_p\Zb_p^T$, they will be orthogonal to $\Zb_p$ and thus, we will have $\Zb_p^T\bv{x}=0$. Since we have $(\Ib-\Zb_p\Zb_p^T)\Ab (\Ib-\Zb_p\Zb_p^T) \bv{x}=\lambda\xb$, we get that $(\Ib-\Zb_p\Zb_p^T)\Ab \xb=\lambda \xb$. Then, we have 
\begin{align}\label{eq:llower}
    \|\Ab \xb\|^2_2 \geq \lambda^2 \geq \sigma^2_{p+1}(\Ab)+ \frac{2\sigma_{p+1}(\Ab)\|\Ab \|_2}{n^{c_2/2-1}}+\frac{\|\Ab \|^2_2}{n^{c_2-2}}.
\end{align}
 Let $\bv{Z}' \in \R^{n \times p+1}$ such that the first $p$ columns of $\bv{Z}'$ are the columns of $\bv{Z}_p$ and $(p+1)$th column is $\bv{x}$. From Theorem~\ref{thm:converge}, we have that $\Ab\Zb_p=\Zb_p\Tilde{\bv{\Lambda}}_p+\bv{E}$ where $\|\bv{E} \|_2 \leq \frac{\| \Ab\|_2}{n^{c_2}}$. Then, we have $\|\bv{E} \|_F \leq \frac{\sqrt{n}\| \Ab\|_2}{n^{c_2}} \leq \frac{\|\Ab \|_2}{n^{c_2-0.5}}$ and thus, we get: 
\begin{align}\label{eq:azp}
    \| \Ab\Zb_p\|_F^2 &\geq (\|\Zb_p\Tilde{\bv{\Lambda}}_p \|_F- \|\bv{E}\|_F)^2 \nonumber \\
    &\geq (\|\Zb_p\Tilde{\bv{\Lambda}}_p \|_F- \frac{\|\Ab\|_2}{n^{c_2-0.5}})^2 \\
    &\geq \|\Zb_p\Tilde{\bv{\Lambda}}_p \|_F^2 +\frac{\|\Ab\|^2_2}{n^{2c_2-1}}-2\|\Zb_p\Tilde{\bv{\Lambda}}_p \|_F\frac{\|\Ab\|_2}{n^{c_2-0.5}} \nonumber\\
    &\geq \|\Zb_p\Tilde{\bv{\Lambda}}_p \|_F^2+\frac{\|\Ab\|^2_2}{n^{2c_2-1}} -\frac{2\|\Ab \|^2_2}{n^{c_2-1}} \nonumber\\
    &=\sum_{i=1}^p \sigma_i^2(\Qb\Qb^T\Ab\Qb\Qb^T)+\frac{\|\Ab\|^2_2}{n^{2c_2-1}}-\frac{2\|\Ab \|^2_2}{n^{c_2-1}},
\end{align}
for some constant $C>0$. The first inequality above follows from the triangle inequality, the second from the bound $\|\bv{E} \|_F \leq \frac{\|\Ab \|_2}{n^{c_2-0.5}}$ and the third from expanding the quadratic expression. The fourth bound follows from $\|\Zb_p\Tilde{\bv{\Lambda}}_p \|_F =\|\Tilde{\bv{\Lambda}}_p \|_F \leq \sqrt{n}\|\Tilde{\bv{\Lambda}}_p \|_2 \leq \sqrt{n}\|\Ab \|_2$ (where we use the fact that $\|\Qb\Qb^T\Ab\Qb\Qb^T \|_2 \leq \|\Ab \|_2$ and $\Zb_p$ is an orthonormal matrix) which gives us $\|\Zb_p\Tilde{\bv{\Lambda}}_p \|_F\|\bv{E}\|_F \leq \frac{\|\Ab \|^2_2}{n^{c_2-1}}$. The final step follows from the fact that $\Zb_p$ and $\Tilde{\bv{\Lambda}}_p$ contain the top $p$ eigenvectors and eigenvalues of $\Qb\Qb^T\Ab\Qb\Qb^T$.
  By the Pythagorean theorem, we have 
 \begin{align}\label{eq:az'}
     \|\Ab\bv{Z}' \|_F^2=\|\Ab\Zb_p \|_F^2+\|\Ab \bv{x} \|_2^2 &> \sum_{i=1}^p \sigma_i^2(\Qb\Qb^T\Ab\Qb\Qb^T)+\frac{\|\Ab\|^2_2}{n^{2c_2-1}}-\frac{2\|\Ab \|^2_2}{n^{c_2-1}}\notag\\
     &+\sigma^2_{p+1}(\Ab)+\frac{2\sigma_{p+1}(\Ab)\|\Ab \|_2}{n^{c_2/2-1}}+\frac{\| \Ab\|^2_2}{n^{c_2-2}},
 \end{align}
where the last inequality follows from the lower bounds in~\eqref{eq:llower} and~\eqref{eq:azp}. From Theorem~\ref{thm:eig_error}, there exists a permutation $S:[p] \rightarrow [p]$ such that for every $i \in [p]$ (and for some constant $C'>0$), $|\lambda_i(\Ab)-\lambda_{S(i)}(\Qb\Qb^T\Ab\Qb\Qb^T)| \leq \frac{\|\Ab \|_2}{n^{c_2}}$ or 
\begin{align*}
    |\lambda^2_i(\Ab)-\lambda^2_{S(i)}(\Qb\Qb^T\Ab\Qb\Qb^T) |\leq \frac{|\lambda_i(\Ab)+\lambda_{S(i)}(\Qb\Qb^T\Ab\Qb\Qb^T)|\|\Ab \|_2}{n^{c_2}} \leq \frac{2\|\Ab \|^2_2}{n^{c_2}},
\end{align*}
where we upper bounded $\lambda_i(\Ab)$ and $\lambda_{S(i)}(\Qb\Qb^T\Ab\Qb\Qb^T)$ by $\|\Ab\|_2$. Thus, we have $\lambda_i^2(\Ab) \leq \lambda^2_{S(i)}(\Qb\Qb^T\Ab\Qb\Qb^T)+\frac{2\|\Ab \|_2^2}{n^{c_2}}$ or $\sigma^2(\Ab) \leq \sigma^2_{S(i)}(\Qb\Qb^T\Ab\Qb\Qb^T)+\frac{2\|\Ab \|_2^2}{n^{c_2}}$ for every $i \in [p]$.  Adding up both sides over $i \in [p]$ we get that $\sum_{i=1}^p \sigma^2_i(\Ab) \leq \sum_{i=1}^p \sigma^2_{i}(\Qb\Qb^T\Ab\Qb\Qb^T)+\frac{2p\|\Ab \|_2^2}{n^{c_2}}\leq \sum_{i=1}^p \sigma^2_{i}(\Qb\Qb^T\Ab\Qb\Qb^T)+\frac{2\|\Ab \|_2^2}{n^{c_2-1}}$. So, we get:
\begin{align}\label{eq:sigmap}
   \sum_{i=1}^{p}\sigma_i^2(\Ab)+\sigma^2_{p+1}(\Ab) &\leq \sum_{i=1}^p \sigma^2_{i}(\Qb\Qb^T\Ab\Qb\Qb^T)+\frac{2\|\Ab \|_2^2}{n^{c_2-1}} + \sigma^2_{p+1}(\Ab) \nonumber\\
   &= \sum_{i=1}^p \sigma^2_{i}(\Qb\Qb^T\Ab\Qb\Qb^T) +\frac{\|\Ab\|^2_2}{n^{2c_2-1}}-\frac{2\|\Ab \|^2_2}{n^{c_2-1}} +\sigma^2_{p+1}(\Ab)+\frac{\| \Ab\|^2_2}{n^{c_2-2}} \nonumber\\
   &+\frac{2\sigma_{p+1}(\Ab)\| \Ab\|_2}{n^{c_2/2-1}}-(\frac{\| \Ab\|^2_2}{n^{c_2-2}}+\frac{\|\Ab\|^2_2}{n^{2c_2-1}}-\frac{4\|\Ab \|^2_2}{n^{c_2-1}}+\frac{2\sigma_{p+1}(\Ab)\| \Ab\|_2}{n^{c_2/2-1}}).
\end{align}
For large enough $n$ and $c_2$, we have $\frac{\| \Ab\|^2_2}{n^{c_2-2}}+\frac{\|\Ab\|^2_2}{n^{2c_2-1}}-\frac{4\|\Ab \|^2_2}{n^{c_2-1}}+\frac{2\sigma_{p+1}(\Ab)\| \Ab\|_2}{n^{c_2/2-1}}>0$. Thus from the lower bound~\eqref{eq:az'} on $\|\Ab\Zb' \|_F^2$ and from~\eqref{eq:sigmap}, we get $\sum_{i=1}^{p+1}\sigma^2_{i}(\Ab) < \|\Ab\Zb' \|_F^2$. But, by the minmax principle for singular values, we have $\max_{\bv{V} \in \R^{n \times (p+1)},\bv{V}^T\bv{V}=\bv{I}}\| \Ab \bv{V}\|_F^2=\sum_{i=1}^{p+1}\sigma_i^2(\Ab)$, which results in a contradiction as $\Zb'$ is an $n \times (p+1)$ orthonormal matrix. Thus, we must have $\| (\Ib-\Zb_p\Zb_p^T) \Ab (\Ib-\Zb_p\Zb_p^T)\|_2 \leq \sigma_{p+1}(\Ab)+\frac{\|\Ab \|_2}{n^{c_2/2-1}}$.
\end{proof}

We now prove our main result which bounds the spectral norm of the deflated matrix $\Pb\Ab\Pb$.

\begin{theorem}\label{thm:sigma_S}
Consider the setting of Theorem~\ref{thm:converge} and let the error bound of Theorem~\ref{thm:converge} hold for some constant $c_3>0$. Let Algorithm~\ref{alg:krylov} be run with the inputs $\Ab \in \R^\n$, block size $l \in [n]$ and $\beta=c_3/2$. Let $S$ be the set of indices as defined in Algorithm~\ref{alg:krylov} such that for any $i \in S$, the eigenvector $\vb_i$ and corresponding eigenvalue $\lambda_i(\Qb^T\Ab\Qb)$ of $\Qb^T\Ab\Qb$ satisfies $\|\Ab\Qb\vb_i-\lambda_i(\Qb^T\Ab\Qb)\Qb\vb_i \|_2 \leq \frac{\|\Ab \|_2}{n^{c_1}}$ for some constant $c_1>0$.  Let $\Pb=\bv{I}-\Zb_S\Zb_S^T$ where $\Zb_S=\Qb\Vb_S$ is the output of Algorithm~\ref{alg:krylov}. Then, for come constants $c_2,c'>0$, we have:
\begin{align*}
   \|\Pb\Ab\Pb \|_2 \leq 2\sigma_{l+1}(\Ab)+\frac{\|\Ab \|_2}{n^{c_2}},
\end{align*}
with probability at least $1-\frac{1}{n^{c'}}$.
\end{theorem}
\begin{proof}

Observe that if $\bv{v}$ is an eigenvector of $\Qb^T \Ab\Qb$ (where $\Qb$ is an orthonormal basis of the Krylov subspace $\mathcal{K}_q$ in line 3 of Algorithm~\ref{alg:krylov}) with corresponding eigenvalue $\lambda$, then $\zb=\Qb \bv{v}$ is an eigenvector of $\Qb \Qb^T\Ab \Qb\Qb^T$ with eigenvalue $\lambda$ as $\Qb$ is an orthonormal matrix. Similarly, for any eigenvector $\zb$ of $\Qb \Qb^T\Ab \Qb\Qb^T$, there is a corresponding eigenvector $\Qb^T\zb$ of $\Qb^T \Ab\Qb$ with the same eigenvalue. Thus, we can interchangeably refer to the eigenvalues and eigenvectors of $\Qb \Qb^T\Ab \Qb\Qb^T$ instead of $\Qb^T \Ab\Qb$ in our proof. Let $\alpha=\max \left(\sigma_{l+1}(\Ab), \frac{\|\Ab \|_2}{n^{c/2}} \right)$ for some constant $c$ as defined in Lemma~\ref{Lem:eig_vec}. We will now prove the Lemma by considering the two cases below:

     \medskip

    \noindent\textbf{Case 1:} Let there be some $k \in [l]$ such that $\sigma_{k}(\Ab) \geq 2\alpha$ and $\sigma_{k+1}(\Ab) < 2\alpha$. Let $p \in [k,l]$ such that $\sigma_p(\Ab) \geq \frac{3}{2}\alpha$, $\sigma_p(\Ab)-\sigma_{p+1}(\Ab) \geq \frac{\|\Ab \|_2}{2n^{c/2+1}}$ as defined in Lemma~\ref{Lem:eig_vec}. Then, from Theorem~\ref{thm:converge}, for any $i \in [p]$, we have $\|\Ab\zb_i-\lambda_i(\Qb\Qb^T\Ab\Qb\Qb^T)\zb_i \|_2 \leq \| \Ab\Zb_p-\Zb_p\Tilde{\bv{\Lambda}}_p\|_F \leq \frac{\|\Ab \|_2}{n^{c_3}}$. Thus, as $\beta=c_3/2$ in Algorithm~\ref{alg:krylov} we have $\frac{\|\Ab \|_2}{n^{c_3}} <\frac{\|\Ab \|_2}{n^{\beta}}$ and the convergence condition in line \ref{alg1:if-condition} of Algorithm~\ref{alg:krylov} is always satisfied for the top $p$ eigenvectors of $\Qb \Qb^T\Ab \Qb\Qb^T$.  So, we must have $[p] \subseteq S$, and $\Zb_S$ must contain at least the top $p$ eigenvectors of $\Qb^T\Ab\Qb$. Thus, observe that using Lemma~\ref{lem:pap}, for some constant $c_4$, we have $\|\Pb \Ab \Pb \|_2 \leq \|(\Ib-\Zb_p\Zb_p^T)\Ab(\Ib-\Zb_p\Zb_p^T) \|_2 \leq \sigma_{p+1}(\Ab)+\frac{\|\Ab \|_2}{n^{c_4}}\leq \sigma_{k+1}(\Ab)+\frac{\|\Ab \|_2}{n^{c_4}}   \leq 2\alpha+\frac{\|\Ab \|_2}{n^{c_4}} \leq 2(\sigma_{l+1}(\Ab)+\frac{\|\Ab \|_2}{n^{c/2}})+\frac{\|\Ab \|_2}{n^{c_4}} \leq 2\sigma_{l+1}(\Ab)+\frac{2\|\Ab \|_2}{n^{c/2}}+\frac{\| \Ab\|_2}{n^{c_4}}$. So, for some constant $c_5>0$ we get:
\begin{align*}
    \|\Pb \Ab \Pb \|_2 \leq  2\sigma_{l+1}(\Ab)+\frac{\|\Ab \|_2}{n^{c_5}}.
\end{align*}

\medskip

    \noindent\textbf{Case 2:} If no such $k \in [l]$ exists, we will have $\sigma_{i}(\Ab) \leq 2\alpha$ for all $i \in [n]$. So, as $\Pb$ is a projection matrix, we have $\|\Pb\Ab\Pb \|_2 \leq \|\Ab \|_2 \leq 2\alpha \leq 2\sigma_{l+1}(\Ab)+\frac{2\|\Ab \|_2}{n^{c/2}}$. This completes the proof after choosingthe constant carefully.
\end{proof}

\subsection{Error Bounds for moment matching with deflation}\label{sec:deflate}

In this section, we prove the final error bounds for Algorithm~\ref{alg:sde}. We first prove the existence of a matrix $\Bb$ that is close in spectral norm to $\Ab$ such that the converged eigenvectors and eigenvalues of $\Qb\Qb^T\Ab\Qb\Qb^T$ as defined in line 7 of Algorithm~\ref{alg:krylov} are a subset of $\Bb$'s true eigenvectors and eigenvalues. First, we state a simplified version of a backward perturbation bound from~\cite{sun1995note}. 

\begin{theorem}[Theorem 3.1 of~\cite{sun1995note}]\label{thm:bkwd_ptr}
Let $\Ab \in \R^{n \times n}$ be a symmetric matrix. Let $\tilde{\bv{X}} \in \R^{n \times l}$ and $\tilde{\bv{\Lambda}} \in \R^{ n\times l}$ (where $\tilde{\bv{\Lambda}}$ is diagonal) be such that $\|\Ab\tilde{\bv{X}}-\tilde{\bv{X}}\tilde{\bv{\Lambda}}\|_2 \leq \Delta$ for some $\Delta>0$. Let $\hat{\bv{R}}=\Ab\tilde{\bv{X}}-\tilde{\bv{X}}\tilde{\bv{\Lambda}}$. Let $\tilde{\bv{X}}=\bv{P}\bv{H}$ be the polar decomposition of $\tilde{\bv{X}}$. Then, there exits a symmetric matrix $\bv{H} \in \R^{n \times n}$ such that $(\Ab+\bv{H})\bv{P}=\bv{P}\tilde{\bv{\Lambda}}$ and $\|\bv{H} \|_F \leq \frac{\sqrt{\|\hat{\bv{R}} \|_F^2+\| \bv{P}^{\perp}\hat{\bv{R}}\|_F^2}}{\sigma_{\min}(\tilde{\bv{X}})}$ where $\bv{P}^{\perp}=\bv{I}-\bv{P}\bv{P}^T$.
\end{theorem}

Using Theorem~\ref{thm:bkwd_ptr}, we state a backward error bound for the deflation algorithm which we will use in the final error analysis.

\begin{lemma}[Wasserstein Error using Backward Error Bound]\label{Lem:bkwd_err}
 Let $\Ab \in \R^\n$ be a symmetric matrix. Let $\Zb_S$ and $\Tilde{\bv{\Lambda}}_S$ be the outputs of Algorithm~\ref{alg:krylov} with $\Ab$ as input. Then there exists a symmetric matrix $\bv{B} \in \R^{n \times n}$ such that 
  \begin{align*}
      \bv{B}\bv{Z}_S=\bv{Z}_S\Tilde{\bv{\Lambda}}_S \text{ and } \|\Ab-\bv{B} \|_2 \leq\frac{\|\Ab \|_2}{n^{\beta-1}},
  \end{align*}
  where $\beta>0$ is the constant defined in Algorithm~\ref{alg:krylov}.
\end{lemma}
\begin{proof}
We will apply the backward perturbation error bound of Theorem~\ref{thm:bkwd_ptr} to prove the existence of $\Bb$. Following Theorem~\ref{thm:bkwd_ptr}, we have $\Tilde{\bv{X}}_1=\bv{Z}_S$ and hence, $\bv{P}_1=\bv{Z}_S$ since $\bv{Z}_S$ is an orthonormal matrix and its polar decomposition is equal to itself. Next, observe that $\sigma_{\min}(\bv{Z}_S)=1$ and $\|\hat{\bv{R}} \|_F^2=\|\bv{A}\bv{Z}_S-\bv{Z}_S \Tilde{\bv{\Lambda}}_S\|^2_F = \sum_{i \in S} \|\Ab\zb_i-\lambda_i(\Qb^T\Ab\Qb)\zb_i \|^2_2 \leq \frac{\|\Ab \|^2_2}{n^{2\beta-1}}$ where the last bound follows from the fact that $\|\Ab\zb_i-\lambda_i(\Qb^T\Ab\Qb)\zb_i \|_2 \leq \frac{\|\Ab \|_2}{n^{\beta}}$ as stated in line \ref{alg1:if-condition} of Algorithm~\ref{alg:krylov}. We also have $\|(\bv{Z}_S)_{\perp}(\bv{Z}_S)^T_{\perp}\hat{\bv{R}} \|_F^2=\|(\bv{I}-\Zb_S\Zb_S^T)(\bv{A}\Zb_S-\bv{Z}_S\Tilde{\bv{\Lambda}}_S) \|_F^2 \leq \|\bv{A}\bv{Z}_S-\bv{Z}_S \Tilde{\bv{\Lambda}}_S\|^2_F \leq \frac{\|\Ab \|^2_2}{n^{2\beta-1}}$ where the last step follows as $\bv{I}-\Zb_S\Zb_S^T$ is a projection matrix. Thus, there exists a symmetric matrix $\bv{H}$ such that $(\bv{A}+\bv{H})\bv{Z}_S=\bv{Z}_S \Tilde{\bv{\Lambda}}_S$ and $\| \bv{H}\|_F \leq \frac{\sqrt{\|\hat{\bv{R}} \|_F^2+\|(\bv{Z}_S)_{\perp}\hat{\bv{R}} \|_F^2}}{\sigma_{\min}(\bv{Z})} \leq \frac{\|\Ab \|_2}{n^{\beta-1}}$. Setting $\bv{B}=\bv{A+H}$ gives us the required matrix.
\end{proof}

We now prove that the output of Algorithm 1 of~\cite{braverman:2022} using the modified moments $\hat{\tau}_i$ as inputs for $i \in N$ (as defined in lines 4-6 of Algorithm~\ref{alg:sde}) must be close in the Wasserstein distance to the spectral density defined by the non-zero eigenvalues of the deflated matrix $\Pb\Ab\Pb$.

\begin{lemma}[Modified SDE bound for deflated matrix]\label{lem:mma}
Let $\Ab \in \R^{n \times n}$ be a rank $n-r $ symmetric matrix for some $r \in [n]$ with spectral density $s_{\Ab}(x)$ such that $|\lambda_1(\Ab)| \leq 1$. Let $\epsilon \in (0,1)$, $N= O\left(\frac{1}{\epsilon}\right)$ and $\tilde{\tau}_1, \ldots, \tilde{\tau}_N$ be the estimates of the top $N$ normalized Chebyshev moments of $\Ab$ estimated using Algorithm 2 of~\cite{braverman:2022}. Define $\hat{\tau}_i=\frac{1}{n-r}(n\tilde{\tau}_i-r\Bar{T}_i(0))$ for $i \in [N]$ where $\Bar{T}_i(0)$ is the $i$\tth normalized Chebyshev polynomial. Let the density function $q(x)$ be the output of Algorithm 1 of~\cite{braverman:2022} with $\hat{\tau}_1, \ldots, \hat{\tau}_N$ as the inputs for the moment estimates. Let $s'_{\Ab}(x)=\frac{1}{n-r}\sum_{j=1}^{n-r}\delta(x-\lambda_j(\Ab))$ be the probability density defined by the top $n-r$ eigenvalues of $\Ab$, $\lambda_1(\Ab), \ldots, \lambda_{n-r}(\Ab)$. Then, with probability at least $1-\delta$, we have:  $$\W_1(s'_{\Ab},q) \leq \frac{n\epsilon}{n-r}.$$ 
Also, estimating the density function $q$ using Algorithm 1 of~\cite{braverman:2022} and also using Algorithm 2 of~\cite{braverman:2022} as a subroutine to estimate the moments, requires $O\left(\frac{b}{\epsilon}\right)$ matrix-vector products where $b=\max(1,\frac{C'}{n\epsilon^2}\log^2 \frac{1}{\epsilon \delta}\log^2 \frac{1}{\epsilon})$. 
\end{lemma}
\begin{proof}
 $\frac{1}{n}\tr{(\Bar{T}_i(\Ab))}=\frac{1}{n}\sum_{j=1}^n\Bar{T}_i(\lambda_j(\Ab))=\frac{1}{n}\sum_{j=1}^{n-r}\Bar{T}_i(\lambda_j(\Ab))+\frac{r}{n}\Bar{T}_i(0)$ for $i \in [N]$ are the top $N$ normalized Chebyshev moments of $s_{\Ab}(x)=\frac{1}{n}\sum_{i=1}^n \delta(x-\lambda_i(\Ab))$. Then, the $i$\tth normalized Chebyshev moment of $s'_{\Ab}(x)$ is given by $\tau'_i= \frac{1}{n-r}\sum_{j=1}^{n-r}\Bar{T}_i(\lambda_j(\Ab))= \frac{n}{n-r}\cdot \frac{1}{n}\sum_{j=1}^{n}\Bar{T}_i(\lambda_j(\Ab))-\frac{1}{n-r} \sum_{j=n-r+1}^{n}\Bar{T}_i(\lambda_j(\Ab))= \frac{n}{n-r}\cdot \frac{1}{n}\tr(\Bar{T}_i(\Ab))-\frac{r}{n-r}\Bar{T}_i(0)$.

By setting $\Delta=\frac{1}{N \ln (eN)}$ (where $N$ is the number of Chebyshev moments estimated), using Lemma 4.2 of~\cite{braverman:2022}, we get that with $\Ab$ as the input, the normalized Chebyshev moment estimates $\hat{\tau}_i$ returned by Algorithm 2 of~\cite{braverman:2022} must satisfy $\left|\tilde{\tau}_i-\frac{1}{n}\tr{(\Bar{T}_i(\Ab))} \right| \leq \frac{1}{N \ln (eN)}$ for all $i \in [N]$.  So, we have: 
\begin{align}\label{Eq:tau}
    \left|\hat{\tau}_i-\tau'_i \right| &= \left|\frac{n}{n-r}\tilde{\tau}_i-\frac{r}{n-r}\Bar{T}_i(0)-\frac{1}{n-r}\sum_{j=1}^{n-r}\Bar{T}_i(\lambda_j(\Ab)) \right| \nonumber\\
    &= \frac{n}{n-r} \left|\tilde{\tau}_i- \frac{1}{n}\tr(\Bar{T}_i(\Ab)) \right| \nonumber\\
    &\leq \frac{n}{(n-r)N \ln (eN)},
\end{align}
where the second step follows from the fact that $\lambda_{n-r+1}(\Ab)=\ldots =\lambda_n(\Ab)=0$ by assumption. The density $q$ returned by Algorithm 1 of~\cite{braverman:2022} is defined on a $(d+1)-$length evenly spaced grid $X_d=[-1,-1+\frac{2}{d},\ldots,1-\frac{2}{d},1]$ for $d=\lceil N^3/2 \rceil$. Let  $z_{s'}=[\tau'_1/1, \tau'_2/2,\ldots,\tau'_N/N]$ and $z_{q}=\mathcal{T}_N^{d+1} q$ where $(\mathcal{T}_N^{d+1}) \in \R^{N \times (d+1)}$ such that $(\mathcal{T}_N^{d+1})_{ij}=\bar{T}_{i}(-1+\frac{2j}{d})/i$ for $i \in [N]$, $j \in \{0,1,\ldots,d\}$. Also, let $z=[\hat{\tau}_1, \hat{\tau}_2/2, \ldots, \hat{\tau}_N/N].$ Then, using triangle inequality, we have 
\begin{align*}
    \| z_q- z_{s'}\|_1=\|z_q-z \|_1+\|z- z_{s'}\|_1.
\end{align*}
Now, $\|z- z_{s'}\|_1 \leq \sum_{i=1}^N (\hat{\tau}_i-\tau'_i)/i \leq \frac{n}{(n-r)N \ln (eN)} \cdot H_n \leq \frac{n}{N(n-r)}$ where we use~\eqref{Eq:tau} to bound $\hat{\tau}_i-\tau'_i$ and $H_n$ is the $n$\tth harmonic number. Next, consider the following distribution $q^*$ on $X_d$ as defined in Lemma 3.4 of~\cite{braverman:2022}:
\begin{align*}
    q^*(x)=\frac{1}{n-r}\sum_{j=1}^{n-r}\delta(x-\argmin_{p \in X_d} |p-\lambda_j(\Ab)|).
\end{align*}
$q^*$ is the distribution corresponding to moving the mass from each $\lambda_j(\Ab)$ to its nearest point on the grid $X_d$. We have $W_1(s'_{\Ab},q^*)\leq \frac{1}{d}$ due to the earth mover's distance interpretation
of the Wasserstein-$1$ distance. Let $z_{q^*}=[\langle q^*,\bar{T}_1 \rangle,\ldots, \langle q^*,\bar{T}_N \rangle/N]$. Now, from Line 3 of Algorithm 1 of~\cite{braverman:2022}, $q$ is defined as the density which minimizes $\|\mathcal{T}_N^{d+1}q-z \|_1$. Thus, we have $\|\mathcal{T}_N^{d+1}q-z  \|_1 \leq \|z_{q^*}-z \|_1$. Then, following the proof of Lemma 3.4 of~\cite{braverman:2022} exactly, we can upper bound $\|z_{q^*}-z \|_1$ by $\frac{2}{N}$ using the properties of Chebyshev polynomials which gives us
$\|\mathcal{T}_N^{d+1}q-z \|_1 \leq \frac{2}{N}$. Thus, we get $ \| z_q- z_{s'}\|_1 \leq \frac{3n}{(n-r)N}$. Finally, following the proof of Lemma 3.4 of~\cite{braverman:2022}, Lemma 3.1 of~\cite{braverman:2022} gives us the final bound on $W_1(s'_{\Ab},q)$.  

Note that the number of matrix-vector products is obtained by setting $\Delta=\frac{1}{N \ln (eN)}$ in Lemma 4.2 of~\cite{braverman:2022}.
\end{proof}

We now prove the final error bound for Algorithm~\ref{alg:sde}. We will first show using Lemma~\ref{Lem:bkwd_err} that there exists a matrix $\Bb$ such that its spectral density is very close to that of $\Ab$ and $\Zb_S$ and $\Tilde{\bv{\Lambda}}_S$ (the output of Algorithm~\ref{alg:krylov}) are a subset of the eigenvectors and eigenvalues of $\Bb$. Then, it is enough to estimate the spectral density of $\Bb$. We already know the eigenvalues $\Tilde{\bv{\Lambda}}_S$ of $\Bb$. So we just need to estimate the spectral density of its remaining eigenvalues (which is equal to the part of the spectral density of the deflated matrix $\Pb\Bb\Pb$ corresponding to the eigenvalues with eigenvectors which lie in the subspace orthogonal to $\Zb_S$). To do this, we first observe that the spectral density of $\Pb\Bb\Pb$ is again close to that of $\Pb\Ab\Pb$. So, it is enough to estimate the spectral density of the eigenvalues of $\Pb\Ab\Pb$ corresponding to its non-deflated eigenvectors. We show this is exactly estimated by Algorithm 1 of~\cite{braverman:2022} by appropriately modifying the Chebyshev moments of $\Pb\Ab\Pb$ to account for the zero eigenvalues corresponding to the deflated eigenvectors.

\begin{reptheorem}{thm:sde1}[SDE with Explicit Deflation]
Let $\Ab \in \R^\n$ be a symmetric matrix. For any $\epsilon, \delta \in (0,1)$, $l \in [n]$, and constants $c,c_1>0$, Algorithm~\ref{alg:sde} performs $O\left(l\log n+\frac{b}{\epsilon} \right)$ matrix-vector products with $\Ab$ where $b=\max\left(1,\frac{1}{n\epsilon^2}\log^2 \frac{n}{\epsilon}\log^2 \frac{1}{\epsilon} \right)$ and computes a probability density function $\q$ such that $$\W_1(s_{\Ab},\text{q}) \leq \epsilon\sigma_{l+1}(\Ab) +\frac{\|\Ab \|_2}{n^{c}},$$ with probability at least $1-\frac{1}{n^{c_1}}.$
\end{reptheorem}
\begin{proof}

    Let $\Zb_S \in \R^{n \times |S|}$ and $\Tilde{\bv{\Lambda}}_S \in \R^{|S| \times |S|}$ be the output of Algorithm~\ref{alg:krylov}, in line \ref{alg2:q-def} of Algorithm~\ref{alg:sde}. Recall that $\Zb_S=\bv{Q}\Vb_S$ where $\Qb$ is an orthonormal basis of the Krylov subspace of $\Ab$ (Definition~\ref{def:krylov}) with a starting block $\bv{X} \in \R^{n \times l}$ containing random $\mathcal{N}(0,1)$ Gaussian entries and the columns of $\Vb_S$ are all the eigenvectors of $\Qb^T\Ab\Qb$ with the corresponding eigenvalues along the diagonal of $\Tilde{\bv{\Lambda}}_S \in \R^{|S| \times |S|}$ such that for any $i \in S$, $\|\Ab\Qb\vb_{i}- (\Tilde{\bv{\Lambda}}_S)_{ii}\Qb\vb_{i}\|_2 \leq \frac{\|\Ab \|_2}{n^{\beta}}$ for some constant $\beta$ as defined in line \ref{alg1:if-condition} of Algorithm~\ref{alg:krylov} (here $\vb_{i}$ is an eigenvector of $\Qb^T\Ab\Qb$ corresponding to eigenvalue $\lambda_i(\Qb^T\Ab\Qb)$). Equivalently, for every $i \in S$, $\zb_i=\Qb\vb_i$ is an eigenvector of $\Qb \Qb^T\Ab \Qb\Qb^T$ corresponding to an eigenvalue $\lambda_i(\Qb \Qb^T\Ab \Qb\Qb^T)$. Now, using Lemma~\ref{Lem:bkwd_err}, there exists a symmetric matrix $\bv{B} \in \R^{n \times n}$ such that $\bv{B}\bv{Z}_S=\bv{Z}_S \Tilde{\bv{\Lambda}}_S$ and $\|\Ab-\bv{B} \|_2 \leq\frac{\|\Ab \|_2}{n^{\beta-1}}$. Using Weyls' inequality (Fact~\ref{fact:weyl}), for all $i\in [n]$ we have
    \begin{equation}\label{Eq:weyl1}
        |\lambda_i(\Ab)-\lambda_i(\Bb)| \leq \frac{\|\Ab \|_2}{n^{\beta-1}}.
    \end{equation}
    Then, $\frac{1}{n}\sum_{i=1}^n |\lambda_i(\Ab)-\lambda_i(\Bb)| \leq \frac{\|\Ab \|_2}{n^{\beta-1}}$. Using the earth mover's interpretation, this implies that: \begin{equation}\label{Eq:w1}
        \W_1(s_{\Ab},s_{\Bb}) \leq \frac{\| \Ab \|_2}{n^{\beta-1}},
    \end{equation}
    where $s_{\Bb}$ is the spectral density of $\Bb$. Since $\beta=c$ as defined in line 2 of Algorithm~\ref{alg:sde}, it is enough to bound $\W_1(s_{\Bb}, q)$.
    
    Since $\bv{B}\bv{Z}_S=\bv{Z}_S \Tilde{\bv{\Lambda}}_S$, the eigenvalues $\lambda_i(\Qb^T\Ab\Qb)$ such that $i \in S$ are also eigenvalues of $\Bb$. Let $S_1 \subseteq [n]$ be the set of indices of the eigenvalues of $\Bb$ such that $|S_1|=|S|$ and for each $i \in S_1$, there exists some  $j \in S$ such that $\lambda_i(\Bb)=\lambda_j(\Qb^T\Ab\Qb)$. Let $[-L_1,L_1]$ contain the support of both $s_{\Bb}$ and $q$ (the output of Algorithm~\ref{alg:sde}). Let $\Pb=\bv{I}-\Zb_S\Zb_S^T$. Note that, from Algorithm~\ref{alg:sde}, we also have $q(x)=\frac{|S|q_1(x)+(n-|S|)q_2(x))}{n}$ where $q_1(x)=\frac{1}{|S|}\sum_{i =1}^{|S|}\delta(x-(\Tilde{\bv{\Lambda}}_S)_{ii})=\frac{1}{|S|}\sum_{i \in S}\delta(x-\lambda_i(\Qb^T\Ab\Qb))$ (as defined in Line \ref{alg2:q-def} of Algorithm~\ref{alg:sde}) and $q_2(x)=q'_2(x/L)$ for $x \in [-L,L]$ where $q'_2$ is a density supported on $[-1,1]$ which is the output of Algorithm 1 of~\cite{braverman:2022} with the modified approximate moments of $\frac{1}{L}\Pb\Ab\Pb$ as defined in lines 4-6 of Algorithm~\ref{alg:sde} (where $\|\Pb\Ab\Pb \|_2 \leq L\leq 2\|\Pb\Ab\Pb \|_2$). $\W_1(s_{\Bb},q)$ can be written as: 
    \begin{align*}
       \W_1(s_{\Bb},q) &= \sup_{h \in \text{1-Lip}}\int_{-L_1}^{L_1} h(x)(s_{\Bb}(x)-q(x)) dx.
    \end{align*}
Let $h^*(x)$ be the 1-Lipschitz function which maximizes the integral above. Then, using the definitons of $s_{\Bb}$ and $q$ we have:
\begin{align}\label{Eq:dist1}
    \W_1(s_{\Bb},q) &= \int_{-L_1}^{L_1} h^*(x) \left(\frac{1}{n}\sum_{i=1}^n\delta(x-\lambda_i(\Bb)) -\frac{1}{n}\sum_{j \in S}\delta(x-\lambda_j(\Qb^T\Ab\Qb))-\frac{n-|S|}{n} q_2(x) \right) dx \nonumber \\
       &= \frac{|S|}{n}\underbrace{\int_{-L_1}^{L_1} h^*(x) \left(\frac{1}{|S|}\sum_{i \in S_1}\delta(x-\lambda_i(\Bb)) - \frac{1}{|S|}\sum_{j \in S} \delta(x-\lambda_j(\Qb^T\Ab\Qb)) \right) dx}_{\text{$I_1$}} \nonumber \\
       &+  \frac{n-|S|}{n}\underbrace{\int_{-L_1}^{L_1} h^*(x) \left( \frac{1}{n-|S|}\sum_{i \notin S_1} \delta(x -\lambda_i(\Bb))-q_2(x)  \right) dx}_\text{$I_2$}.
\end{align}

We have $I_1=0$ by definition of the sets $S_1$ and $S$. Thus, we have $\W_1(s_{\Bb},q)=\frac{n-|S|}{n}I_2$. We now bound $I_2$. 

 Observe that any eigenvalue $\lambda_i(\Bb)$ such that $i \notin S_1$, is also an eigenvalue value of $\Pb\Bb\Pb$. To see this, let $\xb_i$ be the corresponding eigenvector of $\Bb$. Then, observe that $\Pb\Bb\Pb\xb=\Pb\Bb\xb=\lambda_i(\Bb)\xb$ as $\Zb_S^T\xb=0$. Also both $\Pb\Bb\Pb$ and $\Pb\Ab\Pb$ have $|S|$ (where $|S|=|S_1|$) eigenvalues equal to $0$ (with the corresponding eigenvectors being the columns of $\Zb$). Let $S_2 \subseteq [n]$ be the set of indices such that $|S_2|=n-|S_1|=n-|S|$ and for every $j \in S_2$, there exists some $i \in [n]\setminus S_1$ such that $\lambda_j(\Pb\Bb\Pb)=\lambda_i(\Bb)$. So, we can write $I_2$ as:
\begin{align}\label{Eq:i2_eq}
    I_2 &=\int_{-L_1}^{L_1} h^*(x) \left( \frac{1}{n-|S|}\sum_{i \notin S_1} \delta(x -\lambda_i(\Bb))-q_2(x)  \right) dx \nonumber\\
    &= \int_{-L_1}^{L_1} h^*(x) \left( \frac{1}{n-|S|}\sum_{i \in S_2} \delta(x -\lambda_i(\Pb\Bb\Pb))-q_2(x)  \right) dx\nonumber \\
    &= \underbrace{\int_{-L_1}^{L_1} h^*(x) \left( \frac{1}{n-|S|}\sum_{i \in S_2} \delta(x -\lambda_i(\Pb\Bb\Pb))-\frac{1}{n-|S|}\sum_{i \in S_2} \delta(x -\lambda_i(\Pb\Ab\Pb))  \right) dx}_{\text{$t_1$}} \nonumber \\
    &+ \underbrace{\int_{-L_1}^{L_1} h^*(x) \left( \frac{1}{n-|S|}\sum_{i \in S_2} \delta(x -\lambda_i(\Pb\Ab\Pb))-q_2(x)  \right) dx}_{\text{$t_2$}}.
\end{align}
Using the fact that $\Pb$ is a projection matrix and Lemma~\ref{lem:pap}, we again have that $\|\Pb\Ab\bv{P}-\Pb\Bb\bv{P} \|_2 \leq \|\Pb(\Ab-\Bb)\bv{P} \|_2 \leq \| \Ab-\Bb\|_2 \leq \frac{\|\Ab \|_2}{n^{\beta-1}}$. Thus, using Weyls' inequality (Fact~\ref{fact:weyl}) we have $|\lambda_i(\Pb\Bb\Pb)-\lambda_i(\Pb\Ab\Pb)| \leq \frac{\|\Ab \|_2}{n^{\beta-1}}$ for $i \in [n]$. We can bound $t_1$ using the earth movers' interpretation of the Wasserstein-$1$ distance as
\begin{align}\label{eq:t1bound}
    t_1 \leq \frac{1}{n-|S|}\sum_{i \in S_2} |\lambda_i(\Pb\Bb\Pb)-\lambda_i(\Pb\Ab\Pb)| \leq \frac{\|\Ab \|_2}{n^{\beta-1}}.
\end{align}
We bound $t_2$ next. Since all eigenvalues of $\frac{1}{L}\Pb\Ab\Pb$ are in $[-1,1]$ and $\Pb\Ab\Pb$ has at least $n-|S|$ eigenvalues equal to 0 as described previously (corresponding to $\lambda_i(\Pb\Ab\Pb)$ where $i \notin S_2$), according to Lemma~\ref{lem:mma}, the density $\q'_2$ returned by Algorithm 1 of~\cite{braverman:2022} in line 6 of Algorithm~\ref{alg:sde} satisfies the guarantee 
\begin{align}\label{eq:spap}
    \W_1\bigg(s'_{\Pb\Ab\Pb},\q'_2\bigg) \leq \frac{n \epsilon}{(n-|S|) },
\end{align}
with probability at least $1-\delta$ where $s'_{\Pb\Ab\Pb}(x)=\sum_{i \in S_2}\frac{\delta(x-\frac{1}{L}\lambda_i(\Pb\Ab\Pb))}{n-|S|}$. Since $q'_2$ is supported on $[-1,1]$, we set $q_2(x)=q'_2(x/L)$ in line 7 of Algorithm~\ref{alg:sde} when $x \in [-L,L]$ and $q_2(x)=0$ otherwise so that $q_2$ is now supported on $[-L,L]$. So we get that
    \begin{align*}
        t_2  &= \int_{-L}^{L} h^*(x) \left( \frac{1}{n-|S|}\sum_{i \in S_2} \delta(x -\lambda_i(\Pb\Ab\Pb))-q_2(x)  \right)dx  \\
        &= \int_{-1}^{1} h^*(Lx) \left( \frac{1}{n-|S|}\sum_{i \in S_2} \delta(Lx -\lambda_i(\Pb\Ab\Pb))-q_2(Lx)  \right) dx  \\
        &= \int_{-1}^{1} h^*(Lx) \left( \frac{1}{n-|S|}\sum_{i \in S_2} \delta \left(x -\frac{\lambda_i(\Pb\Ab\Pb)}{L} \right)-q'_2(x)  \right) dx \\
        &\leq \sup_{h_1 \in \text{1-Lip}}L\int_{-1}^{1} h_1(x) \left( \frac{1}{n-|S|}\sum_{i \in S_2} \delta \left(x -\frac{\lambda_i(\Pb\Ab\Pb)}{L} \right)-q'_2(x)  \right) dx \\
        &= L\W_1\bigg(s'_{\Pb\Ab\Pb},\q'_2\bigg) \\
        &\leq \frac{n \epsilon}{(n-|S|)} L.
    \end{align*}
     In the first step above, we use the fact that the densities are supported on $[-L,L]$ instaed of $[-L_1,L_1]$. In the second step, we rescale the integrals from $[-L,L]$ to $[-1,1]$. In the third step, we use the fact that $q_2(Lx)=q'_2(x)$ and $\frac{1}{n-|S|}\sum_{i \in S_2} \delta(Lx -\lambda_i(\Pb\Ab\Pb))=\frac{1}{n-|S|}\sum_{i \in S_2} \delta \left(x -\frac{\lambda_i(\Pb\Ab\Pb)}{L} \right)$ when $x \in [-1,1]$. The fourth step follows from the fact that $h^*(Lx)$ is an $L-$Lipschitz function and the final two steps follows from the definition of the Wasserstein-$1$ distance and~\eqref{eq:spap}. 
    
    Since $\|\Pb\Ab\Pb \|_2 \leq L\leq 2\|\Pb\Ab\Pb \|_2$ from Theorem~\ref{thm:sigma_S}, for some constant $c_2>0$, we have $L \leq 2\|\Pb\Ab\Pb \|_2 \leq 2\sigma_{l+1}(\Ab)+\frac{2\|\Ab \|_2}{n^{c_2}}$. Thus, we get:
\begin{align}\label{Eq:deflate}
    t_2 \leq \frac{\epsilon n}{2(n-|S|)} \cdot \left(2\sigma_{l+1}(\Ab)+\frac{2\|\Ab \|_2}{n^{c_2}} \right).
\end{align}
From~\eqref{Eq:i2_eq}, using the bounds on $t_1$ and $t_2$ from~\eqref{eq:t1bound} and~\eqref{Eq:deflate}, we get $I_2 \leq \frac{\|\Ab \|_2}{n^{\beta-1}}+\frac{\epsilon n}{2(n-|S|)} \cdot \left(2\sigma_{l+1}(\Ab)+\frac{2\|\Ab \|_2}{n^{c_2}} \right) $. Using the bound on $I_2$ and the fact that $I_1=0$ in~\eqref{Eq:dist1}, we get:
\begin{align*}
     \W_1(s_{\Bb},q) \leq  \frac{|S|}{n}I_1+ \frac{n-|S|}{n}I_2 \leq \epsilon\sigma_{l+1}(\Ab)+\frac{\epsilon\|\Ab \|_2}{n^{c_2}}+\frac{\|\Ab \|_2}{n^{\beta-1}}.
\end{align*}
Finally, using triangle inequality and~\eqref{Eq:w1}, we get that, for some suitable chosen constant $c_1>0$:
\begin{align*}
    \W_1(s,q) &\leq \W_1(s_{\Ab},s_{\Bb})+ \W_1(s_{\Bb},q) \leq \epsilon \sigma_{l+1}(\Ab)+\frac{\epsilon\|\Ab \|_2}{n^{c_2}}+\frac{2\|\Ab \|_2}{n^{\beta-1}} \\
    &\leq \epsilon \sigma_{l+1}(\Ab)+\frac{\|\Ab \|_2}{n^{c_1}}.
\end{align*}
This completes the proof.

   \medskip

\paragraph{Matrix vector products.} Line \ref{alg2:line1} of Algorithm~\ref{alg:sde} calls Algorithm~\ref{alg:krylov} with $\Ab$ as input which uses $O(l\log n)$ matrix vector products with $\Ab$ to form the Krylov subspace in Line \ref{alg1:krylov-subspace} of Algorithm~\ref{alg:krylov} (since the Krylov subspace has depth $O(\log n)$ and block size $l$). By setting the error parameter $\epsilon$ in Algorithm 2 of~\cite{Musco:2015} to $0.5$, we get an output $L'$ in $O(\log n)$ iterations such that $0.5\|\Pb\Ab\Pb \|_2 \leq L' \leq \|\Pb\Ab\Pb \|_2 $. Then, we can set $L=2L'$ such that $\|\Pb\Ab\Pb \|_2 \leq L \leq 2\|\Pb\Ab\Pb \|_2$ as defined in Line \ref{alg2:line3} of Algorithm~\ref{alg:sde}. Also, Algorithm 2 of~\cite{braverman:2022} is called with the deflated and scaled matrix $\frac{1}{L}\Pb\Ab\Pb$ as input in line 4 of Algorithm~\ref{alg:sde}. According to Lemma~\ref{lem:mma}, this uses $O\left(\frac{b}{\epsilon}\right)$ matrix vector products with $\Pb\Ab\Pb$ where $b=\max\left(1,\frac{1}{n\epsilon^2}\log^2 \frac{n}{\epsilon}\log^2 \frac{1}{\epsilon}\right)$ i.e. $O\left(\frac{b}{\epsilon}\right)$ matrix vector products with $\Ab$ , since $\delta \ge \frac{1}{n^{c'}}$ (for some $c'\geq 0$). Finally, Line \ref{alg2:line6} of Algorithm~\ref{alg:sde} calls Algorithm 2 of~\cite{braverman:2022} which uses at most $O(\log n)$ extra matrix vector products.  Thus, the total number of matrix vector products is $O(l\log n+\frac{b}{\epsilon})$ where $b=\max(1,\frac{C'}{n\epsilon^2}\log^2 \frac{n}{\epsilon }\log^2 \frac{1}{\epsilon})$.
\end{proof}

We now state a simple corollary which shows that we can get Wasserstein-$1$ error depending on the Schatten $1$-norm of $\Ab$ by appropriatley balancing the errors from deflation and moment matching in Algorithm~\ref{alg:sde}.
\begin{repcorollary}{cor:schatten}

Let $\Ab \in \R^\n$ be symmetric. For any $\epsilon \in (0,1)$ and some constant $c>0$, there exists an algorithm that performs $O\left(\frac{\sqrt{n} \log n}{\epsilon} + \sqrt{n} \log^4 n \right)$ matrix vector products with $\Ab$ and computes $M$ such that, with probability at least $1-\frac{1}{n^c}$, $|M-\norm{\bv A}_1| \le \epsilon \norm{\bv A}_1$.  

\end{repcorollary}
\begin{proof}
Set the block size as $l=\frac{\sqrt{n}}{\epsilon}$ and the error parameter $\epsilon'$ to $O(\frac{1}{\sqrt{n}})$ of Algorithm~\ref{alg:sde}. Let $q$ be the output of Algorithm~\ref{alg:sde} using $O\left(l \log n+\frac{b}{\epsilon'} \right)=O \left(\frac{\sqrt{n}}{\epsilon}\log n + \sqrt{n} \log^4 n \right)$ matrix vector products with $\Ab$. Then, from Theorem~\ref{thm:sde1}, we have $\W_1(\s_{\Ab},\text{q}) \leq  \frac{\sigma_{l'+1}(\Ab)}{\sqrt{n}} +\frac{\|\Ab \|_2}{n^{c/4}} \leq \frac{\|\Ab \|
_1}{l\sqrt{n}}+\frac{\|\Ab \|_2}{n^{c/4}}=\frac{\epsilon \| \Ab\|_1}{n}+\frac{\|\Ab \|_2}{n^{c/4}} \leq \frac{2\epsilon \| \Ab\|_1}{n}$. Then, using $q$, we can construct a list of n values $[\tilde{\lambda}_1, \ldots , \tilde{\lambda}_n]$ in time linear in $n$ and $\frac{1}{\epsilon}$ such that $\sum_{i=1}^n |\lambda_i-\tilde{\lambda}_i| \leq \frac{2\epsilon \|\Ab \|_1 \cdot n}{n} \leq2\epsilon \|\Ab \|_1 $ (see~\cite{Cohen-SteinerKongSohler:2018}, theorem B.1 in~\cite{braverman:2022}). Adjusting $\epsilon$ by constant factors gives us the final bound.
\end{proof}

\section{Analysis of Stochastic Lanczos Quadrature}\label{sec:slq}

In this section, we give our error analysis of  Stochastic Lanczos Quadrature (SLQ) (Algorithm~\ref{alg:slq}) by showing that it implicitly performs a deflation of the input matrix. Our analysis shows that for a symmetric $\Ab \in \R^{\n}$, SLQ achieves an error bound of $\epsilon \sigma_{l+1}(\Ab)+\Tilde{O}(\frac{l\|\Ab \|_2}{n})$ using $O(l \log \frac{1}{g_{\min}}+\frac{1}{\epsilon}\log \frac{n \cdot \kappa}{\delta})$ matrix vector products with $\Ab$ for any $l \in [n]$, $\epsilon = \tilde \Omega(1/\sqrt{n})$, failure probability $\delta \in (0,1)$ and where $g_{\min}$ (minimum singular value gap) and $\kappa$ (condition number) are as stated in Theorem~\ref{thm:slq}. Hence, it almost matches the error bounds of the explicit deflation and moment matching algorithm (Algorithm~\ref{alg:sde}) we described in the previous section (up to the additive $\frac{l\|\Ab \|_2}{n}$ factor). Roughly, we show that the large magnitude eigenvalues of $\Ab$ are estimated almost exactly in $O(l \log \frac{1}{g_{\min}}+\frac{1}{\epsilon}\log \frac{n \cdot\kappa}{\delta})$ iterations of Lanczos (Algorithm~\ref{alg:lanczos-slq}). Estimating the spectral density of the small magnitude eigenvalues requires a further $\Tilde{O}(1/\epsilon)$ iterations. Then, we give a simple variant of SLQ, which we call the Variance Reduced SLQ (VRSLQ) which sets the weights of the converged eigenvalues in the final distribution correctly so that we end up getting an error bound of $\epsilon \sigma_{l+1}(\Ab)+\Tilde{O}(\frac{l\sigma_{l+1}(\Ab)}{n})$. We note that the SLQ algorithm described in this paper (Algorithm~\ref{alg:slq}) uses only one random starting vector for simplicity, though in practice, we can get better concentration when the resulting distribution is averaged over multiple  random starting vectors. 

This section is organized as follows. In Section~\ref{sec:slq1}, we first derive a loose error bound of $\epsilon \|\Ab \|_2$ for SLQ using a simple moment matching based analysis. In Section~\ref{sec:slq2}, we derive error bounds for approximating the top eigenvalue and eigenvector of a matrix using the Lanczos algorithm. In Section~\ref{sec:slq3}, we give a tighter error bound of  $\epsilon \sigma_{l+1}(\Ab)+\Tilde{O}(\frac{l\|\Ab \|_2}{n})$ for SLQ by showing that it implicitly performs deflation and approximates the top eigenvalues (by using the error bounds developed in Section~\ref{sec:slq2}) followed by moment matching. In Section~\ref{sec:slq4} we describe the variance reduced version of SLQ and then give the error bounds. 

\textbf{Notations:} Throughout this section, $\mathcal{U}(\mathcal{S}^{n-1})$ denotes a uniform distribution on the unit sphere of dimension $n$. Also, for two random variables $X$ and $Y$, $X:\stackrel{d}{=}Y$ implies they have the same distribution.

\subsection{SLQ bounds via Moment Matching}\label{sec:slq1}

In this section, we derive an error bound for SLQ (Algorithm~\ref{alg:slq}) using a simple moment matching based argument. We begin by describing the Lanczos algorithm (Algorithm~\ref{alg:lanczos-slq}) on which SLQ is based. The Lanczos algorithm iteratively constructs an orthonormal basis of the Krylov subspace $[\bv g, \Ab\bv g, \ldots, \Ab^{m-1}\bv g]$ generated by an input matrix $\Ab$ and a random starting vector $\bv g$ of appropriate dimensions such that $\Tb = \Qb^T\Ab\Qb$ where $\Tb$ is tridiagonal and where $\Qb$ is an orthonormal basis of the Krylov subspace that is computed by the Lanczos algrorithm. 

\begin{algorithm}[H] 
\caption{Lanczos algorithm (\cite{lanczos1952solution, golub2009matrices, chen2021analysis})}
\label{alg:lanczos-slq}
\begin{algorithmic}[1]
\Require{Symmetric $\bv A \in \mathbb{R}^{n\times n}$, starting vector $\bv g \in \R^n$, number of iterations $m$.}
\State Set $\bv{q}_1 = \bv g / \|\bv g\|_2$, $\alpha_1 = \bv{q}_1^T\Ab\bv{q}_1$, $\tilde{\bv{q}}_2 = \Ab\bv{q}_1 - \alpha_1\bv{q}_1$ and initialize $\Tb \in \bv 0^{m\times m}$ with $\Tb_{11} = \alpha_1$.
\For{$i=2,\ldots,m$}
\State Let $\eta_{i-1} = \|\tilde{\bv{q}}_i\|_2$.
\State Compute $\bv{q}_i = \tilde{\bv{q}}_i / \eta_{i-1}$ (i.e., normalize $\tilde{\bv{q}}_i$ to obtain Lanczos vectors $\bv{q}_i$).
\State Compute $\alpha_i = \bv{q}_i^T\Ab\bv{q}_i$.
\State Set $\tilde{\bv{q}}_{i+1} = \Ab\bv{q}_i - \alpha_i\bv{q}_i - \eta_{i-1}\bv{q}_{i-1}$.
\State Set $\Tb_{ii} = \alpha_i$, and $\Tb_{i,i-1} = \Tb_{i-1,i} = \eta_{i-1}$.
\EndFor
\State \Return $\Tb$, $\Qb$ where $\Qb \in \R^{n \times m}$ is a matrix whose $i$\textsuperscript{th} column is $\bv{q}_i$.
\end{algorithmic}
\end{algorithm}

We have the following well known identity for the Lanczos algorithm (for eg. see~\cite{golub2009matrices}), which we prove here for completeness.

\begin{lemma}(\cite{golub2009matrices})
    \label{lem:lanczos-id1}
    Consider Algorithm~\ref{alg:lanczos-slq}  run with input $\Ab\in\R^\n$, starting vector $\bv g \in \R^{n}$, and number of iterations $m$. Let $\Tb \in \R^{m \times m},\Qb \in \R^{n \times m}$ be the outputs of the algorithm. Then, for any $k\in [m-1]$, we have $\Ab^k\bv g = \Qb\Tb^k\Qb^T\bv g$.
\end{lemma}
\begin{proof}
The Krylov subspace generated by Algorithm~\ref{alg:lanczos-slq} is $\mathcal{K}_m=[\bv{g}, \Ab\bv{g}, \Ab^2\bv{g}, \ldots, \Ab^{m-1}\bv{g}]$ after $m$ iterations. Since $\Qb$ is an orthonormal basis of $\mathcal{K}_m$, the columns of $\mathcal{K}_m$ are spanned by the columns of $\Qb$. Let $\xb=\|\bv g \|_2\bv e_1$ (recall $\bv e_1$ is the first standard basis vector with a 1 in the first position). So, we have $\bv g=\Qb\xb$.

    We show that for any $p\in [m]$, $\Ab^p\bv g = \Qb\Tb^p\Qb^T\bv g$, via induction. Observe that for $p=1$, $\Qb\Tb\Qb^T\bv g = \Qb\Tb\Qb^T\Qb\xb = \Qb\Qb^T\Ab\Qb\xb = \Qb\Qb^T\Ab\bv g = \Ab\bv g$, where in the last equality we use the fact that $\Ab\bv g$ is a column of $\mathcal{K}_m$, and so it is spanned by the columns of $\Qb$. This shows that the base case for our induction is true.
    As the inductive hypothesis, assume $\Ab^k\bv g = \Qb\Tb^k\Qb^T\bv g$ for some $k \in [m-2]$. Then observe that $\Qb\Tb^{k+1}\Qb^T\bv g = \Qb\Tb\Qb^T\Qb\Tb^{k}\Qb^T\bv g = \Qb\Tb^{k}\Qb^T\Ab^k\bv g = \Qb \Qb^T \Ab \Qb\Qb^T\Ab^k\bv g $, where in the third equality we use our inductive hypothesis. Since $\Ab^k \bv g$ is a column of $\mathcal{K}_m$, $\Ab^k \bv g$ must be spanned by the columns of $\Qb$. So we have, $\Qb \Qb^T \Ab \Qb\Qb^T\Ab^k\bv g =  \Qb \Qb^T \Ab \Ab^k \bv g=\Qb \Qb^T \Ab^{k+1}\bv g=\Ab^{k+1}\bv g$, where in the last step we  used the fact that $\Ab^{k+1}\bv g$ is also spanned by the columns of $\Qb$. This completes our proof.
\end{proof}

We now state the SLQ algorithm for spectral density estimation.
\begin{algorithm}[H] 
\caption{Stochastic Lanczos Quadrature (adapted from \cite{chen2021analysis})}
\label{alg:slq}
\begin{algorithmic}[1]
\Require{Symmetric $\bv A \in \mathbb{R}^{n\times n}$, number of iterations $m (\leq n)$.}
\State Sample  $\bv g \sim \mathcal{U}(\mathcal{S}^{n-1})$.
\State Run Lanczos (Algorithm~\ref{alg:lanczos-slq}) with inputs $\Ab, \bv g$ and $m$ to compute symmetric tridiagonal matrix $\Tb \in \R^{m\times m}$, orthonormal basis $\Qb \in \R^{n \times m}$. Let the eigenvectors of $\Tb$ be $\vb_1, \ldots, \vb_m$.\label{alg4:line2}
\State Set $f (x) = \sum_{j=1}^m w_j^2\delta(x-\lambda_j(\bv{T}))$ where $w_j = \bv{v}_j^T \bv{e}_1$ where $\bv{e}_1 \in \R^m$ is the first standard basis unit vector (i.e. a 1 in the first position and a $0$ everywhere else)
\State \Return $f (x) $
\end{algorithmic}
\end{algorithm}

We give an error bound for SLQ  by showing that the normalized Chebyshev moments of the output density are approximately equal to the normalized Chebyshev moments of the spectral density of the input matrix $\bv A$. We start by proving a lemma showing that the $j$th Chebyshev moment of the output of SLQ $f$ for any $j \leq m-1$ is exactly given by $\bv{g}^T\Bar{T}_j(\bv{A})\bv{g}$ where $\bv{g}$ is the random starting vector.

\begin{lemma}
\label{lem:eq-slq-lanczos}
Consider Algorithm~\ref{alg:slq} run with input $\Ab\in\R^\n$, number of iterations $m$, and sampled vector $\bv g \sim \mathcal{U}(\mathcal{S}^{n-1})$ in line 1. Let $f(x)=\sum_{i=1}^m w_i^2\delta(x-\lambda_i(\bv{T}))$ be the output of the algorithm. Then, for any $j \in \{0,1,2,\ldots, m-1\}$, $\langle \Bar{T}_j,f \rangle=\bv{g}^T\Bar{T}_j(\bv{A})\bv{g} $.

\end{lemma}

\begin{proof}
    Let $L=\|\Tb \|_2$. The $j^{th}$ normalized Chebyshev moment of $f$ is given by $\langle\Bar{T}_j, f\rangle=\int_{-L}^{L}\Bar{T}_j(x)f(x)dx=\int_{-L}^{L}\Bar{T}_j(x) \sum_{i=1}^m w_i^2\delta(x-\lambda_i(\bv{T}))dx=\sum_{i=1}^m w_i^2\Bar{T}_j(\lambda_i(\Tb))=\sum_{i=1}^m \Bar{T}_j(\lambda_i(\Tb))\bv{e}_1^T \bv{v}_i\bv{v}_i^T\bv{e}_1= \bv{e}_1^T\left(\sum_{i=1}^m \Bar{T}_j(\lambda_i(\Tb))\bv{v}_i\bv{v}_i^T\right) \bv{e}_1=\bv{e}_1^T \Bar{T}_j(\bv{T})\bv{e}_1$.

    Let $\Qb$ be the orthonormal basis computed by the Lanczos algorithm (Algorithm~\ref{alg:lanczos-slq}) in line~\ref{alg4:line2} of Algorithm~\ref{alg:slq}. From Lemma \ref{lem:lanczos-id1} we know that $\Ab^p\bv{g}=\Qb\Tb^p\Qb^T\bv{g}$ for any $p \in [m-1]$. Thus, $\Bar{T}_j(\Ab)\bv{g}=\Qb\Bar{T}_j(\Tb)\Qb^T\bv{g}$ for any $j \leq m-1$. Note that $\Qb \bv{e}_1=\bv{g}$ since $\bv{g}$ is set as the first column of $\Qb$ in Algorithm~\ref{alg:lanczos-slq} ($\bv g$ is a random unit vector). Thus, we get $\bv{e}_1^T \Bar{T}_j(\bv{T})\bv{e}_1=  \bv{g}^T\Qb\Bar{T}_j(\Tb)\Qb^T\bv{g}=\bv{g}^T\Bar{T}_j(\bv{A})\bv{g}$. 
\end{proof}

We next prove that the $i^{th}$ normalized Chebyshev moment of the output of SLQ $f$ is almost equal to the $i$th normalized Chebyshev moment of the SDE of $\Ab$ via the error bounds for a modified hutchinson's trace estimator~\cite{meyer2021hutch++} that uses a random vector on the unit sphere as opposed to a gaussian or a random sign vector. We note that the analysis of hutchinson's using a random vector on the unit sphere is different from the usual analysis which assumes the elements of the random vector are independent and identically distributed with zero mean. We also note that we require that SLQ use a random vector on the unit sphere as opposed to a say, a gaussian vector or a random sign vector is because such a vector has the same distribution as a normalized gaussian vector. As we will see, this helps us leverage the rotational invariance of the gaussian distribution to derive error bounds for the convergence of eigenvectors and eigenvalues while still ensuring that the random starting vector for Lanczos is a unit vector. Hence, we first prove error bounds for hutchinson's trace estimator using a single random vector on the unit sphere in Lemma~\ref{Lem:hutch_unit}.  

\begin{lemma}[Hutchinson's with random vectors on the unit sphere]\label{Lem:hutch_unit}
Let $\Ab \in \R^{n \times n}$, $\delta \in (0,1/2]$ and $\bv g \sim \mathcal{U}(\mathcal{S}^{n-1})$. Then, assuming $n \geq \Omega(\log (1/\delta))$ and for some fixed constant $C>0$, with probability at least $1-\delta$, we have that:
\begin{align*}
    \bigg\lvert \frac{1}{n}\tr{(\Ab)} -\bv g^T \Ab \bv g \bigg\rvert \leq \frac{C \log (1/\delta)}{n}\|\Ab \|_F.
\end{align*}
\end{lemma}
\begin{proof}
We have $\bv{g} :\stackrel{d}{=}\frac{ \bv y_j}{\sqrt{\sum_{j=1}^l  \bv y^2_j}}$ where $\bv y \in \R^n$ is such that 
 $\bv y_j \sim \mathcal{N}(0,1)$ for $j \in [n]$~\cite{chen2021analysis}. Now, using triangle inequality, we have that :
\begin{align}\label{eq:triangle1}
\bigg\lvert \frac{1}{n}\tr{(\Ab)} - \bv g^T \Ab \bv g \bigg\rvert \leq \bigg\lvert \frac{1}{n}\tr{(\Ab)} -\frac{1}{n} \bv y^T \Ab \bv y \bigg\rvert + \bigg\lvert  \bv g^T \Ab \bv g -\frac{1}{n} \bv y^T \Ab \bv y \bigg\rvert.     
 \end{align}
 We will bound the terms individually. The first term is just the error bound for the hutchinson's estimator using a random gaussian vector. The second term can be bounded as the norm of $\bv y$, which is just a Chi-suared distribution, can be shown to concentrate around $n$ using standard concentartion bounds. First observe that, from Lemma 2 of~\cite{meyer2021hutch++}, we have 
 \begin{align}\label{eq:hutch1}
     \bigg\lvert \frac{1}{n}\tr{(\Ab)} -\frac{1}{n} \bv y^T \Ab \bv y \bigg\rvert \leq \frac{1}{n}\bigg\lvert \tr{(\Ab)} - \bv y^T \Ab \bv y \bigg\rvert \leq \frac{\log (1/\delta)}{n}\|\Ab \|_F,
 \end{align}
 with probability at least $1-\delta$. This bounds the first term in~\eqref{eq:triangle1}. Now,  we bound the second term. From~\eqref{eq:hutch1}, we have:
 \begin{align}\label{eq:hutch}
     | \bv{y}^T \Ab \bv{y} | \leq \log (1/\delta)\|\Ab \|_F+\tr{(\Ab)} \leq \log (1/\delta)\|\Ab \|_F +\sqrt{n}\|\Ab \|_F.
 \end{align}
 with probability at least $1-\delta$. Here, the second inequality follows from the fact that $\tr{(\Ab)} \leq \sqrt{n}\|\Ab \|_F$. Observe that by the concentration properties of the Chi-squared distribution, we have $\big \lvert \frac{1}{n}\|\bv{y} \|_2^2-1 \big \rvert \leq  \sqrt{\frac{\log (l/\delta)}{n}}$, with probability at least $1-\delta$, assuming $ n \geq \Omega(\log (l /\delta))$~\cite{Wainwright_2019}. Rearranging, we get 
 \begin{align}\label{eq:chisq}
     \bigg \lvert \frac{n}{\|\bv y \|^2_2} -1\bigg \rvert \leq 2\sqrt{\frac{\log (l/\delta)}{n}}.
 \end{align}
 Thus, we have:
 \begin{align*}
     \bigg\lvert  \bv g^T \Ab \bv g -\frac{1}{n} \bv y^T \Ab \bv y \bigg\rvert &= \bigg\lvert  \frac{\bv y^T \Ab \bv y}{\|\bv{y} \|_2^2} -\frac{1}{n}\bv y^T \Ab \bv y \bigg\rvert \\
     &\leq \frac{1}{n}  |\bv y^T \Ab \bv y| \bigg \lvert  \frac{n}{\|\bv{y} \|_2^2}-1 \bigg \rvert \\
     &\leq \frac{1}{n}  (\log (1/\delta)\|\Ab \|_F+\sqrt{n}\|\Ab \|_F )  \cdot 2\sqrt{\frac{\log (1/\delta)}{n}}\\
     &\leq \bigg(\frac{2\log^{3/2} (1/\delta) }{n\sqrt{n}}+\frac{2\sqrt{\log (1/\delta)}}{n} \bigg) \|\Ab \|_F \\
     &\leq \frac{3 \log (1/\delta)}{n}\|\Ab \|_F.
 \end{align*} 
 where in the second step, we used the triangle inequality and in the third step, we used the upper bounds from~\eqref{eq:hutch} and~\eqref{eq:chisq}. In the final step, we used the fact that $n \geq \Omega(\log (1/\delta))$. Thus, using the upper bounds from above and from~\eqref{eq:hutch1} in~\eqref{eq:triangle1}, we finally get $\bigg\lvert \frac{1}{n}\tr{(\Ab)} -\bv g^T \Ab \bv g \bigg\rvert \leq  \frac{C \log (1/\delta)\| \Ab\|_F}{n}$ for some constant $C>0$.
\end{proof}

\begin{lemma} 
\label{lem:abs-slq-true-A}
Consider the setting of Lemma~\ref{lem:eq-slq-lanczos}. For $n \geq \Omega(\log (1/\delta))$, with probability at least $1-\delta$, $\left|\langle \Bar{T}_i, f \rangle - \langle\Bar{T}_i, s_{\Ab}\rangle\right| \leq \frac{C \log (m/\delta)}{\sqrt{n}}$ for all $i \in \{0,1,2,\ldots, m-1\}$ where $C>0$ is a large constant.
\end{lemma}
\begin{proof}

Observe that $\langle\Bar{T}_i, s_{\Ab}\rangle = \frac{1}{n}\tr(\Bar{T}_j(\Ab))$. From Lemma \ref{lem:eq-slq-lanczos} we have $\langle \Bar{T}_i, f \rangle =  \bv g^T\Bar{T}_i(\Ab)\bv g$ for $i \in \{0,1,2,\ldots, m-1\}$. From Lemma~\eqref{Lem:hutch_unit}, since $n \geq \Omega(\log^2(1/\delta))$, for every $i \in \{0,1,\ldots, m-1 \}$, with probability at least $1-\frac{\delta}{m}$ for a constant $C$, we have:
\begin{align*}
    \left|\frac{1}{n}\tr(\Bar{T}_i(\Ab)) - \bv g^T\Bar{T}_i(\Ab)\bv g \right| \leq \frac{C
    \log (m/\delta)}{n}\|T_i(\Ab) \|_F \leq \frac{C\log(m/\delta)}{\sqrt{n}},
\end{align*}
where in the second inequality, we used the fact that $\|T_i(\Ab) \|_2 \leq 1$ since $\|\Ab\|_2 \leq 1$. Applying a union bound for all $i \in \{0,1,\ldots, m-1 \}$ gives the bound.
\end{proof}

We now state the final result of this section which gives the error bound for SLQ. 

\begin{theorem}\label{thm:slq_main}
    Let $\Ab\in\R^\n$ be a symmetric matrix. Let $f(x)$ be the output of Algorithm \ref{alg:slq} with input $\Ab$ and $m=O\left(\frac{1}{\epsilon}\right) $. Then, for some constant $C$ and for $\epsilon,\delta \in (0,1)$, we have $$W_1(s_{\Ab}, f) \leq \epsilon \|\Ab \|_2+\frac{C\log (1/\epsilon \delta) \log (1/\epsilon)}{\sqrt{n}}  \|\Ab \|_2,$$ with probability at least $1-\delta$. Also, Algorithm~\ref{alg:slq} performs $m = O\left(\frac{1}{\epsilon}\right)$ matrix vector products with $\Ab$.
\end{theorem}
\begin{proof}
 Assume that we run Algorithm~\ref{alg:slq} with $\Bb=\frac{1}{\|\Ab \|_2}\Ab$ as input. Then, let $s_{\Bb}$ be the spectral density of $\Bb$ and let $f_{\Bb}$ be the output of Algorithm~\ref{alg:slq} after $m$ iterations. Then, observe that
 \begin{align*}
     \W_1(s_{\Ab},f)= \|\Ab \|_2 \cdot W_1(s_{\Bb},f_{\Bb}).
 \end{align*}
 Since $\| \Bb\|_2 \leq 1$, as stated in Lemma~\ref{lem:eq-slq-lanczos}, the symmetric tridiagonal matrix $\Tb$, which is the output of Lanczos with input matrix $\Bb$ and starting vector $\bv{g}$ in Line \ref{alg4:line2} of Algorithm~\ref{alg:slq}, can be written as $\bv{T}=\Qb^T\Bb\Qb$ where $\Qb$ is an orthonormal matrix. Thus, $\|\Tb\|_2 = \| \Qb^T\Bb\Qb\|_2 \leq \|\Bb \|_2 \leq 1$. So, the support of density function $f_{\Bb}$ output by Algorithm~\ref{alg:slq} is in $[-1,1]$. Using Lemma 3.1 of \cite{braverman:2022} for any two distributions $s_{\Bb},f_{\Bb}$ with support in $[-1,1]$, we have
    \begin{align*}
        \W_1(s_{\Bb}, f_{\Bb}) \leq \frac{36}{m-1} + 2\sum_{i=1}^{m-1}\frac{|\langle\Bar{T}_i, s_{\Ab} \rangle - \langle\Bar{T}_i, f\rangle|}{i}.
    \end{align*}
    From Lemma \ref{lem:abs-slq-true-A}, $|\langle\Bar{T}_i, s_{\Bb}\rangle - \langle\Bar{T}_i, f_{\Bb}\rangle| \leq \frac{C \log (m/\delta)}{\sqrt{n}}$ with probability $1-\delta$ for $i \in \{0,1,2,\ldots, m-1\}$ as long as $\delta \in \big(\frac{1}{e^n},1 \big)$ (due to the assumption $n \geq \Omega(\log (l/n))$). By setting  $m=O(1/\epsilon)$, this gives us, for constants $C_1$ and $C_2$, $\W_1(s_{\Bb}, f_{\Bb}) \leq C_1\epsilon+\frac{C_2\log (1/\epsilon \delta) \log (1/\epsilon)}{\sqrt{n}} $ from the equation above. Finally, we get $\W_1(s_{\Ab},f) \leq \|\Ab \|_2\W_1(s_{\Bb},f_{\Bb}) \leq C_1\epsilon\| \Ab\|_2+\frac{C_2\log (1/\epsilon \delta) \log (1/\epsilon)}{\sqrt{n}}\|\Ab \|_2 $ 

     Since the Lanczos algorithm is run for $m=O\left(\frac{1}{\epsilon}\right)$ iterations, the number of matrix vector products with $\Ab$ is $O\left(\frac{1}{\epsilon}\right)$. 
\end{proof}

We also note that the second term $\frac{\log (1/\epsilon \delta) \log (1/\epsilon)}{\sqrt{n}}$ in the error bound $\epsilon\| \Ab\|_2+\frac{\log (1/\epsilon \delta) \log (1/\epsilon)}{\sqrt{n}} \| \Ab\|_2$ can be made smaller ($< \epsilon \|\Ab \|_2$) by averaging the resulting distributions over multiple random starting vectors in Algorithm~\ref{alg:slq} instead of a single random starting vector. However, this would complicate the analysis for the improved deflation based bounds for SLQ that we derive in the subsequent section as it would require carefully analyzing the convergence of the `average' distribution for different starting vectors for SLQ. Hence, we do the analysis with a single random starting vector.

\subsection{Error Bounds for Lanczos}\label{sec:slq2}

Since SLQ uses the Lanczos algorithm (Algorithm~\ref{alg:lanczos-slq}) as a subroutine, we now derive eigenvalue and eigenvector approximation error bounds for the Lanczos algorithm, which is a Krylov method with a single random starting vector. We will use some results and proof techniques developed in~\cite{meyer:2024} to derive our bounds. We start by describing the critical observation in~\cite{meyer:2024} that the span of the Krylov subspace generated with a single vector as the starting block is the same as the span of the Krylov subspace with a large starting block with fewer iterations. Let Krylov subspace generated by Lanczos after $q$ iterations with starting vector $\bv{g}$ be $$\mathcal{K}_q(\Ab,\bv g)=[\bv{g}, \ldots, \Ab^{q-1}\bv{g}].$$ Note that here we overload the notation as in Section~\ref{sec:explicit}, we had defined $\mathcal{K}_q(\Ab,\bv g)$ in terms of $\Ab\Ab^T$ and $\Ab \bv g$ while here we define it in terms of $\Ab$ and $\bv g$. The Lanczos algorithm finds an orthonormal basis $\Qb$ of $\mathcal{K}_q(\Ab,\bv g)$ such that $\Tb=\Qb^T\Ab\Qb$ where $\Tb$ is a tridiagonal matrix. We are interested in the bounding the error between the eigenvalues of $\Tb$ and the true eigenvalues of $\Ab$. Since the eigenvalues of $\Tb$ are the same up to a rotation of the orthonormal basis $\Qb$ of the Krylov subspace, the eigenvalues estimated by the Lanczos algorithm depends only on the span of the Krylov subspace $\mathcal{K}$ generated after $q$ iterations and not on the specific $\Qb$ found by Algorithm~\ref{alg:lanczos-slq}. Let $\bv{S}_l$ be such that:
\begin{align}\label{Eq:span1}
    \bv{S}_l = \left[\bv{g}, \Ab\bv{g}, \Ab^2\bv{g}, \ldots, \Ab^{l-1}\bv{g} \right].
\end{align}
Then observe that:
\begin{align}\label{Eq:span2}
    \text{span}(\mathcal{K}_{q}(\Ab, \bv g)) = \text{span}\left( \left[\bv{S}_l, \Ab\bv{S}_l, \Ab^2\bv{S}_l, \ldots, \Ab^{q-l}\bv{S}_l \right] \right).
\end{align}
So, the span of the Krylov subspace generated by Lanczos with $\bv{g}$ as the starting vector after $q$ iterations matches the span of the Krylov subspace generated after $q-l+1$ iterations with $\bv{S}_l$ as the starting block.  So, it is enough to analyze the orthonormal basis of $\mathcal{K}_{q-l+1}(\Ab, \bv{S}_l)$. Let $\bv{Q}$ be the orthonormal basis for $\mathcal{K}_{q-l+1}(\Ab, \bv{S}_l)$ found by the Lanczos algorithm (Algorithm~\ref{alg:lanczos-slq}). Similar to Lemma~\ref{lem:sin_bound}, we first want to show that $\Qb$ approximately spans the top subspace of $\Ab$. However, note that the proof of Lemma~\ref{lem:sin_bound} relies on Lemma~\ref{thm:angle_bound} which assumes that for the starting block $\Xb \in \R^{n \times l}$, we have $\rank(\bv{V}_l^T\Xb) =l$. While this is true for a random Gaussian starting block $\Xb \in \R^{n \times l}$, the starting block $\bv{S}_l$ in~\eqref{Eq:span1} is far from being completely random as its columns are highly correlated. We will first prove that $\text{rank}(\bv{S}_l^T\bv{U}_l)=l$ (where $\Ub_l \in \R^{n \times l}$ contains the first $l$ columns of $\Ub$) by following the approach presented in the proof of Theorem 3 in~\cite{meyer:2024}. Then, we will apply Lemma~\ref{thm:angle_bound} to bound the projection error of the top subspace of $\Ab$ onto $\Qb\Qb^T$. 

\medskip

\noindent\textbf{Gap Dependence.} We note that the number of iterations of Lanczos will depend logarithmically on the inverse of the minimum relative singular value gap $g_{\min}=\min_{i \in [l]}\frac{\sigma_i(\Ab)-\sigma_{i+1}(\Ab)}{\sigma_i(\Ab)}$ among the top $l$ singular values of $\Ab$, similar to the bounds in~\cite{meyer:2024}. We need $g_{\min}>0$ to prove $\rank(\Ub_l^T\bv{S}_l) =l$ by using the fact that a degree $l-1$ polynomial cannot be exactly $0$ of $l$ distinct points. Note that if $g_{\min}=0$, then $\Qb$ will not converge to the the top subspace. For example, when $\Ab$ is an identity matrix, $\Qb$ only spans the starting vector (which is always an eigenvector of the identity matrix) and Lanczos never finds the other eigenvectors. In general, it is reasonable to expect that for most matrices, $g_{\min}$ should be at least a small constant. We also note that the bounds can be made gap independent by a random perturbation analysis as in~\cite{meyer:2024}. 

\begin{lemma}\label{lem:slq_sin}
For a symmetric matrix $\bv{A} \in \R^{n \times n}$ such that $\text{rank}(\Ab) >l$, let $\Qb$ be an orthonormal basis of the Krylov subspace $\mathcal{K}_{q-l+1}(\Ab, \bv{S}_l)$ where $ \bv{S}_l = \left[\bv{g}, \Ab\bv{g}, \Ab^2\bv{g}, \ldots, \Ab^{l-1}\bv{g} \right]$ and $\bv g \sim \mathcal{U}(\mathcal{S}^{n-1})$. Let $g_{\min}=\min_{i \in [l]}\frac{\sigma_i(\Ab)-\sigma_{i+1}(\Ab)}{\sigma_i(\Ab)}$. Let $\alpha= \max \left(\sigma_{l+1}(\Ab), \frac{\|\Ab \|_2g^{c/4}_{\min}}{n^{c/4}} \right)$ for some large constant $c>0$ and let $k \in [l]$ such that $\sigma_{k}(\Ab) \geq 2\alpha$ and $\sigma_{k+1}(\Ab) \leq 2\alpha$. Then, for any $\epsilon \in (0,1]$, $\delta \in (0,1)$ $\kappa=\frac{\|\Ab \|_2}{2\alpha}$ and $q=O(l\log (\frac{1}{g_{\min}})+\frac{1}{\epsilon}\log (\frac{n \cdot \kappa}{\delta})) $, we have $$\|\Ub_{k} - \Qb\Qb^T\Ub_{k}\|_F^2 \leq \frac{g^{cl}_{\min}}{(n\cdot \kappa)^{c/\epsilon}}, $$ with probability $1-\delta$.
\end{lemma}
\begin{proof}
Our proof will utilize the results and proof the techniques from Lemma~\ref{thm:angle_bound}, Lemma~\ref{lem:sin_bound} and Theorem 3 of~\cite{meyer:2024}. We will first prove that $\text{rank}(\bv{S}_l^T\bv{U}_l)=l$ (where $\Ub_l \in \R^{n \times l}$ contains the first $l$ columns of $\Ub$) by following the approach presented in the proof of Theorem 3 in~\cite{meyer:2024}. Noet that we can't directly use Theorem 3 from~\cite{meyer:2024} as it is stated for a gaussian starting vector while our starting vector is random on the unit sphere. Observe that for any $\bv{x} \in \R^l$, $\bv{S}_l \bv{x}= \hat{p}(\bv{A})\bv{g}$ for some degree $l-1$ polynomial $\hat{p}$ with coefficients determined by the entries in $\bv{x}$. Also, by the rotational invariance of the gaussian distribution, we have $(\bv{U}_l^T\bv{g})_i:\stackrel{d}{=}\frac{y_i}{\sqrt{\sum_{j=1}^l y^2_j}}$ where $y_i \sim \mathcal{N}(0,1)$ for $i \in [l]$.  Then, $$\bv{U}_l^T\bv{S}_l \bv{x}=\bv{U}_l^T \hat{p}(\bv{A}) \bv{g} = \bv{U}_l^T \bv{U} \hat{p}(\bv{\Lambda})\bv{U}^T \bv{g}:=\hat{p}(\bv{\Lambda}_l) \frac{\bv{y}}{\| \bv{y}\|_2},$$ where $\bv{\Lambda}_l$ contains the top $l$ eigenvalues of $\Ab$ on its diagonal. By Lemma 4 of~\cite{meyer:2024}, we have $\min_{i \in [l]} y^2_i \geq \frac{2\delta^2}{\pi l^2}$ with probability at least $1-\delta$. Then, we can bound the numerator above as:
\begin{equation}\label{Eq:denom}
     \left\|\hat{p}(\bv{\Lambda}_l) \bv{y} \right\|_2^2 = \sum_{i=1}^l (\hat{p}(\lambda_i(\Ab)))^2y_i^2 \geq \frac{2\delta^2}{\pi l^2 }\sum_{i=1}^l(\hat{p}(\lambda_i(\Ab)))^2.
\end{equation}
Since $\hat{p}$ has degree $l-1$, and none of the eigenvalues are repeated (as $g_{\min}>0$), $\left\|\hat{p}(\bv{\Lambda}_l) \bv{y} \right\|_2^2 >0$. We also have $\|\bv y \|_2 > 0$ with probability 1. Thus, we get that $\|\bv{U}_l^T\bv{S}_l \bv{x} \|_2^2 >0$ for any $\bv{x}$. So, $\sigma_{\min}(\bv{U}_l^T\bv{S}_l) >0$ and hence, $\bv{U}_l^T\bv{S}_l$ is invertible which means $\text{rank}(\bv{S}_l^T\bv{U}_l)=l$. Note that we also have $\sigma_k(\Ab) \geq 2\alpha \geq 2\sigma_{l+1}(\Ab)$ and so $k \in [l]$. So, we can apply Lemma~\ref{thm:angle_bound} to bound $\|\sin \bv{\Theta}(\Qb,\Ub_k) \|_F^2$. Without loss of generality, assume that $q$ is odd. Let $\phi(x)$ be a gap amplifying polynomial of degree $q+1$ consisting of only even powers as defined in Lemma 4.5 of~\cite{drineas2018structural} with parameters $\alpha=\sigma_{l+1}(\Ab)$ and gap $\gamma=\frac{\sigma_{k}(\Ab)}{\sigma_{l+1}(\Ab)}-1 \geq 1$. Note that for even $q$, we can similarly define an amplifying polynomial of degree $q$ consisting of only even powers. For any $i \in [k]$, $\phi(\sigma_i(\Ab)) \geq \sigma_i(\Ab) >0$ as $\sigma_i(\Ab) \geq \sigma_k(\Ab) \geq (1+\gamma)\sigma_{l+1}(\Ab)=(1+\gamma)\alpha$ by our choice of parameters for the gap-amplifying polynomial. Thus, $\phi(\bv{\Lambda}_p)$ is non-singular for any $p \in [l]$. Let $\bv{\Phi}=\Ub\phi(\bv{\Lambda})\Ub^T\bv{S}_l$. The columns of $\bv{\Phi}$ lie in $\text{span}(\mathcal{K}_{q-1})$~\eqref{Eq:span2} and thus $\range(\bv{\Phi})\subseteq \range(\mathcal{K}_{q-1}) = \range(\bv Q)$. Also, $\bv{\Phi}\bv{\Phi}^{\dag}$ and $\bv Q \bv Q^T$ are the orthogonal projectors onto $\range(\bv{\Phi})$ and $\range(\mathcal{K}_{q-1})$ respectively. So, following the proof of Lemma~\ref{thm:angle_bound} we get:
\begin{align}\label{Eq:angle_bound1}
   \|\Ub_k-\Qb\Qb^T\Ub_k \|_F^2 \leq \|\phi(\bv{\Lambda}_{l,\perp}) \|^2_2 \| \phi(\bv{\Lambda}_k)^{-1}\|^2_2\|\bv{U}^T_{l,\perp}\bv{S}_l(\bv{U}^T_{l}\bv{S}_l)^{-1} \|^2_2.
\end{align}
Since $q$ is odd, $q+1$ is even, giving $\phi(\lambda_i(\Ab))=\phi(-\lambda_i(\Ab))=\phi(\sigma_i(\Ab))$ for $i \in [n]$. So, $\phi(\bv{\Lambda}_{l,\perp})=\phi(\bv{\Sigma}_{l,\perp})$ and $\phi(\bv{\Lambda}_k)=\phi(\bv{\Sigma}_k)$. Following the proof of Lemma~\ref{lem:sin_bound} (which in turn uses Lemma 4.5 of~\cite{drineas2018structural} based on the properties of $\phi$), we get the bounds 
\begin{align}\label{eq:phil}
    \|\phi(\bv{\Lambda}_{l,\perp})\|_2=\|\phi(\bv{\Sigma}_{l,\perp}) \|_2 \leq \frac{4\sigma_{l+1}(\Ab)}{2^{(q+1)\min (\sqrt{\gamma},1)}}.
\end{align} and 
\begin{align}\label{eq:phik}
    \|\phi(\bv{\Lambda}_k)^{-1} \|_2=\|\phi(\bv{\Sigma}_k)^{-1} \|_2 \leq \sigma_k(\Ab).
\end{align}
So, the final step is to bound $\|\bv{U}^T_{l,\perp}\bv{S}_l(\bv{U}^T_{l}\bv{S}_l)^{-1} \|^2_2 $. We will bound this by following the proof technique presented in Theorem 3 of~\cite{meyer:2024} to bound the same quantity. We just give an outline of the proof and skip the details since the quantity is the same. Observe that we have:
\begin{align}
   \|\bv{U}^T_{l,\perp}\bv{S}_l(\bv{U}^T_{l}\bv{S}_l)^{-1} \|^2_2  &=\max_{\bv{x}} \frac{\|\bv{U}^T_{l,\perp}\bv{S}_l(\bv{U}^T_{l}\bv{S}_l)^{-1} \bv{x}\|_2}{\|\bv{x} \|_2} \nonumber \\
   &\leq \max_{\bv{x}} \frac{\|\bv{U}^T_{l,\perp}\bv{S}_l\bv{x} \|_2}{\|\bv{U}^T_{l}\bv{S}_l\bv{x} \|_2}=\max_{\text{deg}(\hat{p}) \leq l-1} \frac{\|\bv{U}^T_{l,\perp}\hat{p}(\Ab) \bv{y}\|_2^2}{\|\bv{U}^T_{l} \hat{p}(\Ab)\|_2^2}\notag\\
   &=\max_{\text{deg}(\hat{p}) \leq l-1} \frac{\|\hat{p}(\bv{\Lambda}_{l,\perp})\bv{y}\|_2^2}{\|\hat{p}(\bv{\Lambda}_{l})\bv{y} \|_2^2}.
\end{align}
The denominator is already bounded from below in~\eqref{Eq:denom}. We now bound the numerator following the proof in~\cite{meyer:2024}. Note that  $y^2_i \leq 1+4\log (1/\delta)$ for all $i \in [l]$ with probability at least $1-\delta$ by standard concentration bounds for chi-squared random variables~\cite{wainwright2019high}. Then, by a union bound, $\max_i y_i \leq 5\log (n/\delta)$ for $n>2$.  We thus have:
\begin{align}\label{eq:num}
   \|\hat{p}(\bv{\Lambda}_{l,\perp})\bv{y}\|_2^2 \leq 5\log (n/\delta) \sum_{i=1}^n (\hat{p}(\lambda_i(\Ab)))^2 \leq 5n\log (n/\delta)\max_{i \in [n]} (\hat{p}(\lambda_i(\Ab)))^2
\end{align}
Then, combining~\eqref{eq:num} and~\eqref{Eq:denom}, we get:
\begin{align*}
    \|\bv{U}^T_{l,\perp}\bv{S}_l(\bv{U}^T_{l}\bv{S}_l)^{-1} \|^2_2 \leq \frac{5\pi n l^2\log (n/\delta)}{2\delta^2}\frac{ \max_{i \in [n]} (\hat{p}(\lambda_i(\Ab)))^2}{\sum_{i=1}^l(\hat{p}(\lambda_i(\Ab)))^2}
\end{align*}
We can then bound $\frac{ \max_{i \in [n]} (\hat{p}(\lambda_i(\Ab)))^2}{\sum_{i=1}^l(\hat{p}(\lambda_i(\Ab)))^2} \leq \frac{l}{g^{4l}_{\min}}$ by expanding $\hat{p}$ as a Lagrange interpolating polynomial over $\sigma_1^2(\Ab), \ldots, \sigma_{l}^2(\Ab)$ in exactly the same way as in the proof of Theorem 3 in~\cite{meyer:2024}. We finally get 
\begin{align}\label{Eq:ul}
    \|\bv{U}^T_{l,\perp}\bv{S}_l(\bv{U}^T_{l}\bv{S}_l)^{-1} \|^2_2 \leq \frac{5 \pi n l^3}{2 g^{4l}_{\min}\delta^2} \log (\frac{n}{\delta}).
\end{align}
Then, using the bounds~\eqref{Eq:ul},~\eqref{eq:phik} and~\eqref{eq:phil} on the right hand side of~\eqref{Eq:angle_bound1} and using the fact that $\gamma =\frac{\sigma_k(\Ab)}{\sigma_{l+1}(\Ab)}-1 \geq 1$ (by assumption) we get that for $q=O(l\log (\frac{1}{g_{\min}})+\frac{1}{\epsilon}\log (\frac{n \cdot \kappa}{\delta})) $ (as long as the constant $c>0$ is large enough):
\begin{align*}
    \|\Ub_k-\Qb\Qb^T\Ub_k \|_F^2 \leq O \left(\frac{ \sigma^2_{l+1}(\Ab)}{\sigma^2_{k}(\Ab)}\cdot \frac{nl^3 \log (n/\delta)}{2^{(2q+1)\min(\sqrt{\gamma},1 )} g^{4l}_{\min} \delta^2} \right) \leq \frac{g^{cl}_{\min}}{(n\cdot \kappa)^{c/\epsilon}}.
\end{align*}
\end{proof}

Based on the results from Lemma~\ref{lem:slq_sin}, we can now generalize the error bounds for the randomized block krylov method in Section~\ref{subsec:errblk} to the single vector Lanczos Algorithm~\ref{alg:lanczos-slq}. We will utilize the gap between the eigenvalues ($g_{\min}>0$) to give stronger convergence guarantees for the top $k$ eigenvalues and eigenvectors of $\bv{Q}\bv{Q}^T\bv{A}\bv{Q}\bv{Q}^T$. We first state a stronger version of Theorem~\ref{thm:eig_error} which gives shows that the large magnitude eigenvalues of $\bv{Q}\bv{Q}^T\bv{A}\bv{Q}\bv{Q}^T$ converge to those of $\Ab$ as long as there exists some $k$ such that $\sigma_k(\Ab)$ is larger than $\sigma_{l+1}(\Ab)$ by at least a constant factor. Roughly, we are able to prove the stronger statement as the gaps between the singular values ensure the estimated singular values are also well separated.

\begin{lemma}\label{lem:eig_error_slq}
Consider the setting of Lemma~\ref{lem:slq_sin}. Let $c>0$ be the constant in the error boound of Lemma~\ref{lem:slq_sin}. Then, for every $i \in [k]$, with probability at least $1-\delta$, we have: 
\begin{align*}
    |\lambda_i(\Ab)-\lambda_{i}(\Qb\Qb^T \Ab \Qb \Qb^T)| \leq \frac{\|\Ab \|_2g^{cl/2}_{\min}}{(n\cdot \kappa)^{c/2\epsilon}}.
\end{align*}
\end{lemma}
\begin{proof}
First, we can follow the proofs of Lemmas~\ref{lem:spec1} and~\ref{lem:bound3} along with the stronger bound on $\|\Ub_k-\Qb\Qb^T\Ub_k \|_2$ from Lemma~\ref{lem:slq_sin} to prove that the following guarantees hold with probability at least $1-\delta$:
\begin{enumerate}
    \item $\|\Qb\Qb^T \Ab \Qb \Qb^T \Ub_k-\Ub_k\Ub_k^T\Qb\Qb^T\Ab\Qb\Qb^T\Ub_k \|_2 \leq \frac{\|\Ab \|_2g^{cl}_{\min}}{n^{c/\epsilon-1}( \kappa)^{c/\epsilon}}$,
    \item $| \lambda_i(\bv{U}_k^T\bv{Q}\bv{Q}^T\bv{A}\bv{Q}\bv{Q}^T\bv{U}_k)-\lambda_i(\bv{A})| \leq \frac{\|\Ab \|_2g^{cl}_{\min}}{n^{c/\epsilon-1}( \kappa)^{c/\epsilon}}$.
\end{enumerate}
Next, we can follow the proof of Theorem~\ref{thm:eig_error}, along with the stronger error bounds above and in Lemma~\ref{lem:slq_sin}, with the stronger bound of $\frac{\|\Ab \|_2g^{cl}_{\min}}{n^{c/\epsilon-1}( \kappa)^{c/\epsilon}}$, which (similar to ~\eqref{Eq:eig_gap1}) gives us that for every $i \in [k]$, there exists some $\lambda_j(\bv{Q}\bv{Q}^T\bv{A}\bv{Q}\bv{Q}^T)$ such that: 
\begin{align}\label{Eq:alp}
    |\lambda_j(\bv{Q}\bv{Q}^T\bv{A}\bv{Q}\bv{Q}^T)-\lambda_i(\Ab)| \leq \frac{\|\Ab \|_2g^{cl}_{\min}}{n^{c/\epsilon-1}( \kappa)^{c/\epsilon}},
\end{align}
 Using the min-max principle of singular values, $\sigma_i(\bv{Q}\bv{Q}^T\bv{A}\bv{Q}\bv{Q}^T) \leq \sigma_i(\Ab)$ for all $i \in [n]$. We will now prove that $j =i$ for all $i \in [k]$. Assume, for contradiction, there exists some $i \in [k]$ such that~\eqref{Eq:alp} is only satisfied by some $\lambda_j(\bv{Q}\bv{Q}^T\bv{A}\bv{Q}\bv{Q}^T)$ such that $j>i$. Then, we get:
\begin{align*}
    |\lambda_j(\bv{Q}\bv{Q}^T\bv{A}\bv{Q}\bv{Q}^T)-\lambda_i(\Ab)| &\geq \sigma_i(\Ab)-\sigma_j(\bv{Q}\bv{Q}^T\bv{A}\bv{Q}\bv{Q}^T) \\
    &\geq \sigma_i(\Ab)-\sigma_j(\Ab) \\
    &\geq \sigma_i(\Ab)-\sigma_{i+1}(\Ab) \\
    &\geq g_{\min} \sigma_i(\Ab) \\
    &\geq g_{\min} \sigma_k(\Ab) \geq 2g_{\min}\alpha \geq \frac{2 \|\Ab \|_2g^{(c/4+1)}_{\min}}{n^{c/4}}  \geq \frac{\|\Ab \|_2g^{cl}_{\min}}{n^{c/\epsilon-1}( \kappa)^{c/\epsilon}}.
\end{align*}
for a large enough $c>0$. This contradicts~\eqref{Eq:alp} and thus, we must have $j \leq i$. We can similarly rule out the case $j<i$. Thus, we must have $j=i$ for every $i \in [k]$.
\end{proof}

We next prove a stronger version of Theorem~\ref{thm:converge} below which shows that every eigenvector of $\bv{Q}\bv{Q}^T\bv{A}\bv{Q}\bv{Q}^T$, corresponding to a large magnitude eigenvalue, converges to the corresponding eigenvector of $\Ab$ as long as there exists some $k$ such that $\sigma_k(\Ab)$ is larger than $\sigma_{l+1}(\Ab)$ by at least a constant factor.

\begin{lemma}\label{lemma:converge_slq}
    Consider the setting of Lemma~\ref{lem:slq_sin}. Let the eigenvectors of $\bv{Q}\bv{Q}^T\bv{A}\bv{Q}\bv{Q}^T$ be $\bv{z}_1, \ldots, \bv{z}_n$. Then, for any $i \in [k]$, with probability at least $1-\delta$, for some constant $c_1$, we have 
$$\|\bv{z}_i-\bv{u}_i \|_2 \leq \frac{g^{(c_1l)}_{\min}}{(n \cdot \kappa)^{\frac{c_1}{\epsilon}}} \textrm{ or } \|\bv{z}_i+\bv{u}_i \|_2 \leq \frac{g^{(c_1l)}_{\min}}{(n \cdot \kappa)^{\frac{c_1}{\epsilon}}}.$$ 
\end{lemma}
\begin{proof}
  From Lemma~\ref{lem:slq_sin}, we have that for any eigenvector $\bv{u}_i$ of $\Ab$ for $i \in [k]$, $\bv{Q}\bv{Q}^T\bv{u}_i=\bv{u}_i+\bv{e}_i$ where $\|\bv{e}_i \|_2 \leq \frac{\|\Ab \|_2g^{cl}_{\min}}{(n\cdot \kappa)^{c/\epsilon}}$ for some constant $c>0$. So, we have $\bv{Q}\bv{Q}^T\bv{A}\bv{Q}\bv{Q}^T \bv{u}_i=\bv{Q}\bv{Q}^T\bv{A}(\bv{u}_i+\bv{e}_i)=\lambda_i(\bv{A})\bv{Q}\bv{Q}^T\bv{u}_i+\bv{Q}\bv{Q}^T\bv{Ae}_i=\lambda_i(\Ab)(\bv{u}_i+\bv{e}_i)+\bv{Q}\bv{Q}^T\bv{Ae}_i=\lambda_i(\Ab)\bv{u}_i+\lambda_i(\Ab)\bv{e}_i+\bv{Q}\bv{Q}^T\bv{Ae}_i$. Thus, we get that 
   \begin{equation}\label{Eq:eig_vec}
       \bv{Q}\bv{Q}^T\bv{A}\bv{Q}\bv{Q}^T \bv{u}_i=\lambda_i(\Ab)\bv{u}_i+\bv{r}_i,
   \end{equation}
where  $\|\bv{r}_i \|_2 \leq \frac{2\|\Ab \|_2g^{cl}_{\min}}{(n.\kappa)^{c/\epsilon}}$.

   Since the eigenvectors $\zb_1, \ldots, \zb_n$ form an orthonormal basis for $\R^n$, we can write $\bv{u}_i=\sum_{j=1}^n a_j\bv{z}_j$ where $a_j$ are constants such that $\sum_{j=1}^n a_j^2=1$. Thus, 
   \begin{align}\label{eq:eig_vec1}
       \bv{Q}\bv{Q}^T\bv{A}\bv{Q}\bv{Q}^T \bv{u}_i=\sum_{j=1}^n a_j\lambda_j(\bv{Q}\bv{Q}^T\bv{A}\bv{Q}\bv{Q}^T) \bv{z}_j.
   \end{align}
   From~\eqref{Eq:eig_vec}, we get $\bv{Q}\bv{Q}^T\bv{A}\bv{Q}\bv{Q}^T \bv{u}_i=\sum_{j=1}^n \lambda_i(\Ab)a_j\bv{z}_j+\bv{r}_i$. Hence, from~\eqref{eq:eig_vec1} we have: $\sum_{j=1}^n \lambda_i(\Ab)a_j\bv{z}_j+\bv{r}_i=\sum_{j=1}^n a_j\lambda_j(\bv{Q}\bv{Q}^T\bv{A}\bv{Q}\bv{Q}^T) \bv{z}_j$. Rearranging, we get: $\sum_{j=1}^n (\lambda_j(\bv{Q}\bv{Q}^T\bv{A}\bv{Q}\bv{Q}^T)-\lambda_i(\Ab))a_j \bv{z}_j=\bv{r}_i$. Taking 2-norm on both sides, we get: $$\|\sum_{j=1}^n (\lambda_j(\bv{Q}\bv{Q}^T\bv{A}\bv{Q}\bv{Q}^T)-\lambda_i(\Ab))a_j \bv{z}_j \|_2 \leq \frac{2g^{cl}_{\min}\|\Ab \|_2}{(n\cdot \kappa)^{c/\epsilon}},$$ using the fact that $\|\bv{r}_i\|_2 \leq \frac{2g^{cl}_{\min}\|\Ab \|_2}{(n\cdot \kappa)^{c/\epsilon}}$. Now, $\|\sum_{j=1}^n (\lambda_j(\bv{Q}\bv{Q}^T\bv{A}\bv{Q}\bv{Q}^T)-\lambda_i(\Ab))a_j \bv{z}_j \|_2=\sqrt{\sum_{j=1}^n (\lambda_j(\bv{Q}\bv{Q}^T\bv{A}\bv{Q}\bv{Q}^T)-\lambda_i(\Ab))^2a^2_j }$. Thus, for all $j \in [n]$ we have 
   \begin{equation}\label{Eq:bound1}
       |\lambda_j(\bv{Q}\bv{Q}^T\bv{A}\bv{Q}\bv{Q}^T)-\lambda_i(\Ab)| \cdot |a_j| \leq \frac{2\|\Ab \|_2g^{cl}_{\min}}{(n\cdot \kappa)^{c/\epsilon}}.
   \end{equation} 
   Now, we have $|\lambda_i(\Ab)-\lambda_i(\bv{Q}\bv{Q}^T\bv{A}\bv{Q}\bv{Q}^T)| \leq  \frac{\|\Ab \|_2g^{cl/2}_{\min}}{(n\cdot \kappa)^{c/2\epsilon}}$ for $i \in [k]$, from Lemma~\ref{lem:eig_error_slq}. Also, from our assumptions on the singular value gaps, for any $i \in [k]$ and $j \in [n]$ such that $i \neq j$,  
   \begin{align*}
       |\lambda_i(\Ab)-\lambda_{j}(\Ab)| \geq |\sigma_i(\Ab)-\sigma_j(\Ab)| &\geq  g_{\min}\max(\sigma_i(\Ab),\sigma_j(\Ab)) \\
       &\geq g_{\min}\sigma_i(\Ab) \geq g_{\min}\sigma_{k}(\Ab) \\
       &\geq 2g_{\min}\alpha \\
       &\geq \frac{\|\Ab \|_2g^{3cl/4}_{\min}}{(n\cdot \kappa)^{3c/4\epsilon}},
   \end{align*}
   Using the min-max principle of singular values, $\sigma_i(\bv{Q}\bv{Q}^T\bv{A}\bv{Q}\bv{Q}^T) \leq \sigma_i(\Ab)$ for all $i \in [n]$. So, for $j \neq i$ and $i \in [k]$, using triangle inequality we have that $|\lambda_i(\Ab)-\lambda_j(\bv{Q}\bv{Q}^T\bv{A}\bv{Q}\bv{Q}^T)| \geq |\lambda_i(\Ab)-\lambda_j(\Ab)|-|\lambda_j(\Ab)-\lambda_j(\bv{Q}\bv{Q}^T\bv{A}\bv{Q}\bv{Q}^T)| \geq \frac{\|\Ab \|_2g^{cl/2}_{\min}}{(n\cdot \kappa)^{c/2\epsilon}}-\frac{\|\Ab \|_2g^{3cl/4}_{\min}}{(n\cdot \kappa)^{3c/4\epsilon}}$. Thus, from~\eqref{Eq:bound1}, we get that: 
   \begin{align*}
        |a_j| \leq \frac{\frac{2\|\Ab \|_2g^{cl}_{\min}}{(n\cdot \kappa)^{c/\epsilon}}}{\frac{\|\Ab \|_2g^{cl/2}_{\min}}{(n\cdot \kappa)^{c/2\epsilon}}-\frac{\|\Ab \|_2g^{3cl/4}_{\min}}{(n\cdot \kappa)^{3c/4\epsilon}}} \leq \frac{\frac{2g^{cl}_{\min}}{(n\cdot \kappa)^{c/\epsilon}}}{\frac{g^{cl/2}_{\min}}{(n\cdot \kappa)^{c/2\epsilon}}-\frac{g^{3cl/4}_{\min}}{(n\cdot \kappa)^{3c/4\epsilon}}} \leq \frac{g^{(c_4l)}_{\min}}{(n\cdot \kappa)^{(c_4/\epsilon)}},  
   \end{align*}
   for some constant $c_4>0$. Since $\sum_{j=1}^n a_j^2=1$, we get that $|a_i| \geq 1-\frac{g^{(c_4l)}_{\min}}{(n\cdot \kappa)^{c_4/\epsilon-1}}$. Let $a_i>0$. Then, using triangle inequality, we have $\|\bv{u}_i-\bv{z}_i\|_2=\|\bv{u}_i-a_i\bv{z}_i+\bv{z}_i(a_i-1)\|_2\leq \|\bv{u}_i-a_i\bv{z}_i\|_2+\|\bv{z}_i(a_i-1)\|_2 \leq \|\sum_{j \neq i}a_j\bv{z}_j \|_2+|a_i-1| \leq \frac{C'g^{(c_4l)}_{\min}}{(n\cdot \kappa)^{c_4/2\epsilon-1}}$ where $C'$ is some constant. Similarly, we can show when $a_i<0$, $\|\bv{u}_i+\bv{z}_i \|_2 \leq \frac{C'g^{(c_4l)}_{\min}}{(n\cdot \kappa)^{c_4/2\epsilon-1}}$. We complete the proof by choosing the constant $c_1>0$ suitably. 
\end{proof}

\subsection{Improved Error bounds for SLQ via implicit deflation based analysis}\label{sec:slq3}

In this section, we prove the main error bounds for SLQ (Algorithm~\ref{alg:slq}) by showing that it implicitly performs a deflation followed by moment matching. We first show that there exists a polynomial $r(x)$ of degree $O(l\log (\frac{1}{g_{\min}})+\frac{1}{\epsilon}\log (\frac{n \cdot \kappa}{\delta}))$ which is almost zero on the large magnitude eigenvalues of $\Ab$ and $\Qb^T\Ab\Qb$ which have magnitude greater than $\sigma_k(\Ab)$ and close to one on the small magnitude eigenvalues of $\Ab$ and $\Qb^T\Ab\Qb$ where $k \in [l]$ is an index as defined in 
Lemmas~\ref{lem:eig_error_slq} and~\ref{lemma:converge_slq} such that $\sigma_k(\Ab)$ is larger than $\sigma_{l+1}(\Ab)$ by at least a constant multiplicative gap. Then, a polynomial of the form $t_i(x)=r(x)\Bar{T}_i(x)$ where $\Bar{T}_i(x)$ is the $i$\tth Chebyshev polynomial for $i \in O(1/\epsilon)$ will have degree at most $O(l\log (\frac{1}{g_{\min}})+\frac{1}{\epsilon}\log (\frac{n \cdot \kappa}{\delta}))$ and can be represented in the span of the Krylov subspace~\eqref{Eq:span2} generated by the Lanczos algorithm (Algorithm~\ref{alg:lanczos-slq}) as long as we run Lanczos for at least this many number of iterations. Intuitively, this implies that the polynomial $t_i(x)$ behaves like a Chebyshev polynomial for the small magnitude eigenvalues (with magnitude smaller than $\sigma_k(\Ab)$ and $\sigma_k(\Qb^T\Ab\Qb)$) of $\Ab$ and $\Qb^T\Ab\Qb$ while it is almost zero on the large magnitude eigenvalues of $\Ab$ and $\Qb^T\Ab\Qb$. We then show that the moments of the spectral density of the small magnitude eigenvalues $\Ab$ and the part of the density output by SLQ which only depends on the small magnitude eigenvalues of $\Qb^T\Ab\Qb$ with respect to the polynomial $t_i(x)$ are very close to each other for $i \in O(1/\epsilon)$. Since $t_i(x)$ behaves like a Chebyshev polynomial on these small magnitude eigenvalues, via an argument similar to the Chebyshev moment matching argument in Section~\ref{sec:slq1}, we say that the spectral densities of the small magnitude eigenvalues of $\Ab$ and $\Qb^T\Ab\Qb$ are close to each other. On the other hand, via the results in Section~\ref{sec:slq2}, we know that the large magnitude eigenvalues of $\Ab$ are approximated by the corresponding large magnitude eigenvalues of $\Qb^T\Ab\Qb$. Combining the arguments for the large and small magnitude eigenvalues, we claim the the spectral densities of $\Ab$ and that defined by SLQ algorithm (Algorithm~\ref{alg:slq}) are close to each other.

The polynomial $r(x)$ is defined as $1-g(x)L(x)$ where the polynomial $g(x)$ is a Chebyshev mimizing polynomial (Lemma~\ref{Lem:chebyshev}) and $L(x)$ can be interpreted as a variant of the Lagrange interpolating polynomial through the points $(\lambda_i(\Ab), \frac{1}{g(\lambda_i(\Ab))})$ for $i \in [k]$. We describe this in more detail now. Suppose we have a set of $k$ `basis' polynomials such that the $i$'th polynomial in the set is almost $1$ at $\lambda_i(\Ab)$ and $\lambda_i(\Qb^T\Ab\Qb)$ and almost zero at all other $\lambda_j(\Ab)$ and $\lambda_j(\Qb^T\Ab\Qb)$. Then, $r(x)$ can be defined as 1 minus the sum of these polynomials i.e. we will have $r(\lambda_i(\Ab))$ and $r(\lambda_i(\Qb^T\Ab\Qb))$ almost to 1 for every $i \in [k]$ and almost $0$ at all other eigenvalues. As a first step, in Lemma~\ref{Lem:p(x)} for each $i \in [k]$ we first define a polynomial $p_i(x)$ which is almost $1$ at $\lambda_i(\Ab)$ and $\lambda_i(\Qb^T\Ab\Qb)$ (referred to as $\tilde{\lambda}_i(\Ab)$ in the lemmas below) and zero at all other $\lambda_j(\Ab)$ and $\lambda_j(\Qb^T\Ab\Qb)$ for $j \neq i$ and $j \in [l]$. However, outside the top $l$ eigenvalues, i.e. on the small magnitude eigenvalues of $\Ab$ and $\Qb^T\Ab\Qb$,  $p_i(x)$ can be potentially be very large (up to $(2/g_{\min})^{O(l)}$). To ensure the polynomials $p_i(x)$ do not blow up on the small magnitude eigenvalues of $\Ab$ and $\Qb^T\Ab\Qb$, we will multiply each $p_i(x)$ with a corresponding minimizing polynomial (defined in Lemma~\ref{Lem:chebyshev}) which squishes its values down to almost $0$ on the small magnitude eigenvalues (while keeping its values almost same on the large magnitude eigenvalues). The final set of polynomials is each $p_i(x)$ multiplied by its corresponding minimizing polynomials. We can then define $r(x)$ as $1$ minus this polynomial.

We first define the polynomials $p_i(x)$ for $i \in [k]$ in the lemma below. In the lemmas below, $\tilde{\lambda}_i(\Ab)=\lambda_i(\Qb^T\Ab\Qb)$ for ease of notation.
\begin{lemma}\label{Lem:p(x)}
Consider the setting of Lemma~\ref{lem:slq_sin}. Let $\tilde{\lambda}_i(\Ab)=\lambda_i(\Qb\Qb^T\Ab\Qb\Qb^T)$ for $i \in [n]$. Then, for $i \in [k]$, 
\begin{align*}
    p_i(x)=\prod_{j \in [l], j \neq i} \left( \frac{x-\lambda_j(\Ab)}{\lambda_i(\Ab)-\lambda_j(\Ab)} \right) \prod_{j \in [l], j \neq i} \left( \frac{x-\tilde{\lambda}_j(\Ab)}{\lambda_i(\Ab)-\tilde{\lambda}_j(\Ab)} \right).
\end{align*} 
Then, for some constant $c>0$, with probability at least $1-\delta$:
 \begin{enumerate}
 \item For any $i \in [k]$ $p_i(\lambda_i(\Ab))=1$ and $p_j(\lambda_i(\Ab))=0$ for $j \neq i$.
 \item For any $i \in \{k+1, \ldots, l\}$ and $j \in [k]$, $p_j(\lambda_i(\Ab))=0$ and $p_j(\tilde{\lambda}_i(\Ab))=0$.
     \item  $|p_i(x)| \leq \frac{2^{3l}}{g^{2l}_{\min}}$ when $|x| \leq |\lambda_k(\Ab)|$ for $i \in [k]$.
\item For any $i \in [k]$, $|p_i(\tilde{\lambda}_i(\Ab))-1| \leq \frac{g^{cl/2}_{\min}}{(n\cdot \kappa)^{c/2\epsilon-4}}$ and $p_j(\tilde{\lambda}_i(\Ab))=0$ for $j \neq i$.
 \end{enumerate}
\end{lemma}
\begin{proof}
 
 First observe that, from Lemma~\ref{lem:eig_error_slq}, $|\lambda_i(\Ab)-\tilde{\lambda}_i(\Ab)| \leq \frac{\|\Ab \|_2g^{cl}_{\min}}{(n\cdot \kappa)^{c/\epsilon}}$ for $i \in [k]$ with probability at least $1-\delta$ for some constant $c>0$. The first two claims are straightforward from the definition of $p_i(x)$. So we proceed to prove the third claim.

 \medskip
 
\noindent\textbf{Claim 3:} We now bound $|p_i(x)|$ when $|x| \leq \sigma_{k}(\Ab)$ to prove the third claim. We have $$|p_i(x)|=\prod_{j \in [l], j \neq i} \left|  \frac{x-\lambda_j(\Ab)}{\lambda_i(\Ab)-\lambda_j(\Ab)}\right| \prod_{j \in [l], j \neq i} \left|\frac{x-\tilde{\lambda}_j(\Ab)}{\lambda_i(\Ab)-\tilde{\lambda}_j(\Ab)} \right|.$$ We bound each term separately.

   $\left|  \frac{x-\lambda_j(\Ab)}{\lambda_i(\Ab)-\lambda_j(\Ab)}\right|$: For any $i \in [k]$, $\left|  \frac{x-\lambda_j(\Ab)}{\lambda_i(\Ab)-\lambda_j(\Ab)}\right| \leq \frac{2\max (|x|,\sigma_j(\Ab))}{|\lambda_i(\Ab)-\lambda_j(\Ab)|} \leq \frac{2\max (\sigma_i(\Ab),\sigma_j(\Ab))}{|\lambda_i(\Ab)-\lambda_j(\Ab)|} \leq \frac{2}{g_{\min}}$ where the second inequality follows from the fact that $\left|  x-\lambda_j(\Ab)\right| \leq 2\max (|x|,\sigma_j(\Ab))$, the third inequality from the fact that $|x| \leq \sigma_{k}(\Ab) \leq \sigma_i(\Ab)$ for $i \in [k]$ and the final inequality follows from the definition of $g_{\min}$.

    $\left|\frac{x-\tilde{\lambda}_j(\Ab)}{\lambda_i(\Ab)-\tilde{\lambda}_j(\Ab)} \right|$: We consider the cases $j \in [k]$ and $k \leq j \leq l$ separately. 

    \medskip
    
    \noindent\textbf{Case 1.} ($j \in [k]$): We have $\left|\frac{x-\tilde{\lambda}_j(\Ab)}{\lambda_i(\Ab)-\tilde{\lambda}_j(\Ab)} \right| \leq \frac{2\max(|x|, |\tilde{\lambda}_j(\Ab)|)}{|\lambda_i(\Ab)-\lambda_j(\Ab)| -\frac{\|\Ab \|_2g^{cl}_{\min}}{(n\cdot \kappa)^{c/\epsilon}}} \leq \frac{2\max(|x|, |\lambda_j(\Ab)|)}{|\lambda_i(\Ab)-\lambda_j(\Ab)| -\frac{\|\Ab \|_2g^{cl}_{\min}}{(n\cdot \kappa)^{c/\epsilon}}}.$ Here, for the numerator we used the fact that $|\tilde{\lambda}_j(\Ab)| \leq |\lambda_j(\Ab)|$ by the minimax principle and for the denominator we used triangle inequality to get $|\lambda_i(\Ab)-\tilde{\lambda}_j(\Ab)| \geq |\lambda_i(\Ab)-\lambda_j(\Ab)|- |\lambda_j(\Ab)-\tilde{\lambda}_j(\Ab)|$ and the fact that $|\lambda_j(\Ab)-\tilde{\lambda}_j(\Ab)| \leq \frac{\|\Ab \|_2g^{cl}_{\min}}{(n\cdot \kappa)^{c/\epsilon}}$ for $j \in [k]$. Since $|x| \leq |\lambda_i(\Ab)|$ for $i \in [k]$ we have: 
    \begin{align*}
        \frac{2\max(|x|, |\lambda_j(\Ab)|)}{|\lambda_i(\Ab)-\lambda_j(\Ab)| -\frac{\|\Ab \|_2g^{cl}_{\min}}{(n\cdot \kappa)^{c/\epsilon}}} &\leq \frac{2\max(|\lambda_i(\Ab)|, |\lambda_j(\Ab)|)}{|\lambda_i(\Ab)-\lambda_j(\Ab)| -\frac{\|\Ab \|_2g^{cl}_{\min}}{(n\cdot \kappa)^{c/\epsilon}}} \\        
        &\leq \frac{2}{\frac{|\lambda_i(\Ab)-\lambda_j(\Ab)|}{\max(|\lambda_i(\Ab)|, |\lambda_j(\Ab)|)} -\frac{\|\Ab \|_2g^{cl}_{\min}}{\max(|\lambda_i(\Ab)|, |\lambda_j(\Ab)|) (n\cdot \kappa)^{c/\epsilon}}}\\
        &\leq \frac{2}{g_{\min}-\frac{\|\Ab \|_2g^{cl}_{\min}}{\max(|\lambda_i(\Ab)|, |\lambda_j(\Ab)|) (n\cdot \kappa)^{c/\epsilon}}} \\
        &\leq \frac{4}{g_{\min}} .
    \end{align*}
    Here, the second to last step follows from the definition of $g_{\min}$ and the last step follows from the assumptions that $\max(|\lambda_i(\Ab)|,|\lambda_j(\Ab)|) \geq \sigma_k(\Ab) \geq 2\alpha \geq \frac{2\|\Ab \|_2g^{(cl)/4}_{\min}}{(n\cdot \kappa)^{c/4}}$ (the first step follows from the fact that $i,j \in [k]$), which gives us $\frac{\|\Ab \|_2g^{cl}_{\min}}{\max(|\lambda_i(\Ab)|,|\lambda_j(\Ab)|)(n\cdot \kappa)^{c/\epsilon}} \leq \frac{\|\Ab \|_2g^{3cl/4}_{\min} (n\cdot \kappa)^{c/4}}{2(n\cdot \kappa)^{c/\epsilon}}\leq  \frac{g_{\min}}{2}$. Thus, we get $\left|\frac{x-\tilde{\lambda}_j(\Ab)}{\lambda_i(\Ab)-\tilde{\lambda}_j(\Ab)} \right|  \leq \frac{4}{g_{\min}}$ for $j \in [k]$ and $|x| \leq \sigma_k(\Ab)$.

    \medskip

    \noindent\textbf{Case 2.} ($j \in \{k+1, \ldots, l\}$): For $k+1 \leq j \leq p$ and $i \in [k]$, we have $|\tilde{\lambda}_j(\Ab)| \leq |\lambda_j(\Ab)| \leq |\lambda_i(\Ab)|$ where the first inequality follows from our assumption and the second from the fact that $j \geq i$. So, $|\lambda_i(\Ab)-\tilde{\lambda}_j(\Ab)| \geq |\lambda_i(\Ab)|-|\tilde{\lambda}_j(\Ab)| \geq |\lambda_i(\Ab)|-|\lambda_j(\Ab)|=\sigma_i(\Ab)-\sigma_j(\Ab)$. So, we get $$\left|\frac{x-\tilde{\lambda}_j(\Ab)}{\lambda_i(\Ab)-\tilde{\lambda}_j(\Ab)} \right| \leq  \frac{2\max(|x|,|\tilde{\lambda}_j(\Ab)|)}{\sigma_i(\Ab)-\sigma_j(\Ab)} \leq \frac{2\max(|\lambda_i(\Ab)|,|\lambda_j(\Ab)|)}{\sigma_i(\Ab)-\sigma_j(\Ab)} \leq \frac{2}{g_{\min}},$$ where in the second inequality we used the fact that $|x| \leq |\lambda_i(\Ab)|$ and $|\tilde{\lambda}_j(\Ab)| \leq |\lambda_j(\Ab)|$ by the minimax principle and the last step follows from the definition of $g_{\min}$. So, we finally have $\left|\frac{x-\tilde{\lambda}_j(\Ab)}{\lambda_i(\Ab)-\tilde{\lambda}_j(\Ab)} \right| \leq \frac{2}{g_{\min}}$ for $j \in \{k+1, \ldots, l\}$.
    
    Combining the two cases above, we get $\left|\frac{x-\tilde{\lambda}_j(\Ab)}{\lambda_i(\Ab)-\tilde{\lambda}_j(\Ab)} \right| \leq \frac{4}{g_{\min}}$ when $|x| \leq \sigma_k(\Ab)$. Thus, plugging in the upper bound on each term of $|p_i(x)|$ , we get $|p_i(x)| < \frac{2^{3l}}{g^{2l}_{\min}}$ when $|x| \leq \sigma_k(\Ab)$ for $i \in [k]$. 

    \medskip

    \noindent\textbf{Claim 4:} We now prove the fourth claim of our theorem. For any $i,j \in [k]$, $p_j(\tilde{\lambda}_i(\Ab))=0$ for $j \neq i$ from the definition of $p_j(x)$. This gives the second part of the claim. We now prove the first part of the claim. Observe that, 
    \begin{align}\label{eq:ptilde}
        \left| p_i(\tilde{\lambda}_i(\Ab)) \right|=\prod_{j \in [l], j \neq i} \left| \frac{\tilde{\lambda}_i(\Ab)-\lambda_j(\Ab)}{\lambda_i(\Ab)-\lambda_j(\Ab)} \right|\prod_{j \in [l], j \neq i} \left| \frac{\tilde{\lambda}_i(\Ab)-\tilde{\lambda}_j(\Ab)}{\lambda_i(\Ab)-\tilde{\lambda}_j(\Ab)} \right|.
    \end{align}
    We will bound each term in the product individually. First, observe that $$\left| \frac{\tilde{\lambda}_i(\Ab)-\lambda_i(\Ab)}{\lambda_i(\Ab)-\lambda_j(\Ab)} \right| \leq \frac{\|\Ab \|_2g^{cl}_{\min}}{(n\cdot \kappa)^{c/\epsilon}|\lambda_i(\Ab)-\lambda_j(\Ab)|} \leq \frac{\|\Ab \|_2g^{cl/2}_{\min}}{(n\cdot \kappa)^{c/2\epsilon}} ,$$ where the second step follows from bounding the numerator using fact that $|\tilde{\lambda}_i(\Ab)-\lambda_i(\Ab)| \leq \frac{\|\Ab \|_2g^{cl}_{\min}}{(n\cdot \kappa)^{c/\epsilon}}$ for $i \in [k]$ and last step follows from the fact that for $i,j \in [k+1]$ and $i \neq j$ (and a large enough $c$):
    \begin{align}\label{eq:last}
        |\lambda_i(\Ab)-\lambda_j(\Ab)| \geq |\sigma_i(\Ab)-\sigma_j(\Ab)| \nonumber &\geq g_{\min}\max(\sigma_i(\Ab),\sigma_j(\Ab)) \nonumber\\ &\geq g_{\min}\sigma_k(\Ab) \nonumber\\&\geq \frac{2\|\Ab \|_2g^{(c/4)+1}_{\min}}{(n\cdot \kappa)^{c/4}} \nonumber \\
        &\geq \frac{\|\Ab \|_2g^{cl/2}_{\min}}{(n\cdot \kappa)^{c/2\epsilon}}.
    \end{align}
    Thus, we get 
    \begin{align}\label{eq:b1}
        \left| \frac{\tilde{\lambda}_i(\Ab)-\lambda_j(\Ab)}{\lambda_i(\Ab)-\lambda_j(\Ab)}-1\right|&=\left| \frac{\tilde{\lambda}_i(\Ab)-\lambda_i(\Ab)}{\lambda_i(\Ab)-\lambda_j(\Ab)} \right|\nonumber \\ &\leq \frac{\|\Ab \|_2g^{cl}_{\min} }{(n\cdot \kappa)^{c/\epsilon}|\lambda_i(\Ab)-\lambda_j(\Ab)|} \nonumber\\& \leq\frac{g^{cl/2}_{\min}}{(n\cdot \kappa)^{c/2\epsilon}} .
    \end{align}
    Next, we have  
    \begin{align}\label{eq:b2}
        \left| \frac{\tilde{\lambda}_i(\Ab)-\tilde{\lambda}_j(\Ab)}{\lambda_i(\Ab)-\tilde{\lambda}_j(\Ab)}-1 \right| &= \left| \frac{\tilde{\lambda}_i(\Ab)-\lambda_i(\Ab)}{\lambda_i(\Ab)-\tilde{\lambda}_j(\Ab)}\right|\notag\\
        &\leq \frac{\|\Ab \|_2g^{cl}_{\min} }{(n\cdot \kappa)^{c/\epsilon}|\lambda_i(\Ab)-\tilde{\lambda}_j(\Ab)|} \nonumber \\
        &\leq  \frac{\|\Ab \|_2g^{cl}_{\min} }{(n\cdot \kappa)^{c/\epsilon}(|\lambda_i(\Ab)-\lambda_j(\Ab)|- \frac{\|\Ab \|_2g^{cl}_{\min}}{(n\cdot \kappa)^{c/\epsilon}})}  \nonumber \\
        &\leq \frac{1 }{n^{c/\epsilon} \frac{g^{cl/2}_{\min}}{g^{cl}_{\min} (n\cdot \kappa)^{c/2\epsilon}}-1} \leq 2\frac{g^{cl/2}_{\min}}{(n\cdot \kappa)^{c/2\epsilon}}.
    \end{align}
     In the second step above, we used the fact that $|\lambda_i(\Ab)-\tilde{\lambda}_i(\Ab)| \leq \frac{\|\Ab \|_2g^{cl}_{\min}}{(n\cdot \kappa)^{c/\epsilon}}$ for bounding the numerator, in the third step  we used triangle inequality for bounding the denominator and the fourth step follows from the fact that $|\lambda_i(\Ab)-\lambda_j(\Ab)| \geq \frac{\|\Ab \|_2g^{cl/2}_{\min}}{(n\cdot \kappa)^{c/2\epsilon}}$ as described in~\eqref{eq:last}. Thus, from the bounds in~\eqref{eq:b1} and~\eqref{eq:b2} we have 
    \begin{align}\label{eq:upp}
         \left( \frac{\tilde{\lambda}_i(\Ab)-\lambda_j(\Ab)}{\lambda_i(\Ab)-\lambda_j(\Ab)}\right) \left( \frac{\tilde{\lambda}_i(\Ab)-\tilde{\lambda}_j(\Ab)}{\lambda_i(\Ab)-\tilde{\lambda}_j(\Ab)} \right)  &\leq (1+\left(\frac{g^{cl/2}_{\min}}{(n\cdot \kappa)^{c/2\epsilon}} \right))\notag\\
         &\cdot(1+2\left(\frac{g^{cl/2}_{\min}}{(n\cdot \kappa)^{c/2\epsilon}} \right)) \nonumber\\
         &\leq 1+3\left(\frac{g^{cl/2}_{\min}}{(n\cdot \kappa)^{c/2\epsilon}} \right)\notag\\
         &+ 2\left(\frac{g^{cl}_{\min}}{(n\cdot \kappa)^{c/\epsilon}} \right) \nonumber\\
         &\leq 1+5\left(\frac{g^{cl/2}_{\min}}{(n\cdot \kappa)^{c/2\epsilon}} \right),
     \end{align}
     for a large enough $c$ and also, 
     \begin{align}\label{eq:low}
         \left( \frac{\tilde{\lambda}_i(\Ab)-\lambda_j(\Ab)}{\lambda_i(\Ab)-\lambda_j(\Ab)}\right) \left( \frac{\tilde{\lambda}_i(\Ab)-\tilde{\lambda}_j(\Ab)}{\lambda_i(\Ab)-\tilde{\lambda}_j(\Ab)} \right)  &\geq \left(1-\left(\frac{g^{cl/2}_{\min}}{(n\cdot \kappa)^{c/2\epsilon}} \right) \right)\notag\\ 
         &\cdot\left(1-2\left(\frac{g^{cl/2}_{\min}}{(n\cdot \kappa)^{c/2\epsilon}} \right) \right) \nonumber \\
         &\geq \left(1-3\left(\frac{g^{cl/2}_{\min}}{(n\cdot \kappa)^{c/2\epsilon}} \right) \right).
     \end{align}
     Multiplying the upper bounds in~\eqref{eq:upp} for each $j \in [l]$ together, we get 
     \begin{align}\label{eq:upfinal}
         \prod_{j \in [l], j \neq i}\left( \frac{\tilde{\lambda}_i(\Ab)-\lambda_j(\Ab)}{\lambda_i(\Ab)-\lambda_j(\Ab)}\right) \left( \frac{\tilde{\lambda}_i(\Ab)-\tilde{\lambda}_j(\Ab)}{\lambda_i(\Ab)-\tilde{\lambda}_j(\Ab)} \right) &\leq \left(1+5\left(\frac{g^{cl/2}_{\min}}{(n\cdot \kappa)^{c/2\epsilon}} \right) \right)^{l} \nonumber\\
         &\leq 1+\sum_{r=1}^l {l \choose r} \left(5\left(\frac{g^{cl/2}_{\min}}{(n\cdot \kappa)^{c/2\epsilon}}\right)\right)^r \nonumber\\
         &\leq  1+5l\left(\frac{g^{cl/2}_{\min}}{(n\cdot \kappa)^{c/2\epsilon-2}} \right) \nonumber\\
         &\leq 1+\frac{5g^{cl/2}_{\min}}{(n\cdot \kappa)^{(c/2\epsilon)-3}},
     \end{align}
     for a large enough $c$. In the last step above, we bounded $l$ by $n$. In the third step above, we bounded each term ${l \choose r} (5(1/(n\cdot \kappa))^{c/2\epsilon})^r$ (for $r \in [l]$) in the binomial expansion of $(1+5\left(\frac{1}{n\cdot \kappa} \right)^{c/2\epsilon})^{l} $ as ($a=(n\cdot \kappa)$ below) 
     \begin{align*}
        \left(\frac{el}{r} \right)^r \left(5\left(\frac{1}{a}\right)^{c/2\epsilon}\right)^r \leq \left(5l^2\left(\frac{1}{a}\right)^{c/2\epsilon}\right)^r \leq \left( \frac{5}{a^{c/2\epsilon-2}}\right)^r \leq  \frac{5}{a^{c/2\epsilon-2}},    
     \end{align*}
     where in the first step we used the well-known upper bound on the binomial coefficient $ {l \choose r} \leq (\frac{el}{r} )^r$  and in the final step, we used the fact that $ \frac{5}{(n\cdot \kappa)^{c/2\epsilon-2}} \leq 1$ for large $c>0$. Similarly, multiplying the lower bounds in~\eqref{eq:low} for each $j \in [l]$ together, we get 
     \begin{align}\label{eq:lowfinal}
         \prod_{j \in [l], j \neq i}\left( \frac{\tilde{\lambda}_i(\Ab)-\lambda_j(\Ab)}{\lambda_i(\Ab)-\lambda_j(\Ab)}\right) \left( \frac{\tilde{\lambda}_i(\Ab)-\tilde{\lambda}_j(\Ab)}{\lambda_i(\Ab)-\tilde{\lambda}_j(\Ab)} \right) &\geq \left(1-\frac{3g^{cl/2}_{\min}}{(n\cdot \kappa)^{c/2\epsilon}}\right)^{l} \nonumber\\
         &\geq 1+\sum_{r=1}^l {l \choose r} (-1)^r\left(3\left(\frac{g^{cl/2}_{\min}}{(n\cdot \kappa)^{c/2\epsilon}}\right)\right)^r \nonumber\\
         &\geq 1-\frac{3lg^{cl/2}_{\min}}{(n\cdot \kappa)^{c/2\epsilon-2}} \nonumber\\
         &\geq 1-\frac{3g^{cl/2}_{\min}}{(n\cdot \kappa)^{(c/2\epsilon)-3}},
     \end{align}
     for large enough $c$. Thus, using the bounds in~\eqref{eq:upfinal},~\eqref{eq:lowfinal} and~\eqref{eq:ptilde}, we get $|p_i(\tilde{\lambda}_i(\Ab)) -1| \leq \frac{5g^{cl/2}_{\min}}{(n\cdot \kappa)^{(c/2\epsilon)-3}}\leq \frac{g^{cl/2}_{\min}}{(n\cdot \kappa)^{(c/2\epsilon)-4}}.$ This proves the first property in claim 4.
\end{proof}

We now define a Chebyshev Minimizing polynomial similar to the polynomial defined in Lemma 5 of~\cite{Musco:2015}.

\begin{lemma}[Chebyshev Minimizing Polynomial]\label{Lem:chebyshev}
For values $\xi>0$, gap $\gamma > 0$, and some even integer $q \geq 1$, there exists a polynomial $g(x)$ of degree $q$ such that:
\begin{enumerate}
    \item $g(x)=(1+\gamma)\xi$ for $|x|=(1+\gamma)\xi$.
    \item $g(x) \geq |x|$ for all $|x| \geq (1+\gamma)\xi$.
    \item $g(x) \leq \frac{\xi}{2^{q \log (1+\gamma)-1}}$ for $x \in [-\xi,\xi] $.
\end{enumerate}
\end{lemma}
\begin{proof}
We define the polynomial $g(x)$ following the proof of Lemma 5 in~\cite{Musco:2015}. Let 
\begin{equation}\label{Eq:g(x)}
    g(x)= (1+\gamma)\xi \frac{T_q(x/\xi)}{T_q(1+\gamma)},
\end{equation}
where $T_q(x)$ is a Chebyshev polynomial of degree $q$ where $q$ is even. Since the degree $q$ is even, we have $T_q(x)=T_q(-x)$ and thus, $g(x)=g(-x)$. The proofs of the first two properties follow the proof of the first two properties of Lemma 5 of~\cite{Musco:2015} exactly. To prove the third property, first observe that using the property of Chebyshev polynomials that $T_q(x) \leq 1$ for $x \in [-1,1]$, we have $T_q(x/\xi) \leq 1$. So, it suffices to show that  $T_q(1+\gamma) \geq 2^{q \log (1+\gamma)-1}$. Using equation 15 of~\cite{Musco:2015} which gives the expression for $T_q(x)$ for $|x|>1$, we have $T_q(x) \geq \frac{1}{2}(1+\gamma)^q =\frac{2^{q \log (1+\gamma)}}{2}=2^{q \log (1+\gamma)-1}$. This completes the proof of the third property.
\end{proof}

We now define the polynomial $r(x)$ in the lemma below.

\begin{lemma}\label{Lem:r(x)}
Consider the setting of Lemmas~\ref{Lem:p(x)}. Let $c_1,c_2>0$ be some constants. Let $g(x)$ be a degree $q=O(l\log (\frac{1}{g_{\min}})+\frac{1}{\epsilon}\log (\frac{n\kappa}{\delta}))$ Chebyshev Minimizing polynomial as defined in Lemma~\ref{Lem:chebyshev} for some $q$, with parameters $\xi=\alpha$ and $\gamma =\frac{\sigma_k(\Ab)}{\sigma_{l+1}(\Ab)}-1$. Let $$r(x)=1-g(x)\sum_{i=1}^k\frac{p_i(x)}{g(\lambda_i(\Ab))}.$$ Then,  with probability at least $1-\delta$, we have:
\begin{enumerate}
    \item $r(x)=0$ for any $x \in  \{\lambda_1(\Ab), \ldots, \lambda_k(\Ab) \}$.
    \item  $r(x)=1$ for any $x \in \{ \lambda_{k+1}(\Ab), \ldots, \lambda_l(\Ab)\}$ and any $x \in \{ \tilde{\lambda}_{k+1}(\Ab), \ldots, \tilde{\lambda}_l(\Ab)\}$.
    \item $|r(x)| \leq \frac{g^{c_1 l}_{\min}}{(n\cdot \kappa)^{c_1/\epsilon}}$ for any $x \in \{\tilde{\lambda}_1(\Ab), \ldots, \tilde{\lambda}_k(\Ab) \}$.
    \item $|r(x)-1| \leq \frac{g^{c_2l}_{\min}}{(n\cdot \kappa)^{c_2/\epsilon}}$ for $|x| \leq \sigma_{l+1}(\Ab)$.    
\end{enumerate}
\end{lemma}
\begin{proof}
 First observe that, from Lemma~\ref{lem:eig_error_slq}, $|\lambda_i(\Ab)-\tilde{\lambda}_i(\Ab)| \leq \frac{\|\Ab \|_2g^{cl}_{\min}}{(n\cdot \kappa)^{c/\epsilon}}$ for $i \in [k]$ with probability at least $1-\delta$. Also note that $\gamma \geq 1$ by the assumption that $\sigma_k(\Ab) \geq 2\alpha$. We now prove the main claims.

The first property follows directly from property 1 of Lemma~\ref{Lem:p(x)}. The second property follows directly from property 2 of Lemma~\ref{Lem:p(x)}. We prove the third property below.

\medskip

\noindent\textbf{Claim 3:} For the third property, observe that for any $i,j \in [k]$ and $j \neq i$, we have $p_j(\tilde{\lambda}_i(\Ab))=0$ from property 4 of Lemma~\ref{Lem:p(x)}, so  we have 
\begin{align}\label{eq:rtilde}
    r(\tilde{\lambda}_i(\Ab))= 1-g(\tilde{\lambda}_i(\Ab))\sum_{j=1}^k \frac{p_j(\tilde{\lambda}_i(\Ab))}{g(\lambda_i(\Ab))}=1-p_i(\tilde{\lambda}_i(\Ab))\frac{g(\tilde{\lambda}_i(\Ab))}{g(\lambda_i(\Ab))}.
\end{align}
 From property 4 of Lemma~\ref{Lem:p(x)}, we get that for any $i \in [k]$:
 \begin{align}\label{eq:prop4}
     |p_i(\tilde{\lambda}_i(\Ab))-1| \leq \frac{g^{cl/2}_{\min}}{(n\cdot \kappa)^{c/2\epsilon-4}}.
 \end{align}

First observe that from definition of $g(x)$ in~\eqref{Eq:g(x)} in Lemma~\ref{Lem:chebyshev}, we have $\frac{g(x)}{g(y)}=\frac{T_q(x/\xi)}{T_q(y/\xi)}$ where $T_q(x)$ is the $q$\textsuperscript{th} Chebyshev polynomial for some even $q$. Now, for any $i \in [k]$, $|\lambda_i(\Ab)| \geq |\lambda_k(\Ab)| \geq 2\alpha $ and so $\frac{|\lambda_i(\Ab)|}{\alpha} \geq 2$. From Lemma~\ref{lem:eig_error_slq}, we have that for any $i \in [k]$,
 \begin{align}\label{eq:l1}
     |\tilde{\lambda}_i(\Ab)|  &\geq  |\lambda_{i}(\Ab)| -\frac{\|\Ab \|_2g^{cl}_{\min}}{(n\cdot \kappa)^{c/\epsilon}}\notag\\
     &\geq |\lambda_{k}(\Ab)| -\frac{\|\Ab \|_2g^{cl}_{\min}}{(n\cdot \kappa)^{c/\epsilon}} \notag\\
     &\geq 2\alpha-\frac{\|\Ab \|_2g^{cl}_{\min}}{(n\cdot \kappa)^{c/\epsilon}}.
 \end{align}
 Now, $\alpha \geq \frac{\|\Ab \|_2g^{c/4}_{\min}}{(n\cdot \kappa)^{c/4}} \geq \frac{\|\Ab \|_2g^{cl}_{\min}}{(n\cdot \kappa)^{c/\epsilon}}$. Thus, from~\eqref{eq:l1}, we get $ |\tilde{\lambda}_i(\Ab)| \geq \alpha$. Thus, we also have $\frac{|\tilde{\lambda}_i|}{\alpha} \geq 1$. Also, since $q$ is even we have $T_q(x)=T_q(|x|)$. Thus, from~\eqref{Eq:T(x)} we have:
\begin{align}\label{eq:ratiog}
    \frac{g(\tilde{\lambda}_i(\Ab))}{g(\lambda_i(\Ab))}=\frac{T_q(\tilde{\lambda}_i(\Ab)/\alpha)}{T_q(\lambda_i(\Ab)/\alpha)}=\frac{ \bigg(\frac{\tilde{\sigma}_i(\Ab)}{\alpha}+\sqrt{(\frac{\tilde{\sigma}_i(\Ab)}{\alpha})^2-1} \bigg)^q+\bigg(\frac{\tilde{\sigma}_i(\Ab)}{\alpha}-\sqrt{(\frac{\tilde{\sigma}_i(\Ab)}{\alpha})^2-1} \bigg)^q}{ \bigg(\frac{\sigma_i(\Ab)}{\alpha}+\sqrt{(\frac{\sigma_i(\Ab)}{\alpha})^2-1} \bigg)^q+\bigg(\frac{\sigma_i(\Ab)}{\alpha}-\sqrt{(\frac{\sigma_i(\Ab)}{\alpha})^2-1} \bigg)^q}.
\end{align}
where $\tilde{\sigma}_i(\Ab)=|\tilde{\lambda}_i(\Ab)|$. We will first upper bound this ratio. Since $\frac{\tilde{\sigma}_i(\Ab)}{\alpha}-\sqrt{(\frac{\tilde{\sigma}_i(\Ab)}{\alpha})^2-1}  \leq 1$, the numerator in~\eqref{eq:ratiog} can be upper bounded as $\bigg(\frac{\tilde{\sigma}_i(\Ab)}{\alpha}+\sqrt{(\frac{\tilde{\sigma}_i(\Ab)}{\alpha})^2-1} \bigg)^q+1$. Now, by the minimax principle, we have $\tilde{\sigma}_{i}(\Ab) \leq \sigma_i(\Ab)$. Thus, the denominator in~\eqref{eq:ratiog} is lower bounded by $\bigg(\frac{\sigma_i(\Ab)}{\alpha}+\sqrt{(\frac{\sigma_i(\Ab)}{\alpha})^2-1} \bigg)^q \geq \bigg(\frac{\tilde{\sigma}_i(\Ab)}{\alpha}+\sqrt{(\frac{\tilde{\sigma}_i(\Ab)}{\alpha})^2-1} \bigg)^q$. Using the upper bound on the numerator and the lower bound on the denominator in~\eqref{eq:ratiog}, we get:
\begin{align}\label{eq:gup}
   \frac{g(\tilde{\lambda}_i(\Ab))}{g(\lambda_i(\Ab))} \leq \frac{\bigg(\frac{\tilde{\sigma}_i(\Ab)}{\alpha}+\sqrt{(\frac{\tilde{\sigma}_i(\Ab)}{\alpha})^2-1} \bigg)^q+1}{\bigg(\frac{\tilde{\sigma}_i(\Ab)}{\alpha}+\sqrt{(\frac{\tilde{\sigma}_i(\Ab)}{\alpha})^2-1} \bigg)^q} &\leq 1+ \frac{1}{\bigg(\frac{\tilde{\sigma}_i(\Ab)}{\alpha}+\sqrt{(\frac{\tilde{\sigma}_i(\Ab)}{\alpha})^2-1} \bigg)^q}  \nonumber\\
   &\leq 1+\frac{1}{\big(\frac{\tilde{\sigma}_i(\Ab)}{\alpha} \big)^q}.
\end{align}
Observe that since $\alpha \geq \frac{\|\Ab \|_2g^{c/4}_{\min}}{n^{c/4}}$, we have
\begin{align}\label{eq:alph1}
    \frac{\|\Ab \|_2g^{cl}_{\min}}{\alpha (n\cdot \kappa)^{c/\epsilon}} \leq \frac{g^{cl}_{\min} (n\cdot \kappa)^{c/4}}{g^{c/4}_{\min} (n\cdot \kappa)^{c/\epsilon}} \leq \frac{g^{cl/2}_{\min}}{(n\cdot \kappa)^{c/2\epsilon}}.
\end{align}
Dividing both sides of~\eqref{eq:l1} by $\alpha$ and using the fact $\frac{g^{cl/2}_{\min}}{(n\cdot \kappa)^{c/2\epsilon}} \leq \frac{1}{2}$, we get:
 \begin{align*}
     \frac{\tilde{\sigma}_i(\Ab)}{\alpha}  \geq  2-  \frac{g^{cl}_{\min}}{\alpha (n\cdot \kappa)^{c/\epsilon}}  \geq 2-\frac{1}{2} \geq \frac{3}{2}.
 \end{align*}
Thus, using the bound above in~\eqref{eq:gup}, we get that:
\begin{align}\label{Eq:glup}
     \frac{g(\tilde{\lambda}_i(\Ab))}{g(\lambda_i(\Ab))} \leq 1+\frac{1}{\big(\frac{\tilde{\sigma}_i(\Ab)}{\alpha} \big)^q} \leq  1+\frac{1}{(3/2)^q} \leq 1+\frac{g^{c_3l}_{\min}}{(n\cdot \kappa)^{c_3/\epsilon}},
\end{align}
for some large constant $c_3>0$. Here, in the lasts step, we used the fact that $q=O(l\log (\frac{1}{g_{\min}})+\frac{1}{\epsilon}\log (\frac{n}{\delta}))$. We now lower bound $\frac{g(\tilde{\sigma}_i(\Ab))}{g(\sigma_i(\Ab))}$.

First observe that the function $x-\sqrt{x^2-1}$ is a decreasing function with respect to $x$. Thus, we have $\bigg(\frac{\tilde{\sigma}_i(\Ab)}{\alpha}-\sqrt{(\frac{\tilde{\sigma}_i(\Ab)}{\alpha})^2-1} \bigg) \geq \bigg(\frac{\sigma_i(\Ab)}{\alpha}-\sqrt{(\frac{\sigma_i(\Ab)}{\alpha})^2-1} \bigg)$ as $\frac{\sigma_i(\Ab)}{\alpha} \geq \frac{\tilde{\sigma}_i(\Ab)}{\alpha}$. This implies that we have 
\begin{align*}
    &\bigg(\frac{\tilde{\sigma}_i(\Ab)}{\alpha}-\sqrt{(\frac{\tilde{\sigma}_i(\Ab)}{\alpha})^2-1} \bigg)\bigg(\frac{\sigma_i(\Ab)}{\alpha}+\sqrt{(\frac{\sigma_i(\Ab)}{\alpha})^2-1} \bigg) \\
    &\geq \bigg(\frac{\sigma_i(\Ab)}{\alpha}-\sqrt{(\frac{\sigma_i(\Ab)}{\alpha})^2-1} \bigg)\bigg(\frac{\tilde{\sigma}_i(\Ab)}{\alpha}+\sqrt{(\frac{\tilde{\sigma}_i(\Ab)}{\alpha})^2-1} \bigg).
\end{align*}
So, from~\eqref{eq:ratiog}, we have:
\begin{align}\label{eq:lambdai}
    \frac{g(\tilde{\lambda}_i(\Ab))}{g(\lambda_i(\Ab))} \geq \frac{ \bigg(\frac{\tilde{\sigma}_i(\Ab)}{\alpha}+\sqrt{(\frac{\tilde{\sigma}_i(\Ab)}{\alpha})^2-1} \bigg)^q}{ \bigg(\frac{\sigma_i(\Ab)}{\alpha}+\sqrt{(\frac{\sigma_i(\Ab)}{\alpha})^2-1} \bigg)^q}.
\end{align}
Now, from Lemma~\ref{lem:eig_error_slq} and using triangle inequality, we have $\tilde{\sigma}_i(\Ab) \geq \sigma_i(\Ab)-\frac{g^{cl}_{\min}\|\Ab \|_2}{(n\cdot \kappa)^{c/\epsilon}}$ for $i \in [k]$. Dividing both sides by $\alpha$, we get $\frac{\tilde{\sigma}_i(\Ab)}{\alpha} \geq \frac{\sigma_i(\Ab)}{\alpha}-\frac{g^{cl}_{\min}\|\Ab \|_2}{\alpha (n\cdot \kappa)^{c/\epsilon}} \geq \frac{\sigma_i(\Ab)}{\alpha}-\frac{g^{cl/2}_{\min}}{(n\cdot \kappa)^{c/2\epsilon}}$ where in the last step we used~\eqref{eq:alph1} to upper bound $\frac{g^{cl}_{\min}\|\Ab \|_2}{\alpha (n\cdot \kappa)^{c/\epsilon}}$. So, using the lower bound on $\frac{\tilde{\sigma}_i(\Ab)}{\alpha}$ in~\eqref{eq:lambdai}, we get:
\begin{align}\label{eq:num1}
     \frac{g(\tilde{\lambda}_i(\Ab))}{g(\lambda_i(\Ab))} \geq \frac{ \bigg(\frac{\sigma_i(\Ab)}{\alpha}-\frac{g^{cl/2}_{\min}}{(n\cdot \kappa)^{c/2\epsilon}}+\sqrt{(\frac{\sigma_i(\Ab)}{\alpha}-\frac{g^{cl/2}_{\min}}{(n\cdot \kappa)^{c/2\epsilon}})^2-1} \bigg)^q}{ \bigg(\frac{\sigma_i(\Ab)}{\alpha}+\sqrt{(\frac{\sigma_i(\Ab)}{\alpha})^2-1} \bigg)^q}.
\end{align}
Observe that we have 
\begin{align*}
    &\sqrt{(\frac{\sigma_i(\Ab)}{\alpha}-\frac{g^{cl/2}_{\min}}{(n\cdot \kappa)^{c/2\epsilon}})^2-1} \\
    &=\sqrt{\big((\frac{\sigma_i(\Ab)}{\alpha})^2-1 \big)-\big( \frac{2g^{cl/2}_{\min}\sigma_i(\Ab)}{(n\cdot \kappa)^{c/2\epsilon}\alpha}-\frac{g^{cl}_{\min}}{(n\cdot \kappa)^{c/\epsilon}}\big)} \\
    &\geq \sqrt{(\frac{\sigma_i(\Ab)}{\alpha})^2-1}-\sqrt{\frac{2g^{cl/2}_{\min}\sigma_i(\Ab)}{(n\cdot \kappa)^{c/2\epsilon}\alpha}-\frac{g^{cl}_{\min}}{(n\cdot \kappa)^{c/\epsilon}}},
\end{align*}
from the fact that $\sqrt{a-b} \geq \sqrt{a}-\sqrt{b}$. Thus, using the lower bound on the numerator in~\eqref{eq:num1} we get:
\begin{align}\label{eq:num2}
     \frac{g(\tilde{\lambda}_i(\Ab))}{g(\lambda_i(\Ab))} &\geq \left(1-\frac{\bigg(\frac{g^{cl/2}_{\min}}{(n\cdot \kappa)^{c/2\epsilon}}+ \sqrt{\frac{2g^{cl/2}_{\min}\sigma_i(\Ab)}{(n\cdot \kappa)^{c/2\epsilon}\alpha}-\frac{g^{cl}_{\min}}{(n\cdot \kappa)^{c/\epsilon}}}\bigg)}{\bigg(\frac{\sigma_i(\Ab)}{\alpha}+\sqrt{(\frac{\sigma_i(\Ab)}{\alpha})^2-1} \bigg)} \right)^q  \nonumber\\
     &\geq  \left(1-\frac{2\sqrt{\frac{2g^{cl/2}_{\min}\sigma_i(\Ab)}{(n\cdot \kappa)^{c/2\epsilon}\alpha}}}{\frac{\sigma_i(\Ab)}{\alpha}} \right)^q \nonumber \\
     &= \left(1-2\sqrt{\frac{2g^{cl/2}_{\min}\alpha}{(n\cdot \kappa)^{c/2\epsilon}\sigma_i(\Ab)}} \right)^q \nonumber\\
     &\geq \left(1-\frac{2g^{cl/4}_{\min}}{(n\cdot \kappa)^{c/4\epsilon}} \right)^q.
\end{align}
The second step above follow from upper bounding the numerator of $\frac{\bigg(\frac{1}{(n\cdot \kappa)^{c/2\epsilon}}+ \sqrt{\frac{2\sigma_i(\Ab)}{(n\cdot \kappa)^{c/\epsilon}\alpha}-\frac{1}{(n\cdot \kappa)^{c/\epsilon}}}\bigg)}{\bigg(\frac{\sigma_i(\Ab)}{\alpha}+\sqrt{(\frac{\sigma_i(\Ab)}{\alpha})^2-1} \bigg)}$ by $2\sqrt{\frac{2\sigma_i(\Ab)}{(n\cdot \kappa)^{c/\epsilon}\alpha}}$ and its denominator by $\frac{\sigma_i(\Ab)}{\alpha}$. The final step follows from the fact that $\sigma_i(\Ab) \geq 2\alpha$ for any $i \in [k]$. Now, observe that we have:
\begin{align}\label{eq:num3}
   \left(1-\frac{2g^{cl/4}_{\min}}{(n\cdot \kappa)^{c/4\epsilon}}  \right)^q &=1+ \sum_{t=1}^q \binom{q}{t} (-1)^t (2g^{cl/4}_{\min}/(n\cdot \kappa)^{c/4\epsilon})^t  \nonumber\\
   &\geq 1- \sum_{t=1}^q (eq/t)^t (2g^{cl/4}_{\min}/(n\cdot \kappa)^{c/4\epsilon})^t \nonumber\\
   &\geq 1-\sum_{t=1}^q \left(\frac{2eq g^{cl/4}_{\min}}{(n\cdot \kappa)^{c/4\epsilon}} \right)^t \nonumber\\
   &\geq 1-\sum_{t=1}^q \frac{g^{cl/4}_{\min}}{(n\cdot \kappa)^{(c/4\epsilon)-2}} \nonumber\\
   &=1-\frac{qg^{cl/4}_{\min}}{(n\cdot \kappa)^{(c/4\epsilon)-2}} \nonumber\\
   &\geq 1-\frac{g^{cl/4}_{\min}}{(n\cdot \kappa)^{(c/4\epsilon)-3}}.
\end{align}
In the second step above, we upper bounded the binomial coefficients $\binom{q}{t}$ by $(eq/t)^t$. In the fourth and final step above, we bounded $q$ by $n$. In the fourth step, we also used the fact that $\frac{2eq}{(n\cdot \kappa)^{c/2\epsilon}} <<1$ for large enough $c$. So, using the lower bound from~\eqref{eq:num3} in~\eqref{eq:num2}, we get:
\begin{align}\label{eq:glow}
     \frac{g(\tilde{\lambda}_i(\Ab))}{g(\lambda_i(\Ab))} \geq 1-\frac{g^{cl/4}_{\min}}{(n\cdot \kappa)^{(c/4\epsilon)-3}}.
\end{align}
Now, from the upper bounds on $p_i(\tilde{\lambda}_i(\Ab))$ and $\frac{g(\tilde{\lambda}_i(\Ab))}{g(\lambda_i(\Ab))}$ from~\eqref{eq:prop4} and~\eqref{Eq:glup} respectively, we get:
\begin{align*}
    p_i(\tilde{\lambda}_i(\Ab))\frac{g(\tilde{\lambda}_i(\Ab))}{g(\lambda_i(\Ab))} \leq 1+\frac{g^{c_4l}_{\min}}{(n\cdot \kappa)^{c_4/\epsilon}},
\end{align*}
for some large enough $c_4$. Similarly, from the lower bounds on $p_i(\tilde{\lambda}_i(\Ab))$ and $\frac{g(\tilde{\lambda}_i(\Ab))}{g(\lambda_i(\Ab))}$ from~\eqref{eq:prop4} and~\eqref{eq:glow} respectively, we get:
\begin{align*}
    p_i(\tilde{\lambda}_i(\Ab))\frac{g(\tilde{\lambda}_i(\Ab))}{g(\lambda_i(\Ab))} \geq 1-\frac{g^{c_5l}_{\min}}{(n\cdot \kappa)^{c_5/\epsilon}},
\end{align*}
for some $c_5$. Finally, using the upper and lower bounds on $p_i(\tilde{\lambda}_i(\Ab))\frac{g(\tilde{\lambda}_i(\Ab))}{g(\lambda_i(\Ab))}$ in~\eqref{eq:rtilde}, we get:
\begin{align*}
    |r(\tilde{\lambda}_i(\Ab))|= \bigg| 1-p_i(\tilde{\lambda}_i(\Ab))\frac{g(\tilde{\lambda}_i(\Ab))}{g(\lambda_i(\Ab))} \bigg| \leq \frac{g^{c_8 l}_{\min}}{(n\cdot \kappa)^{c_8/\epsilon}},
\end{align*}
for some constant $c_8$. This gives us the third claim.

\medskip

\noindent\textbf{Claim 4:} For the fourth and final property, using property 3 for $|p_i(x)|$ from Lemma~\ref{Lem:p(x)}, we get that for any $|x| \leq \sigma_{l+1}(\Ab)$, 
\begin{align*}
    |r(x)-1|\leq \sum_{i=1}^k\frac{|g(x)|}{|g(\lambda_i(A))|}|p_i(x)|&\leq \frac{\alpha}{\alpha 2^{(Cl\log (1/g_{\min})+\frac{C}{\epsilon}\log (n \kappa/\delta))}} \cdot \frac{2^{3l}}{g_{\min}^{2l}}\\
    &\leq \frac{g^{c_9 l}_{\min}}{(n\cdot \kappa)^{c_9/\epsilon}},
\end{align*} 
for some large constant $c_9>0$. Here, in the second step, we also used property 2 of $g(x)$ from Lemma~\ref{Lem:chebyshev}, i.e., $g(\lambda_i(\Ab)) \geq |\lambda_i(\Ab)| \geq 2\alpha$ for $|\lambda_i(\Ab)| \geq |\lambda_k(\Ab)| \geq 2\alpha$ for lower bounding bounding the denominator, and property 3 for $g(x)$ when $|x| \leq \sigma_{l+1}(\Ab) \leq \alpha$ which gives $g(x) \leq \frac{\alpha}{2^{q\log 2 -1}}$ from Lemma~\ref{Lem:chebyshev} for upper bounding the numerator.
\end{proof}

\begin{definition}[Moment Matching Polynomial]\label{def:mom_match}
    Let $t_i(x)=\Bar{T}_i(x)r(x)$ where $\Bar{T}_i(x)$ is the $i^{th}$ normalized Chebyshev polynomial for 
$i \in [O(1/\epsilon)]$ and $r(x)$ is a degree $O(l\log (\frac{1}{g_{\min}})+\frac{1}{\epsilon}\log (\frac{(n \cdot \kappa)}{\delta}))$ polynomial as defined in Lemma~\ref{Lem:r(x)}. Thus, $t_i(x)$ has degree $O(l\log (\frac{1}{g_{\min}})+\frac{1}{\epsilon}\log (\frac{n \cdot \kappa}{\delta}))$. 
\end{definition}

We run Lanczos with $\Ab$ as input for $O(l\log (\frac{1}{g_{\min}})+\frac{1}{\epsilon}\log (\frac{(n \cdot \kappa))}{\delta}))$ iterations. We will now show that the moments of the spectral density of $\Ab$ as well as the output of Algorithm~\ref{alg:slq} (SLQ) with respect to $t_i(x)$ is approximately equal to the $i$\textsuperscript{th} normalized Chebyshev moment of the SDE of $\Ab$ and the output of SLQ respectively for all $i  \in [O(1/\epsilon)]$.

\begin{lemma}\label{Lem:slq_mom}
Consider the setting of Lemma~\ref{Lem:r(x)}. Let $\Ab \in \R^{n \times n}$ be such that $\|\Ab \|_2 \geq 1$ and $\kappa=\|\Ab \|_2$. Let $k \in [l]$ such that $\sigma_k(\Ab) \geq 1$ and $\sigma_{k+1}(\Ab) \leq 1$. Let $t_i(x)=\Bar{T}_i(x)r(x)$ for $i \in O(\frac{1}{\epsilon})$ be the polynomials in Definition~\ref{def:mom_match}. Let $f(x)$ be the final output after running Algorithm \ref{alg:slq} with $\Ab$ as input for $m=O(l\log (\frac{1}{g_{\min}})+\frac{1}{\epsilon}\log (\frac{n \cdot \kappa}{\delta}))$ iterations for some $\epsilon,\delta \in (0,1)$. Then, with probability at least $1-\delta$, for $i \in O(\frac{1}{\epsilon})$, we have (for constant $c_1,c_2>0$) $$\bigg| \langle t_i,s_{\Ab} \rangle-\frac{1}{n}\sum_{j=k+1}^n \Bar{T}_i(\lambda_j(\Ab)) \bigg|  \leq \frac{g^{c_1l}_{\min}}{(n\cdot \kappa)^{c_1/\epsilon}}$$ and $$\bigg| \langle t_i,f \rangle-\sum_{j=k+1}^m w_j^2\Bar{T}_i(\lambda_j(\Tb)) \bigg| \leq \frac{g^{c_2l}_{\min}}{(n\cdot \kappa)^{c_2/\epsilon}}.$$
\end{lemma}
\begin{proof}
We have $\bv{T}=\Qb^T \Ab \Qb$ where $\Tb \in \R^{m \times m}$ is the output of Algorithm~\ref{alg:slq} and $\Qb$ is the orthonormal basis of the Krylov subspace generated by the Lanczos algorithm in~\ref{alg:lanczos-slq}. Note that the nonzero eigenvalues of $\Qb\Qb^T \Ab \Qb \Qb^T$ are the same as those of $\Qb^T\Ab\Qb$ i.e. $\lambda_i(\Tb)=\lambda_i(\Qb\Qb^T\Ab\Qb\Qb^T)$ for $i \in [m]$.

\medskip

\noindent\textbf{Moments of $s_{\Ab}$:} We have:
 \begin{align*}
     \langle t_i,s_{\Ab} \rangle &= \frac{1}{n}\sum_{j=1}^n t_i(\lambda_j(\Ab)) \\
     &= \frac{1}{n}\sum_{j=1}^k \Bar{T}_i(\lambda_j(\Ab))r(\lambda_j(\Ab))+\frac{1}{n}\sum_{j=k+1}^l \Bar{T}_i(\lambda_j(\Ab))r(\lambda_j(\Ab))+ \frac{1}{n}\sum_{j=l+1}^n \Bar{T}_i(\lambda_j(\Ab))r(\lambda_j(\Ab)) \\
     &= 0+\frac{1}{n}\sum_{j=k+1}^l \Bar{T}_i(\lambda_j(\Ab)) + \frac{1}{n}\sum_{j=l+1}^n \Bar{T}_i(\lambda_j(\Ab))r(\lambda_j(\Ab)),
 \end{align*}
  where the second step follows from the definition of $t_i(x)$ and $s_{\Ab}(x)$ and the last step follows from properties 1 and 2 of $r(x)$ in lemma~\ref{Lem:r(x)}. Finally, by property 4 of $r(x)$ in lemma~\ref{Lem:r(x)}, we have $r(\lambda_j(\Ab)) =1+e_j $ such that $|e_j| \leq \frac{g^{c_7l}_{\min}}{(n\cdot \kappa)^{c_9/\epsilon}}$ for $j \in \{l+1, \ldots, n\}$ and for some constant $c_9>0$.  Also, note that $|\Bar{T}_j(\lambda_j(\Ab)| \leq \sqrt{\frac{2}{\pi}}$ for $j \in \{l+1, \ldots, n\}$ as $|\lambda_j(\Ab)| \leq |\lambda_{l+1}(\Ab)|=1$ by our assumptions. So, using triangle inequality we finally get $$\bigg|\langle t_i,s_{\Ab} \rangle-\frac{1}{n}\sum_{j=k+1}^n \Bar{T}_i(\lambda_j(\Ab)) \bigg| \leq \frac{1}{n}\sum_{j=l+1}^n |\Bar{T}_i(\lambda_j(\Ab)||e_j| \leq \frac{g^{c_9 l}_{\min}}{(n\cdot \kappa)^{c_9/\epsilon}},$$ for some constant $c_9>0$. This proves the first claim of the lemma. We now prove the second claim.

  \medskip
  
 \noindent\textbf{Moments of $f$:} We have
  \begin{align}\label{Eq:moments}
    \langle t_i,f \rangle &= \sum_{j=1}^m w_j^2 t_i(\lambda_j(\Tb)) \nonumber\\
    &=\sum_{j=1}^k w_j^2 \Bar{T}_i(\lambda_j(\Tb)) r(\lambda_j(\Tb))+\sum_{j=k+1}^l w_j^2 \Bar{T}_i(\lambda_j(\Tb)) r(\lambda_j(\Tb))+ \sum_{j=l+1}^m w_j^2 \Bar{T}_i(\lambda_j(\Tb)) r(\lambda_j(\Tb)) \nonumber \nonumber\\
     &= \sum_{j=1}^k w_j^2 \Bar{T}_i(\lambda_j(\Tb)) r(\lambda_j(\Tb))+\sum_{j=k+1}^l w_j^2\Bar{T}_i(\lambda_j(\Tb)) + \sum_{j=l+1}^m w_j^2\Bar{T}_i(\lambda_j(\Tb))r(\lambda_j(\Tb)),
 \end{align}
where the last step follows from the properties of $r(x)$ in Lemma~\ref{Lem:r(x)}.

Suppose $|\lambda_j(\Tb)| \geq 1 $ for some $j \in [k]$. So, from~\eqref{Eq:T(x)}, we get that $T_i(\lambda_j(\Tb))$ can be written as $$T_i(\lambda_j(\Tb)) =\frac{1}{2}\bigg[\bigg(|\lambda_j(\Tb)|+\sqrt{|\lambda_j(\Tb)|^2-1}\bigg)^i+\bigg(|\lambda_j(\Tb)|-\sqrt{|\lambda_j(\Tb)|^2-1}\bigg)^i \bigg] \leq (2|\lambda_j(\Tb)|)^{i}.$$ Thus, 
\begin{align}\label{eq:bartx}
    \Bar{T}_i(\lambda_j(\Tb)) \leq \sqrt{\frac{2}{\pi}}(2|\lambda_j(\Tb)|)^{i}.
\end{align}
Thus, for any $j \in [k]$, we have $\Bar{T}_i(\lambda_j(\Tb)) \leq \sqrt{\frac{2}{\pi}}\max ((2|\lambda_j(\Tb)|)^{k},1)$. We can also bound each $w_j^2$ as $(w_j)^2 =(\vb_i^T \bv{e}_1)^2 \leq 1$ for all $j \in [m]$. From property 3 of lemma~\ref{Lem:r(x)}, we have $r(\lambda_j(\Tb)) \leq \frac{g^{c_6 l}_{\min}}{(n\cdot \kappa)^{c_6/\epsilon}}$ for $j \in [k]$ and some constant $c_6$. So, using the bounds on $w_j$, $r(\lambda_j(\Tb))$ and~\eqref{eq:bartx}, for any $ j\in [k]$ such that that $|\lambda_j(\Tb)| \geq 1$, we get:
\begin{align}\label{eq:wk}
    | w_j^2 \Bar{T}_i(\lambda_j(\Tb)) r(\lambda_j(\Tb))| &\leq \sqrt{\frac{2}{\pi}}(2\|\Ab \|_2)^i \bigg(\frac{g^{c_6l}_{\min}}{(n\cdot \kappa)^{c_6/\epsilon}} \bigg) \nonumber\\
    &\leq \sqrt{\frac{2}{\pi}}(2\|\Ab \|_2)^{c'/\epsilon}\bigg(\frac{g^{c_6l}_{\min}}{(n\cdot \kappa)^{c_6/\epsilon}} \bigg) \nonumber \\
    &\leq \frac{g^{c_8 l}_{\min}}{(n\cdot \kappa)^{c_8/\epsilon}},
\end{align}
 for a large enough constant $c_8$, where we also use the facts that $i \leq O(\frac{1}{\epsilon})$ and $\kappa=\|\Ab \|_2$ by assumption. Similarly, $|w_j^2 \Bar{T}_i(\lambda_j(\Tb)) r(\lambda_j(\Tb))| \leq 1$ if $|\lambda_j(\Tb)| \leq 1$ for any $j \in [k]$. From property 2 of Lemma~\ref{Lem:r(x)}, we have $r(\lambda_j(\Tb))=1$ for $j \in \{k+1, \ldots,l\}$. Also, since $|\lambda_j(\Tb)| \leq |\lambda_j(\Ab)| \leq |\lambda_{l+1}(\Ab)| = 1$ for all $j \in \{l+1,\ldots,m\}$, from property 4 of Lemma~\ref{Lem:r(x)}, we again have $r(\lambda_j(\Tb))=1+e_j$ such that $|e_j|\leq \frac{g^{c_2l}_{\min}}{(n\cdot \kappa)^{c_2/\epsilon}}$ for some constant $c_2>0$. Also, note that we again have $|\Bar{T}_j(\lambda_j(\Tb))| \leq 1$ for $j \in \{l+1, \ldots,m\}$. So, we have $$\sum_{j=l+1}^m w_j^2\Bar{T}_i(\lambda_j(\Tb))r(\lambda_j(\Tb)) =\sum_{j=l+1}^m w_j^2\Bar{T}_i(\lambda_j(\Tb)) +\sum_{j=l+1}^m w_j^2\Bar{T}_i(\lambda_j(\Ab)e_j. $$
 Finally, from~\eqref{Eq:moments}, using~\eqref{eq:wk} and the bound above (where $|e_j|\leq \frac{g^{c_2l}_{\min}}{(n\cdot \kappa)^{c_2/\epsilon}}$) and using triangle inequality we have that 
 \begin{align*}
    | \langle t_i,f \rangle-\sum_{j=k+1}^m w_j^2\Bar{T}_i(\lambda_j(\Tb))| &\leq \sum_{j=1}^k w_j^2|\Bar{T}_i(\lambda_j(\Ab)||r(\lambda_j(\Tb))|+ \sum_{j=l+1}^m w_j^2|\Bar{T}_i(\lambda_j(\Ab)||e_j| \\
    &\leq \frac{g^{c_{10}l}_{\min}}{(n\cdot \kappa)^{c_{10}/\epsilon}} ,   
 \end{align*}
 for some large constant $c_{10} >0$. This proves the second claim. Note that everything above holds with probability at least $1-\delta$ after adjusting $\delta$ by constant factors (as the bounds in Lemma~\ref{Lem:r(x)} hold with probability at least $1-\delta$).
\end{proof}

We now bound the difference between the moments of $s_{\Ab}$ and $f$ with respect to the polynomials $t_i(x)$ for each $i \in O(1/\epsilon)$.

\begin{lemma}\label{Lem:slq_hutch}
Consider the setting of Lemma~\ref{Lem:slq_mom}. Then, for $n \geq \Omega(\log (1/\delta))$ and $i \in \{0,1,\ldots,\lceil\frac{C_1}{\epsilon} \rceil\}$ where $C_1>0$ is some constant, $$|\langle t_i,f \rangle-\langle t_i,s_{\Ab} \rangle|\leq \frac{C \log (1/\epsilon\delta)}{\sqrt{n}},$$ for some constant $C>0$ with probability at least $1-\delta$.
\end{lemma}
\begin{proof}
Observe that $ \langle t_i,f \rangle=\sum_{j=1}^m w_i^2 t_i(\lambda_j(\Tb))=\sum_{j=1}^m (\bv{e}_1^T \bv{v}_j)^2 t_i(\lambda_j(\Tb))=\sum_{j=1}^m \bv{e}_1^T \bv{v}_j \bv{v}_j^T \bv{e}_1 t_i(\lambda_j(\Tb))=\bv{e}_1 (\sum_{j=1}^m t_i(\lambda_j(\Tb))\bv{v}_j \bv{v}_j^T )\bv{e}_1=\bv{e}_1^Tt_i(\lambda_j(\Tb)) \bv{e}_1 $. From Lemma~\ref{Lem:slq_mom}, we have $t_i(x)=\Bar{T}_i(x)r(x)$ for $i \in \lceil\frac{C_1}{\epsilon} \rceil$. Then, since $r(x)$ is a polynomial of degree $C_2(l\log (\frac{1}{g_{\min}})+\frac{1}{\epsilon}\log (\frac{(n \cdot \kappa)}{\delta}))$ for some constant $C_2>0$ and $\Bar{T}_i(x)$ has degree at most $\frac{C_1}{\epsilon}$, $t_i(x)=\Bar{T}_i(x)r(x)$ has degree at most $(C_1+C_2)(l\log (\frac{1}{g_{\min}})+\frac{1}{\epsilon}\log (\frac{(n \cdot \kappa)}{\delta}))$. Following the proof of Lemma~\ref{lem:eq-slq-lanczos}, as long as Algorithm~\ref{alg:slq} is run for $m$ iterations such that  $(C_1+C_2)(l\log (\frac{1}{g_{\min}})+\frac{1}{\epsilon}\log (\frac{(n \cdot \kappa)}{\delta})) \leq m$, we have $\langle t_i,f \rangle=\bv{g}^Tt_i(\Ab)\bv{g}$.

We also have $ \langle t_i,s_{\Ab} \rangle=\frac{1}{n}\tr(t_i(\Ab))$. Thus, from Lemma~\ref{Lem:hutch_unit}, we have that (for a single repetition of the hutchinson's estimator and number), for each $i \in O(1/\epsilon)$, with probability at least $1-O(\delta/\epsilon)$, for some constant $C>0$:
\begin{align*}
    |\langle t_i,f \rangle-\langle t_i,s_{\Ab} \rangle|=|\bv{g}^Tt_i(\Ab)\bv{g}-\frac{1}{n}\tr(t_i(\Ab))| \leq \frac{C\log (1/\epsilon \delta)}{n}\|t_i(\Ab) \|_F \leq \frac{C\log (1/\epsilon \delta)}{\sqrt{n}}\|t_i(\Ab) \|_2.
\end{align*}
Now, observe that for $j \in [k]$, from property 1 of Lemma~\ref{Lem:r(x)}, we get that $t_i(\lambda_j(\Ab))=\Bar{T}_i(\lambda_j(\Ab))r(\lambda_j(\Ab))=0$. Also, recall that from the assumptions of Lemma~\ref{Lem:slq_mom}, we have that $\sigma_{k+1}(\Ab)=1$. Now, for $j \geq k+1$, $|t_i(\lambda_j(\Ab))|=|\Bar{T}_i(\lambda_j(\Ab))r(\lambda_j(\Ab))| \leq 2$ where we use property 2 of Lemma~\ref{Lem:r(x)} and the fact that $|T(x)| \leq 1$ for $|x| \leq 1$. Thus, we have $\|t_i(\Ab) \|_2 \leq \max_{j \in \{k+1, \ldots, n\}}|t_i(\lambda_j(\Ab))| \leq 2$. So, we finally get: 
\begin{align*}
     |\langle t_i,f \rangle-\langle t_i,s_{\Ab} \rangle| \leq \frac{C\log (1/\epsilon \delta)}{\sqrt{n}}.
\end{align*}
Taking a union bound over all $i \in O(1/\epsilon)$ completes the claim.
\end{proof}

Next we bound the weights $w^2_i$ in the distribution $f$, the output of SLQ. We will show that $w_i^2$ for the top $k$ weights is at most $\Tilde{O}(1/n)$. This will help us bound the Wasserstein error from the spectral density of the large eigenvalues of $\Tb$.

\begin{lemma}\label{Lem:wts}
Consider the setting of Lemma~\ref{Lem:slq_mom}. Let $w_i=\vb_i^T\bv{e}_1$ be the weights in the output of distribution of SLQ (Algorithm~\ref{alg:slq}) where $\bv{e}_1 \in \R^m$ is the first standard basis vector and $\vb_i$ is the $i$\textsuperscript{th} eigenvector of $\Qb^T\Ab\Qb$. Then, for $\delta \in (0,1)$ such that $n \geq \Omega(\log (1/\delta))$, for all $i \in [k]$, with probability at least $1-\delta$, we have:
\begin{align*}
     w_i^2 \leq  \frac{C\sqrt{\log (k/\delta)}}{n},
\end{align*}
 where $C>0$ is a large constant.
\end{lemma}
\begin{proof}
Let $\zb_i=\Qb\vb_i$ where $\Qb$ is the orthonormal basis of the Krylov subspace generated by the Lanczos algorithm (Algorithm~\ref{alg:lanczos-slq}) after $m$ iterations. Then, $\zb_i$ are the eigenvectors of $\Qb\Qb^T\Ab \Qb \Qb^T$. Also, $\bv{g}=\Qb\bv{e}_1$ or $\Qb^T \bv{g}=\bv{e}_1$. From Lemma~\ref{lemma:converge_slq}, we have that for every $i \in [k]$, $\zb_i=\ub_i+\bv{b}_i$ where $\ub_i$ is the $i$\textsuperscript{th} eigenvector of $\Ab$ and $$\|\bv{b}_i \|_2 \leq \frac{g^{cl}_{\min}}{(n\cdot \kappa)^{c/\epsilon}},$$ for some constant $c>0$, with probability at least $1-\delta$. Also, we have $w_i=\vb_i^T \bv{e}_1= (\Qb\vb_i)^T (\Qb \bv{e}_1) = \zb_i^T\bv{g} = (\ub_i+\bv{b}_i)^T\bv{g}=\ub_i^T \bv{g}+\bv{b}_i^T \bv{g} $. $w_i^2=(\bv{u}_i^T \bv{g})^2 +(\bv{b}_i^T \bv{g})^2 +2 (\bv{u}_i^T \bv{g})(\bv{b}_i^T \bv{g})$

Since $\bv{g}$ is sampled from a uniform distribution on the unit sphere, $\bv{g} \stackrel{d}{:=} \frac{\bv y_i}{\sqrt{\sum_{j=1}^n \bv y_j^2}}$ where $\bv y \in \R^n$ is a random vector such that $\bv y_i \sim \mathcal{N}(0,1)$ for $i \in [n]$. So, we have:
\begin{align}\label{Eq:w}
    w_i^2&=(\bv{u}_i^T \bv{g})^2 +(\bv{b}_i^T \bv{g})^2 +2 (\bv{u}_i^T \bv{g})(\bv{b}_i^T \bv{g}) \notag\\
    &\stackrel{d}{=} \frac{(\sum_{j=1}^n \bv{u}_{ij} \bv y_j)^2}{\sum_{r=1}^n \bv y_r^2}+ \frac{(\sum_{j=1}^n \bv{b}_{ij} \bv y_j)^2}{\sum_{r=1}^n \bv y_r^2} + 2 \frac{(\sum_{j=1}^n \bv{u}_{ij} \bv y_j)(\sum_{j=1}^n \bv{b}_{ij} \bv y_j)}{\sum_{r=1}^n \bv y_r^2}.
\end{align}
We will now bound $\sum_{r=1}^n \bv y_r^2$, $\sum_{j=1}^n \bv{u}_{ij}y_j$ and $\sum_{j=1}^n \bv{b}_{ij} \bv y_j$ individually. Note that $\sum_{r=1}^n \bv y_r^2$ is a Chi-squared random variable with $n$ degrees of freedom. Then, using well-known tail bounds for Chi-squared variables (see 2.21~\cite{wainwright2019high}) we have that with probability at least $1-\delta$, $|\sum_{r=1}^n \bv y_r^2-n| \leq 2 \sqrt{2n\log \frac{2}{\delta}}$. This gives us $\sum_{r=1}^n \bv y_r^2 \geq \frac{n}{2}$ (for $\delta \in (\Omega(1/e^n),1)$ ). Next, observe that since $\ub_i^T \bv{y}=\sum_{j=1}^n \bv{b}_{ij}\bv y_j$ is a linear combination of $\mathcal{N}(0,1)$ random variables and $\|\ub \|_2=1$,  $\ub_i^T \bv{y}$ is another $\mathcal{N}(0,1)$ random variable. So, $(\ub_i^T \bv{y})^2$ is a Chi-squared random variable and using the Chi-squared tail bound again, we have $(\ub_i^T \bv{y})^2 \leq 1+2\sqrt{2\log \frac{2}{\delta}} \leq 3\sqrt{2\log \frac{2}{\delta}}$ with probability at least $1-\delta$. So, we have $\frac{(\sum_{j=1}^n \bv{u}_{ij}y_j)^2}{\sum_{r=1}^n y_r^2} \leq \frac{6\sqrt{2 \log (2/\delta)}}{n}$ with probability at least $1-2\delta$. Next, using Cauchy Schwartz inequality, $\sum_{j=1}^n \bv{b}_{ij}\bv y_j \leq \|\bv{b}_{ij} \|^2 \| \bv{y} \|^2$. So, $\frac{(\sum_{j=1}^n \bv{b}_{ij}\bv y_j)^2}{\sum_{r=1}^n \bv y_r^2} \leq \|\bv{b}_i \|^2_2 \leq \frac{ g^{2cl}_{\min}}{(n\cdot \kappa)^{2c/\epsilon}}$. Also, $2 \frac{(\sum_{j=1}^n \bv{u}_{ij} \bv y_j)(\sum_{j=1}^n \bv{b}_{ij}\bv y_j)}{\sum_{r=1}^n \bv y_r^2} \leq 2\|\bv{b}_i \|_2 \leq \frac{2g^{cl}_{\min}}{(n\cdot \kappa)^{c/\epsilon}}$. So, taking a union bound over all these events from~\eqref{Eq:w} for all $i \in [k]$, and adjusting $\delta$ by a factor of $2k$, we get that with probability at least $1-\delta$, for every $i \in [k]$: 
\begin{align*}
    w_i^2 \leq \frac{6\sqrt{2 \log (2k/\delta)}}{n}+\frac{3g^{cl}_{\min}}{(n\cdot \kappa)^{c/\epsilon}} \leq O \left( \frac{\sqrt{\log (k/\delta)}}{n}\right).
\end{align*}
\end{proof}

We are now ready to prove our final theorem in this section which bounds the Wasserstein-$1$ error between $s_{\Ab}$ and $f$. Let $\alpha=\max( \sigma_{l+1}(\Ab),\frac{\|\Ab \|_2 g^{c/4}_{\min}}{n^{c/4}})$ for some constant $c>0$. We will now consider two cases in the proof: when there exists some $k \in [l]$ such that $\sigma_{k}(\Ab)$ is at least a constant multiplicative factor larger than $\alpha$ and when there isn't any such $k$. When there is such a $k$, we will use the bounds from Section~\ref{sec:slq2} and Lemma~\ref{Lem:wts} to show that the Wasserstein-$1$ error of the parts of $s_{\Ab}$ and $f$ defined by the large magnitude eigenvalues of $\Ab$ and $\Tb$ is at most $\Tilde{O}(l\|\Ab \|_2/n)$ Then, we will use Lemma~\ref{Lem:slq_mom} to show that the Chebyshev moments of the spectral density defined by the small magnitude eigenvalues ($\leq \sigma_{k+1}(\Ab)$) of $\Ab$ and $\Tb$ are approximately equal. So, the Wasserstein-$1$ error of the parts of $s_{\Ab}$ and $f$ defined by the small eigenvalues is bounded by $\epsilon \sigma_{k+1}(\Ab)$. When there isn't such a $k$, all singular values of $\Ab$ are small $(< 2\alpha)$, and we can use the bounds using moment matching in Section~\ref{sec:slq1} to bound the error by $\epsilon \alpha$ in this case. 
Before proving the theorem, we state a couple of results which we will use in our proof. We first state a result from~\cite{braverman:2022} on  uniform approximation of Lipschitz
continuous functions by a Chebyshev series:
\begin{fact}[Fact 3.2 of~\cite{braverman:2022}]\label{fact1}
Let $f:[-1,1] \rightarrow \R$ be a Lipschitz continuous function with Lipschitz constant $\lambda >0$. Then, for every $N \in 4\mathbb{N}^{+}$, there exists $N+1$ constants $0 \leq b_N <\ldots < b_0=1$ such that the polynomial $\bar{f}_N=\sum_{k=0}^N b_k \langle f, w \cdot \Bar{T}_k \rangle \Bar{T}_k$ has the property that $\max_{x \in [-1,1]} |f(x)-\bar{f}_N(x)| \leq \frac{18\lambda}{N}$.
\end{fact}

We now state another result from~\cite{braverman:2022} bounding the magnitude of the inner-product of a Lipschitz function $f$ with the k-th normalized Chebyshev polynomial $\Bar{T}_k$:

\begin{fact}[Fact 3.3 of~\cite{braverman:2022}]\label{fact2}
Let $f:[-1,1] \rightarrow \R$ be a Lipschitz continuous function with Lipschitz constant $\lambda >0$. Then, for any $k \geq 0$, we have that $|f, w\cdot \Bar{T}_k|=|\int_{-1}^{1}f(x)\Bar{T}_k(x)w(x)dx| \leq \frac{2\lambda}{k}$.    
\end{fact}

\begin{reptheorem}{thm:slq}
Let $\Ab\in\R^\n$ be symmetric and consider any $l \in [n]$ and $\epsilon, \delta \in (0,1)$. Let $g_{\min}=\min_{i \in [l]} \frac{\sigma_i(\Ab)-\sigma_{i+1}(\Ab)}{\sigma_i(\Ab)}$ and $\kappa=\frac{\|\Ab \|_2}{2\alpha}$. Let $\alpha=\max\left( \sigma_{l+1}(\Ab),\frac{\|\Ab \|_2 g^{c/4}_{\min}}{n^{c/4}} \right)$ for some constant $c>0$. Algorithm~\ref{alg:slq} (SLQ) run for $m = O(l\log \frac{1}{g_{\min}}+\frac{1}{\epsilon}\log \frac{n \cdot \kappa}{\delta})$ iterations performs $m$ matrix vector products with $\bv A$ and outputs a probability density function $f$ such that, with probability at least $1-\delta$, for a fixed constant $C$,
$$W_1(s_{\Ab},f) \leq \epsilon \cdot \sigma_{l+1}(\bv A) + \frac{C\log (n/\epsilon)\log (1/\epsilon) }{\sqrt{n}} \cdot \sigma_{l+1}(\Ab) + \frac{Cl\log (1/\epsilon)\sqrt{\log (l/\delta)}}{n}\|\Ab \|_2.$$
\end{reptheorem}
\begin{proof}

We will prove the theorem for two complementary cases and analyze them separately below:

\medskip

\noindent\textbf{Case 1:} Let there be some $k \in [l]$ such that $\sigma_k(\Ab) \geq 2\alpha$ and $\sigma_{k+1}(\Ab) < 2\alpha$. Assume that we run SLQ (Algorithm~\ref{alg:slq}) with input $\bv{B}=\frac{1}{2\alpha}\Ab$ instead of $\Ab$ for $m=O(l\log \frac{1}{g_{\min}}+\frac{1}{\epsilon}\log \frac{(n \cdot \kappa)}{\delta})$ iterations. The output of Lanczos (Algorithm~\ref{alg:lanczos-slq}) will be the scaled tridiagonal matrix $\Tb_1=\frac{1}{2\alpha}\Tb $ after $m$ iterations. Observe that $\Bb$ satisfies all the conditions of  Lemma~\ref{Lem:slq_mom} since $\sigma_{k}(\Bb)=\frac{\sigma_k(\Ab)}{2\alpha}\geq 1$, $\sigma_{k+1}(\Bb)\leq 1$, $\|\Bb \|_2=\frac{\|\Ab \|_2}{2\alpha} \geq 1$ and $\kappa_{\Bb}= \| \Bb\|_2$.

Let the output of Algorithm~\ref{alg:slq} with $\Bb$ as input be $f_{\Bb}(x)=\sum_{j=1}^m w_j^2 \delta(x-\lambda_{j}(\Tb_1))=\sum_{j=1}^m w_j^2 \delta \left(x-\frac{\lambda_j(\Tb)}{\sigma_{l+1}(\Ab)} \right)=f\left(\frac{x}{\sigma_{l+1}(\Ab)} \right)$ where $w_j =\vb_j^T\bv{e}_1$ (recall that $\vb_j$ is the eigenvector of $\Tb$ corresponding to $\lambda_j(\Tb)$). Also, the spectral density of $\Bb$ is given by $s_{\Bb}(x)=\frac{1}{n}\sum_{j=1}^m \delta(x-\lambda_j(\Tb_1))=\frac{1}{m}\sum_{j=1}^m \delta(x-\frac{\lambda_j(\Tb)}{\sigma_{l+1}(\Ab)})=s_{\Ab}\left(\frac{x}{\sigma_{l+1}(\Ab)} \right)$. Thus, we must have:
 \begin{align}\label{Eq:w1_b}
     W_1(s,f) \leq 2\alpha \cdot W_1(s_{\Bb},f_{\Bb}).
 \end{align}
Let $L=\|\Bb \|_2=\frac{1}{2\alpha} \|\Ab \|_2$. Then, the spectrum of both $\Bb$ and $\Tb_1$ are in $[-L,L]$ (as $|\lambda_i(\Tb_1)| \leq |\lambda_i(\Bb)|$ for all $i \in [m]$ by the minimax principle). So, we have:
      \begin{align}\label{Eq:w1_1}
          W_1(s_{\Bb},f_{\Bb}) &= \sup_{h \in \text{1-Lip}}\int_{-L}^{L} h(x) \left(s_{\Bb}(x)-f_{\Bb}(x) \right) dx \nonumber \\
          &=\int_{-L}^{L} h^*(x) \left(\frac{1}{n}\sum_{i=1}^n \delta(x-\lambda_i(\Bb))- \sum_{i=1}^m w_i^2\delta(x-\lambda_i(\Tb_1)) \right) dx \nonumber \\
          &= \underbrace{\int_{-L}^{L} h^*(x) \left(\frac{1}{n}\sum_{i=1}^k \delta(x-\lambda_i(\Bb))- \sum_{i=1}^k w_i^2\delta(x-\lambda_i(\Tb_1)) \right) dx}_\text{$I_1$} \nonumber \\
          &+ \underbrace{\int_{-L}^{L}h^*(x) \left(\frac{1}{n}\sum_{i=k+1}^n \delta(x-\lambda_i(\Bb))- \sum_{i=k+1}^m w_i^2\delta(x-\lambda_i(\Tb_1)) \right) dx}_\text{$I_2$}.
      \end{align}
      Here, $h^*(x)$ is a $1$-Lipschitz function that maximizes the integral for computing the Wasserstein distance. We will bound $I_1$ and $I_2$ separately. 

Note that $\sigma_{i}(\Bb) \leq \sigma_{k+1}(\Bb)\leq 1 $ for $i \in \{k+1, \ldots, n\}$ and  $\sigma_{i}(\Tb_1) \leq \sigma_{i}(\Bb) \leq 1$  for $i \in \{k+1, \ldots, m\}$ by the minimax principle. So, the support of $\frac{1}{n}\sum_{i=k+1}^n \delta(x-\lambda_i(\Bb))$ and $\sum_{i=k+1}^m w_i^2\delta(x-\lambda_i(\Tb_1)$ is in $[-1,1]$ and we can write $I_2$ as follows:
\begin{align*}
   I_2 &= \int_{-1}^{1} h^*(x) \left(\frac{1}{n}\sum_{i=k+1}^n \delta(x-\lambda_i(\Bb))- \sum_{i=k+1}^m w_i^2\delta(x-\lambda_i(\Tb_1)) \right) dx = \int_{-1}^{1} h^*(x) \left( r_1(x) - r_2(x)) \right) dx,
\end{align*}
where $r_1(x)=\sum_{i=k+1}^n \frac{1}{n}\delta(x-\lambda_i(\Bb))$ and $r_2(x)=\sum_{i=k+1}^m  w_i^2\delta(x-\lambda_i(\Tb_1))$.

Let $N=O(1/\epsilon)$. We will now bound this by following the Chebyshev moment matching proof of Lemma 3.1 in~\cite{braverman:2022} where we match $N$ normalized Chebyshev moments of $r_1(x)$ and $r_2(x)$. Let $\bar{h}_N(x)=\sum_{i=0}^N b_i \langle h^*, w.\Bar{T}_i \rangle \Bar{T}_i$ (where $w(x)=\frac{1}{\sqrt{1-x^2}}$) be the function from Fact~\ref{fact1} for constants $0\leq b_N \ldots \leq b_0=1$  such that $\max_{x \in [-1,1]}|h^*(x)-\bar{h}_N(x)| \leq \frac{18}{N}$. Then, using triangle inequality, the integral can be upper bounded:
\begin{align}\label{Eq:I2}
I_2 &\leq \underbrace{\int_{-1}^{1} |h^*(x)-\bar{h}_N(x)|(r_1(x)-r_2(x)) dx}_{t_1} + \underbrace{\int_{-1}^{1} \bar{h}_N(x) (r_1(x)-r_2(x)) dx}_{t_2}.
\end{align}
Since $\max_{x \in [-1,1]}|h^*(x)-\bar{h}_N(x)| \leq \frac{18}{N}$, $\int_{-1}^{1}r_1(x)=\frac{n-k}{n}$ and $\int_{-1}^{1} r_2(x) = \sum_{i=k+1}^m w_i^2$, we have $$t_1 \leq \frac{18}{N}\left( \frac{n-k}{n}+\sum_{i=k+1}^m w_i^2 \right) \leq \frac{36}{N},$$ where we used the fact that $\sum_{i=k+1}^m w_i^2 \leq 1$. Next, we bound $t_2$ using the Chebyshev series expansion of $\bar{h}_N(x)$. First note that $r_1(x)-r_2(x) \in [-1,1]$ and so its Chebsyshev series expansion is $\sum_{i=0}^{\infty} \langle r_1-r_2, \Bar{T}_i \rangle \Bar{T}_i$. This gives us:
\begin{align*}
    t_2 & = \int_{-1}^{1} \bar{h}_N(x) w(x) \frac{r_1(x)-r_2(x)}{w(x)} dx =\int_{-1}^{1} \bar{h}_N(x) w(x) \cdot \sum_{i=0}^{\infty} \langle r_1-r_2, \Bar{T}_i \rangle \Bar{T}_i dx  \\
    &=  \int_{-1}^{1}  w(x) \left(\sum_{i=0}^N b_i \langle h^*, w.\Bar{T}_i \rangle \Bar{T}_i \right) \left(\sum_{i=0}^{\infty} \langle r_1-r_2, \Bar{T}_i \rangle \Bar{T}_i \right) dx.
\end{align*}
By the orthogonality of Chebyshev polynomials under weight $w(x)$ and since $\langle \Bar{T}_k,w\Bar{T}_k \rangle =1$ for $k \in [N]$, we can bound $t_2$ as:
\begin{align}\label{Eq:t2}
    t_2 &\leq \sum_{i=0}^N \langle h^*, w.\Bar{T}_i \rangle \cdot (\langle r_1, \Bar{T}_i \rangle -\langle r_2, \Bar{T}_i \rangle) \nonumber \\
    &\leq \langle h^*, w.\Bar{T}_0 \rangle \cdot (\langle r_1, \Bar{T}_0 \rangle -\langle r_2, \Bar{T}_0 \rangle) +\sum_{i=1}^N |\langle h^*, w.\Bar{T}_i \rangle| \cdot |\langle r_1, \Bar{T}_i \rangle -\langle r_2, \Bar{T}_i \rangle| \nonumber \\
    &\leq \frac{\langle h^*, w.\Bar{T}_0 \rangle}{\sqrt{\pi}} \cdot \left( \frac{n-k}{n}-\sum_{j=k+1}^m w_j^2 \right) +\sum_{i=1}^N \frac{2}{i} \cdot |\langle r_1, \Bar{T}_i \rangle -\langle r_2, \Bar{T}_i \rangle| \nonumber \\
    &\leq  \frac{\langle h^*, w.\Bar{T}_0 \rangle}{\sqrt{\pi}} \cdot \left( \frac{n-k}{n}-\sum_{j=k+1}^m w_j^2 \right) +\sum_{i=1}^N \frac{2}{i} \cdot \bigg| \frac{1}{n}\sum_{j=k+1}^n \Bar{T}_i(\lambda_j(\Bb)) -\sum_{j=k+1}^m w_j^2 \Bar{T}_i(\lambda_j(\Tb_1)) \bigg|,
\end{align}
where for the first step, we used the fact that $\frac{\hat{b}_{N}(i)}{\hat{b}_{N}(0)} \leq 1$, and $ \int_{-1}^{1}\langle r_1-r_2, \Bar{T}_i \rangle = \langle r_1, \Bar{T}_i \rangle -\langle r_2, \Bar{T}_i \rangle$ for $i \in [N]$. For the second step, we bound the sum from $1$ to $N$ by its absolute value. For the third inequality, we used the fact that $\Bar{T}_o(x)=\frac{1}{\sqrt{\pi}}$ and so $\langle r_1, \Bar{T}_0 \rangle=\frac{n-k}{n\sqrt{\pi}}$ and $\langle r_2, \Bar{T}_0 \rangle=\frac{\sum_{j=k+1}^m w_j^2}{\sqrt{\pi}}$ for the first term and use Fact~\ref{fact2} which gives us $|\langle h^*, w.\Bar{T}_i \rangle| \leq \frac{2}{i}$ for $ i\in [N]$ for the second term. For the final step, we replace $\langle r_1, \Bar{T}_i \rangle$ and $\langle r_2, \Bar{T}_i \rangle$ by evaluating the integrals. Let $t_i(x)$ be the polynomial defined in Lemma~\ref{Lem:slq_mom}. Then, using triangle inequality, for constants $c_9,c_{10}>0$ and $C>0$, we get for any $i \in [N]$:
\begin{align}\label{Eq:t22}
 \bigg| \frac{1}{n}\sum_{j=k+1}^n \Bar{T}_i(\lambda_j(\Bb)) -\sum_{j=k+1}^m w_j^2 \Bar{T}_i(\lambda_j(\Tb_1)) \bigg| &\leq  \bigg| \frac{1}{n}\sum_{j=k+1}^n \Bar{T}_i(\lambda_j(\Bb))-  \langle t_i,s_{\Bb} \rangle \bigg| \nonumber \\
 &+  \bigg| \sum_{j=k+1}^m w_j^2 \Bar{T}_i(\lambda_j(\Tb_1))-  \langle t_i,f_{\Bb} \rangle \bigg| + \bigg| \langle t_i,f_{\Bb} \rangle -\langle t_i,s_{\Bb} \rangle  \bigg| \nonumber \\
 &\leq \frac{g^{c_9 l}_{\min}}{(n \cdot \kappa)^{c_9/\epsilon}}+\frac{g^{c_{10} l}_{\min}}{(n \cdot \kappa)^{c_{10}/\epsilon}} \nonumber \\
 &+\frac{C \log (N/\delta)}{\sqrt{n}}.
\end{align}
To bound $\bigg| \langle t_i,f_{\Bb} \rangle -\langle t_i,s_{\Bb} \rangle  \bigg|$, the final step uses Lemma~\ref{Lem:slq_hutch} and the fact that $\kappa_{\Bb}=\|\Bb \|_2=\frac{\|\Ab \|_2}{2\alpha}=\kappa$. Putting together the bounds on $t_1$ and $t_2$ from~\eqref{Eq:t2} and~\eqref{Eq:t22}, and bounding $\sum_{i=1}^N \frac{2}{i} \leq \log N$, we get from~\eqref{Eq:I2} (for some constant $c_{11}$):
\begin{align}\label{Eq:I2_final}
    I_2 \leq  \frac{36}{N}+\frac{\langle h^*, w.\Bar{T}_0 \rangle}{\sqrt{\pi}} \cdot \left( \frac{n-k}{n}-\sum_{j=k+1}^m w_j^2 \right)+\frac{2g^{c_{11}l}_{\min} \log N}{(n \cdot \kappa)^{c_{11}/\epsilon}}+\frac{C \log (N/\delta)\log N}{\sqrt{n}}.
\end{align}
We now bound $I_1$. We have the following:
\begin{align}\label{Eq:I1}
    I_1
    &= \int_{-L}^{L} h^*(x) \left(\frac{1}{n}\sum_{i=1}^k \delta(x-\lambda_i(\Bb))- \sum_{i=1}^k w_i^2\delta(x-\lambda_i(\Tb_1)) \right) dx \nonumber \\
    &= \frac{k}{n}\int_{-L}^{L} h^*(x) \left(\frac{1}{k}\sum_{i=1}^k \delta(x-\lambda_i(\Bb))- \frac{1}{k}\sum_{i=1}^k \delta(x) \right) dx \nonumber\\
    &+ \sum_{j=1}^k w_j^2 \int_{-L}^{L} h^*(x) \left( \frac{\sum_{i=1}^k w_i^2 \delta(x) }{\sum_{j=1}^k w_j^2} -\frac{\sum_{i=1}^k w_i^2 \delta(x-\lambda_i(\Tb_1))}{\sum_{j=1}^k w_j^2} \right) dx \nonumber \\
    &+ \int_{-L}^{L} h^*(x) \left( \frac{1}{n}\sum_{i=1}^k\delta(x) -\sum_{i=1}^k w_i^2 \delta(x) \right) dx.
\end{align}
We now bound the three terms above. First observe that $\frac{1}{k}\sum_{i=1}^k \delta(x-\lambda_i(\Bb))$ and $\frac{1}{k}\sum_{i=1}^k \delta(x)$ are density functions of distributions (defined by dirac deltas at $\lambda_1(\Bb),\ldots, \lambda_k(\Bb)$ with weights $\frac{1}{n}$ and at $0$ with weight $1$ respectively). So, $\sup_{h \in 1-\text{Lip}}\int_{-L}^{L} h^*(x) \left(\frac{1}{k}\sum_{i=1}^k \delta(x-\lambda_i(\Bb))- \frac{1}{k}\sum_{i=1}^k \delta(x) \right)$ is the Wasserstein-$1$ distance between the distributions. Using the Earth mover's distance interpretation of Wasserstein-$1$ distance, we have:
\begin{align*}
    \int_{-L}^{L} h^*(x) \left(\frac{1}{k}\sum_{i=1}^k \delta(x-\lambda_i(\Bb))- \frac{1}{k}\sum_{i=1}^k \delta(x) \right) &\leq \sup_{h \in 1-\text{Lip}}\int_{-L}^{L} h(x) \left(\frac{1}{k}\sum_{i=1}^k \delta(x-\lambda_i(\Bb))- \frac{1}{k}\sum_{i=1}^k \delta(x) \right) \\
    &\leq \frac{\sum_{i=1}^k|\lambda_i(\Bb)|}{k} \leq \|\Bb\|_2.
\end{align*}
Similarly, by interpreting $\frac{\sum_{i=1}^k w_i^2 \delta(x) }{\sum_{j=1}^k w_j^2}$ and $\frac{\sum_{i=1}^k w_i^2 \delta(x-\lambda_i(\Tb_1))}{\sum_{j=1}^k w_j^2}$ as probability densities and using the earth mover's distance interpretation of Wasserstein-$1$ distance we have:
\begin{align*}
  \int_{-L}^{L} h^*(x) \left( \frac{\sum_{i=1}^k w_i^2 \delta(x) }{\sum_{j=1}^k w_j^2} -\frac{\sum_{i=1}^k w_i^2 \delta(x-\lambda_i(\Tb_1))}{\sum_{j=1}^k w_j^2} \right) dx \leq \sum_{i=1}^k \frac{w_i^2 |\lambda_i(\Tb_1)|}{\sum_{j=1}^k w_j^2} \leq  \sum_{i=1}^k \frac{w_i^2 \|\Tb_1 \|_2}{\sum_{j=1}^k w_j^2}.
\end{align*}
Finally, to bound $\int_{-L}^{L} h^*(x) \left( \frac{1}{n}\sum_{i=1}^k\delta(x) -\sum_{i=1}^k w_i^2 \delta(x) \right) dx$, we again use the proof technique outlined in Lemma 3.1 of~\cite{braverman:2022} and for bounding $I_2$ here. Observe that similar to~\eqref{Eq:I2} we have:
\begin{align*}
    &\int_{-L}^{L} h^*(x) \left( \frac{1}{n}\sum_{i=1}^k\delta(x) -\sum_{i=1}^k w_i^2 \delta(x) \right) dx \\
    &\leq \underbrace{\int_{-1}^{1} (h^*(x)-\Bar{h}_N(x))\left( \frac{1}{n}\sum_{i=1}^k\delta(x) -\sum_{i=1}^k w_i^2 \delta(x) \right) dx}_{\text{$t_1$}} + \underbrace{\int_{-1}^{1} \bar{h}_N(x) \left( \frac{1}{n}\sum_{i=1}^k\delta(x) -\sum_{i=1}^k w_i^2 \delta(x) \right) dx}_{\text{$t_2$}}.
\end{align*}
We can bound $t_1$ and $t_2$ the same way as in the case of $I_2$  to get $t_1 \leq \frac{36}{N}$ and (similar to~\eqref{Eq:t2})
\begin{align*}
    t_2 &\leq \frac{\langle h^*,w\cdot \Bar{T}_o \rangle}{\sqrt{\pi}}\left( \frac{k}{n}-\sum_{j=1}^k w_j^2\right)+\sum_{i=1}^N \frac{2}{i}\cdot \bigg| \frac{1}{n}\sum_{j=1}^k\Bar{T}_i(0) -\sum_{j=1}^k w_j^2 \Bar{T}_i(0) \bigg| \\
    &\leq \frac{\langle h^*,w\cdot \Bar{T}_o \rangle}{\sqrt{\pi}}\left( \frac{k}{n}-\sum_{j=1}^k w_j^2\right)+ \sum_{i=1}^N \frac{2}{i} \bigg| \frac{k}{n} -\sum_{j=1}^k w_j^2  \bigg|.
\end{align*}
From Lemma~\ref{Lem:wts}, we have $\sum_{j=1}^k w_j^2 \leq \frac{Ck \sqrt{\log 2/\delta}}{n}$ for some large $C>0$. So, using the upper bounds on all the terms on the right hand side of~\eqref{Eq:I1} (and using the fact that $\| \bv{T}_1\|_2 \leq \|\Bb \|_2$), we get that:
\begin{align}
    I_1 &\leq \frac{k}{n}\|\Bb \|_2 + (\sum_{i=1}^k w_i^2) \|\Tb_1 \|_2 + \frac{\langle h^*,w\cdot \Bar{T}_o \rangle}{\sqrt{\pi}}\left( \frac{k}{n}-\sum_{j=1}^k w_j^2\right)+ \sum_{i=1}^N \frac{2}{i} \bigg| \frac{k}{n} -\sum_{j=1}^k w_j^2  \bigg| \notag\\
    &\leq \frac{k(1+C\sqrt{\log k/\delta})}{n}\|\Bb \|_2 +  \frac{\langle h^*,w\cdot \Bar{T}_o \rangle}{\sqrt{\pi}}\left( \frac{k}{n}-\sum_{j=1}^k w_j^2\right)+ \frac{3Ck\log N\sqrt{\log k/\delta}}{n}\label{Eq:I1_final},
\end{align}
where in the last step we used triangle inequality to bound the final term. Finally, using the bounds on $I_1$ and $I_2$ from~\eqref{Eq:I1_final} and~\eqref{Eq:I2_final} respectively in~\eqref{Eq:w1_1} (and using the fact that $\sum_{j=1}^m w_j^2=1$), we get that:
\begin{align*}
    W_1(s_{\Bb},f_{\Bb}) &\leq \frac{36}{N}+\frac{k(1+4C\log N\sqrt{\log k/\delta})}{n}\|\Bb \|_2  +\frac{2g^{c_{11}l}_{\min} \log N}{(n \cdot \kappa)^{c_{11}/\epsilon}}+\frac{C \log (N/\delta)\log N}{\sqrt{n}} \\
    &\leq \frac{36}{N}+\frac{5Ck\log N\sqrt{\log k/\delta}}{n}\|\Bb \|_2 +\frac{2C \log (N/\delta)\log N}{\sqrt{n}},
\end{align*}
where we also use the fact that the constant $c_{11}$ can be chosen to be large enough so that the third term in the first inequality can be made small enough. Recall that we  set $N=O(\frac{1}{\epsilon})$. From~\eqref{Eq:w1_b}, we also have $W_1(s,f) \leq 2\alpha W_1(s_{\Bb},f_{\Bb})$. Thus, we finally get (using the fact $\Bb=\frac{\Ab}{2\alpha}$):
\begin{align*}
    W_1(s,f) &\leq 72\epsilon \alpha+\frac{5Ck\log (1/\epsilon)\sqrt{\log k/\delta}}{n}\|\Ab \|_2 + \frac{4C \log (1/\epsilon\delta)\log (1/\epsilon) \alpha}{\sqrt{n}} \\
    &\leq \frac{5Cl\log (1/\epsilon)\sqrt{\log (k/\delta)}}{n}\|\Ab \|_2 + \left(72\epsilon+\frac{4C \log (1/\epsilon\delta )\log (1/\epsilon) }{\sqrt{n}} \right)\sigma_{l+1}(\Ab) \\
    &\leq \frac{5Cl\log (1/\epsilon)\sqrt{\log (l/\delta)}}{n}\|\Ab \|_2 +\left(72\epsilon+\frac{4C \log (1/\epsilon\delta )\log (1/\epsilon) }{\sqrt{n}} \right)\sigma_{l+1}(\Ab).
\end{align*}
where we also used the fact that $ 2 \alpha < 2\sigma_{l+1}(\Ab) +\frac{2\|\Ab \|_2 g^{c/4}_{\min}}{n^{c/4}}$ and $k \leq l$. Also note that the term $\left(72\epsilon+\frac{4C \log (1/\epsilon\delta )\log (1/\epsilon) }{\sqrt{n}} \right)\frac{2\|\Ab \|_2 g^{c/4}_{\min}}{n^{c/4}}< \frac{Cl\log (1/\epsilon)\sqrt{\log (l/\delta)}}{n}\|\Ab \|_2$ for a large $c$ and hence is absorbed into that term by making the constant $C$ larger. We get the final bound by adjusting $\epsilon$ by constant factors. This concludes the case when such a $k$ exists.

\medskip

\noindent\textbf{Case 2:} Now, when such a $k$ doesn't exist, i.e. we have $\|\Ab \|_2 \leq  2 \alpha < 2\sigma_{l+1}(\Ab) +\frac{2\|\Ab \|_2 g^{c/4}_{\min}}{n^{c/4}}$, from Theorem~\ref{thm:slq_main} we directly get error $\W_1(s_{\Ab},f) \leq 2\epsilon \alpha +\frac{C \log (1/\epsilon \delta) \log (1/\epsilon)}{\sqrt{n}}2\alpha \leq \epsilon\sigma_{l+1}(\Ab)+ \frac{2C \log (1/\epsilon \delta) \log (1/\epsilon)}{\sqrt{n}}\sigma_{l+1}(\Ab)+\frac{C\|\Ab \|_2}{n} $ after adjusting $\epsilon$ by constants. 

Finally, observe that the condition $n \geq \Omega(\log (1/\delta))$ in Lemmas~\ref{Lem:slq_mom} and~\ref{lem:abs-slq-true-A} must always be satisfied if we want a non-trivial ($\leq n$) number of matrix vector products with $\Ab$. Adjusting $\delta$ by some constant factors gives us the final bound.
\end{proof}

\subsection{Variance reduced SLQ}\label{sec:slq4}

We introduce our variant of SLQ with better error guarantees here. The algorithm, which we call the Variance reduced SLQ is described in Algorithm~\ref{alg:slq+}. The main difference between this algorithm and Algorithm~\ref{alg:slq} is that here, we set the weights $w_i^2$ corresponding to the \textit{converged} large magnitude eigenvalues of $\Tb$ (the output of Lanczos) to $\frac{1}{n}$ (and adjust the remaining weights so that the square of the weights sum to one). This is described in lines 4-9 of Algorithm~\ref{alg:lanczos-slq}. From Lemma~\ref{lem:eig_error_slq}, we know that in the presence of a constant multiplicative gap between $\sigma_k(\Ab)$ and $\sigma_{l+1}(\Ab)$, the top $k$ eigenvalues of $\Tb$ will approximately be equal to those of $\Ab$. Hence the \textit{correct} weights corresponding to these eigenvalues must be $\frac{1}{n}$. This lets us avoid the $\Tilde{O}(\frac{l}{n}\|\Ab \|_2)$ error SLQ was incurring on the large magnitude eigenvalues. However, we do not actually know the value of $k$ or even if there exists such a $k \in [l]$ with a constant multiplicative gap. Hence, we check two convergence conditions for each of the top $l$ indices which will always be true for the top $k \leq l$ eigenvectors of $\Tb$ (the output of Lanczos) and corresponding weights $(\bv{v}_j^T\bv{e}_1)^2$ if such a $k \in [l]$ exists. 

The two convergence conditions to find the indices for which we set the weights to $\frac{1}{n}$ are given in line 5 of Algorithm~\ref{alg:slq+}. The first condition ($\|\Ab \Qb \vb_j-  \lambda_j(\Tb) \Qb\vb_j\|_2 \leq \frac{\|\Ab \|_2}{n^{\beta/\epsilon}}$) checks if an eigenvector  $\vb_j$ of  $\Tb$ (and its corresponding eigenvalue) is also approximately an eigenvector and eigenvalue of $\Ab$. Note that from Lemma~\ref{lemma:converge_slq} this condition will always be true for the top $k$ eigenvalues and eigenvectors of $\Tb$ as long as there is at least a constant multiplicative factor gap between $\sigma_k(\Ab)$ and $\sigma_{l+1}(\Ab)$. This helps us set the correct weights of $\frac{1}{n}$ for the top $k$ eigenvalues. However, there might be some eigenvectors of $\Tb$ outside of the top $k$ indices for which this condition is also true. Unfortunately, we can't guarantee that the corresponding eigenvalues of $\Tb$ outside of the top $k$ have converged to an eigenvalue of $\Ab$. The second condition ($(\bv{v}_j^2\bv{e}_1)^2 \leq \frac{\sqrt{\log (l/\delta)}}{n}$) ensures that in case this happens, we don't incur too much error by correcting the weights to $\frac{1}{n}$. Note that from Lemma~\ref{Lem:wts}, this condition will also be true for the top $k$ weights of the SLQ distribution as long as there is at least a constant multiplicative factor gap between $\sigma_k(\Ab)$ and $\sigma_{l+1}(\Ab)$. But in case there is spurious convergence of an eigenvector of $\Tb$ without the corresponding eigenvalue of $\Tb$ converging to an eigenvalue of $\Ab$, this condition essentially bound the error we incur by correcting the weights by $\frac{\sqrt{\log(l/\delta)}}{n}$. Also, since the corresponding eigenvalue of $\Ab$ (and $\Tb$) is at most $\sigma_{k+1}(\Ab)$, the total error we incur by incorrectly setting weights for these spurious eigenvectors is at most $\Tilde{O}(\frac{l\sigma_{k+1}(\Ab)}{n}) \leq \Tilde{O}(\frac{l\sigma_{l+1}(\Ab)}{n})$ as opposed to $\Tilde{O}(\frac{l\|\Ab \|_2}{n})$ as the magnitude of these eigenvalues will always be less than $\sigma_{k+1}(\Ab) \leq O(\sigma_{l+1}(\Ab))$. 

We now analyze the algorithm below. The proof is similar to that of Theorem~\ref{thm:slq} except now, we use the Backward Perturbation Bound (Lemma~\ref{Lem:bkwd_err}) to first show that there exists some matrix with the same converged eigenvalues and eigenvectors of $\Ab$ and which is spectrally close to $\Ab$ (as we had done in the proof of Theorem~\ref{thm:sde1}). This helps us bound the error incurred on the converged eigenvectors. For bounding the Wasserstein error of the spectral density corresponding to the eigenvalues with non-converged eigenvectors (which will also have small magnitude), we use the moment matching analysis of Theorem~\ref{thm:slq} again.

\begin{savenotes}
\begin{algorithm}[H] 
\caption{Variance reduced Stochastic Lanczos Quadrature }
\label{alg:slq+}
\begin{algorithmic}[1]
\Require{Symmetric $\bv A \in \mathbb{R}^{n\times n}$, number of iterations $m (\leq n)$, convergence parameters $l \in [n]$, $\beta,C>0$.}\footnote{refer to Theorem~\ref{thm:varslq} for instructions on setting these parameters.}
\State Sample $\bv g \sim \mathcal{U}(\mathcal{S}^{n-1})$. 
\State Run Lanczos (Algorithm~\ref{alg:lanczos-slq}) on $\Ab, \bv g$ for $m$ steps to compute symmetric tridiagonal matrix $\Tb \in \R^{m\times m}$ and orthonormal basis $\Qb \in \R^{n \times m}$ such that $\Tb=\Qb^T\Ab \Qb$~\cite{golub2009matrices}. Let the eigenvectors of $\Tb$ be $\vb_1, \ldots, \vb_m$.
\State Set $S=\{\}$
\For{$j=1, \ldots, l$} 
\If{$\|\Ab \Qb \vb_j-  \lambda_j(\Tb) \Qb\vb_j\|_2 \leq \frac{\|\Ab \|_2}{n^{\beta}}$ and $(\bv{v}_j^T \bv{e}_1)^2 \leq \frac{C\sqrt{\log (l/\delta)}}{n}$\footnote{$\delta \in (0,1)$ is the failure probability.}} \label{varslq:line1}
\State $S=S \cup \{j \}$
\EndIf
\EndFor
\State Set $f (x) =  \sum_{j \in S} \frac{1}{n}\delta(x-\lambda_j(\bv{T}))+ \bigg(\frac{1-\frac{|S|}{n}}{\sum_{i \in [m]\setminus S} w_i^2}\bigg)\sum_{j \in [m]\setminus S} w_j^2\delta(x-\lambda_j(\bv{T}))$, where $w_j = \bv{v}_j^T \bv{e}_1$.
\State \Return $f (x)$.
\end{algorithmic}
\end{algorithm}
\end{savenotes}

\begin{reptheorem}{thm:varslq}
Let $\Ab\in\R^\n$ be symmetric and consider any $l \in [n]$ and $\epsilon, \delta \in (0,1)$. Let $g_{\min}=\min_{i \in [l]} \frac{\sigma_i(\Ab)-\sigma_{i+1}(\Ab)}{\sigma_i(\Ab)}$ and $\kappa=\frac{\|\Ab \|_2}{2\alpha}$. Let $\alpha=\max\left( \sigma_{l+1}(\Ab),\frac{\|\Ab \|_2 g^{c/4}_{\min}}{n^{c/4}} \right)$ for some constant $c>0$. Algorithm~\ref{alg:slq+} run for $m = O(l\log \frac{1}{g_{\min}}+\frac{\log (n \cdot \kappa)}{\epsilon})$ iterations, performs $m$ matrix vector products with $\bv A$ and outputs a distribution $f$ such that, with probability at least $1-\delta$, for large fixed constants $C>0$ and $c_2>0$,
$$W_1(s_{\Ab},f) \leq \epsilon \cdot \sigma_{l+1}(\bv A) + \left (\frac{C \log(n/\epsilon) \log(1/\epsilon)}{\sqrt{n}} + \frac{C l \log(1/\epsilon) \sqrt{\log (l/\delta)}}{n} \right) \cdot \sigma_{l+1}(\bv A) + \frac{\| \Ab\|_2}{n^{c_2}}.$$
\end{reptheorem}
\begin{proof}
Throughout the proof, for ease of writing, we will abuse notation slightly and assume that the first convergence condition in line~\ref{varslq:line1} of Algorithm~\ref{alg:slq+} is given by $\|\Ab \Qb \vb_j-  \lambda_j(\Tb) \Qb\vb_j\|_2 \leq \frac{\|\Ab \|_2}{n^{\beta/\epsilon}}$ i.e. the first parameter is $\beta/\epsilon$ instead of just $\beta$.
Let $k \in [l]$ be such that $\sigma_k(\Ab) \geq 2\alpha$ and $\sigma_{k+1}(\Ab) < 2\alpha$. We will again consider two cases as in Theorem~\ref{thm:slq}, i.e. when such a $k$ exists and when it doesn't. 
\medskip

\noindent\textbf{Case 1:} We first consider the case when such a $k$ exists. In this case, similar to proof of case 1 of Theorem~\ref{thm:slq_main}, let $\Bb=\frac{1}{2\alpha}\Ab$ and let $\Tb_1=\frac{1}{2\alpha}\Tb$ where $\bv{T}_1$ is the output of running Lanczos with $\Bb$ as input. Then, as in~\eqref{Eq:w1_b}, we have: 
 \begin{align}\label{Eq:sf}
     W_1(s_{\Ab},f) \leq 2\alpha W_1(s_{\Bb},f_{\Bb}),
\end{align}
 where $s_{\Bb}(x)=\sum_{j=1}^{m} \frac{1}{m}\delta(x-\lambda_i(\Bb))$ and $f_{\Bb}(x)=\sum_{j \in S} \frac{1}{n}\delta(x-\lambda_j(\bv{T}_1)+ \bigg(\frac{1-\frac{|S|}{n}}{\sum_{i \in [m]\setminus S} w_i^2}\bigg)\sum_{j \in [m]\setminus S} w_j^2\delta(x-\lambda_j(\bv{T}_1))$, where $w_j=\vb_j^T\bv{e}_1$ for each $j \in [m]$ and $S$ is the set of indices containing the \textit{converged} eigenvectors as defined in Algorithm~\ref{alg:slq+}. Let $\Vb_S \in \R^{m \times |S|}$ contain the set of all eigenvectors $\vb_i$ of $\Tb_1$ such that $i \in S$. Let $\Zb_S=\Qb\Vb_S$ where $\Qb$ is the orthonormal basis of the Krylov subspace generated by the Lanczos algorithm (Algorithm~\ref{alg:lanczos-slq}). Then, $\|\Bb\Zb_S-\Zb_S\Tilde{\bv{\Lambda}} \|_F=\sqrt{\sum_{i\in S} \|\Bb\Qb\vb_i-\lambda_i(\Ab)\Qb\vb_i\|^2_2} \leq \frac{\|\Bb \|_2}{n^{(\beta-1)/\epsilon}}$ where $\Tilde{\bv{\Lambda}}$ is a matrix with the eigenvalues of $\Tb_1$ corresponding to the eigenvectors of $\Vb_S$ on its diagonal and zeros everywhere else. Then, using the backwards error bound (Lemma~\ref{Lem:bkwd_err}) we get that there exists a matrix $\Cb$ such that $\Cb\Zb_S=\Zb_S\Tilde{\bv{\Lambda}}$ and $\|\Bb-\Cb \|_2 \leq \frac{2\|\Bb \|_2}{n^{(\beta-1)/\epsilon}}$. Using Weyls' inequality (Fact~\ref{fact:weyl}), we get that for all $i \in [n]$
 \begin{align}\label{eq:weyl_slq}
     |\lambda_i(\Bb)-\lambda_i(\Cb)| \leq \frac{2\|\Bb \|_2}{n^{(\beta-1)/\epsilon}}.
 \end{align}
 Then, $\frac{1}{n}\sum_{i=1}^n |\lambda_i(\Bb)-\lambda_i(\Cb)| \leq \frac{2\|\Bb \|_2}{n^{(\beta-1)/\epsilon}}$ which implies that 
 \begin{align}\label{Eq:sb}
     \W_1(s_{\Bb},s_{\Cb}) \leq \frac{2\|\Bb \|_2}{n^{(\beta-1)/\epsilon}}.
 \end{align}
 where $s_{\Cb}(x)$ is the spectral density function of $\Cb$. Let $S_1 \subseteq [n]$ such $|S_1|=|S|$ and for every $j \in S_1$, there exists an $i \in S$ such that $\lambda_j(\Cb)=\lambda_i(\Tb_1)$. We know such a set $S_1$ must exist as $\Tilde{\bv{\Lambda}}$ are eigenvalues of $\Cb$. Let $L=\max(\|\Cb \|_2,\|\Bb \|_2)$. Then we have:
 \begin{align}\label{Eq:sc}
     \W_1(s_{\Cb},f_{\Bb}) &=\sup_{h \in 1-\text{Lip}}\int_{-L}^{L} h(x)(s_{\Cb}(x)-f_{\Bb}(x)) dx \nonumber\\
     &= \int_{-L}^{L} h^*(x) \bigg( \frac{1}{n}\sum_{i=1}^n \delta(x-\lambda_i(\Cb))- \sum_{j \in S} \frac{1}{n}\delta(x-\lambda_j(\bv{T}_1) \nonumber \\
     &- \bigg(\frac{1-\frac{|S|}{n}}{\sum_{i \in [m]\setminus S} w_i^2}\bigg)\sum_{j \in [m]\setminus S} w_j^2\delta(x-\lambda_j(\bv{T}_1))\bigg) dx \nonumber \\
     &= \underbrace{\frac{|S|}{n}\int_{-L}^{L} h^*(x) \bigg(\frac{1}{|S|}\sum_{i \in S_1}\delta(x-\lambda_i(\Cb))-\frac{1}{|S|}\sum_{i \in S}\delta(x-\lambda_i(\Tb_1))  \bigg) dx}_{I_1} \nonumber \\
     &+ \underbrace{\frac{n-|S|}{n}\int_{-L}^{L} h^*(x) \bigg(\frac{1}{n-|S|}\sum_{i \in [n]\setminus S_1}\delta(x-\lambda_i(\Cb))-\frac{1}{\sum_{j \in [m]\setminus S}w^2_j}\sum_{i \in [m]\setminus S}w^2_i\delta(x-\lambda_i(\Tb_1))  \bigg) dx}_{I_2}.
 \end{align}
 Here, in the second step, $h^*(x)$ is the function that maximizes the integral in the first step. By definition of the set $S_1$ we have $I_1=0$. Now we bound $I_2$. Let $s_1(x)=\frac{1}{n-|S|}\sum_{i \in [n]\setminus S_1}\delta(x-\lambda_i(\Cb))$, $s_2(x)=\frac{1}{n-|S|}\sum_{i \in [n]\setminus S_1}\delta(x-\lambda_i(\Bb))$ and $f_1(x)=\frac{1}{\sum_{j \in [m]\setminus S}w^2_j}\sum_{i \in [m]\setminus S}w^2_i\delta(x-\lambda_i(\Tb_1)$ be three density functions. Using triangle inequality, and the fact that $h^*(x)$ is a 1-Lipschitz function, we have:
  \begin{align}\label{Eq:I2_sql+}
      I_2 &= \frac{n-|S|}{n}\int_{-L}^{L} h^*(x)(s_1(x)-f_1(x)) dx  \nonumber\\
      &\leq \frac{n-|S|}{n} \W_1(s_1,f_1) \leq \frac{n-|S|}{n} \big( \W_1(s_1,s_2) +\W_1(s_2,f_1) \big).
  \end{align}
  Using~\eqref{eq:weyl_slq} and the earth movers' interpretation of the 1-Wasserstein distance, we can bound $W_1(s_1,s_2)$ as 
  \begin{align}\label{Eq:s12}
      \W_1(s_1,s_2) \leq \frac{1}{n-|S|}\sum_{i \in [n]\setminus S_1}|\lambda_i(\Cb)-\lambda_i(\Bb)| \leq \frac{2\|\Bb \|_2}{n^{(\beta-1)/\epsilon}}.
  \end{align}
  We now bound $W_1(s_2,f_1)$. Let $\zb_i=\Qb \vb_i$. Recall that $\vb_i$ are eigenvectors of $\Qb^T\Bb\Qb$ for eigenvalues $\lambda_i(\Qb^T\Bb\Qb)$, $\zb_i$ are the eigenvectors of $\Qb\Qb^T\Bb \Qb\Qb^T$ corresponding to eigenvalues $\lambda_i(\Qb^T\Bb \Qb)=\lambda_i(\Qb\Qb^T\Bb \Qb\Qb^T)$. From Lemma~\ref{lemma:converge_slq}, we have that for every $i \in [k]$, either $\|\bv{z}_i-\bv{u}_i \|_2 \leq \frac{g^{cl/4}_{\min}}{(n\cdot \kappa)^{c/2\epsilon-1}}$ or $\|\bv{z}_i+ \bv{u}_i \|_2 \leq \frac{g^{cl/4}_{\min}}{(n\cdot \kappa)^{c/2\epsilon-1}}$. From Lemma~\ref{lem:eig_error_slq}, we have that for $i \in [k]$ , $$|\lambda_i(\Bb)-\lambda_i(\Qb\Qb^T\Bb\Qb\Qb^T)| \leq \frac{g^{cl}_{\min}\|\Bb \|_2}{(n\cdot \kappa)^{c/\epsilon}}.$$  Let us assume we have $\|\bv{z}_i-\bv{u}_i \|_2 \leq \frac{g^{cl/4}_{\min}}{(n\cdot \kappa)^{c/2\epsilon-1}}$ for some $i \in [k]$. Using triangle inequality and spectral submultiplicativity, we get that for $i \in [k]$:
  \begin{align*}
      \|\Bb\bv{z}_i -\lambda_i(\Qb^T\Bb\Qb)\zb_i\| &\leq \|\Bb(\zb_i-\ub_i) \|_2 +\|\Bb\ub_i-\lambda_i(\Qb^T\Bb\Qb)\ub_i \|_2+\|\lambda_i(\Qb^T\Bb\Qb)(\ub_i-\zb_i) \|_2 \\
      &\leq \frac{g^{cl/4}_{\min}\|\Bb \|_2}{(n\cdot \kappa)^{c/2\epsilon-1}}+\|(\lambda_i(\Bb)-\lambda_i(\Qb\Qb^T\Bb\Qb\Qb^T))\ub_i \|_2 +\frac{g^{cl/4}_{\min}\|\Bb \|_2}{(n\cdot \kappa)^{c/2\epsilon-1}} \\
      &\leq \frac{3g^{cl/4}_{\min}\|\Bb \|_2}{(n\cdot \kappa)^{c/2\epsilon-3}}.
  \end{align*}
 For the second step, we also used the fact that $\lambda_i(\Qb\Qb^T\Bb \Qb\Qb^T) =\lambda_i(\Qb^T\Bb \Qb)$. We can similarly prove that $\|\Bb\bv{z}_i -\lambda_i(\Qb\Qb^T\Bb\Qb\Qb^T)\zb_i\| \leq \frac{3g^{cl/4}_{\min}\|\Bb \|_2}{(n\cdot \kappa)^{c/2\epsilon-3}}$ when $\|\bv{z}_i+\bv{u}_i \|_2 \leq \frac{g^{cl/4}_{\min}}{(n\cdot \kappa)^{c/2\epsilon-1}}$. Moreover, from Lemma~\ref{Lem:wts}, for any $i \in [k]$, we have $w_i^2=(\vb_i^T\bv{e}_1)^2 \leq \frac{C\sqrt{\log (k/\delta)}}{n} \leq \frac{C\sqrt{\log (l/\delta)}}{n}$ for some constant $C>0$. Thus, if we set the constants $\beta,C_1$ in Algorithm~\ref{alg:slq+} such that $\beta \leq \frac{c}{2}-3\epsilon$, and $C_1\geq C$, then  $\vb_1, \ldots \vb_k$ must satisfy the conditions in line 5 that $\|\Bb \Qb \vb_j-  \lambda_j(\Tb_1) \Qb\vb_j\|_2 \leq \frac{\|\Bb \|_2}{n^{\beta/\epsilon}}$ and $(\vb_j^T\bv{e}_1)^2 \leq \frac{C_1\sqrt{\log (l/\delta)}}{n}$, i.e. $[k] \subseteq S.$ Thus, $\max_{i \in [n]\setminus S} |\lambda_i(\Bb)| \leq |\lambda_{k+1}(\Bb)| \leq 1$. Also, by the minimax principle, $\max_{i \in [m]\setminus S} |\lambda_i(\Tb_1)| \leq |\lambda_{k+1}(\Tb_1)| \leq |\lambda_{k+1}(\Bb)| \leq 1$. Thus, the support of $s_2(x)$ and $f_1(x)$ is in $[-1,1]$. 
  To bound $W_1(s_2,f_1)$, we use Lemma 3.1 of~\cite{braverman:2022} to get that for any $N \in 4\mathbb{N}^+$, we have:
  \begin{align*}
      \W_1(s_2,f_1) &\leq \frac{36}{N}+2\sum_{i=1}^{N}\frac{|\langle \Bar{T}_i,s_2 \rangle- \langle \Bar{T}_i,f_1 \rangle|}{i},
  \end{align*}
  where $\Bar{T}_i$ is the $i$th normalized Chebyshev polynomial.
  Then, we get:
  \begin{align}\label{Eq:s2}
    \frac{n-|S|}{n} W_1(s_2,f_1) &\leq  \frac{n-|S|}{n}\cdot\frac{36}{N} +2\sum_{i=1}^{N}\frac{1}{i} \bigg|\frac{1}{n}\sum_{j \in [n]\setminus S_1}\Bar{T}_i(\lambda_j(\Bb))- \frac{\frac{n-|S|}{n}}{\sum_{p \in [m]\setminus S}w^2_p}\sum_{j \in [m]\setminus S} w_j^2\Bar{T}_i(\lambda_j(\Tb_1)) \bigg|.
  \end{align}
  We can set $N=O\big(\frac{1}{\epsilon} \big)$ since the total number of iterations of Lanczos is $m=\Tilde{O} \left(l +\frac{1}{\epsilon} \right)$. Then, using triangle inequality, and the fact that $[k] \subseteq S$ and $[k] \subseteq S_1$, we get that for any $i \in [N]$:
\begin{align}\label{Eq:var_slq_diff}
 &\bigg|\frac{1}{n}\sum_{j \in [n]\setminus S_1}\Bar{T}_i(\lambda_j(\Bb))- \frac{\frac{n-|S|}{n}}{\sum_{p \in [m]\setminus S}w^2_p}\sum_{j \in [m]\setminus S} w_j^2\Bar{T}_i(\lambda_j(\Tb_1)) \bigg| \nonumber \\ 
 &\leq \bigg|\frac{1}{n}\sum_{j=k+1}^n\Bar{T}_i(\lambda_j(\Bb))- \sum_{j=k+1}^m w_j^2\Bar{T}_i(\lambda_j(\Tb_1)) \bigg| +\bigg| \frac{1}{n}\sum_{j \geq k+1, j \in S_1 }\Bar{T}_i(\lambda_j(\Bb)) - \sum_{j \geq k+1, j \in S} w_j^2 \Bar{T}_i(\lambda_j(\Tb_1)) \bigg| \nonumber \\
 &+  \bigg| \frac{\frac{n-|S|}{n}}{\sum_{p \in [m]\setminus S}w^2_p}\sum_{j \in [m]\setminus S} w_j^2\Bar{T}_i(\lambda_j(\Tb_1))- \sum_{j \in [m] \setminus S} w_j^2 \Bar{T}_i(\lambda_j(\Tb_1)) \bigg|.
\end{align}
We will now bound the three terms above individually. Let $f_2(x)=\sum_{i=1}^m w_i^2 \delta(x-\lambda_i(\Tb_1))$, i.e., the output of SLQ with $\Bb$ as input. Next, observe that $\Bb$ satisfies all the conditions of  Lemma~\ref{Lem:slq_mom} since $\sigma_{k}(\Bb) =\frac{\sigma_k(\Ab)}{2\alpha}\geq 1$ and $\sigma_{k+1}(\Bb)\leq 1$. Let $t_i(x)$ be the polynomial defined in Lemma~\ref{Lem:slq_mom} for $i \in [N]$. Then, for the first term, using triangle inequality, we have that:
\begin{align*}
    \bigg|\frac{1}{n}\sum_{j=k+1}^n\Bar{T}_i(\lambda_j(\Bb))- \sum_{j=k+1}^m w_j^2 \Bar{T}_i(\lambda_j(\Tb_1)) \bigg| &\leq \bigg| \frac{1}{n}\sum_{j =k+1}^n\Bar{T}_i(\lambda_j(\Bb))-  \langle t_i,s_{\Bb} \rangle \bigg|   \nonumber \\
    &+ \bigg|\sum_{j=k+1}^m w_j^2\Bar{T}_i(\lambda_j(\Tb_1))  - \langle t_i,f_2 \rangle \bigg| 
    + \bigg| \langle t_i,f_2 \rangle -\langle t_i,s_{\Bb} \rangle  \bigg|.
\end{align*}
From Lemma~\ref{Lem:slq_mom}, we have that $\bigg|\sum_{j=k+1}^m w_j^2 \Bar{T}_i(\lambda_j(\Tb_1))  -  \langle t_i,f_2 \rangle \bigg| \leq \frac{g^{c_3 l}_{\min}}{(n \cdot \kappa)^{c_3/\epsilon}}$ and $\bigg| \frac{1}{n}\sum_{j =k+1}^n \Bar{T}_i(\lambda_j(\Bb))-  \langle t_i,s_{\Bb} \rangle \bigg| \leq \frac{g^{c_4 l}_{\min}}{(n \cdot \kappa)^{c_4/\epsilon}}$ for constants $c_3,c_4>0$ for all $i \in O(\frac{1}{\epsilon})$. From Lemma~\ref{Lem:slq_hutch}, we have that $\bigg| \langle t_i,f_2 \rangle -\langle t_i,s_{\Bb} \rangle  \bigg| \leq \frac{C_2\log(N/\delta) }{\sqrt{n}} \leq \frac{C_2\log(1/\epsilon\delta) }{\sqrt{n}}$ for some constant $C_2$ and for all $i \in O(1/\epsilon)$ with probability at least $1-\delta$. Thus, we get that 
\begin{align*}
    &\bigg|\frac{1}{n}\sum_{j=k+1}^n\Bar{T}_i(\lambda_j(\Bb))- \sum_{j=k+1}^m w_j^2 \Bar{T}_i(\lambda_j(\Tb_1)) \bigg| \\ &\leq \frac{g^{c_3 l}_{\min}}{(n \cdot \kappa)^{c_3/\epsilon}}+ \frac{g^{c_4 l}_{\min}}{(n \cdot \kappa)^{c_4/\epsilon}} + \frac{C_2\log(1/\epsilon\delta) }{\sqrt{n}} \leq \frac{2C_2\log(1/\epsilon\delta) }{\sqrt{n}},
\end{align*}
for all $i \in O(\frac{1}{\epsilon})$. We now bound the second term in~\eqref{Eq:var_slq_diff}. Note that since $|\lambda_j(\Bb)| \leq 1 $ for $j \geq k+1$, we have that $\Bar{T}_i(\lambda_j(\Bb)) \leq \sqrt{\frac{2}{\pi}}$. Also, $|\lambda_j(\Tb_1)| \leq |\lambda_{j}(\Bb)| \leq 1$ for $j \geq k+1$ and $w_j^2 \leq \frac{C_1 \sqrt{\log (l/\delta)}}{n}$ for all $ j \in S$.  Thus, we get that: 
\begin{align*}
    &\bigg| \frac{1}{n}\sum_{j \geq k+1, j \in S_1 }\Bar{T}_i(\lambda_j(\Bb)) - \sum_{j \geq k+1, j \in S} w_j^2 \Bar{T}_i(\lambda_j(\Tb_1)) \bigg| \\
    &\leq \frac{1}{n}\sum_{j \geq k+1, j \in S_1 } |\Bar{T}_i(\lambda_j(\Bb))| + \sum_{j \geq k+1, j \in S} w_j^2 |\Bar{T}_i(\lambda_j(\Tb_1))| \\
    &\leq \sqrt{\frac{2}{\pi}}\frac{|S_1|}{n}+  \sqrt{\frac{2}{\pi}}\frac{C_1|S|\sqrt{\log(l/\delta)}}{n} \leq \frac{C_3 l \sqrt{\log (l/\delta)}}{n}.
\end{align*}
for some constant $C_3>0$ and where we used the fact that $|S_1|=|S|$ and $|S| \leq l$ in the last step. We finally bound the last term in~\eqref{Eq:var_slq_diff}. Observe that we have (for some constant $C_4>0$):
\begin{align*}
    &\left| \frac{\frac{n-|S|}{n}}{\sum_{p \in [m]\setminus S}w^2_p}\sum_{j \in [m]\setminus S} w_j^2\Bar{T}_i(\lambda_j(\Tb_1))- \sum_{j \in [m] \setminus S} w_j^2 \Bar{T}_i(\lambda_j(\Tb_1)) \right| \\
    &\leq  \left| \sum_{j \in [m]\setminus S} w_j^2 \Bar{T}_i(\lambda_j(\Tb_1))\left(\frac{\frac{n-|S|}{n}}{\sum_{p \in [m]\setminus S}w^2_p}-1 \right) \right| \\
    &\leq \left| \frac{\frac{n-|S|}{n} - \sum_{p \in [m]\setminus S}w^2_p}{\sum_{p \in [m]\setminus S}w^2_p}  \right| + \left|  \sum_{p \in [m]\setminus S}w^2_p \sqrt{\frac{2}{\pi}} \right| \\
    &\leq \sqrt{\frac{2}{\pi}} \left| \sum_{p \in S} w^2_p -\frac{|S|}{n} \right| \\
    &\leq \sqrt{\frac{2}{\pi}}\frac{C_2|S|\sqrt{\log(l/\delta)}}{n} + \sqrt{\frac{2}{\pi}}\frac{|S|}{n} \\
    &\leq \frac{C_4l\sqrt{\log(l/\delta)}}{n}.
\end{align*}
In the second step above, we used the fact that since $[k] \subseteq S$, we have $\lambda_j(\Tb_1) \leq \lambda_k(\Tb_1) \leq \lambda_k(\Bb) \leq 1$ for $j \in [m] \setminus S$. So we have $\Bar{T}_i(\lambda_j(\Tb_1)) \leq \sqrt{\frac{2}{\pi}}$  for $j \in [m] \setminus S$. In the third step, we used the fact that $\sum_{i \in [m]} w_i^2=1$ and in the fourth step, we bounded $w_p^2 \leq \frac{C_2\sqrt{\log (l/\delta)}}{n}$. Finally, using the upper bounds on the three terms on the right hand side of~\eqref{Eq:var_slq_diff}, and observing that for large enough we get that for $i \in O(\frac{1}{\epsilon})$ (for constants $C_2, C_5>0$), with probability at least $1-\delta$:
\begin{align*}
\left|\frac{1}{n}\sum_{j \in [n]\setminus S_1}\Bar{T}_i(\lambda_j(\Bb))- \frac{\frac{n-|S|}{n}}{\sum_{p \in [m]\setminus S}w^2_p}\sum_{j \in [m]\setminus S} w_j^2\Bar{T}_i(\lambda_j(\Tb_1)) \right| \leq \frac{C_2 \log (1/\epsilon \delta)}{\sqrt{n}}+\frac{C_5 l\sqrt{\log(l/\delta)}}{n}.
\end{align*}
From~\eqref{Eq:s2}, we get that (using the fact that $N=O\left( \frac{1}{\epsilon} \right)$):
\begin{align*}
     \frac{n-|S|}{n}W_1(s_2,f_1) &\leq C_6\epsilon \frac{n-|S|}{n} + \left(\frac{C_6 \log (1/\epsilon \delta)}{\sqrt{n}}+\frac{C_6 l\sqrt{\log(l/\delta)}}{n} \right) \sum_{i=1}^N \frac{1}{i} \\
     &\leq  C_6\epsilon +\frac{C_6 \log (1/\epsilon \delta) \log (1/\epsilon)}{\sqrt{n}} +\frac{C_6 l\sqrt{\log(l/\delta) }\log(1/\epsilon)}{n},
\end{align*}
for large constant $C_6>0$. Finally, using the bounds on $\W_1(s_2,f_1)$ from above, and on $\W_1(s_1,s_2)$ from~\eqref{Eq:s12}, in~\eqref{Eq:I2_sql+}, we get that:
\begin{align*}
    I_2 \leq \frac{2\|\Bb \|_2}{n^{(\beta-1)/\epsilon}} + C_6\epsilon +\frac{C_6 \log (1/\epsilon \delta) \log (1/\epsilon)}{\sqrt{n}} +\frac{C_6 l\sqrt{\log(l/\delta) }\log(1/\epsilon)}{n},
\end{align*}
where we also used the fact that $\frac{n-|S|}{n} \leq 1$. Thus, from~\eqref{Eq:sc}, we get that $\W_1(s_{\Cb},f_{\Bb}) \leq \frac{2\|\Bb \|_2}{n^{(\beta-1)/\epsilon}} + C_6\epsilon +\frac{C_6 \log (1/\epsilon \delta) \log (1/\epsilon)}{\sqrt{n}} +\frac{C_6 l\sqrt{\log(l/\delta) }\log(1/\epsilon)}{n}$. Then, from~\eqref{Eq:sb} and using triangle inequality, we get that:
\begin{align*}
   \W_1(s_{\Bb},f_{\Bb}) &\leq  \W_1(s_{\Bb},s_{\Cb})+\W_1(s_{\Cb},f_{\Bb}) \\
   &\leq \frac{4\|\Bb \|_2}{n^{(\beta-1)/\epsilon}} + C_6\epsilon +\frac{C_6 \log (1/\epsilon \delta) \log (1/\epsilon)}{\sqrt{n}} +\frac{C_6 l\sqrt{\log(l/\delta) }\log(1/\epsilon)}{n}.
\end{align*}
Finally, from~\eqref{Eq:sf}, we get that:
\begin{align*}
     \W_1(s_{\Ab},f) &\leq 2\alpha \left( \frac{4\|\Bb \|_2}{n^{(\beta-1)/\epsilon}} + C_6\epsilon +\frac{C_6 \log (1/\epsilon \delta) \log (1/\epsilon)}{\sqrt{n}} +\frac{C_6 l\sqrt{\log(l/\delta) }\log(1/\epsilon)}{n} \right) \\
     &\leq \frac{4\| \Ab\|_2}{n^{(\beta-1)/\epsilon}}+ 2\alpha \left( C_6\epsilon+ \frac{C_6 \log (1/\epsilon \delta) \log (1/\epsilon)}{\sqrt{n}} +\frac{C_6 l\sqrt{\log(l/\delta) }\log(1/\epsilon)}{n}\right),
\end{align*}
where the second step follows from the facts that $\Bb=\frac{\Ab}{2\alpha}$. Also note that we have $\alpha \leq \sigma_{l+1}(\Ab)+\frac{\|\Ab \|_2g^{c/4}_{\min}}{n^{c/4}}$. Simplifying the above expression we get that for some large constants $C'>0$ and $c_2>0$ such that $\frac{4}{n^{(\beta-1)/\epsilon}} << \frac{1}{n^{c_2}}$, we have:
\begin{align*}
    \W_1(s_{\Ab},f) &\leq C'\epsilon \sigma_{l+1}(\Ab)  +\frac{C'\log (1/\epsilon \delta) \log (1/\epsilon)}{\sqrt{n}}\sigma_{l+1}(\Ab) +\frac{ C'l\sqrt{\log(l/\delta) }\log(1/\epsilon)}{n}\sigma_{l+1}(\Ab) +\frac{\| \Ab\|_2}{n^{c_2}}.
\end{align*}

\medskip

\noindent\textbf{Case 2:} We now consider the case when such a $k$ doesn't exist. In this case, we have $\|\Ab \|_2 \leq 2\alpha$. Then, following case 2 of Theorem~\ref{thm:slq}, let $\Bb=\frac{\Ab}{2\alpha}$ such that $\|\Bb \|_2 \leq 1$. Then, following case 1, we can again apply the backwards error to get a matrix $\Cb$ such that $\|\Bb-\Cb \|_2 \leq \frac{2\|\Bb \|_2}{n^{(\beta-1)/\epsilon}}$. We can again split $ \W_1(s_{\Cb},f_{\Bb})$ into the integrals $I_1$ and $I_2$ as in~\eqref{Eq:sc} such that $I_1=0$ and $I_2$ can be bounded as $I_2 \leq \frac{n-|S|}{n} \big( \W_1(s_1,s_2) +\W_1(s_2,f_1) \big)$ (as in~\eqref{Eq:I2_sql+}) such that $\W_1(s_1,s_2) \leq \frac{2\|\Bb \|_2}{n^{(\beta-1)/\epsilon}}$ (as in~\eqref{Eq:s12}) and $\frac{n-|S|}{n}\W_1(s_2,f_1)$ is bounded using Lemma 3.1 of~\cite{braverman:2022} by the Chebyshev moments of $s_2$ and $f_1$ as in~\eqref{Eq:s2}. Then, using triangle inequality, we can bound the moments of $f_1$ and $s_2$ again as in~\eqref{Eq:var_slq_diff} such that we have:
\begin{align}
 &\bigg|\frac{1}{n}\sum_{j \in [n]\setminus S_1}\Bar{T}_i(\lambda_j(\Bb))- \frac{\frac{n-|S|}{n}}{\sum_{p \in [m]\setminus S}w^2_p}\sum_{j \in [m]\setminus S} w_j^2\Bar{T}_i(\lambda_j(\Tb_1)) \bigg| \nonumber \\ 
 &\leq \bigg|\frac{1}{n}\sum_{j \in [n]}\Bar{T}_i(\lambda_j(\Bb))- \sum_{j \in [m]} w_j^2\Bar{T}_i(\lambda_j(\Tb_1)) \bigg| +\bigg| \frac{1}{n}\sum_{j  \in S_1 }\Bar{T}_i(\lambda_j(\Bb)) - \sum_{j \in S} w_j^2 \Bar{T}_i(\lambda_j(\Tb_1)) \bigg| \nonumber \\
 &+  \bigg| \frac{\frac{n-|S|}{n}}{\sum_{p \in [m]\setminus S}w^2_p}\sum_{j \in [m]\setminus S} w_j^2\Bar{T}_i(\lambda_j(\Tb_1))- \sum_{j \in [m] \setminus S} w_j^2 \Bar{T}_i(\lambda_j(\Tb_1)) \bigg|. 
\end{align}
Since $w_j^2 \leq \frac{C_2\sqrt{\log (l/\delta)}}{n}$ for $j \in S$ and $\lambda_j(\Tb_1) \leq \lambda_j(\Bb) \leq 1$, the second and third terms on the right hand side above can be bounded by $O\left( \frac{l\sqrt{\log (l/\delta)}}{n} \right)$ as in case 1. To bound the first term, observe that $\frac{1}{n}\sum_{j \in [n]}\Bar{T}_i(\lambda_j(\Bb))=\frac{1}{n}\tr(\Bar{T}_i(\Bb))$ and $\sum_{j \in [m]} w_j^2\Bar{T}_i(\lambda_j(\Tb_1))$ is the $i$th Chebyshev moment of the output of SLQ (Algorithm~\ref{alg:slq}) with $\Bb$ as input. Thus, from Lemma~\ref{lem:abs-slq-true-A} we get that $\bigg|\frac{1}{n}\sum_{j \in [n]}\Bar{T}_i(\lambda_j(\Bb))- \sum_{j \in [m]} w_j^2\Bar{T}_i(\lambda_j(\Tb_1)) \bigg| \leq O \left( \frac{\log (N/\delta)}{\sqrt{n}} \right) \leq  O \left( \frac{\log (1/\epsilon\delta)}{\sqrt{n}} \right)$. Thus, we get $\bigg|\frac{1}{n}\sum_{j \in [n]\setminus S_1}\Bar{T}_i(\lambda_j(\Bb))- \frac{\frac{n-|S|}{n}}{\sum_{p \in [m]\setminus S}w^2_p}\sum_{j \in [m]\setminus S} w_j^2\Bar{T}_i(\lambda_j(\Tb_1)) \bigg| \leq O \left( \frac{\log (1/\epsilon\delta)}{\sqrt{n}} +\frac{l\sqrt{\log (l/\delta)}}{n} \right)$.The rest of the proof follows the proof of case 1 which gives us the final bound of $$\W_1(s,f) \leq \epsilon \sigma_{l+1}(\Ab)  +\frac{C'\log (1/\epsilon \delta) \log (1/\epsilon)}{\sqrt{n}}\sigma_{l+1}(\Ab) +\frac{ C'l\sqrt{\log(l/\delta) }\log(1/\epsilon)}{n}\sigma_{l+1}(\Ab) +\frac{\| \Ab\|_2}{n^{c_2}}.$$

Finally, observe that adjusting $\delta$ by some constant factors gives us the final bound.
\end{proof}

\section{Lower Bound}\label{sec:lowerBound}

We now prove the lower bound on the number of matrix vector queries required by any algorithm to estimate the spectral density of any matrix $\Ab$ upto Wasserstein error $\epsilon \sigma_{l+1}(\Ab)$. Our proof proceeds via a reduction of the spectral density estimation problem to the problem of distinguishing between two Wishart matrices with ranks very close to each other (Theorem 17 of~\cite{SwartworthWoodruff:2023}). Our lower bound of $O(l +\frac{1}{\epsilon})$ shows that our upper bounds for estimating the SDE via explicit deflationa s well as SLQ are nearly tight (upto polylog factors).

\begin{reptheorem}{thm:lower_bound}
 Any (possibly randomized) algorithm that given symmetric $\bv A \in \R^{n \times n}$ outputs $\tilde s_{\bv A}$ such that, with probability at least $1/2$, $W_1(s_\bv{A},\tilde s_{\bv A}) \le \epsilon \sigma_{l+1}(\Ab)$ for $\epsilon \in (0,1)$ and $l \in [n]$ must make $\Omega \left(l +\frac{1}{\epsilon} \right)$ (possibly adaptively chosen) matrix-vector product queries to $\Ab$.
 \end{reptheorem}
 \begin{proof}
      Let $\mathcal{A}$ be an adaptive algorithm that estimates the spectral density of $\Ab $ up to error $\epsilon \sigma_{l+1}(\Ab)$ in the Wasserstein-$1$ norm. Let $W(n,r)$ be the $n$ dimensional Wishart distribution with $r$ degrees of freedom i.e.,  the distribution of $\bv{G}\bv{G}^T$ where $\bv{G} \in \R^{n \times r}$ has i.i.d. standard normal entries. We will use Theorem 17 of~\cite{SwartworthWoodruff:2023} which states that at least $\Omega(r)$ (possibly adaptive) matrix vector queries are required by any adaptive algorithm to distinguish between two Wishart matrices $W(n,r)$ and $W(n,r+2)$ with probability at least $\frac{2}{3}$. We prove the lower bound by considering the two cases: $l > \frac{1}{\epsilon}$ and $l \leq \frac{1}{\epsilon}$.

     \medskip
     
  \noindent\textbf{Case 1.} $\left(l  > \frac{1}{\epsilon} \right)$:  The non-zero eigenvalues of the Wishart ensembles $W(n,l)$ and $W(n,l+2)$ are bounded between $\frac{n}{2}$ and $2n$ with probability at least $5/6$ as long as $n \geq Cl$ for some constant $C$ \cite{vershynin2018high}. Let $n=Cl$.  Consider the Wishart ensembles $\Ab_1=W(n,l)$ and $\Ab_2=W(n,l+2)$. Observe that $\Ab_1$ and $\Ab_2$ have ranks of $l$ and $l+2$ respectively. So, $\sigma_{l+3}(\Ab_1)=\sigma_{l+3}(\Ab_2)=0$. Let $\mathcal{A}$ use $k$ matrix-vector products with the input matrix to estimate the spectral density $s_{\bv A}(x)$ up to error $\epsilon \sigma_{l+3}(\Ab)=0$ in both cases with probability at least $0.5$ i.e. the spectral density of the input matrix is estimated exactly. So, the rank of $\Ab$ in both cases is given exactly by $n\int_{n/2}^{2n} s_{\bv A}(x) dx$ with probability at least $0.5$. Hence, $\mathcal{A}$ can distinguish between $W(n,l)$ and $W(n,l+2)$ with probability at least $\frac{2}{3}$. Thus, from Theorem 17 of~\cite{SwartworthWoodruff:2023}, we must have $k = \Omega(l)$.

\medskip

\noindent\textbf{Case 2.} $\left(l \leq \frac{1}{\epsilon} \right)$: In this case, let $r=\lfloor \frac{1}{\epsilon} \rfloor$ and let us consider the normalized Wishart ensembles $\Ab_1=\frac{1}{2n}W(n,r)$ and $\Ab_2=\frac{1}{2n}W(n,r+2)$ where $n=2Cr$. Let either $\Ab_1$ or $\Ab_2$ be the input to $\mathcal{A}$. Since the nonzero eigenvalues of $\Ab_1$ and $\Ab_2$ lie in $[0.25,1]$ with probability at least $5/6$, by setting the error parameter to $\frac{\epsilon}{1000C}$ we can estimate the spectral density of the input to $\mathcal{A}$ up to error $\frac{\epsilon}{1000C}\max(\sigma_{l+1}(\Ab_1),\sigma_{l+1}(\Ab_2)) \leq  \frac{\epsilon}{1000C}$ in the Wasserstein-$1$ distance. Let $\lambda_1 \geq \ldots \geq \lambda_n$ be the true eigenvalues of the input to $\mathcal{A}$. Then, using the estimated spectral density returned by $\mathcal{A}$, we can construct a list of n values $[\tilde{\lambda}_1, \ldots , \tilde{\lambda}_n]$ in time linear in $n$ and $\frac{1}{\epsilon}$ such that $\sum_{i=1}^n |\lambda_i-\tilde{\lambda}_i| \leq \frac{3\epsilon n}{1000C} \leq \frac{6}{1000}$ (see~\cite{Cohen-SteinerKongSohler:2018}, theorem B.1 in~\cite{braverman:2022}). So, we have $|\lambda_{r+1}-\tilde{\lambda}_{r+1}| \leq \sum_{i=1}^n |\lambda_i-\tilde{\lambda}_i| \leq \frac{6}{1000}$. When, $\Ab_1$ is the input we gt $|\tilde{\lambda}_{r+1}| \leq \frac{6}{1000}$ since $\lambda_{r+1}=0$ in this case.  When, $\Ab_3$ is the input, via triangle inequality, we get $|\tilde{\lambda}_{r+1}| \geq |\lambda_{r+1}|- \frac{6}{1000} \geq 0.25-\frac{6}{1000}$ since $\lambda_{r+1}>0.25$ in this case. So, we can distinguish between $\Ab_1$ and $\Ab_2$ with probability at least $\frac{2}{3}$. So, using Theorem 17 of~\cite{SwartworthWoodruff:2023}, we again have $k \geq \Omega(\frac{1}{\epsilon})$.

Finally, observe that setting $\delta=\frac{1}{n^{c'}}$ for some constant $c'$ gives us the final bound.
\end{proof}

\section{Empirical Evaluation}\label{sec:exp}

In this section, we compare the empirical performance of the algorithms studied in this paper and several other standard algorithms in approximating the spectral densities of several synthetic and publicly available matrices. We observe that the SLQ algorithm (Algorithm \ref{alg:slq}) and its variance reduced variant, VR-SLQ (Algorithm \ref{alg:slq+}) generally out perform explicit moment matching methods like KPM and Chebyshev moment matching in the Wasserstein distance metric. This finding aligns with our results in Theorem \ref{thm:slq} and \ref{thm:varslq} that these methods perform implicit deflation for {any deflation parameter $\ell$} allowing them to adapt to the matrix spectrum to achieve stronger error bounds.

\subsection{Datasets}

Our comparisons are performed on the following matrices:

\begin{itemize}
    \item The \emph{Gaussian matrix}, which is a $5000\times 5000$ matrix, constructed by first drawing $5000$ eigenvalues from a Gaussian distribution to get $\bv{\Lambda} \in \mathcal{N}(0,1)^{5000}$, normalizing $\bv{\Lambda} = \bv{\Lambda} / \norm{\bv\Lambda}_{\infty}$, then generating a random orthonoromal matrix $\Vb\in \R^{5000\times 5000}$, and finally computing $\Ab = \Vb\mathbf{\Lambda}\Vb^T$. This matrix is generated from the descriptions of the Gaussian matrix in \cite{braverman:2022}.

    \item The \emph{Uniform matrix} is constructed analogously to the Gaussian matrix, except its eigenvalues are drawn uniformly and independently from $[-1,1]$. This matrix is generated from the descriptions of the Uniform matrix in \cite{braverman:2022}.

    \item The \emph{inverse spectrum matrix} is a $5000\times 5000$ diagonal matrix with entries $1, 1/2, \ldots, 1/5000$.

    \item The \emph{power law spectrum matrix}  is a $5000\times 5000$ diagonal matrix with entries $1, 1/2^2, \ldots, 1/2^{5000}$.

    \item The \emph{low-rank matrix} is a $5000\times 5000$ diagonal matrix with $100$ entries drawn uniformly and independently from a Gaussian distribution $\mathcal{N}(0,1)$, and normalized as $\bv{\Lambda} = \bv{\Lambda} / \norm{\bv\Lambda}_{\infty}$. The remaining diagonal entries are set at $0$.

    \item Finally, \emph{Erdos992} is the adjacency matrix of the Erd\"os collaboration network (snapshot of 1992) with $6100$ vertices and containing $15030$ undirected edges. It is taken from a publicly available sparse matrix collection \cite{davis2011university}. The adjacency matrix has one eigenvalue at $1$, roughly $1000$ eigenvalues at $0$, and roughly $500$ eigenvalues of magnitude greater than $0.2$.
\end{itemize}

We note that all algorithms considered in this paper are `rotationally invariant' in that their performance should not depend on the actual eigenvector basis of $\bv A$. For this reason, most of our test matrices are simply diagonal matrices. And we indeed see no systematic difference in performance depending on the eigenvector basis.

\subsection{Implementation Details}

We consider three baseline SDE algorithms: 1) stochastic Lanczos quadrature (SLQ), 2) the Chebyshev moment matching (CMM) algorithm (Algorithm 1 of \cite{braverman:2022}), and 3) the Jackson damped kernel polynomial method (KPM) (Algorithm 6 of \cite{braverman:2022}) which is a popular moment matching algorithm that can be thought of as an approximation  of CMM. We implement two versions of our explicit deflation algorithm (Algorithm \ref{alg:sde}) -- one of which uses CMM after deflation (the one we analyze), and one that uses KPM. We call these algorithms {def-CMM} and {def-KPM} respectively. Along with these algorithms, we also compare the performance of SLQ algorithm (Algorithm \ref{alg:slq}) and VR-SLQ (Algorithm \ref{alg:slq+}).

Since we test for relatively small $n$, as the number of iterations/moments performed by each algorithm increases, we may reach small $\epsilon = \tilde o(1/\sqrt{n})$ values for which more than one random vectors in Hutchinson's method in the explicit moment matching methods or more than one independent trial of SLQ are needed. For simplicity, for all moment matching  methods we use 15 random vectors for Hutchinson's method. For SLQ-based methods we perform $15$ independent trials of the method (using $15$ independent random starting vectors) and average together the output densities to obtain our final spectral density estimate. For block Krylov based deflation we perform $15$ iterations to generate the Krylov subspace.

We compare each algorithm based on the Wasserstein-$1$ error achieved for a fixed number of total matrix-vector queries to the input matrix. Since the iterations in Krylov and number of vectors in Hutchinson's algorithm are fixed at 15, to ensure that def-CMM or def-KPM uses the same number of matrix-vector queries as other algorithms we then need to split the moment budget of CMM or KPM to accomodate for the matrix-vector queries due to block Krylov. We split this in the ratio $1:3$, i.e., for every $2$ moment computations of CMM or KPM in def-CMM or def-KPM, we use a block size of $3$ to compute deflation via block Krylov algorithm. Note, for block size of $3$ and $\ell$ iterations, block Krylov method uses $6\ell$ matrix-vector queries. We vary the total matrix-vector query budget for all algorithms, i.e., increase the total moments approximated by CMM-based algorithms, or increase the total number of iterations in SLQ-variants, and report the Wasserstein-$1$ error in Figure \ref{fig:main plot} across all matrices.

To compute $N$ moment estimates, the CMM algorithm \cite{braverman:2022} is evaluated on a discrete and evenly-spaced grid of length $d+1$ in the interval $[-1,1]$. Theoretically, the algorithm requires setting $d=\lceil N^3/3\rceil$. In our experiments, we found setting $d=20000$ to be sufficient. To solve the $\ell_1$-regression problem in the CMM algorithm, we use HiGHS solvers \cite{huangfu2018parallelizing} within SciPy \cite{2020SciPy-NMeth}. To ensure that the solution vector falls within the probability simplex we use Algorithm 1 of \cite{condat2016fast}.
The eigenvalues of any matrix are computed using numpy \cite{numpyeigvals}. 

We repeat each experiment for $t=10$ independent trials. In the plots in Figure \ref{fig:main plot}, the $x$-axis denotes the total matrix-vector queries used by each algorithm per trial, and the $y$-axis represents the corresponding Wasserstein-$1$ SDE error. The bold lines in the plot represent the mean error across $10$ trials. The $10$\tth and the $90$\tth percentile of the observed errors are represented by the faded envelope around the bold lines. 

\paragraph{Code.} All codes are written in Python and available at \url{https://github.com/archanray/SDE_SLQ}.

\subsection{Summary of Results}

We observe that across all of the matrices, SLQ (Algorithm \ref{alg:slq}) and VR-SLQ (Algorithm \ref{alg:slq+}) outperform the explicit moment matching-based algorithms that we test. Among these two algorithms, VR-SLQ more often outperforms SLQ, especially when the spectrum of the input matrix contains only a few large eigenvalues, as in this case, the variance reduction step can have a significant positive effect on the spectral density estimate.
We also observe that the variants of Algorithm \ref{alg:sde} (def-CMM and def-KPM) more often outperform naive CMM and KPM, in particular for matrices with only a few large eigenvalues, as expected.

\begin{figure}[H]
    \centering
    \begin{subfigure}[t]{\textwidth}
        \centering
        \includegraphics[width=0.32\textwidth]{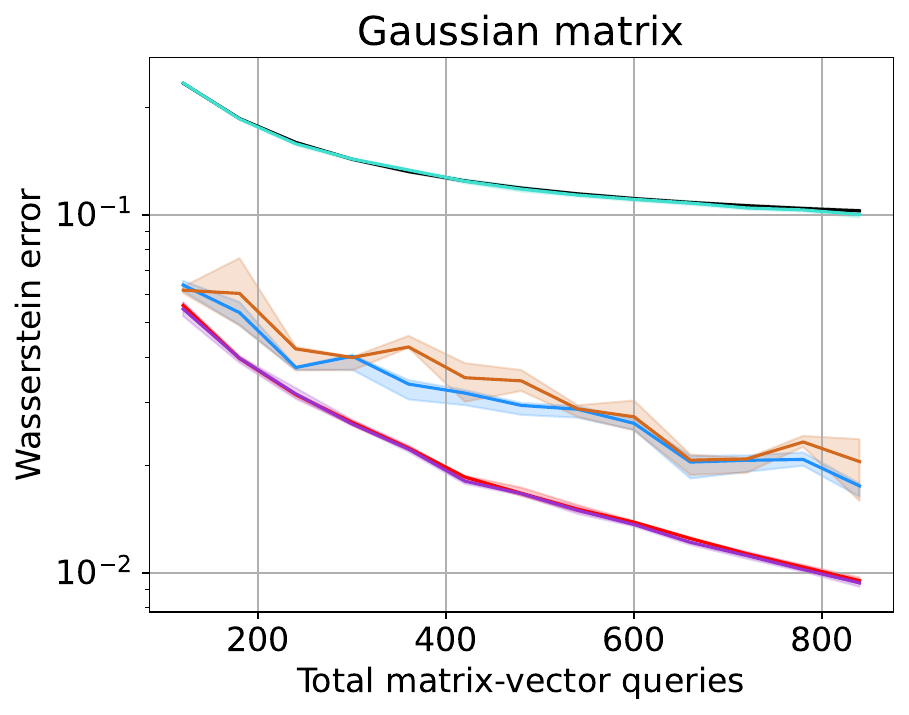}
        \includegraphics[width=0.32\textwidth]{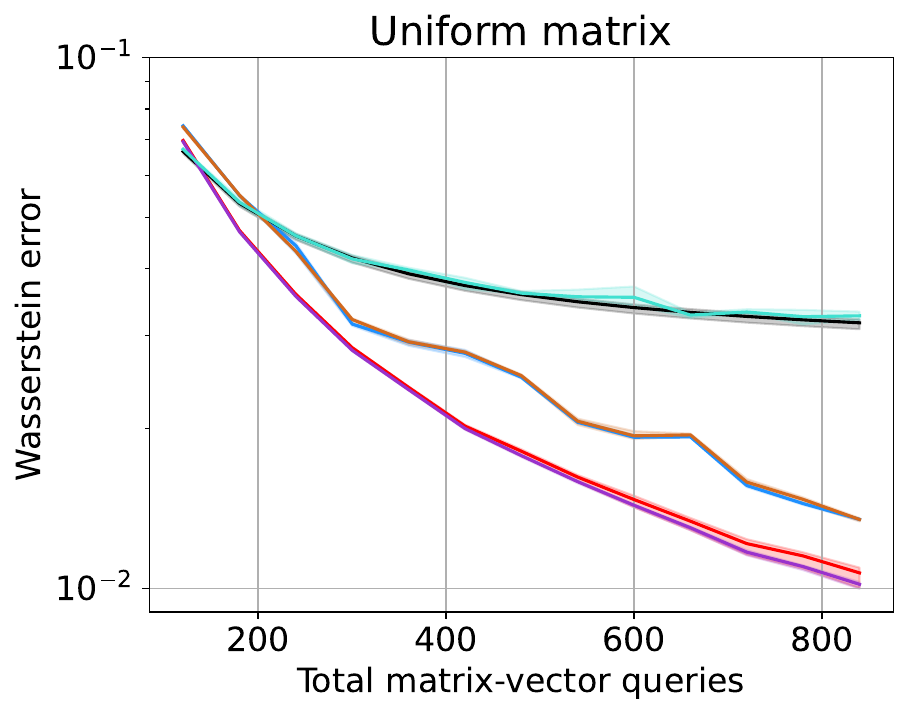}
        \includegraphics[width=0.32\textwidth]{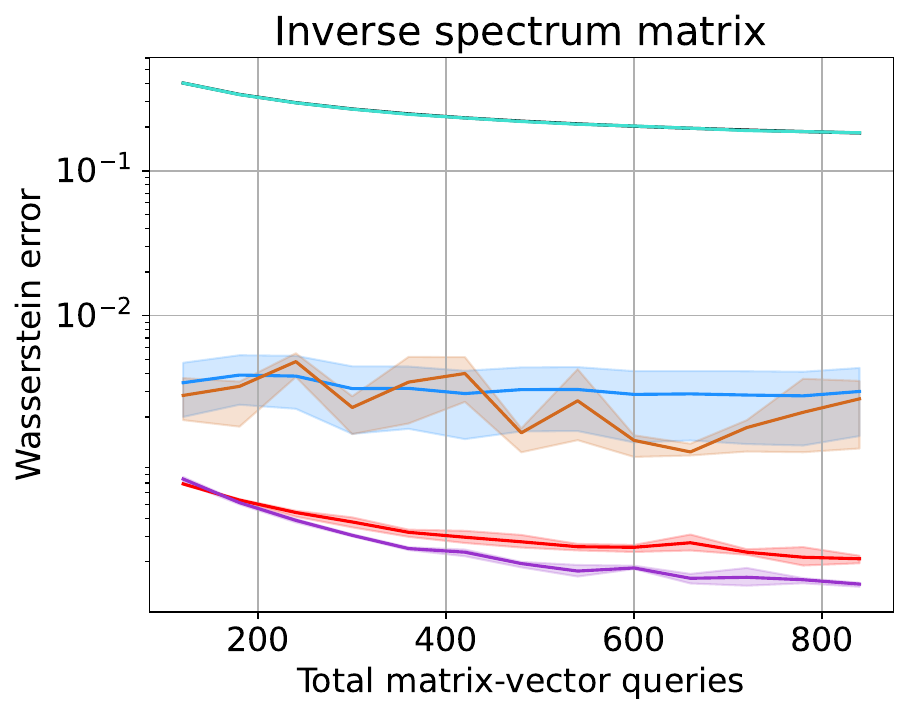}
    \end{subfigure}
    \begin{subfigure}[t]{\textwidth}
        \centering
        \includegraphics[width=0.32\textwidth]{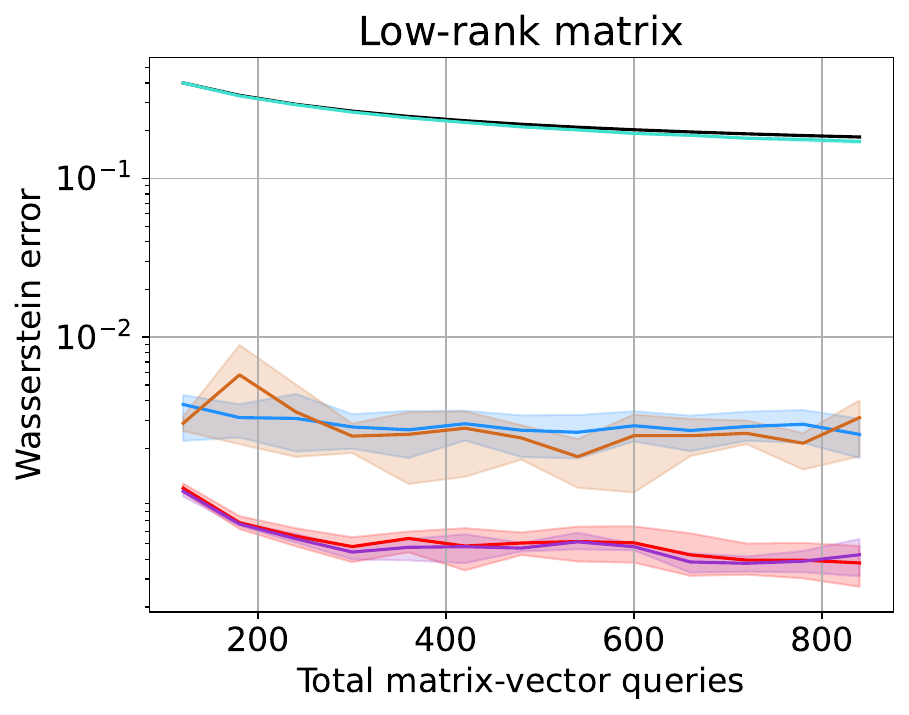}
        \includegraphics[width=0.32\textwidth]{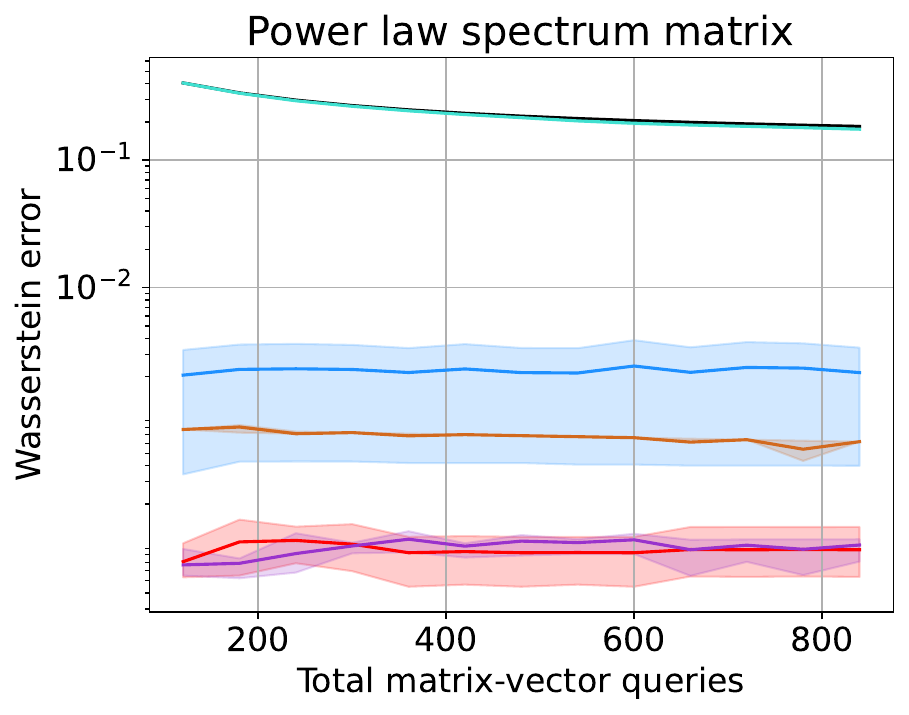}
        \includegraphics[width=0.32\textwidth]{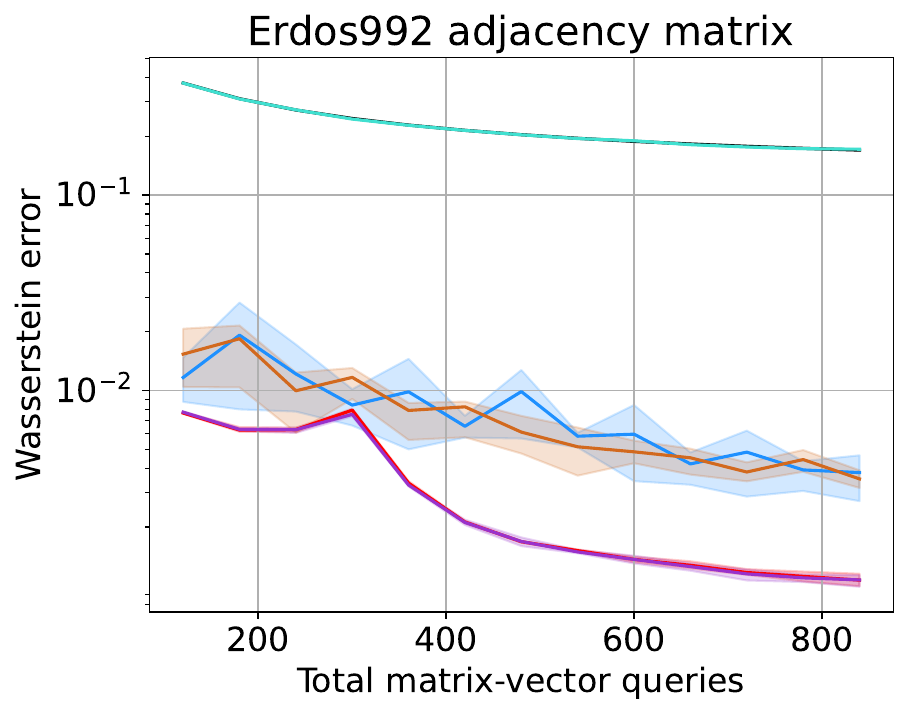}
    \end{subfigure}
    \begin{subfigure}[t]{\textwidth}
        \centering
        \includegraphics[width=0.9\textwidth]{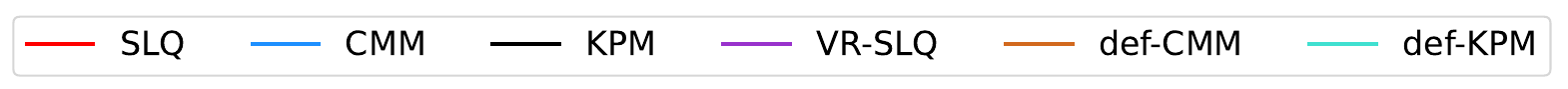}
    \end{subfigure}
    \caption{\textbf{Wasserstein-$1$ error of spectral density estimation approximation algorithms}. In the figures above, we plot the Wasserstein-$1$ error of approximating the spectral density of several matrices using the algorithms presented in the paper and some  baseline algorithms. We observe that in almost all cases, VR-SLQ algorithm outperforms all other SDE algorithms. We also observe that the variants of our deflation algorithm (Algorithm \ref{alg:sde}) generally outperform the corresponding baseline of CMM and KPM.}
\label{fig:main plot}
\end{figure}

\section*{Acknowledgements}

Cameron, Rajarshi, and Archan were partially supported by NSF grants 1934846, 2046235, and 2427363. Part of the work was completed while Archan was at the University of Massachusetts Amherst, NY. Christopher was supported by NSF grant 2045590.
\bibliography{refs}

\end{document}